\documentclass[12pt]{article}
\usepackage{geometry}
\usepackage[round, authoryear,sectionbib]{natbib}
\usepackage[utf8]{inputenc}
\usepackage{lineno}
\usepackage{titlesec}
\usepackage{comment}
\usepackage{tabularx,multirow}
\usepackage{graphicx,xcolor}
\usepackage{amsthm,amsmath, amssymb}
\usepackage{mathtools}
\usepackage{textcomp}
\usepackage{changepage}
\usepackage{tabularx}
\usepackage{enumitem}
\usepackage{hyperref}
\usepackage{xurl}
\hypersetup{breaklinks=true}
\usepackage{listings}
\usepackage{zref-totpages}
\usepackage[nottoc,numbib]{tocbibind}
\usepackage{xcolor}
\usepackage[all]{nowidow}
\usepackage{booktabs}
\usepackage[T1]{fontenc}
\usepackage{tikz}
\usepackage{pgfplots}
\pgfplotsset{compat=1.18}
\usetikzlibrary {decorations.pathreplacing,
                 arrows.meta,
                 calc,
                 plotmarks,
                 shapes,
                 datavisualization,
                 backgrounds,
                 fit,
                 positioning,
                 patterns,}
\theoremstyle{definition}
\newtheorem{Def}{Definition}

\theoremstyle{remark}
\newtheorem{Ex}{Example}
\newtheorem{Rmk}{Remark}
\newtheorem{Not}{Notation}

\theoremstyle{plain}
\newtheorem{Cor}{Corollary}
\newtheorem{Lem}{Lemma}
\newtheorem{Thm}{Theorem}
\newtheorem{Task}{Task}

\numberwithin{equation}{section}
\numberwithin{Def}{section}
\numberwithin{Alg}{section}
\numberwithin{Ex}{section}
\numberwithin{Rmk}{section}
\numberwithin{Not}{section}
\numberwithin{Cor}{section}
\numberwithin{Lem}{section}
\numberwithin{Thm}{section}
\numberwithin{Task}{section}
\numberwithin{Obs}{section}

\definecolor{codegreen}{rgb}{0,0.6,0}
\definecolor{codegray}{rgb}{0.5,0.5,0.5}
\definecolor{codepurple}{rgb}{0.58,0,0.82}
\definecolor{backcolour}{rgb}{0.95,0.95,0.95}

\lstdefinestyle{mystyle}{
    backgroundcolor=\color{backcolour},   
    commentstyle=\color{codegreen},
    keywordstyle=\color{magenta},
    numberstyle=\tiny\color{codegray},
    stringstyle=\color{codepurple},
    basicstyle=\ttfamily\footnotesize,
    breakatwhitespace=false,         
    breaklines=true,                 
    captionpos=b,                    
    keepspaces=true,                 
    numbers=left,                    
    numbersep=5pt,                  
    showspaces=false,                
    showstringspaces=false,
    showtabs=false,                  
    tabsize=2
}

\tikzset{
    individual/.style={fill, circle, inner sep = 0.8mm, anchor = center},
    inheritance/.style={->, >=Stealth, thick},
    survival/.style={->, >=Stealth, thick, dashed},
    death/.style={->, >=Rays, thick, dashed},
    generation description/.style={below = 5mm of #1,
                                   anchor = east,
                                   font = \footnotesize},
    relationship/.style={decorate,
                         decoration = {brace,
                                       mirror,
                                       raise = 2mm,
                                       amplitude = 2mm}},
    relationship description/.style={midway, below, yshift = -4mm},
    gene pass description/.style={midway, above, sloped, font = \footnotesize\sffamily},
    queen/.style={fill, circle, inner sep = 0.8mm, anchor = center},
    replacement queen/.style={draw, circle, thick, inner sep = 0.8mm, anchor = center},
    drone/.style={draw, inner sep = 0.8mm, anchor = center},
    group/.style={draw, circle, inner sep = 2.7mm, anchor = center},
    group base/.style={fill = white, circle, inner sep = 3mm, anchor = center},
    worker group/.style={draw, pattern = crosshatch dots, circle, inner sep = 2.7mm, anchor = center},
    drones/.pic={
        \path (0,0) coordinate (c)
            node[above right = 1.4mm and 1.4mm of c.center, drone] {}
            node[above left = 1.4mm and 1.4mm of c.center, drone] {}
            node[below right = 1.4mm and 1.4mm of c.center, drone] {}
            node[below left = 1.4mm and 1.4mm of c.center, drone] {}
        ;
    },
    queens/.pic={
        \path (0,0) coordinate (c)
            node[above right = 1.4mm and 1.4mm of c.center, queen] {}
            node[above left = 1.4mm and 1.4mm of c.center, queen] {}
            node[below right = 1.4mm and 1.4mm of c.center, queen] {}
            node[below left = 1.4mm and 1.4mm of c.center, queen] {}
        ;
    },
    mating/.style={->, >=Stealth, thick, dashed},
    not considered/.style={black!40},
}

\title{Optimum Contribution Selection for Honeybees}

\author{Manuel Du\textsuperscript{1}*, Richard Bernstein\textsuperscript{1}, Andreas Hoppe\textsuperscript{1}}

\date{}

\begin{document}

\maketitle

\noindent{} \textsuperscript1 Institute for Bee Research Hohen Neuendorf, Friedrich-Engels-Str.~32, 16540~Hohen~Neuendorf, Germany

\noindent{} * manuel.du@hu-berlin.de
\bigskip

\bigskip\bigskip


\bigskip

\newpage{}

\section*{On this manuscript}

In 1997, T. H. E. Meuwissen published a groundbreaking article titled '\emph{Maximizing the response of selection with a predefined rate of inbreeding}' \citep{meuwissen97}, in which he provided an optimized solution for the trade-off between genetic response and inbreeding avoidance in animal breeding. Evidently, this issue is highly relevant for the honeybee with its small breeding population sizes. However, the genetic peculiarities of bees have thus far prevented an application of the theory to this species. The present manuscript intends to fill this desideratum. It develops the necessary bee-specific theory and introduces a small R script that implements Optimum Contribution Selection (OCS) for honeybees. 

While researching for this manuscript, we found it rather cumbersome that even though Meuwissen's theory is 28 years old and has sparked research in many new directions, to our knowledge, there is still no comprehensive textbook on the topic. Instead, all relevant information had to be extracted from several articles, leading to a steep learning curve. We anticipate that many honeybee breeding scientists with a putative interest in OCS for honeybees have little to no experience with classical OCS. Thus, we decided to embed our new derivations into a general introduction to OCS that then specializes more and more to the honeybee case. The result are these {\ztotpages} pages, of which we hope that at least the first sections can also be of use for breeding theorists concerned with other species than honeybees.\\[1cm]
Hohen Neuendorf, April 2025\\
\emph{Manuel Du}

\paragraph{Acknowledgments}
This work was supported by the Deutsche Forschungsgemeinschaft (DFG, German Research Foundation) – (Grant no. 462225818 to M.\,D.).

\newpage

\tableofcontents

\newpage{}

\section{Introduction}
Breeding endeavors in both animals and plants follow one fundamental idea: If a trait of interest is at least partly influenced by genetics, then it can be improved over generations by selecting only the best individuals for reproduction \citep{lush37}.
For most economically interesting traits, the role of genetics in the determination of phenotypes is only a partial one, because phenotypes are also influenced by environmental effects and further residual effects. The Animal Model, and, in combination with it, breeding value estimation by BLUP (best linear unbiased prediction), was introduced by \citet{henderson75best} and allows for a separation of these effects. Thereby, breeders are enabled to select the \emph{genetically} best rather than just the \emph{phenotypically} best individuals for reproduction -- at least with some statistical accuracy.

However, a genetic foundation of a trait will also often lead to the situation that the best individuals are closely related because related individuals share large portions of their genetic information. If only closely related individuals reproduce, inbreeding effects will occur, articulating themselves in a loss of genetic variance and other forms of inbreeding depression \citep{bienefeld89, gutierrezreinoso22}. Sustainable breeding schemes therefore have to create good genetic revenue while at the same time keeping inbreeding at an acceptable level. On a theoretical basis, an (at least in certain aspects) \emph{optimal} solution to this problem was presented by \citet{meuwissen97}. Meuwissen's fundamental idea was to maximize the genetic revenue in each generation under the restriction that the overall increase in inbreeding per generation is limited. This strategy, named \emph{Optimum Contribution Selection} (OCS), has soon gained attention by numerous scientists in plant and animal breeding and many articles on this idea and variations thereof have been published \citep{henryon15, wang17, wellmann19optimum}.

The honeybee, however, is an organism for which there is no theory of OCS thus far. Honeybees come with a number of genetic and biological peculiarities that prevent a straightforward application of Meuwissen's theory to this species. These peculiarities include the fact that honeybees are haplo-diploid, express phenotypes as colonies rather than individuals, and that honeybee queens mate only once in their lives -- with multiple drones.

This text aims to transfer Meuwissen's theory of OCS to the honeybee. To do so, it will proceed in several steps. Firstly, Section~\ref{sec::dipl} will review traditional OCS for diploids, at least in those aspects that will become important later on. Secondly, Section~\ref{sec::hbpec} will recapitulate the peculiarities of the honeybee that prevent a direct application of OCS and the following Section~\ref{sec::ocshb} will develop an adequate theory of OCS for this species. Section~\ref{sec::solvetask} will show, how  solvers for honeybee specific OCS can be implemented and, finally, Section~\ref{sec::demo} will demonstrate our implementation using two examples. Even though parts of the presentation in Sections~\ref{sec::dipl} and~\ref{sec::hbpec} may be non-standard and possibly new, these two sections are generally written in the spirit of a review. The later Sections~\ref{sec::ocshb} to~\ref{sec::demo} present original research.

\section{Optimum Contribution Selection in diploids}\label{sec::dipl}
We start by recapitulating how OCS works in diploid species. At times, we may slightly deviate from the usual presentation of the topic but use equivalent formulations that will facilitate the transfer to the honeybee case discussed in Section~\ref{sec::ocshb}.
We will discuss the topic in increasing complexity, meaning that we start with the case of discrete generations and a monoecious population (Section~\ref{sec::monoec}), then pass on to diecious populations (Section~\ref{sec::diec}), and finally treat overlapping generations (Section~\ref{sec::overl}).

\subsection{Discrete generations}

\begin{Not}
    \begin{enumerate}[label = (\roman*)]
        \item For a finite set $\mathcal A$ we denote the number of its elements by $|\mathcal A|$. The empty set is denoted by $\varnothing$, i.\,e. $|\varnothing|=0$.
        \item For two sets $\mathcal A$ and $\mathcal B$, we denote their union by $\mathcal A\cup\mathcal B$, their intersection by $\mathcal A\cap\mathcal B$, and their difference by $\mathcal A\backslash\mathcal B$. 
        \item If $\mathcal A$ and $\mathcal B$ are disjoint (i.\,e. $\mathcal A\cap\mathcal B=\varnothing$), we may write $\mathcal A\sqcup\mathcal B$ instead of $\mathcal A\cup\mathcal B$ in order to emphasize this property.
        \item For the union, intersection, or disjoint union of multiple sets $\mathcal A_i$ with $i$ from an index set $I$, we write $\bigcup\limits_{i\in I}\mathcal A_i$, $\bigcap\limits_{i\in I}\mathcal A_i$, and $\bigsqcup\limits_{i\in I}\mathcal A_i$, respectively.
    \end{enumerate}
\end{Not}

Throughout this section, we assume a population $\mathcal P=\bigsqcup\limits_{t\in\mathbb N}\mathcal P_t$ that spans over several disjoint generations. Each generation $\mathcal P_t$ ($t\in\mathbb N$) is a finite set of individuals, which produce the next generation $\mathcal P_{t+1}$ in such a way that each individual $I\in\mathcal P_{t+1}$ has two parents from $\mathcal P_t$.

\begin{Not}
    We denote the number of individuals in generation $\mathcal P_t$ by
        \[N_t:=\left|\mathcal P_t\right|.\]
\end{Not}

We further assume that each individual is equipped with diploid genetics according to the additive Infinitesimal Model of \citet{fisher18}, meaning that each individual receives half of its genetic information from each of its parents and that the expected genetic value of an offspring is the average genetic value of its parents.

We further assume that the genetic values of individuals are accessible in the sense that an unbiased estimate of the breeding value of each individual $I\in\mathcal P_t$ exists. In practice, such estimated breeding values are usually derived by a BLUP procedure.

\begin{Not}
    We denote the estimated breeding value of an individual $I\in\mathcal P_t$ by $\hat u_I$.
\end{Not}

\begin{Def}\label{def::grppp}
    For a group $\mathcal I\subseteq\mathcal P_t$ of individuals, we define the estimated breeding value of $\mathcal I$ to be the average of the individuals' breeding values:
        \[\hat u_{\mathcal I} := \frac1{|\mathcal I|}\sum_{I\in\mathcal I}\hat u_I.\]
\end{Def}

\begin{Rmk}\label{rmk::nobars}
    \begin{enumerate}[label = (\roman*)]
        \item In Definition~\ref{def::grppp}, we speak of a \emph{group} $\mathcal I$ of individuals. This is not to be understood in the algebraic sense of the term \emph{group}. Mathematically precise would be to call $\mathcal I$ a \emph{set}. However, in everyday language, it is more common to speak of a group of animals rather than a set of animals which is why this expression was chosen.

        \item Following Definition~\ref{def::grppp}, the average estimated breeding value of the population in generation $\mathcal P_t$ is denoted $\hat u_{\mathcal P_t}$.

        \item \label{item::nobars} In general, it may help the reader to note that when (lower) indices appear in calligraphic font ($\mathcal I$, $\mathcal P$, etc.), it usually means that averages are taken.
    \end{enumerate}
\end{Rmk}

We may even take this one step further:

\begin{Def}\label{def::setset}
    For a finite set $\mathfrak I=\{\mathcal I_1,...,\mathcal I_{|\mathfrak I|}\}$ of groups of individuals, we define the estimated breeding value of $\mathfrak I$ to be the average of the breeding values of the groups in $\mathfrak I$:
        \[\hat u_{\mathfrak I} := \frac1{|\mathfrak I|}\sum_{\mathcal I\in\mathfrak I}\hat u_{\mathcal I}.\]
\end{Def}

\begin{Rmk}
    Note that in this average, all groups $\mathcal I\in\mathfrak I$ obtain equal weight, independent of their size. Thus, if we write $\tilde{\mathcal I}:=\bigcup\limits_{\mathcal I\in\mathfrak I}\mathcal I$, we in general have
        \[\hat u_{\mathfrak I}\neq\hat u_{\tilde{\mathcal I}}.\]
\end{Rmk}

\subsubsection{Monoecious populations} \label{sec::monoec}

We first discuss a monoecious population, meaning that any two individuals from a generation may produce offspring together and different sexes do not play a role. This is the case in several plant species \citep{lewis42}. We further assume that \emph{selfing} is allowed, meaning that an individual can also produce offspring with itself or, differently put, that the two parents of an individual may be identical.

\begin{Rmk}
    Although this setting is theoretically easier than the diecious case with distinguished male and female individuals, it was not treated in the original derivation of OCS by \citet{meuwissen97}. Probably, the first treatment of this case is that by \citet{kerr98}.
\end{Rmk}

We are in the situation that the $N_t$ individuals of generation $\mathcal P_t$ are to pass on their genes to the next generation $\mathcal P_{t+1}$. In general, the contributions of different individuals $I\in\mathcal P_t$ to the next generation $\mathcal P_{t+1}$ will not be equal; individuals with a greater number of offspring have a higher genetic contribution to the next generation.

\begin{Not}\label{not::corig}
    The fraction of the gene pool of generation $\mathcal P_{t+1}$ that was passed on from individual $I\in\mathcal P_t$ is denoted by $c_{I}\in [0,1]$.
\end{Not}

\begin{Ex}
    Assume a population $\mathcal P$ with $\mathcal P_1=\left\{A_1,B_1,C_1\right\}$ and $\mathcal P_2=\left\{A_2,B_2,C_2\right\}$. Individuals $A_1$ and $B_1\in\mathcal P_1$ are the parents of $A_2\in\mathcal P_2$, $A_1$ and $C_1\in\mathcal P_1$ are the parents of $B_2\in\mathcal P_2$, and $C_2\in\mathcal P_2$ is the result of a selfing of $C_1\in\mathcal P_1$.
    
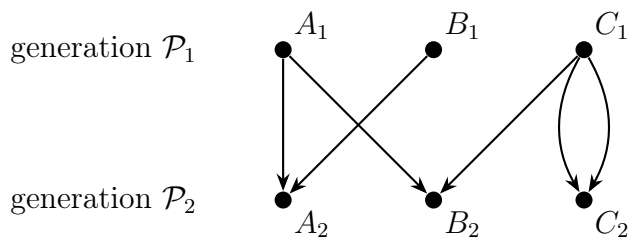
\begin{center}
        \begin{tikzpicture}
            \path (0,2) node[anchor = east] (P1) {generation $\mathcal P_1$}
                  (0,0) node[anchor = east] (P2) {generation $\mathcal P_2$};
            \path (1,2) node[individual] (A1) {}
                            node[above right] {$A_1$}
                  (3,2) node[individual] (B1) {}
                            node[above right] {$B_1$}
                  (5,2) node[individual] (C1) {}
                            node[above right] {$C_1$}
                  (1,0) node[individual] (A2) {}
                            node[below right] {$A_2$}
                  (3,0) node[individual] (B2) {}
                            node[below right] {$B_2$}
                  (5,0) node[individual] (C2) {}
                            node[below right] {$C_2$};
            \draw[inheritance] (A1) -- (A2);
            \draw[inheritance] (A1) -- (B2);
            \draw[inheritance] (B1) -- (A2);
            \draw[inheritance] (C1) -- (B2);
            \draw[inheritance] (C1) to[out=300,in=60] (C2);
            \draw[inheritance] (C1) to[out=240,in=120] (C2);
        \end{tikzpicture}
    \end{center}
    If we fix a locus, there are a total of six ($=3\cdot2$) alleles present in generation $\mathcal P_2$. Of these, two are inherited from $A_1\in\mathcal P_1$, one is inherited from $B_1\in\mathcal P_1$ and three are inherited from $C_1\in\mathcal P_1$. Consequently, we have
        \[c_{A_1}=\frac26=\frac13,\quad
          c_{B_1}=\frac16,\quad
          \text{and}\quad
          c_{C_1}=\frac36=\frac12.\]
\end{Ex}

\begin{Rmk}
    \begin{enumerate}[label = (\roman*)]
        \item We always assume a closed population, which means that all genetic information of generation $\mathcal P_{t+1}$ comes from the individuals in generation $\mathcal P_t$. Consequently, their contributions need to add up to unity:
        \begin{equation}\label{eq::csum}
            \sum_{I\in\mathcal P_t}c_{I}=1.
        \end{equation}

        \item We assume additive genetics, which means that the relative genetic contribution $c_{I}$ of individual $I\in\mathcal P_t$ to the next generation equals its relative contribution to the next generation's average breeding value. While the value of the genetics passed on from individual $I\in\mathcal P_t$ may deviate from its own breeding value due to Mendelian sampling, its expectation is precisely $\hat u_{I}$, meaning that
        \begin{equation}\label{eq::eupt}
            \mathbb E\left[\hat{u}_{\mathcal P_{t+1}}\right]
            =\sum_{I\in\mathcal P_t}c_{I}\hat{u}_{I}.
        \end{equation}
    \end{enumerate}
\end{Rmk}

The goal of a breeding program is to maximize the expectation of $\hat{u}_{\mathcal P_{t+1}}$ given the estimated breeding values $\hat{u}_I$ of individuals $I\in\mathcal P_t$.

\begin{Rmk}\label{rmk::noconstr}
    Without any further constraints, the way to maximize $\mathbb E\left[\hat{u}_{\mathcal P_{t+1}}\right]$ is to choose the individual $I^{\ast}\in\mathcal P_t$ with the maximum estimated breeding value
        \[\hat{u}_{I^{\ast}}
        =\max\left\{\hat{u}_{I}:I\in\mathcal{P}_t\right\},\]
    and let $I^{\ast}$ have all the offspring via selfing. This would mean that
        \[c_{I}=\begin{cases}1,&\text{if }I=I^{\ast}\\
                             0,&\text{otherwise}\end{cases},\]
    and we would have
        \[\mathbb E\left[\hat{u}_{\mathcal P_{t+1}}\right]=\hat{u}_{I^{\ast}}.\]
\end{Rmk}

However, this strategy would not only maximize $\mathcal P_{t+1}$'s average estimated breeding value but also its inbreeding coefficients. Inbreeding is a measure for the relatedness of parents. Identity of parents, as in selfing, can be seen as the closest form of relationship. In the long run, high inbreeding in a breeding program will lead to depletion of genetic variance and possibly further depression effects. Therefore, while a certain degree of inbreeding is unavoidable in a finite closed population, excessively high inbreeding rates are generally to be avoided.

\begin{Def}\label{def::inbr}
    The inbreeding coefficient $f_{I}$ of an individual $I\in\mathcal P_t$ is defined as the probability for the two alleles at a random locus of $I$ to be identical by descent (ibd).
\end{Def}

While this definition works on the scale of individuals, we are interested in a global measure of how inbred an entire population is. At first glance, a strategy to assess the inbreeding of a whole population may be to simply calculate the average inbreeding coefficient
    \[f_{\mathcal P_t}:=\frac1{N_t}\sum_{I\in\mathcal P_t}f_{I}\]
of its individuals. However, it turns out that this measure is not stably transported over generations. This is illustrated by the following example:

\begin{Ex}\label{ex::avginbr}
    Assume a very small population with only two individuals per generation. Generation $\mathcal P_t$ consists of two fully inbred (i.\,e. homozygous by descent at all loci) but unrelated individuals $A_t$ and $B_t$. Generation $\mathcal P_{t+1}$ consists of two common children $A_{t+1}$ and $B_{t+1}$ of $A_t$ and $B_t$. Finally, generation $\mathcal P_{t+2}$ consists of two common children $A_{t+2}$ and $B_{t+2}$ of $A_{t+1}$ and $B_{t+1}$.
    
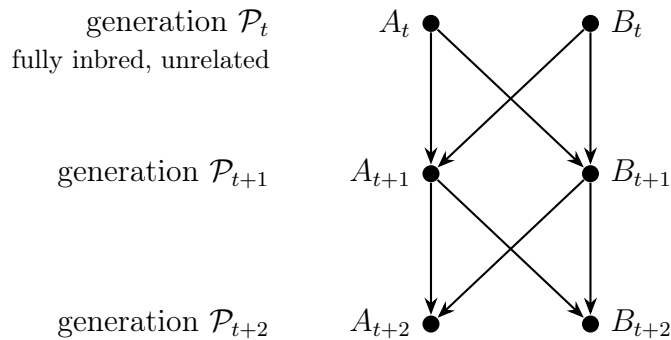
\begin{center}
        \begin{tikzpicture}[every node/.style={draw=none},node distance=2cm and 2cm]
            \path (0,4) coordinate (Pt0)
                        node[left = 0mm of Pt0]
                            {generation $\mathcal P_t$}
                        node[generation description = Pt0]
                                 {fully inbred, unrelated}
                        coordinate[below = of Pt0] (Pt1)
                        node[left = 0mm of Pt1]
                            {generation $\mathcal P_{t+1}$}
                        coordinate[below = of Pt1] (Pt2)
                        node[left = 0mm of Pt2]
                            {generation $\mathcal P_{t+2}$}
                        node[right = of Pt0, individual] (At0) {}
                        node[left = 0mm of At0] {$A_{t}$}
                        node[right = of At0, individual] (Bt0) {}
                        node[right = 0mm of Bt0]
                            {$B_{t}$}
                        node[right = of Pt1, individual] (At1) {}
                        node[left = 0mm of At1]
                            {$A_{t+1}$}
                        node[right = of At1, individual] (Bt1) {}
                        node[right = 0mm of Bt1]
                            {$B_{t+1}$}
                        node[right = of Pt2, individual] (At2) {}
                        node[left = 0mm of At2]
                            {$A_{t+2}$}
                        node[right = of At2, individual] (Bt2) {}
                        node[right = 0mm of Bt2]
                            {$B_{t+2}$};
            \draw[inheritance] (At0) -- (At1);
            \draw[inheritance] (At0) -- (Bt1);
            \draw[inheritance] (Bt0) -- (At1);
            \draw[inheritance] (Bt0) -- (Bt1);
            \draw[inheritance] (At1) -- (At2);
            \draw[inheritance] (At1) -- (Bt2);
            \draw[inheritance] (Bt1) -- (At2);
            \draw[inheritance] (Bt1) -- (Bt2);
        \end{tikzpicture}
    \end{center}
    \begin{enumerate}[label = (\roman*)]
        \item In generation $\mathcal P_t$, both individuals have the inbreeding coefficient $f_{A_t}=f_{B_t}=1$ because they are fully inbred. So,
            \[f_{\mathcal P_t}=\frac12\left(f_{A_t}+f_{B_t}\right)=1.\]
        \item An individual $I\in\mathcal P_{t+1}$ has inherited its two alleles at any given locus from its two unrelated parents from generation $\mathcal P_t$. The probability of these alleles to be ibd is thus 0 and we conclude
            \[f_{\mathcal P_{t+1}}=\frac12\left(f_{A_{t+1}}+f_{B_{t+1}}\right)=0.\]
        \item If we look at an individual $I\in\mathcal P_{t+2}$, we see that the first allele at a given locus ultimately stems from either $A_t\in\mathcal P_t$ or $B_t\in\mathcal P_t$, and the same holds independently for the second allele. The two alleles are ibd if and only if they come from the same grandparent, whence we conclude that
        $f_{A_{t+2}}=f_{B_{t+2}}=\frac12$ and in consequence
            \[f_{\mathcal P_{t+2}}=\frac12\left(f_{A_{t+2}}+f_{B_{t+2}}\right)=\frac12.\]
    \end{enumerate}
    Thus, if we want to judge the population in terms of inbreeding, we get very different results, depending on the generation we look at, even though the actual risk of losing genetic diversity may not have changed so drastically over the generations.
\end{Ex}

Instead, a more stable measure and better indicator for the genetic variance that remains in generation $\mathcal P_t$ of the population is the average kinship $k_{\mathcal P_t,\mathcal P_t}$.

Typically, kinships are calculated between individuals. However, the notion is easily extended to finite groups of individuals. The following definition can for example be found in \citep{jimenezmena16}.

\begin{Def}\label{def::grpkin}
    For two groups $\mathcal I, \mathcal J\subseteq\mathcal P_t$, we fix a locus and sample one of the $2\cdot|\mathcal I|$ alleles that are assembled at this locus in group $\mathcal I$ and one of the $2\cdot|\mathcal J|$ alleles that are assembled at this locus in group $\mathcal J$. The kinship $k_{\mathcal I,\mathcal J}$ between groups $\mathcal I$ and $\mathcal J$ is defined as the probability of the two sampled alleles to be ibd.
\end{Def}

\begin{Rmk}
    If $\mathcal I$ and $\mathcal J$ comprise common individuals ($\mathcal I\cap\mathcal J\neq\varnothing$), the drawing process has to be thought of as "with replacement".
\end{Rmk}

\begin{Def}
    \begin{enumerate}[label = (\roman*)]
        \item For individuals $I,J\in\mathcal P_t$, their kinship can then simply be defined as the kinship between the single-elemented groups containing $I$ and $J$, respectively:
            \[k_{I,J}=k_{\{I\},\{J\}}.\]

        \item The kinship between an individual $I\in\mathcal P_t$ and a set of individuals $\mathcal J\subseteq\mathcal P_t$ is defined as
            \[k_{I,\mathcal J}=k_{\{I\},\mathcal J}.\]
    \end{enumerate}
\end{Def}

\begin{Rmk}\label{rmk::kinsh}
    \begin{enumerate}[label = (\roman*)]
        \item Our definition of kinship between individuals via single-elemented sets is equivalent to the standard definition of kinship by \citet{malecot48}.

        \item \label{item::relation}The notion of \emph{relationship coefficients}, as introduced by \citet{wright22}, is very closely related to Malécot's kinship coefficients: The relationship between two individuals is twice their kinship.
    \end{enumerate}
\end{Rmk}

As the following lemma shows, the kinship between two groups of individuals is the average kinship between the individuals of the two groups.

\begin{Lem}
    Let $\mathcal I,\mathcal J\subseteq\mathcal P_t$ be two finite groups of individuals. Then
        \[k_{\mathcal I,\mathcal J}=\frac1{|\mathcal I|\cdot|\mathcal J|}\sum_{I\in\mathcal I}\sum_{J\in\mathcal J}k_{I,J}.\]
\end{Lem}

\begin{proof}
    Fix an individual $I\in\mathcal I$ and an individual $J\in\mathcal J$. In the thought allele drawing experiment to determine $k_{\mathcal I,\mathcal J}$, the probability that the first allele is drawn from $I$ is $\frac2{2\cdot|\mathcal I|}=\frac1{|\mathcal I|}$ and the probability that the second allele is drawn from $J$ is $\frac2{2\cdot|\mathcal J|}=\frac1{|\mathcal J|}$. If the two drawn alleles are indeed from $I$ and $J$, respectively, their probability to be ibd is $k_{I,J}$. The assertion follows.
\end{proof}

\begin{Ex}\label{ex::fformula}
    To familiarize ourselves with the notion of kinship coefficients (between individuals), we derive their values for some simple cases.
    \begin{enumerate}[label=(\roman*)]
        \item Evidently, if two individuals $I,J\in\mathcal P_t$ do not share any common ancestors, their probability to share ibd alleles and thus their kinship coefficient is
            \[k_{I,J}=0.\]
        
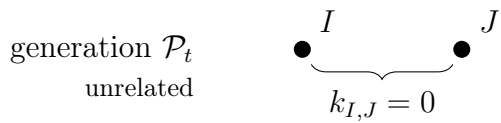
\begin{center}
            \begin{tikzpicture}
                \path (0,0) coordinate (Pt0)
                    node[left = 0cm of Pt0] {generation $\mathcal P_t$}
                    node[generation description = Pt0] {unrelated}
                    node[right = 1.3cm of Pt0, individual] (I) {}
                    node[above right = 0cm and 0cm of I] {$I$}
                    node[right = 2cm of I, individual] (J) {}
                    node[above right = 0cm and 0cm of J] {$J$};
                \draw[relationship]  (I) -- (J)
                    node[relationship description]{$k_{I,J}=0$};
            \end{tikzpicture}
        \end{center}

        \item Assume two half-siblings, i.\,e., two individuals $I,J\in\mathcal P_{t+1}$ that share one non-inbred parent $P\in\mathcal P_t$ with the respective other parents being neither related to $P$ nor to each other. The only possibility for $I$ and $J$ to share ibd alleles is via $P$. For the drawing procedure to end up with ibd alleles one would have to
        \begin{itemize}
            \item draw from both $I$ and $J$ the respective allele that was inherited from $P$ (probability each time $\frac12$), and
            \item $P$ would have to have inherited the same of its two alleles to both $I$ and $J$ (probability $\frac12$).
        \end{itemize}
        The probability for all this to happen and thus the kinship of $I$ and $J$ is
            \[k_{I,J}=\frac1{2^3}=\frac18.\]
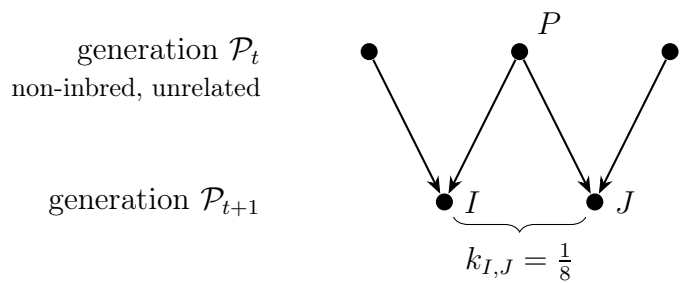
        \begin{center}
            \begin{tikzpicture}
                \path (0,0) coordinate (Pt0)
                    node[left = 0mm of Pt0] {generation $\mathcal P_t$}
                    node[generation description = Pt0]
                        {non-inbred, unrelated}
                    coordinate[below = 2cm of Pt0] (Pt1)
                    node[left = 0mm of Pt1] {generation $\mathcal P_{t+1}$}
                    node[right = 1.3cm of Pt0, individual] (X1) {}
                    node[right = 2cm of X1.center, individual] (P) {}
                    node[above right = 0cm and 0cm of P] {$P$}
                    node[right = 2cm of P.center, individual] (X2) {}
                    node[below left = 2cm and 1cm of P.center, individual]
                        (I) {}
                    node[below right = 2cm and 1cm of P.center, individual]
                        (J) {}
                    node[right = 0cm of I] {$I$}
                    node[right = 0cm of J] {$J$};
                \draw[inheritance] (X1) -- (I);
                \draw[inheritance] (P) -- (I);
                \draw[inheritance] (X2) -- (J);
                \draw[inheritance] (P) -- (J);
                \draw[relationship] (I) -- (J)
                    node[relationship description]{$k_{I,J}=\frac18$};
            \end{tikzpicture}
        \end{center}

        \item Next, we assume full-siblings (i.\,e., two individuals $I,J\in\mathcal P_{t+1}$ with $I\neq J$) that are common offspring of individuals $P$ and $Q\in\mathcal P_t$. We further assume that $P$ and $Q$ are neither inbred nor related to each other. At any fixed locus, $P$ and $Q$ together assemble four different alleles and when we pick a random allele of $I$, this could be any of these four alleles with equal probability. Likewise, when we draw a random allele from $J$, we independently also end up with a random allele from the four combined alleles of $P$ and $Q$. The probability of these two drawings to render the same allele is
            \[k_{I,J}=\frac14.\]
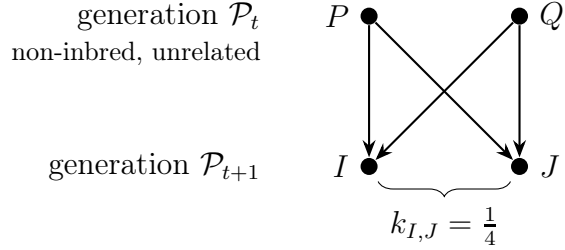
        \begin{center}
            \begin{tikzpicture}
                \path (0,0) coordinate (Pt0)
                    node[left = 0cm of Pt0] {generation $\mathcal P_t$}
                    node[generation description = Pt0]
                        {non-inbred, unrelated}
                    coordinate[below = 2cm of Pt0] (Pt1)
                    node[left = 0cm of Pt1] {generation $\mathcal P_{t+1}$}
                    node[right = 1.3cm of Pt0, individual] (P) {}
                    node[left=0cm of P] {$P$}
                    node[right = 2cm of P.center, individual] (Q) {}
                    node[right=0cm of Q] {$Q$}
                    node[below = 2cm of P.center, individual] (I) {}
                    node[left=0cm of I] {$I$}
                    node[right = 2cm of I.center, individual] (J) {}
                    node[right=0cm of J] {$J$};
                \draw[inheritance] (P) -- (I);
                \draw[inheritance] (Q) -- (I);
                \draw[inheritance] (P) -- (J);
                \draw[inheritance] (Q) -- (J);
                \draw[relationship]  (I) -- (J)
                    node[relationship description]{$k_{I,J}=\frac14$};
            \end{tikzpicture}
        \end{center}

        \item \label{item::lastly} Lastly, we look at the kinship of an individual $I\in\mathcal P_t$ with itself. This can also be interpreted as the kinship between two identical twins. If we draw two alleles from $I$ (with replacement), we have a probability of $\frac12$ that we draw both times the very same allele (which, of course, is ibd to itself). With the complementary probability $\frac12$, we draw the two different alleles of $I$. The probability for these to be ibd is precisely the inbreeding coefficient of $I$:
            \[k_{I,I}=\frac12+\frac{f_{I}}{2}.\]
        This means that $k_{I,I}$ will always take on values between $\frac12$ (non-inbred) and $1$ (fully inbred).
    \end{enumerate}
\end{Ex}

\begin{Ex}\label{ex::avginbrr}
    We pick up our Example~\ref{ex::avginbr} and now look at $k_{\mathcal P_t,\mathcal P_t}$, $k_{\mathcal P_{t+1},\mathcal P_{t+1}}$, and $k_{\mathcal P_{t+2},\mathcal P_{t+2}}$. In each generation, the pool from which we draw consists of four alleles.
    
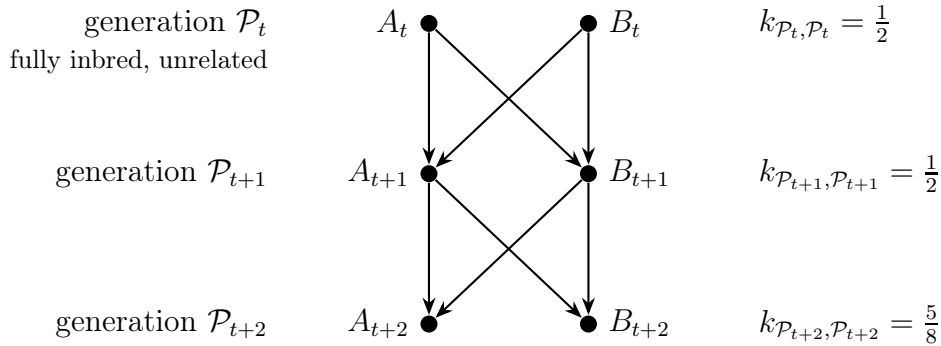
\begin{center}
        \begin{tikzpicture}[every node/.style={draw=none},node distance=2cm and 2cm]
            \path (0,4) coordinate (Pt0)
                        node[left = 0mm of Pt0]
                            {generation $\mathcal P_t$}
                        node[generation description = Pt0]
                                 {fully inbred, unrelated}
                        coordinate[below = of Pt0] (Pt1)
                        node[left = 0mm of Pt1]
                            {generation $\mathcal P_{t+1}$}
                        coordinate[below = of Pt1] (Pt2)
                        node[left = 0mm of Pt2]
                            {generation $\mathcal P_{t+2}$}
                        node[right = of Pt0, individual] (At0) {}
                        node[left = 0mm of At0] {$A_{t}$}
                        node[right = of At0, individual] (Bt0) {}
                        node[right = 0mm of Bt0]
                            {$B_{t}$}
                        node[right = of Pt1, individual] (At1) {}
                        node[left = 0mm of At1]
                            {$A_{t+1}$}
                        node[right = of At1, individual] (Bt1) {}
                        node[right = 0mm of Bt1]
                            {$B_{t+1}$}
                        node[right = of Pt2, individual] (At2) {}
                        node[left = 0mm of At2]
                            {$A_{t+2}$}
                        node[right = of At2, individual] (Bt2) {}
                        node[right = 0mm of Bt2]
                            {$B_{t+2}$}
                        node[right = of Bt0]
                            {$k_{\mathcal P_t,\mathcal P_t}=\frac12$}
                        node[right = of Bt1]
                            {$k_{\mathcal P_{t+1},\mathcal P_{t+1}}
                            =\frac12$}
                        node[right = of Bt2]
                            {$k_{\mathcal P_{t+2},\mathcal P_{t+2}}
                            =\frac58$}
                            ;
            \draw[inheritance] (At0) -- (At1);
            \draw[inheritance] (At0) -- (Bt1);
            \draw[inheritance] (Bt0) -- (At1);
            \draw[inheritance] (Bt0) -- (Bt1);
            \draw[inheritance] (At1) -- (At2);
            \draw[inheritance] (At1) -- (Bt2);
            \draw[inheritance] (Bt1) -- (At2);
            \draw[inheritance] (Bt1) -- (Bt2);
        \end{tikzpicture}
    \end{center}
    \begin{enumerate}[label = (\roman*)]
        \item In generation $\mathcal P_t$, we draw our two alleles with equal probability either from the same or from different individuals. In the former case, the alleles will be identical since both individuals in $\mathcal P_t$ are fully inbred. In the second case, the alleles will not be ibd, because the individuals in $\mathcal P_t$ are mutually unrelated. The overall probability to end up with ibd alleles in generation $\mathcal P_t$ is thus
            \[k_{\mathcal P_t,\mathcal P_t}=\frac12.\]

        \item In generation $\mathcal P_{t+1}$, the four alleles consist of two identical alleles that come from $A_t\in\mathcal P_t$ and two identical alleles that come from $B_t\in\mathcal P_t$. Because $A_t$ and $B_t$ are unrelated, the two pairs of identical alleles in generation $\mathcal P_{t+1}$ are non-identical. Thus, as in generation $\mathcal P_t$, we have two different pairs of identical alleles in generation $\mathcal P_{t+1}$ and
            \[k_{\mathcal P_{t+1},\mathcal P_{t+1}}=\frac12.\]

        \item Finally, we look at generation $\mathcal P_{t+2}$. Here, we distinguish two cases for our drawing of alleles. Case 1 is that we draw the very same allele twice. As there are four alleles, the probability for this to happen is $\frac14$. In this case, the two drawn alleles are necessarily ibd. Case 2 is that we draw two different alleles and thus has the complementary probability of $\frac34$. In this case, each of the two alleles independently originates ultimately either from $A_t$ or $B_t\in\mathcal P_t$. If they originate from the same individual, they are ibd (because both $A_t$ and $B_t$ are fully inbred), if they originate from different individuals, they are not ibd (because $A_t$ and $B_t$ are unrelated). Both possibilities occur with probability $\frac12$. So, in total, the probability of drawing two ibd alleles is
            \[k_{\mathcal P_{t+2},\mathcal P_{t+2}}=\frac14\cdot1+\frac34\cdot\frac12=\frac58.\]
    \end{enumerate}
\end{Ex}

\begin{Rmk}\label{rmk::inbkin}
    \begin{enumerate}[label = (\roman*)]
        \item We see that unlike the average inbreeding $f_{\mathcal P_t}$, the average kinship $k_{\mathcal P_t,\mathcal P_t}$ shows positive values for all generations and appears much more stable. A monotonous increase of $k_{\mathcal P_t,\mathcal P_t}$ in $t$ is what we generally expect for a closed population. (Note, however, that it is possible to construct situations where $k_{\mathcal P_{t+1},\mathcal P_{t+1}} < k_{\mathcal P_t,\mathcal P_t}$.)

        \item \label{item::inbkin} It should be noted that Example~\ref{ex::avginbr} ($=$ Example~\ref{ex::avginbrr}) was chosen in an extreme way to illustrate the shortcomings of $f_{\mathcal P_t}$. In many breeding schemes, $f_{\mathcal P_t}$ and $k_{\mathcal P_t,\mathcal P_t}$ show very similar behavior with
        \[k_{\mathcal P_t,\mathcal P_t}\approx f_{\mathcal P_{t+1}}.\]
    But also in this case one should stick with $k_{\mathcal P_t,\mathcal P_t}$ rather than $f_{\mathcal P_{t}}$, because the former value \emph{looks one generation further into the future}.
    \end{enumerate}
\end{Rmk}

Next, we want to take a closer look at the development of $k_{\mathcal P_t,\mathcal P_t}$ over the generations. We will show the following

\begin{Lem}\label{lemma::avkin}
    With the notation as introduced above, we have
    \begin{equation*}
        k_{\mathcal P_{t+1},\mathcal P_{t+1}} =
              \sum_{I,J\in\mathcal P_t}c_{I}c_{J}k_{I,J}
            - \frac1{4N_{t+1}}\sum_{I\in\mathcal P_t}c_{I}f_{I}
            + \frac1{4N_{t+1}}.
    \end{equation*}
\end{Lem}

\begin{proof}
    Fix a random locus and two non-identical individuals $I,J\in\mathcal P_t$. If we draw a random allele at this locus from the pool of $2N_{t+1}$ alleles in generation $\mathcal P_{t+1}$, the probability that this allele was inherited from individual $I\in\mathcal P_t$ is $c_{I}$ because this number signifies the proportion of alleles in generation $P_{t+1}$ that comes from $I$. If we then draw a random second allele, the probability that this allele comes from individual $J\in\mathcal P_t$ is $c_{J}$. Provided that we indeed drew a first allele that came from $I$ and a second allele that came from $J$, the probability of these alleles to be ibd is $k_{I,J}$ by the definition of kinship coefficients. Thus, the total probability that the two drawn alleles are coming from $I$ and $J$, respectively, and are ibd is $c_{I}c_{J}k_{I,J}$.

    At first glance, this suggests that the average kinship $k_{\mathcal P_{t+1},\mathcal P_{t+1}}$ should amount to
        \[\sum_{I,J\in\mathcal P_t}c_{I}c_{J}k_{I,J}.\]
    However, note that above we had assumed $I$ and $J$ to be non-identical. Of course, there is also the possibility that we pick two alleles from generation $\mathcal P_{t+1}$ that both come from the same individual $I\in\mathcal P_t$. With the same consideration as above, the probability for this to happen is $c_{I}^2$. But if both drawn alleles come indeed from individual $I\in\mathcal P_t$, their probability to be ibd will be higher than $k_{I,I}$:

    The number of alleles that are passed from individual $I\in\mathcal P_t$ to generation $\mathcal P_{t+1}$ is $2c_IN_{t+1}$. If we draw two of these alleles (with replacement!), the probability to pick the very same allele twice is $\frac1{2c_IN_{t+1}}$ and in this case their chance to be ibd is 1. With the complementary probability of $1-\frac1{2c_IN_{t+1}}$, we pick two different alleles which are independent samples of $I$'s two alleles with probability $k_{I,I}$ to be ibd. Thus, the probability, of picking two ibd alleles in generation $\mathcal P_{t+1}$ provided that both alleles were passed on from individual $I\in\mathcal P_t$ is
    \begin{align*}
        \frac1{2c_IN_{t+1}} + \left(1-\frac1{2c_IN_{t+1}}\right)k_{I,I}=k_{I,I}+\frac1{2c_IN_{t+1}}(1-k_{I,I}).
    \end{align*}

    This means that
    \begin{equation*}
        k_{\mathcal P_{t+1},\mathcal P_{t+1}} =
              \sum_{I,J\in\mathcal P_t}c_{I}c_{J}k_{I,J}
            + \frac1{2N_{t+1}}
                    \sum_{I\in\mathcal P_t}c_{I}\left(1-k_{I,I}\right).
    \end{equation*}
    The assertion follows by exploiting $\sum_{I\in\mathcal P_t}c_I=1$ (Equation~\ref{eq::csum}), replacing $k_{I,I}$ with $\frac12+\frac{f_{I}}{2}$ (Example~\ref{ex::fformula}\,\ref{item::lastly}), and simplifying.
\end{proof}

\begin{Rmk}
    \begin{enumerate}[label = (\roman*)]
        \item The reasoning why the term $-\frac1{4N_{t+1}}\sum_{I\in\mathcal P_t}c_{I}f_{I}$$ + \frac1{4N_{t+1}}$ needs to be added is fairly subtle. The original derivation of OCS by \citet{meuwissen97} apparently overlooks this summand and only works with $\sum_{I,J\in\mathcal P_t}c_{I}c_{J}k_{I,J}$ for $k_{\mathcal P_{t+1},\mathcal P_{t+1}}$. Since then, many authors have adapted Meuwissen's formula, seemingly without questioning it. However, the (correct) formula presented here is not unknown. It can, for example, be found in \citep{wellmann09} or \citep{wellmann19key}.

        \item For the reader unfamiliar with Lemma~\ref{lemma::avkin}, it is instructive to once again check out Example~\ref{ex::avginbr} ($=$ Example~\ref{ex::avginbrr}). Here, for all  three generations we have $N_t=N_{t+1}=N_{t+2}=2$ and further $c_{A_t}=c_{B_t}=c_{A_{t+1}}=c_{B_{t+1}}=\frac12$. Readers are invited to check out and calculate for themselves that and why Lemma~\ref{lemma::avkin} actually holds in this special case. By doing so, they will find that the additional term is zero in the calculation of $k_{\mathcal P_{t+1},\mathcal P_{t+1}}$ but becomes positive in the calculation of $k_{\mathcal P_{t+2},\mathcal P_{t+2}}$.

        \item In our analysis, we focus on the average kinship $k_{\mathcal P_t,\mathcal P_t}$. The same value for $k_{\mathcal P_t,\mathcal P_t}$ can be reached either with highly inbred but mutually barely related individuals or by non-inbred individuals that share some stronger kinships between each other. This can very well be seen by comparing the first two generations of Example~\ref{ex::avginbr} ($=$ Example~\ref{ex::avginbrr}) with each other. In practice, the second option will be favored, because it features a lower risk of inbreeding depression at the same level of genetic diversity within the population. But when it comes to keeping the average kinship for the next generation low, it appears that a population of highly inbred individuals is to be preferred, because high inbreeding coefficients $f_{I}$ lower the value of the additional term in Lemma~\ref{lemma::avkin}. \citet{wellmann19key} thus argue that the additional term may be deliberately left out or could be altered in a way that ameliorates the preference of inbred individuals for reproduction. The validity of their argument appears unclear. While the goal to keep the average kinship low most probably will indeed lead to a preferred \emph{selection} of inbred individuals, it is not evident if such a selection scheme also \emph{produces} highly inbred individuals at a higher rate.
    \end{enumerate}
\end{Rmk}

From now on, we will pass to the use of vector and matrix notation, with which our results can be formulated more concisely.

\begin{Not}
    \begin{enumerate}[label = (\roman*)]
        \item We combine the $N_t$ contributions $c_{I}$ of individuals $I\in\mathcal P_t$ to the next generation $\mathcal P_{t+1}$ to one vector
            \[\mathbf c_t=\left(c_{I}\right)_{I\in\mathcal P_t}\in\mathbb R_{\geq0}^{\mathcal P_t}.\]

        \item Similarly, we combine the estimated breeding values $\hat u_{I}$ of the individuals $I\in\mathcal P_t$ to a vector
            \[\hat{\mathbf u}_t=\left(\hat u_{I}\right)_{I\in\mathcal P_t}\in\mathbb R^{\mathcal P_t},\]

        \item and their inbreeding coefficients $f_{I}$ to
            \[\mathbf f_t=\left(f_{I}\right)_{I\in\mathcal P_t}\in\mathbb R^{\mathcal P_t}.\]

        \item Furthermore, the kinships between individuals $I\in\mathcal P_t$ are combined to a matrix
            \[\mathbf K_t=\left(k_{I,J}\right)_{I,J\in\mathcal P_t}\in\mathbb R^{\mathcal P_t\times \mathcal P_t}.\]

        \item Finally, we let $\boldsymbol{1}_t\in\mathbb R^{\mathcal P_t}$ be the vector with 1 as every entry:
            \[\boldsymbol 1_t=\left(1\right)_{I\in\mathcal P_t}.\]
    \end{enumerate}
\end{Not}

\begin{Rmk}\label{rmk::onone}
    \begin{enumerate}[label = (\roman*)]
        \item Some readers may not be familiar with the notation $\mathbb R^{\mathcal P_t}$, i.\,e. the real numbers to the power of a set. Formally, $\mathbb R^{\mathcal P_t}$ is defined as the vector space of all functions $\mathcal P_t\to\mathbb R$ and is isomorphic to $\mathbb R^{N_t}$. The conceptual advantage of working with $\mathbb R^{\mathcal P_t}$ instead of $\mathbb R^{N_t}$ is that it does not require to chose a (more or less arbitrary) ordering of the individuals $I\in\mathcal P_t$. Readers who are confused by $\mathbb R^{\mathcal P_t}$ are advised to simply think of $\mathbb R^{N_t}$ instead.

        \item By standard theory on quantitative genetics, the kinship matrix $\mathbf K_t$ is symmetric and positive definite \citep{lange97mathematical}.

        \item \label{item::scp} The scalar product of $\boldsymbol{1}_t\in\mathbb R^{\mathcal P_t}$ with another vector $\mathbf v\in\mathbb R^{\mathcal P_t}$ is the sum of all elements of $\mathbf v$. The average of the entries of $\mathbf v$ can thus be written as $\frac1{N_t}\mathbf 1_t^{\top}\mathbf v$.
    \end{enumerate}
\end{Rmk}

\begin{Lem}
    Using this notation, the key findings so far read as follows:
    \begin{align}
        \hat u_{\mathcal P_t} &=
            \frac1{N_t}\boldsymbol 1_t^{\top}\hat{\mathbf u}_t,
                \label{eq::preu}\\
        \mathbb E\left[\hat u_{\mathcal P_{t+1}}\right] &=
            \mathbf c_t^{\top}\hat{\mathbf u}_t,
                \label{eq::u}\\
        k_{\mathcal P_t,\mathcal P_t} &=
            \frac1{N_t^2}\boldsymbol 1_t^{\top}\mathbf K_t\boldsymbol 1_t,\\
        k_{\mathcal P_{t+1},\mathcal P_{t+1}} &=
            \mathbf c_t^{\top}\mathbf K_t\mathbf c_t
                - \frac{\mathbf c_{t}^{\top}\mathbf f_t}{4N_{t+1}}
                + \frac{1}{4N_{t+1}},
                    \label{eq::k}\\
        \boldsymbol 1_t^{\top}\mathbf c_t &= 1.
            \label{eq::c}
    \end{align}
\end{Lem}

\begin{Rmk}
    By equations~\ref{eq::u} and~\ref{eq::k}, we see that once $\mathbf u_t$, $\mathbf K_t$ (and thereby $\mathbf f_t$), and $N_{t+1}$ are given, the values for $\mathbb E\left[\hat u_{\mathcal P_{t+1}}\right]$ and $k_{\mathcal P_{t+1},\mathcal P_{t+1}}$ are fully determined by the vector $\mathbf c_t$ of contributions. We may thus see $\mathbb E\left[\hat u_{\mathcal P_{t+1}}\right]=\mathbb E\left[\hat u_{\mathcal P_{t+1}}\right](\mathbf c_t)$ and $k_{\mathcal P_{t+1},\mathcal P_{t+1}}=k_{\mathcal P_{t+1},\mathcal P_{t+1}}(\mathbf c_t)$ as functions in $\mathbf c_t$. The general idea behind OCS is to determine $\mathbf c_t$ so that $\mathbb E\left[\hat u_{\mathcal P_{t+1}}\right]$  is maximized under the constraint that $k_{\mathcal P_{t+1},\mathcal P_{t+1}}$ shall not exceed a given value $k_{t+1}^{\ast}$ that is judged to be an acceptable average kinship.
\end{Rmk}

Thus, the basic task of Optimum Contribution Selection can be formulated as follows:

\begin{Task}
    Given a generation $\mathcal P_t$, $\hat{\mathbf u}_t\in\mathbb R^{\mathcal P_t}$, and $\mathbf K_t\in\mathbb R^{\mathcal P_t\times \mathcal P_t}$ (symmetric and positive definite), as well as the required number of individuals of the next generation, $N_{t+1}$, and a maximum acceptable kinship level $k_{t+1}^{\ast}$, maximize the function
        \[\mathbb E\left[\hat u_{\mathcal P_{t+1}}\right]:
            \mathbb R_{\geq0}^{\mathcal P_t}\to\mathbb R,\quad
            \mathbf c_t\mapsto \mathbf c_t^{\top}\hat{\mathbf u}_t\]
    under the constraints
        \[\mathbf 1_t^{\top}\mathbf c_t=1\]
    and
        \[\mathbf c_t^{\top}\mathbf K_t\mathbf c_t
            - \frac{\mathbf c_{t}^{\top}\mathbf f_t}{4N_{t+1}}
            + \frac{1}{4N_{t+1}}\leq k_{t+1}^{\ast}.\]
\end{Task}

We will not proceed by explaining how this task can be tackled. Instead, we will derive the corresponding tasks for the diecious setting as well as for settings with overlapping generations. In Section~\ref{sec::ocshb}, we will derive similar tasks for honeybee populations. Finally, in Section~\ref{sec::solvetask}, we observe that all the tasks follow a general pattern. Then we will discuss how to solve tasks of this pattern in general.

\subsubsection{Diecious populations} \label{sec::diec}

In a diecious population, the $N_t$ individuals of a generation $\mathcal P_t$ fall into two categories, females and males. We may thus write $\mathcal P_t$ as a disjoint union
\[\mathcal P_t=\mathcal F_t\sqcup\mathcal M_t,\]
where $\mathcal F_t$ and $\mathcal M_t$ are the sets of female and male individuals of generation $\mathcal P_t$, respectively.

Furthermore, inheritance is organized in a way that each individual $I\in\mathcal P_{t+1}$ has one female parent $F\in\mathcal F_t$ (called \emph{dam}) and one male parent $M\in\mathcal M_t$ (called \emph{sire}).

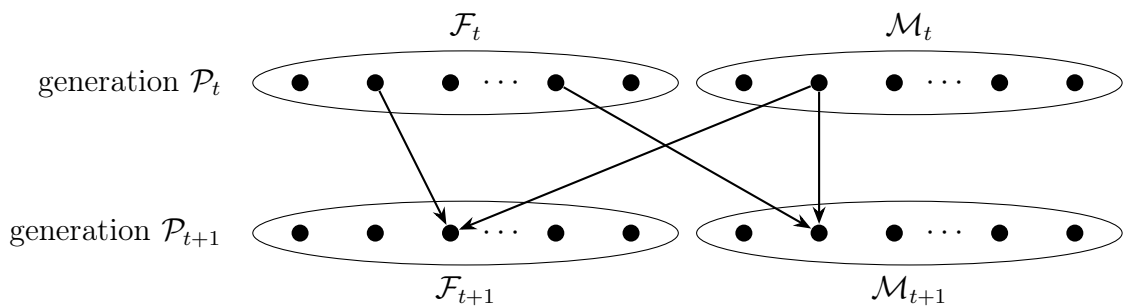
\begin{center}
    \begin{tikzpicture}
        \path (0,0) coordinate (Pt0)
            node[left = 0mm of Pt0] {generation $\mathcal P_t$}
            coordinate[below = 2cm of Pt0] (Pt1)
            node[left = 0mm of Pt1] {generation $\mathcal P_{t+1}$}

            node[right = 0.9cm of Pt0, individual] (At0) {}
            coordinate[right = 5mm of At0.center] (AAt0)
            node[right = 1cm of At0.center, individual] (Bt0) {}
            node[right = 1cm of Bt0.center, individual] (Ct0) {}
            node[right = 7mm of Ct0.center, anchor = center] (Dt0) {$\cdots$}
            node[right = 7mm of Dt0.center, individual] (Et0) {}
            node[right = 1cm of Et0.center, individual] (Ft0) {}
            coordinate[left = 5mm of Ft0.center] (FFt0)

            node[right = 1.5cm of Ft0.center, individual] (at0) {}
            coordinate[right = 5mm of at0.center] (aat0)
            node[right = 1cm of at0.center, individual] (bt0) {}
            node[right = 1cm of bt0.center, individual] (ct0) {}
            node[right = 7mm of ct0.center, anchor = center] (dt0) {$\cdots$}
            node[right = 7mm of dt0.center, individual] (et0) {}
            node[right = 1cm of et0.center, individual] (ft0) {}
            coordinate[left = 5mm of ft0.center] (fft0)

            node[right = 0.9cm of Pt1, individual] (At1) {}
            coordinate[right = 5mm of At1.center] (AAt1)
            node[right = 1cm of At1.center, individual] (Bt1) {}
            node[right = 1cm of Bt1.center, individual] (Ct1) {}
            node[right = 7mm of Ct1.center, anchor = center] (Dt1) {$\cdots$}
            node[right = 7mm of Dt1.center, individual] (Et1) {}
            node[right = 1cm of Et1.center, individual] (Ft1) {}
            coordinate[left = 5mm of Ft1.center] (FFt1)

            node[right = 1.5cm of Ft1.center, individual] (at1) {}
            coordinate[right = 5mm of at1.center] (aat1)
            node[right = 1cm of at1.center, individual] (bt1) {}
            node[right = 1cm of bt1.center, individual] (ct1) {}
            node[right = 7mm of ct1.center, anchor = center] (dt1) {$\cdots$}
            node[right = 7mm of dt1.center, individual] (et1) {}
            node[right = 1cm of et1.center, individual] (ft1) {}
            coordinate[left = 5mm of ft1.center] (fft1)
            ;

        \begin{scope}[on background layer]
            \path node [fit=(AAt0) (FFt0), ellipse, inner sep = 3mm, draw] (F0) {}
                node [fit=(aat0) (fft0), ellipse, inner sep = 3mm, draw] (M0) {}
                node [fit=(AAt1) (FFt1), ellipse, inner sep = 3mm, draw] (F1) {}
                node [fit=(aat1) (fft1), ellipse, inner sep = 3mm, draw] (M1) {}
                node [above = 0mm of F0] {$\mathcal F_t$}
                node [above = 0mm of M0] {$\mathcal M_t$}
                node [below = 0mm of F1] {$\mathcal F_{t+1}$}
                node [below = 0mm of M1] {$\mathcal M_{t+1}$};
        \end{scope}

        \draw[inheritance] (Bt0) -- (Ct1);
        \draw[inheritance] (bt0) -- (Ct1);
        \draw[inheritance] (Et0) -- (bt1);
        \draw[inheritance] (bt0) -- (bt1);
    \end{tikzpicture}
\end{center}

\begin{Not}
    We denote the numbers of female and male individuals of generation $\mathcal P_t$ by $N_{t}^{\mathcal F}:=\left|\mathcal F_t\right|$ and $N_{t}^{\mathcal M}:=\left|\mathcal M_t\right|$, respectively. So,
        \[N_t=N_t^{\mathcal F}+N_t^{\mathcal M}.\]
\end{Not}

\begin{Rmk}
    Writing $\mathcal P_t=\mathcal F_t\sqcup\mathcal M_t$ as a disjoint union gives rise to a natural isomorphism
        \[\mathbb R^{\mathcal P_t}\cong\mathbb R^{\mathcal F_t}\oplus\mathbb R^{\mathcal M_t},\]
    separating vector entries belonging to male and female individuals.
\end{Rmk}

\begin{Not}\label{not::einsfm}
    \begin{enumerate}[label = (\roman*)]
        \item \label{item::einsfm} Under this isomorphism, $\mathbf c_t$ becomes $\mathbf c_{t}^{\mathcal F}\oplus\mathbf c_t^{\mathcal M}$, $\hat{\mathbf u}_t$ becomes $\hat{\mathbf u}_{t}^{\mathcal F}\oplus\hat{\mathbf u}_{t}^{\mathcal M}$, and $\mathbf f_t$ becomes $\mathbf f_{t}^{\mathcal F}\oplus\mathbf f_{t}^{\mathcal M}$. Of course, also the vector of ones, $\mathbf 1_t$, can be split up into $\mathbf 1_{t}^{\mathcal F}\oplus\mathbf 1_{t}^{\mathcal M}$.

        \item For the matrix $\mathbf K_t$, we distinguish four blocks $\mathbf K_{t}^{\mathcal F\mathcal F}$, $\mathbf K_{t}^{\mathcal F\mathcal M}$, $\mathbf K_{t}^{\mathcal M\mathcal F}$, and $\mathbf K_{t}^{\mathcal M\mathcal M}$, containing the kinships within and between sex classes,
            \[\mathbf K_t = \begin{pmatrix}\mathbf K_{t}^{\mathcal F\mathcal F} & \mathbf K_{t}^{\mathcal F\mathcal M} \\
                                           \mathbf K_{t}^{\mathcal M\mathcal F} & \mathbf K_{t}^{\mathcal M\mathcal M}
                            \end{pmatrix}.\]
    \end{enumerate}
\end{Not}

\begin{Rmk}
    As $\mathbf K_t$ is symmetric, we have
        \[\mathbf K_{t}^{\mathcal M\mathcal F}=\left(\mathbf K_{t}^{\mathcal F\mathcal M}\right)^{\top}.\]
\end{Rmk}

Now, one should take a moment to think about what is the variable to maximize. The straightforward approach is to still use the average breeding value
    \[\hat{u}_{\mathcal P_t} =
    \frac1{N_t}\mathbf 1_t^{\top}\hat{\mathbf u}_t=
    \frac1{N_t^{\mathcal F}+N_{t}^{\mathcal M}}
    \left(\left(\mathbf 1_{t}^{\mathcal F}\right)^{\top}\hat{\mathbf u}_{t}^{\mathcal F}+\left(\mathbf 1_{t}^{\mathcal M}\right)^{\top}
        \hat{\mathbf u}_{t}^{\mathcal M}\right).\]
But consider the following

\begin{Ex}\label{ex::weightedmean}
    Assume a population of eight female individuals with breeding values
        \[\hat{\mathbf u}_{t}^{\mathcal F}=\begin{pmatrix}20.3 & 17.9 & 18.4 & 19.0 & 17.0 & 19.5 & 20.1 & 19.8 \end{pmatrix}^{\top}\]
    and two male individuals with breeding values
        \[\hat{\mathbf u}_{t}^{\mathcal M}=\begin{pmatrix}8.3 & 9.7 \end{pmatrix}^{\top}.\]
    Then the average breeding value of this population is
        \[\hat{u}_{\mathcal P_t}=\frac{20.3 + 17.9 + 18.4 + 19.0 + 17.0 + 19.5 + 20.1 + 19.8 + 8.3 + 9.7}{10} = 17.0.\]
    But even if we only let the best female of generation $\mathcal P_t$ (corresponding to the first entry in $\hat{\mathbf u}_{t}^{\mathcal F}$) and the best male (corresponding to the second entry in $\hat{\mathbf u}_{t}^{\mathcal M}$) have offspring with each other, the expected average breeding value of the next generation would be $\frac{20.3+9.7}{2}=15.0$. So the value $\hat{u}_{\mathcal P_t}$ does not adequately represent the population's inherent value for breeding purposes. Instead, the value
        \[\frac{\hat{u}_{\mathcal F_t}+\hat{u}_{\mathcal M_t}}2 =
        \frac{\frac1{N_{t}^{\mathcal F}}\left(\mathbf 1_{t}^{\mathcal F}\right)^{\top}\hat{\mathbf u}_{t}^{\mathcal F}
            +\frac1{N_{t}^{\mathcal M}}\left(\mathbf 1_{t}^{\mathcal M}\right)^{\top}\hat{\mathbf u}_{t}^{\mathcal M}}{2}\]
    appears to be a more appropriate choice. It reflects the expected breeding value of a common offspring of a randomly chosen female and a randomly chosen male. Following Definition~\ref{def::setset}, by declaring the (two-elemented) set of sexes
        \[\mathfrak S_t:=\{\mathcal F_t,\mathcal M_t\},\]
    we may write this value as $\hat u_{\mathfrak S_t}$.

    In our example, we have
        \[\hat{u}_{\mathcal F_t}=\frac{20.3 + 17.9 + 18.4 + 19.0 + 17.0 + 19.5 + 20.1 + 19.8}8 = 19.0\]
    and
        \[\hat{u}_{\mathcal M_t}=\frac{8.3 + 9.7}2 = 9.0,\]
    from which results
        \[\hat{u}_{\mathfrak S_t}=\frac{19.0 + 9.0}2 = 14.0.\]
\end{Ex}

So, instead of maximizing $\hat{u}_{\mathcal P_t}$ over the generations, in theory one should aim to maximize $\hat{u}_{\mathfrak S_t}$. However, in order to predict $\mathbb E\left[\hat{u}_{\mathfrak S_{t+1}}\right]$ from $\hat{\mathbf u}_{t}^{\mathcal F}\oplus\hat{\mathbf u}_{t}^{\mathcal M}$, one would have to split up the vector $\mathbf c_t\cong\mathbf c_{t}^{\mathcal F}\oplus\mathbf c_{t}^{\mathcal M}$ into the respective contributions towards male and female offspring. But, in general, it is not predictable if a particular offspring of two individuals will be male or female and if one assumes equal probability for both options, the expected value of $\hat{u}_{\mathfrak S_{t+1}}$ is precisely $\mathbb E\left[\hat{u}_{\mathcal P_{t+1}}\right]$. Therefore, after a little theoretical detour, one again ends up with the original choice of maximizing $\hat{u}_{\mathcal P_t}$.

\begin{Rmk}
    \begin{enumerate}[label = (\roman*)]
        \item By the use of sexed semen, it would indeed be possible to predefine separate contributions towards males and females in the next generation. We are not aware if the resulting theory for maximizing $\hat{u}_{\mathfrak S_t}$ has been worked out. However, even with sexed sperm, the practical relevance is likely to be low. The large difference between $\hat{u}_{\mathcal P_t}$ and $\hat{u}_{\mathfrak S_t}$ in Example~\ref{ex::weightedmean} comes from the big differences between male and female individuals in both number and average breeding values. While in a breeding scheme with sexed semen there are indeed likely more females than males, there is no reason to assume greatly differing average breeding values between the sexes. Therefore, the difference between $\hat{u}_{\mathcal P_t}$ and $\hat{u}_{\mathfrak S_t}$ will be practically negligible.

        \item Similarly, one may ask the question if $k_{\mathcal P_t,\mathcal P_t}$ is still the right measure for the genetic diversity within a diecious population. For a monoecious population with selfing, $k_{\mathcal P_t,\mathcal P_t}$ is the expected value of the average inbreeding $f_{\mathcal P_{t+1}}$ of the next generation under panmixia. For a diecious population, this is no longer the case as the expectation for $f_{\mathcal P_{t+1}}$ here equals the average kinship between male and female individuals, i.\,e.
            \[\mathbb E\left[f_{\mathcal P_{t+1}}|\text{ panmixia }\right] =
            k_{\mathcal F_t,\mathcal M_t} =
            \frac{1}{N_{t}^{\mathcal F}\cdot N_{t}^{\mathcal M}}\left(\mathbf 1_{t}^{\mathcal F}\right)^{\top}\mathbf K_{t}^{\mathcal F\mathcal M}\mathbf 1_{t}^{\mathcal M}.\]
        However, the fact that $k_{\mathcal P_t,\mathcal P_t}$ is the expectation for $f_{\mathcal P_{t+1}}$ is not the reason why we chose this value for our analysis in the monoecious case. Particularly, because the assumption of panmixia is anyway massively violated under directed selection. Furthermore, to predict such an alternative measure for interrelatedness in the population would again require control over the sex of offspring.
        Lastly, there seems to be no reason to divert from $k_{\mathcal P_t,\mathcal P_t}$ as the relevant measure for genetic interrelatedness, also in diecious populations.
    \end{enumerate}
\end{Rmk}

By the above considerations, it follows that Equations~\ref{eq::preu} to~\ref{eq::k} do not require any changes in the diecious case. But one can reformulate them in a way that female and male components become apparent. The equations then turn into

\begin{align}
    \hat u_{\mathcal P_t}
        &= \frac1{N_t^{\mathcal F}+N_{t}^{\mathcal M}}\cdot
            \left(\left(\boldsymbol 1_{t}^{\mathcal F}\right)^{\top}\hat{\mathbf u}_{t}^{\mathcal F} +
                \left(\boldsymbol 1_{t}^{\mathcal M}\right)^{\top}\hat{\mathbf u}_{t}^{\mathcal M}\right), \label{eq::preudi}\\
    \mathbb E\left[\hat u_{\mathcal P_{t+1}}\right]
        &= \left(\mathbf c_{t}^{\mathcal F}\right)^{\top}\hat{\mathbf u}_{t}^{\mathcal F}
            + \left(\mathbf c_{t}^{\mathcal M}\right)^{\top}\hat{\mathbf u}_{t}^{\mathcal M}, \label{eq::udi}\\
    k_{\mathcal P_t,\mathcal P_t}
        &= \frac1{(N_{t}^{\mathcal F}+N_{t}^{\mathcal M})^2}
            \left(\left(\boldsymbol 1_{t}^{\mathcal F}\right)^{\top}\mathbf K_{t}^{\mathcal F\mathcal F}\boldsymbol 1_{t}^{\mathcal F}
                + 2\left(\boldsymbol 1_{t}^{\mathcal F}\right)^{\top}\mathbf K_{t}^{\mathcal F\mathcal M}\boldsymbol 1_{t}^{\mathcal M}
                + \left(\boldsymbol 1_{t}^{\mathcal M}\right)^{\top}\mathbf K_{t}^{\mathcal M\mathcal M}\boldsymbol 1_{t}^{\mathcal M}\right),\\
    k_{\mathcal P_{t+1},\mathcal P_{t+1}}
        &= \left(\mathbf c_{t}^{\mathcal F}\right)^{\top}\mathbf K_{t}^{\mathcal F\mathcal F}\mathbf c_{t}^{\mathcal F}
            + 2\left(\mathbf c_{t}^{\mathcal F}\right)^{\top}\mathbf K_{t}^{\mathcal F\mathcal M}\mathbf c_{t}^{\mathcal M}
            + \left(\mathbf c_{t}^{\mathcal M}\right)^{\top}\mathbf K_{t}^{\mathcal M\mathcal M}\mathbf c_{t}^{\mathcal M}\nonumber\\
            &\quad\quad\quad\quad\quad\quad\quad\quad\quad\quad\quad\quad\quad\quad\quad
            - \frac{\left(\mathbf c_{t}^{\mathcal F}\right)^{\top}\mathbf f_{t}^{\mathcal F}+\left(\mathbf c_{t}^{\mathcal M}\right)^{\top}\mathbf f_{t}^{\mathcal M}}{4N_{t+1}}
            + \frac{1}{4N_{t+1}}. \label{eq::kdi}
\end{align}

\begin{Rmk}
    These reformulations of Equations~\ref{eq::preu} to~\ref{eq::k} do not reveal further insights and appear clumsy in comparison with the original. Throughout the literature, equations are therefore usually reported in the version of Equations~\ref{eq::preu} to~\ref{eq::k}. The reason for adding the alternative formulations of Equations~\ref{eq::preudi} to~\ref{eq::kdi} is that it may prepare the reader for what to expect when turning to honeybees later on.
\end{Rmk}

Finally, the condition imposed on $\mathbf c_t$ by Equation~\ref{eq::c} actually needs modification in the diecious case. While in the monoecious case, we only had the requirement that all contributions add up to unity, we now need to obey the fact that each individual $I\in\mathcal P_{t+1}$ has exactly one male and one female parent. So, male and female individuals of generation $\mathcal P_t$ have to contribute equally to generation $\mathcal P_{t+1}$, which is captured by the two equations
\begin{equation}\label{eq::ceins}
    \left(\mathbf 1_{t}^{\mathcal F}\right)^{\top}\mathbf c_{t}^{\mathcal F}=\frac12
\end{equation}
and
\begin{equation}\label{eq::czwei}
    \left(\mathbf 1_{t}^{\mathcal M}\right)^{\top}\mathbf c_{t}^{\mathcal M}=\frac12.
\end{equation}

Once more, we are able to formulate the task of OCS in case of a diecious population:

\begin{Task}
    Given a sex-divided generation $\mathcal P_t=\mathcal F_t\sqcup\mathcal M_t$ and
    \begin{itemize}
        \item $\hat{\mathbf u}_{t}^{\mathcal F}\in\mathbb R^{\mathcal F_t}, \hat{\mathbf u}_{t}^{\mathcal M}\in\mathbb R^{\mathcal M_t}$,

        \item matrices $\mathbf K_{t}^{\mathcal F\mathcal F}\in\mathbb R^{\mathcal F_t\times \mathcal F_t}$, $\mathbf K_{t}^{\mathcal F\mathcal M}\in\mathbb R^{\mathcal F_t\times \mathcal M_t}$, and $\mathbf K_{t}^{\mathcal M\mathcal M}\in\mathbb R^{\mathcal M_t\times \mathcal M_t}$ such that
            $\mathbf K_t:=\begin{pmatrix}\mathbf K_{t}^{\mathcal F\mathcal F} & \mathbf K_{t}^{\mathcal F\mathcal M} \\
            \left(\mathbf K_{t}^{\mathcal F\mathcal M}\right)^{\top} & \mathbf K_{t}^{\mathcal M\mathcal M}\end{pmatrix}$
        is symmetric and positive definite,

        \item the required number of individuals of the next generation, $N_{t+1}$,

        \item and a maximum acceptable kinship level $k_{t+1}^{\ast}$,
    \end{itemize}
    maximize the function
        \[\mathbb E\left[\hat u_{\mathcal P_{t+1}}\right]:\mathbb R_{\geq0}^{\mathcal F_t}\oplus \mathbb R_{\geq0}^{\mathcal M_t}\to\mathbb R,\quad
            \mathbf c_{t}^{\mathcal F}\oplus\mathbf c_{t}^{\mathcal M}
                \mapsto \left(\mathbf c_{t}^{\mathcal F}\right)^{\top}\hat{\mathbf u}_{t}^{\mathcal F}+\left(\mathbf c_{t}^{\mathcal M}\right)^{\top}\hat{\mathbf u}_{t}^{\mathcal M}\]
    under the constraints
        \[\left(\mathbf 1_{t}^{\mathcal F}\right)^{\top}\mathbf c_{t}^{\mathcal F}=\frac12,\]
        \[\left(\mathbf 1_{t}^{\mathcal M}\right)^{\top}\mathbf c_{t}^{\mathcal M}=\frac12,\]
    and
    \begin{align*}
        \left(\mathbf c_{t}^{\mathcal F}\right)^{\top}\mathbf K_{t}^{\mathcal F\mathcal F}\mathbf c_{t}^{\mathcal F}
            + 2\left(\mathbf c_{t}^{\mathcal F}\right)^{\top}\mathbf K_{t}^{\mathcal F\mathcal M}\mathbf c_{t}^{\mathcal M}
            &+ \left(\mathbf c_{t}^{\mathcal M}\right)^{\top}\mathbf K_{t}^{\mathcal M\mathcal M}\mathbf c_{t}^{\mathcal M}\\
            &- \frac{\left(\mathbf c_{t}^{\mathcal F}\right)^{\top}\mathbf f_{t}^{\mathcal F}+\left(\mathbf c_{t}^{\mathcal M}\right)^{\top}\mathbf f_{t}^{\mathcal M}}{4N_{t+1}}
            + \frac{1}{4N_{t+1}}
                \leq k_{t+1}^{\ast}.
    \end{align*}
\end{Task}

\subsection{Overlapping generations} \label{sec::overl}

So far, we have discussed a situation where individuals of a generation $\mathcal P_t$ produced the next generation $\mathcal P_{t+1}$ at one specific time and afterwards ceased to play a role in the breeding system. In most real-life breeding systems, this will not be the case. Instead, individuals may have several offspring that are born at different points in time and generations may overlap and intermingle. To model this situation, we no longer distinguish disjoint generations but rather different reproductive periods, meaning time frames during which new individuals are born and some old individuals leave the population. For simplicity, one can assume such a time frame to be one year but rescaling to different time frames is easily possible. Hopefully not too misleading, we still call the state $\mathcal P_t$ of a population $\mathcal P$ at a time $t$ a \emph{generation}. So, formally, what we do is to allow subsequent generations $\mathcal P_t$ and $\mathcal P_{t+1}$ to share common individuals.

\begin{Rmk}
    The first derivation of OCS with overlapping generations was worked out by \citep{meuwissen98}, who followed a slightly different modeling approach than we do.
\end{Rmk}

\begin{Not}\label{not::nsclassic}
    \begin{enumerate}[label = (\roman*)]
        \item Each generation $\mathcal P_{t+1}$ now splits into two disjoint sets
            \[\mathcal P_{t+1}=\mathcal N_{t+1}\sqcup\mathcal S_{t+1},\]
        where
            \[\mathcal N_{t+1}:=\mathcal P_{t+1}\backslash\mathcal P_t\]
        are the newly born individuals of generation $\mathcal P_{t+1}$ and
            \[\mathcal S_{t+1}:=\mathcal P_{t+1}\cap\mathcal P_t\]
        are the individuals that survived from generation $\mathcal P_t$.

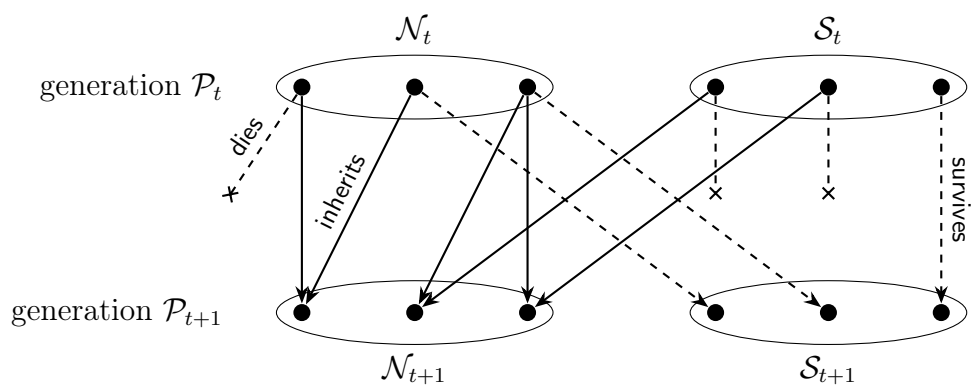
        \begin{center}
            \begin{tikzpicture}
                \path (0,0) coordinate (Pt0)
                    node[left = 0mm of Pt0] {generation $\mathcal P_t$}
                    coordinate[below = 3cm of Pt0] (Pt1)
                    node[left = 0mm of Pt1] {generation $\mathcal P_{t+1}$}

                    node[right = 0.9cm of Pt0, individual] (At0) {}
                    coordinate[right = 5mm of At0.center] (AAt0)
                    node[right = 1.5cm of At0.center, individual] (Bt0) {}
                    node[right = 1.5cm of Bt0.center, individual] (Ct0) {}
                    coordinate[left = 5mm of Ct0.center] (CCt0)
                    node[right = 2.5cm of Ct0.center, individual] (Dt0) {}
                    coordinate[right = 5mm of Dt0.center] (DDt0)
                    node[right = 1.5cm of Dt0.center, individual] (Et0) {}
                    node[right = 1.5cm of Et0.center, individual] (Ft0) {}
                    coordinate[left = 5mm of Ft0.center] (FFt0)

                    node[right = 0.9cm of Pt1, individual] (At1) {}
                    coordinate[right = 5mm of At1.center] (AAt1)
                    node[right = 1.5cm of At1.center, individual] (Bt1) {}
                    node[right = 1.5cm of Bt1.center, individual] (Ct1) {}
                    coordinate[left = 5mm of Ct1.center] (CCt1)
                    node[right = 2.5cm of Ct1.center, individual] (Dt1) {}
                    coordinate[right = 5mm of Dt1.center] (DDt1)
                    node[right = 1.5cm of Dt1.center, individual] (Et1) {}
                    node[right = 1.5cm of Et1.center, individual] (Ft1) {}
                    coordinate[left = 5mm of Ft1.center] (FFt1)
                    ;

                \begin{scope}[on background layer]
                    \path node [fit=(AAt0) (CCt0), ellipse, inner sep = 3mm, draw] (N0) {}
                        node [fit=(DDt0) (FFt0), ellipse, inner sep = 3mm, draw] (S0) {}
                        node [fit=(AAt1) (CCt1), ellipse, inner sep = 3mm, draw] (N1) {}
                        node [fit=(DDt1) (FFt1), ellipse, inner sep = 3mm, draw] (S1) {}
                        node [above = 0mm of N0] {$\mathcal N_t$}
                        node [above = 0mm of S0] {$\mathcal S_t$}
                        node [below = 0mm of N1] {$\mathcal N_{t+1}$}
                        node [below = 0mm of S1] {$\mathcal S_{t+1}$};
                \end{scope}

                \draw[inheritance] (At0) -- (At1);
                \draw[inheritance] (Bt0) -- (At1)
                    node[gene pass description]{inherits};
                \draw[inheritance] (Ct0) -- (Bt1);
                \draw[inheritance] (Dt0) -- (Bt1);
                \draw[inheritance] (Ct0) -- (Ct1);
                \draw[inheritance] (Et0) -- (Ct1);
                \draw[survival] (Bt0) -- (Dt1);
                \draw[survival] (Ct0) -- (Et1);
                \draw[survival] (Ft0) -- (Ft1)
                    node[gene pass description]{survives};
                \draw[death] (At0) -- ++(-1cm,-1.5cm)
                    node[gene pass description]{dies};
                \draw[death] (Dt0) -- ++(0cm,-1.5cm);
                \draw[death] (Et0) -- ++(0cm,-1.5cm);
            \end{tikzpicture}
        \end{center}

        \item We let $N_t^{\mathcal N}:=\left|\mathcal N_t\right|$ and $N_t^{\mathcal S}:=\left|\mathcal S_t\right|$, so that
            \[N_t=N_{t}^{\mathcal N}+N_{t}^{\mathcal S}.\]

        \item \label{item::nsclassicc}Each newly born individual $I\in\mathcal N_{t+1}$ has two parents from $\mathcal P_t$. At a given locus, the individuals in $\mathcal N_{t+1}$ assemble a total of $2N_{t+1}^{\mathcal N}$ alleles and we may again assign to each individual $I\in\mathcal P_t$ the fraction $c_{I,t}$ of these alleles that were inherited from $I$. As in the situation with discrete generations, this gives rise to a vector $\mathbf c_t=(c_{I,t})_{I\in\mathcal P_t}\in\mathbb R_{\geq0}^{\mathcal P_t}$ with
        \begin{equation}\label{eq::ccc}
            \mathbf 1_t^{\top}\mathbf c_t=1.
        \end{equation}
    \end{enumerate}
\end{Not}

\begin{Rmk}
    \begin{enumerate}[label = (\roman*)]
        \item When developing the theory for discrete generations, we had denoted the contribution of an individual $I\in\mathcal P_t$ to the next generation simply by $c_I$ (Notation~\ref{not::corig}). In case of overlapping generations, however, we have to add the index $t$ to the notation (i.\,e. $c_{I,t}$). Individual $I$ may also be alive at a different time $t'\neq t$ and then have a different contribution $c_{I,t'}\neq c_{I,t}$ to generation $\mathcal P_{t'+1}$.

        \item For the same reason, the additional index $t$ is also added in the following notations regarding survival and estimated breeding values. Note in particular that also the estimated breeding value of an individual will generally change over time (whereas the true breeding value remains constant).

        \item Equation~\ref{eq::ccc} describes the case of a monoecious population. In case of a diecious population and a separation $\mathcal P_t=\mathcal F_t\sqcup\mathcal M_t$ into females and males, we need to impose
            \[\left(\mathbf 1_t^{\mathcal F}\right)^{\top}\mathbf c_t^{\mathcal F}=\frac12\quad\text{and}\quad
              \left(\mathbf 1_t^{\mathcal M}\right)^{\top}\mathbf c_t^{\mathcal M}=\frac12,\]
        just as in Equations~\ref{eq::ceins} and~\ref{eq::czwei}.
    \end{enumerate}
\end{Rmk}

\begin{Not}\label{not::basc}
    \begin{enumerate}[label = (\roman*)]
        \item \label{item::basci}To each individual $I\in\mathcal P_t$ we assign the binary value $s_{I,t}\in\{0,1\}$, indicating if $I$ survives to generation $\mathcal P_{t+1}$:
            \[s_{I,t}:=\begin{cases}1&\text{if }I\in\mathcal P_{t+1},\\0&\text{otherwise}\end{cases}.\]
        This gives rise to a binary vector $\mathbf s_t\in\{0,1\}^{\mathcal P_t}\subseteq\mathbb R^{\mathcal P_t}$ with
        \begin{equation}
            \mathbf 1_t^{\top}\mathbf s_t=N_{t+1}^{\mathcal S}.
        \end{equation}

        \item Each individual $I\in\mathcal P_t$ has an estimated breeding value $\hat u_{I,t}$, resulting in a vector $\hat{\mathbf u}_{t}\in\mathbb R^{\mathcal P_t}$.
        We calculate the average breeding value as
        \begin{equation}
            \hat u_{\mathcal P_t}=\frac1{N_t}\mathbf 1_t^{\top}\hat{\mathbf u}_t.
        \end{equation}
    \end{enumerate}
\end{Not}

\begin{Rmk}\label{rmk::notgood}
    \begin{enumerate}[label = (\roman*)]
        \item The term \emph{survival} is not necessarily to be understood literally. For our purposes, it makes no difference whether an individual dies or becomes (irreversably) infertile. Also if breeding scheme restrictions only allow individuals up to a specific age to reproduce, all older individuals may be considered \emph{dead}.
        
        \item \label{item::notgood} We will soon discuss why this value $\hat u_{\mathcal P_t}$ may not be ideal in order to assess the overall genetic quality of $\mathcal P_t$ and how it could be replaced with a more refined value. Nevertheless, we will calculate, how this value $\hat u_{\mathcal P_t}$ is transported over the years as we think that these calculations are instructive.
    \end{enumerate}
\end{Rmk}

\begin{Thm}\label{thm::eusimple}
    We have
        \[\mathbb E\left[\hat u_{\mathcal P_{t+1}}\right]=\frac{\left(N_{t+1}^{\mathcal N}\mathbf c_t+\mathbf s_t\right)^{\top}\hat{\mathbf u}_t}{N_{t+1}}.\]
\end{Thm}

\begin{proof}
    In generation $\mathcal P_{t+1}$, we can calculate separate expected average estimated breeding values $\mathbb E\left[\hat{u}_{\mathcal N_{t+1}}\right]$ for the newly born individuals and $\mathbb E\left[\hat{u}_{\mathcal S_{t+1}}\right]$ for the survivor individuals. The total average $\mathbb E\left[\hat{u}_{\mathcal P_{t+1}}\right]$ will then be a weighted mean of these two values:
        \[\mathbb E\left[\hat{u}_{\mathcal P_{t+1}}\right] =
            \frac{N_{t+1}^{\mathcal N}\mathbb E\left[\hat{u}_{\mathcal N_{t+1}}\right]+N_{t+1}^{\mathcal S}\mathbb E\left[\hat{u}_{\mathcal S_{t+1}}\right]}{N_{t+1}}.\]
    The calculation of the expected average estimated breeding value of the newly born individuals of generation $\mathcal P_{t+1}$ is in complete analogy to the case of discrete generations and we have
        \[\mathbb E\left[\hat{u}_{\mathcal N_{t+1}}\right]=\mathbf c_t^{\top}\hat{\mathbf u}_t.\]
    We turn to $\mathbb E\left[\hat u_{\mathcal S_{t+1}}\right]$. If the binary survival vector $\mathbf s_t$ is known, the calculation is simple because it is precisely the surviving individuals from $\mathcal P_t$ that contribute to the group of survivors in $\mathcal P_{t+1}$, i.\,e.
        \[\hat u_{\mathcal S_{t+1}}=\frac1{N_{t+1}^{\mathcal S}}\mathbf s_t^{\top}\hat{\mathbf u}_t.\]
    Putting all formulas together and simplifying finally yields the assertion.
\end{proof}

\begin{Rmk}
    \begin{enumerate}[label = (\roman*)]
        \item In practice, selection decisions to produce a generation $\mathcal P_{t+1}$ from generation $\mathcal P_t$ are made at some time $t^{\ast}\in[t,t+1]$. If $t^{\ast}$ is close to $t+1$, one can already be (relatively) certain if an individual $I\in\mathcal P_t$ will also be in $\mathcal P_{t+1}$. If, however, $t^{\ast}$ is close to $t$, one has to make an assumption which individuals of $\mathcal P_t$ one still expects to be a part of $\mathcal P_{t+1}$. The optimum contributions determined by OCS will then only be optimal if the assumptions turn out to be true.

        \item In mammals with long gestation times and juvenile phases (i.\,e. long generation intervals), selection decisions are indeed made early and the question on what to assume for the survival vector $\mathbf s_t$ is highly relevant. With a very coarse approach, one could simply assume  that all individuals below a certain age threshold will survive while all individuals above it will die. But, of course, breeders may fine-tune this approach by including, for example, the health status of individuals into the assumption of $\mathbf s_t$.
        
        \item Alternatively, one can interpret $\mathbf s_t$ as a random vector of Bernoulli-distributed variables, meaning that each individual $I\in\mathcal P_t$ is attributed with a probability $p_{I,t}$ to survive and consequently a probability $1-p_{I,t}$ to die. The derivations by \citet{wellmann19key} essentially follow this approach with equal survival probabilities for all members of a sex\,$\times$\,year class. It results in a vector $\mathbf p_{t}$ of survival probabilities in generation $\mathcal P_t$ with
            \[\mathbf p_{t}=\mathbb E\left[\mathbf s_t\right].\]

        \item It appears plausible to assume the individual random variables $s_{I}$ to be mutually independent, so that
            \[\mathrm{var}(\mathbf s_t)=\mathrm{diag}\left(p_{I,t}(1-p_{I,t})\right)_{I\in\mathcal P_t}.\]

        \item One should note that by turning $\mathbf s_t$ into a random vector, also the surviving population size $N_{t+1}^{\mathcal S}=\mathbf 1_{t}^{\top}\mathbf s_t$ and thus the total population size $N_{t+1}=N_{t+1}^{\mathcal N}+N_{t+1}^{\mathcal S}$ become random variables. By linearity of the expectation, we have
            \[\mathbb E\left[N_{t+1}^{\mathcal S}\right]=\mathbf 1_t^{\top}\mathbf p_{t}\]
        and
            \[\mathbb E\left[N_{t+1}\right]=N_{t+1}^{\mathcal N}+\mathbf 1_t^{\top}\mathbf p_{t}.\]

        \item In honeybees, which are the main target of this manuscript, $t^{\ast}$ will turn out to be very close to $t+1$. Queen selection decisions are usually made in early spring \citep{buchler24} and the new generation is ready after a few weeks. Since winter is the main season for honeybee queens to die \citep{bruckner23, gray23, tang23}, it is reasonable to assume that at the time of selection decision in honeybees, it is already known which queens from former years will still be around in the current year.

        \item While it is possible to derive (or at least approximate) the relevant formulas for OCS with survival probabilities $p_{I,t}$, we will not pursue this path here. The ultimate goal of this manuscript is to develop a working theory of OCS for honeybees for which $\mathbf s_t$ can safely be assumed as known at the time of selection decisions. Thus, we do not want to unnecessarily get too much diverted.

        Therefore, we will stick with the assumption that we know which individuals survive and thus with the formula from Theorem~\ref{thm::eusimple}:
            \[\mathbb E\left[\hat u_{\mathcal P_{t+1}}\right]=\frac{\left(N_{t+1}^{\mathcal N}\mathbf c_t+\mathbf s_t\right)^{\top}\hat{\mathbf u}_t}{N_{t+1}}.\]
    \end{enumerate}
\end{Rmk}

However, we had mentioned in Remark~\ref{rmk::notgood}\,\ref{item::notgood} that $\hat u_{\mathcal P_t}$ is not necessarily the best measurement for the genetic quality of a population. We illustrate this with a little example:

\begin{Ex}
    Assume a population $\mathcal P$, where $\mathcal P_t=\{A,B,C\}$ consists of three individuals. While $A\in\mathcal N_{t}$, $B$ and $C\in\mathcal S_t$ are old already, both being born at a time $t'\ll t$. In $\mathcal P_{t+1}$, we have again three individuals: While $A$ and $B$ survive to be in $\mathcal S_{t+1}$, $C$ dies but is replaced by a common child $D$ of $B$ and $C$.
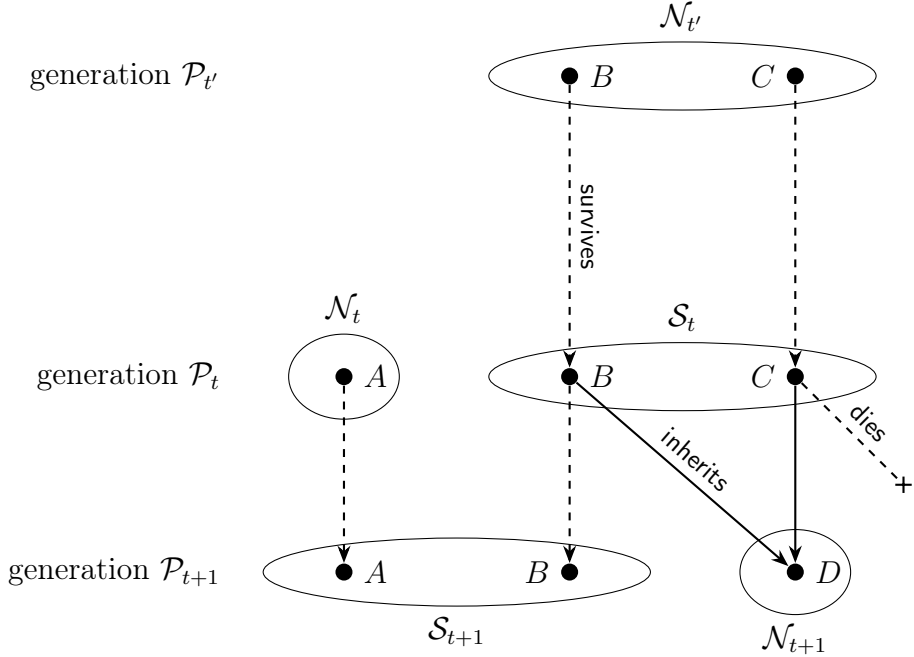
    \begin{center}
        \begin{tikzpicture}
            \path (0,0) coordinate (Pts)
                node[left = 0mm of Pts] {generation $\mathcal P_{t'}$}
                    coordinate[below = 4cm of Pts] (Pt0)
                    node[left = 0mm of Pt0] {generation $\mathcal P_{t}$}
                    coordinate[below = 2.6cm of Pt0] (Pt1)
                    node[left = 0mm of Pt1] {generation $\mathcal P_{t+1}$}

                    coordinate[right = 15mm of Pts] (Ats)
                    node[right = 3cm of Ats, individual] (Bts) {}
                        node[right = 0cm of Bts] {$B$}
                    node[right = 3cm of Bts.center, individual] (Cts) {}
                        node[left = 0cm of Cts] {$C$}

                    node[right = 15mm of Pt0, individual] (At0) {}
                        node[right = 0cm of At0] {$A$}
                    node[right = 3cm of At0.center, individual] (Bt0) {}
                        node[right = 0cm of Bt0] {$B$}
                    node[right = 3cm of Bt0.center, individual] (Ct0) {}
                        node[left = 0cm of Ct0] {$C$}

                    node[right = 15mm of Pt1, individual] (At1) {}
                        node[right = 0cm of At1] {$A$}
                    node[right = 3cm of At1.center, individual] (Bt1) {}
                        node[left = 0cm of Bt1] {$B$}
                    node[right = 3cm of Bt1.center, individual] (Dt1) {}
                        node[right = 0cm of Dt1] {$D$};

                \draw[survival] (Bts) -- (Bt0)
                    node[gene pass description] {survives};
                \draw[survival] (Cts) -- (Ct0);
                \draw[survival] (At0) -- (At1);
                \draw[survival] (Bt0) -- (Bt1);
                \draw[death] (Ct0) -- ++(1.5cm,-1.5cm)
                    node[gene pass description] {dies};
                \draw[inheritance] (Bt0) -- (Dt1)
                    node[gene pass description] {inherits};
                \draw[inheritance] (Ct0) -- (Dt1);

                \begin{scope}[on background layer]
                    \path node [fit=(Bts) (Cts), ellipse, inner sep = 2mm, draw] (Nts) {}
                            node [above = 0mm of Nts] {$\mathcal N_{t'}$}
                        node [fit=(At0.east) (At0.west), ellipse, inner sep = 4mm, draw] (Nt0) {}
                            node [above = 0mm of Nt0] {$\mathcal N_{t}$}
                        node [fit=(Bt0) (Ct0), ellipse, inner sep = 2mm, draw] (St0) {}
                            node [above = 0mm of St0] {$\mathcal S_{t}$}
                        node [fit=(At1) (Bt1), ellipse, inner sep = 2mm, draw] (St1) {}
                            node [below = 0mm of St1] {$\mathcal S_{t+1}$}
                        node [fit=(Dt1.east) (Dt1.west), ellipse, inner sep = 4mm, draw] (Nt1) {}
                            node [below = 0mm of Nt1] {$\mathcal N_{t+1}$};
                \end{scope}
        \end{tikzpicture}
    \end{center}

    If we look at $\hat{u}_{\mathcal P_{t+1}}$, we calculate it as
        \[\hat{u}_{\mathcal P_{t+1}}=\frac{\hat{u}_{A}+\hat{u}_{B}+\hat{u}_{D}}{3},\]
    meaning that all three individuals $A$, $B$, and $D$ contribute equally towards $\hat{u}_{\mathcal P_{t+1}}$. But is this justified? Individual $B$ is already very old and likely to die soon. So it will no longer exercise a great influence on future generations. In contrast, individuals $A$ and $D$ may still live for many years and can have numerous offspring. So when we want to assess the total value for breeding in the population, we should use a weighted mean of the individuals' breeding values, the weights representing the number of offspring they are still expected to have in their remaining lifespan.
\end{Ex}

\begin{Not}
    We may attribute to each individual $I\in\mathcal P_t$ a weight $w_{I,t}$ with which it should contribute to the average breeding value, giving rise to a vector $\mathbf w_t\in\mathbb R^{\mathcal P_t}$. We then may consider
        \[\hat{u}_{\mathcal P_t,\mathbf w_t}:=\mathbf w_t^{\top}\hat{\mathbf u}_t\]
    as the relevant value to maximize over time.
\end{Not}

\begin{Rmk}
    The question arises, how these weights $w_{I,t}$ can be determined in practice. Typically, weights are chosen for the different $\text{age}\times\text{sex}$ classes, and all individuals with the same age and sex receive the same weight \citep{wellmann19key}. Older individuals receive lower weights than younger ones. While this assumption is practical, it is not entirely necessary and other weights may very well be given. For the purpose of maximizing the expected value of $\hat{u}_{\mathcal P_t,\mathbf w_t}$, it is, however, necessary that the weights for the individuals in $\mathcal P_{t+1}$ are already available at selection time $t^{\ast}<t+1$. Thus, all newly born individuals in $\mathcal N_{t+1}$ should receive the same weight $w_{\mathcal N_{t+1}}$ and for all individuals alive at time $t$ it should be known what their weights will be at time $t+1$ in case they survive. Evidently, non-surviving individuals should not play a role in the calculation of the average breeding value at time $t+1$, so $s_{I,t}=0$ should imply $w_{I,t}=0$.
\end{Rmk}

The expected weighted average breeding value can be calculated as follows:

\begin{Thm}\label{thm::euw}
    We have
        \[\mathbb E\left[\hat u_{\mathcal P_{t+1},\mathbf w_{t+1}}\right]=\left(w_{\mathcal N_{t+1}}N_{t+1}^{\mathcal N}\mathbf c_t+\mathbf w_{t+1}\right)^{\top}\hat{\mathbf u}_t.\]
\end{Thm}

\begin{proof}
    This follows in complete analogy to Theorem~\ref{thm::eusimple}.
\end{proof}

\begin{Rmk}
    Ultimately, it is not really important whether one works with $\hat{u}_{\mathcal P_t}$ or with $\hat{u}_{\mathcal P_t,\mathbf w_t}$. Each time, the goal will be to maximize the expected (weighted) average breeding value in $\mathcal P_{t+1}$ via an optimal choice of the vector $\mathbf c_t$, which in both cases is achieved by maximizing the product $\mathbf c_t^{\top}\hat{\mathbf u}_t$. This is also the case if one works with survival probabilities $p_{I,t}$.
\end{Rmk}

We turn to the question of average kinship. Also here, one may think of weighted averages but we will consider simply the value $k_{\mathcal P_t,\mathcal P_t}=\frac{1}{N_t^2}\mathbf 1_t^{\top}\mathbf K_t\mathbf 1_t$ which we will want to restrict.
Thus, the main remaining question is how to predict $k_{\mathcal P_{t+1}}$ in dependence of the vectors $\mathbf c_t$ and $\mathbf s_t$ of contributions from individuals in $\mathcal P_t$.

\begin{Not}
    With suitable ordering of the individuals in $\mathcal P_{t+1}$, the kinship matrix $\mathbf K_{t+1}$ can be decomposed into four blocks
        \[\mathbf K_{t+1}=\begin{pmatrix}\mathbf K_{t+1}^{\mathcal N\mathcal N} & \mathbf K_{t+1}^{\mathcal N\mathcal S}\\
                                         \mathbf K_{t+1}^{\mathcal S\mathcal N} & \mathbf K_{t+1}^{\mathcal S\mathcal S}\end{pmatrix},\]
    where $\mathbf K_{t+1}^{\mathcal N\mathcal N}$ and $\mathbf K_{t+1}^{\mathcal S\mathcal S}$ denote the kinships between newly born and between surviving individuals, respectively, and $\mathbf K_{t+1}^{\mathcal N\mathcal S}=\left(\mathbf K_{t+1}^{\mathcal S\mathcal N}\right)^{\top}$ denote the kinships between newly born and surviving individuals.
\end{Not}

\begin{Rmk}
    By that, we have
    \begin{equation}\label{eq::olk}
        k_{\mathcal P_{t+1},\mathcal P_{t+1}}=\frac{1}{N_{t+1}^2}\left(\left(N_{t+1}^{\mathcal N}\right)^2k_{\mathcal N_{t+1},\mathcal N_{t+1}}
            + 2N_{t+1}^{\mathcal N}N_{t+1}^{\mathcal S}k_{\mathcal N_{t+1},\mathcal S_{t+1}}+\left(N_{t+1}^{\mathcal S}\right)^2k_{\mathcal S_{t+1},\mathcal S_{t+1}}\right).
    \end{equation}
\end{Rmk}

We will proceed by separate calculations of the terms $k_{\mathcal N_{t+1},\mathcal N_{t+1}}$, $k_{\mathcal N_{t+1},\mathcal S_{t+1}}$, and $k_{\mathcal S_{t+1},\mathcal S_{t+1}}$.

\begin{Lem}\label{lem::olk}
    We have
    \begin{align}
        k_{\mathcal N_{t+1},\mathcal N_{t+1}} &= \mathbf c_t^{\top}\mathbf K_t\mathbf c_t-\frac{\mathbf c_{t}^{\top}\mathbf f_t}{4N_{t+1}^{\mathcal N}}+\frac{1}{4N_{t+1}^{\mathcal N}},
            \label{eq::olnn}\\
        k_{\mathcal N_{t+1},\mathcal S_{t+1}} &= \frac{1}{N_{t+1}^{\mathcal S}}\mathbf c_t^{\top}\mathbf K_t\mathbf s_t,
            \label{eq::olns}\\
        k_{\mathcal S_{t+1},\mathcal S_{t+1}} &= \frac1{\left(N_{t+1}^{\mathcal S}\right)^2}\mathbf s_t^{\top}\mathbf K_t\mathbf s_t.
            \label{eq::olss}
    \end{align}
\end{Lem}

\begin{proof}
    \begin{enumerate}[label = (\roman*)]
        \item Regarding Equation~\ref{eq::olnn}, we are precisely in the same situation as in the case of non-overlapping generations: a group of older individuals generates a cohort of new individuals. Thus, the assertion follows with the same arguments as in Lemma~\ref{lemma::avkin}.

        \item We turn to the average kinship $k_{\mathcal N_{t+1},\mathcal S_{t+1}}$ between newly born individuals and survivors in generation $\mathcal P_{t+1}$. We derive $k_{\mathcal N_{t+1},\mathcal S_{t+1}}$ via an allele drawing process. If we fix a locus, generation $\mathcal P_{t+1}$ assembles a total of $2N_{t+1}$ alleles, of which $2N_{t+1}^{\mathcal N}$ are in $\mathcal N_{t+1}$ and $2N_{t+1}^{\mathcal S}$ are in $\mathcal S_{t+1}$. The desired value $k_{\mathcal N_{t+1},\mathcal S_{t+1}}$ is the probability to end up with ibd alleles when drawing one allele from $\mathcal N_{t+1}$ and one allele from $\mathcal S_{t+1}$.

        Let $A^{\mathcal N}$ be the allele drawn from $\mathcal N_{t+1}$ and $A^{\mathcal S}$ be the allele from $\mathcal S_{t+1}$. Because the population is closed, (copies of) both alleles were already present in $\mathcal P_t$. From there, $A^{\mathcal N}$ was inherited and $A^{\mathcal S}$ simply stayed in the same individual. We fix two individuals $I,J\in\mathcal P_t$ (possibly identical). The probability that $A^{\mathcal N}$ was inherited from $I$ is $c_{I,t}$ because this is the fraction of the $2N_{t+1}^{\mathcal N}$ newly born alleles that was inherited from $I$. Next, we look whether $A^{\mathcal S}$ may have come from $J$. This is only possible, if $J$ survived to be in $\mathcal P_{t+1}$, i.\,e. if $s_{J,t}=1$. In that case, $J$ is responsible for two of the $2N_{t+1}^{\mathcal S}$ alleles of the survivor part of the population. So, in total, the probability that $A^{\mathcal S}$ came from $J$ is $\frac{s_{J,t}}{N_{t+1}^{\mathcal S}}$. If $A^{\mathcal N}$ and $A^{\mathcal S}$ did indeed come from $I$ and $J$, respectively, they are just two randomly drawn alleles from these individuals, so their probability to be ibd is $k_{I,J}$. Thus, the probability that $A^{\mathcal N}$ and $A^{\mathcal S}$ came from $I$ and $J$ and are ibd is $\frac{c_{I,t}k_{I,J}s_{J,t}}{N_{t+1}^{\mathcal S}}$. Summing over all possible choices for $I$ and $J$ yields
            \[k_{\mathcal N_{t+1},\mathcal S_{t+1}}
                = \sum_{I,J\in\mathcal P_t}\frac{c_{I,t}k_{I,J}s_{J,t}}{N_{t+1}^{\mathcal S}}
                = \frac{\mathbf c_t^{\top}\mathbf K_t\mathbf s_t}{N_{t+1}^{\mathcal S}}.\]

        \item \label{item::olkiii} We come to $k_{\mathcal S_{t+1},\mathcal S_{t+1}}$. Kinships between surviving individuals do not change over the years. Therefore, $k_{\mathcal S_{t+1}\mathcal S_{t+1}}$, the average of kinships of survivors in $\mathcal P_{t+1}$, equals the average of kinships between those individuals in $\mathcal P_t$ that survive to the next year,
        i.\,e. $\frac1{\left(N_{t+1}^{\mathcal S}\right)^2}\mathbf s_t^{\top}\mathbf K_t\mathbf s_t$.
    \end{enumerate}
\end{proof}

In total, this gives us
\begin{Thm}
    \[k_{\mathcal P_{t+1},\mathcal P_{t+1}}=\frac{\left(N_{t+1}^{\mathcal N}\right)^2}{N_{t+1}^2}\mathbf c_t^{\top}\mathbf K_t\mathbf c_t
        + \frac{N_{t+1}^{\mathcal N}}{4N_{t+1}^2}\mathbf c_t^{\top}(8\mathbf K_t\mathbf s_t-\mathbf f_t)
        + \frac{\mathbf s_t^{\top}\mathbf K_t\mathbf s_t+N_{t+1}^{\mathcal N}}{4N_{t+1}^2}\]
\end{Thm}

\begin{proof}
    This follows by inserting the results of Lemma~\ref{lem::olk} into Equation~\ref{eq::olk} and simplifying.
\end{proof}

We may now formulate the task of OCS with overlapping generations.

\begin{Task}
    \begin{enumerate}[label = (\roman*)]
        \item\label{item::casei} Given a generation $\mathcal P_t$ of monoecious individuals and
        \begin{itemize}
            \item a vector $\hat{\mathbf u}_t\in\mathbb R^{\mathcal P_t}$ of estimated breeding values,

            \item a survival vector $\mathbf s_t\in\{0,1\}^{\mathcal P_t}$,

            \item a symmetric and positive definite kinship matrix $\mathbf K_t\in\mathbb R^{\mathcal P_t\times\mathcal P_t}$,

            \item the required number of newly created individuals of the next generation, $N_{t+1}^{\mathcal N}$,

            \item and a maximum acceptable kinship level $k_{t+1}^{\ast}$,
        \end{itemize}
        let $N_{t+1}:=N_{t+1}^{\mathcal N}+\mathbf 1_t^{\top}\mathbf s_t$ and maximize the function
            \[\mathbb E\left[\hat u_{\mathcal P_{t+1}}\right]:\mathbb R_{\geq0}^{\mathcal P_t}\to\mathbb R,\quad
            \mathbf c_{t}\mapsto \frac{\left(N_{t+1}^{\mathcal N}\mathbf c_t+\mathbf s_t\right)^{\top}\hat{\mathbf u}_t}{N_{t+1}}\]
        under the constraints
            \[\mathbf 1_t^{\top}\mathbf c_t=1\]
        and
            \[\frac{\left(N_{t+1}^{\mathcal N}\right)^2}{N_{t+1}^2}\mathbf c_t^{\top}\mathbf K_t\mathbf c_t
        + \frac{N_{t+1}^{\mathcal N}}{4N_{t+1}^2}\mathbf c_t^{\top}(8\mathbf K_t\mathbf s_t-\mathbf f_t)
        + \frac{\mathbf s_t^{\top}\mathbf K_t\mathbf s_t+N_{t+1}^{\mathcal N}}{4N_{t+1}^2} \leq k_{t+1}^{\ast}\]

        \item Given a sex-divided generation $\mathcal P_t=\mathcal F_t\sqcup\mathcal M_t$ of diecious individuals and the other values as in \ref{item::casei} the task remains the same, only the condition on $\mathbf c_t\cong\mathbf c_t^{\mathcal F}\oplus\mathbf c_t^{\mathcal M}$ has to be changed to
            \[\left(\mathbf 1_t^{\mathcal F}\right)^{\top}\mathbf c_t^{\mathcal F}=\frac12\]
        and
            \[\left(\mathbf 1_t^{\mathcal M}\right)^{\top}\mathbf c_t^{\mathcal M}=\frac12.\]
    \end{enumerate}
\end{Task}

This time, we abstain from writing everything in terms of female and male components.

\begin{Rmk}
    Because $N_{t+1}^{\mathcal N}$, $N_{t+1}$, $\mathbf s_t$, and $\hat{\mathbf u}_t$ are known constants, the function
        \[\mathbb E\left[\hat u_{\mathcal P_{t+1}}\right]:\mathbb R_{\geq0}^{\mathcal P_t}\to\mathbb R,\quad
            \mathbf c_{t}\mapsto \frac{\left(N_{t+1}^{\mathcal N}\mathbf c_t+\mathbf s_t\right)^{\top}\hat{\mathbf u}_t}{N_{t+1}}\]
    is maximized precisely when the simpler function
        \[\mathbf c_t\mapsto\mathbf c_t^{\top}\hat{\mathbf u}_t\]
    is maximized.
\end{Rmk}

\section{Honeybee peculiarities}\label{sec::hbpec}
\subsection{Reproductive biology of honeybees}

In most farm animals, performances are attributed to individuals. Each individual dairy cow has a lactation yield, each individual piglet has a weaning weight. In honeybees, however, phenotypes are generally only measured on the level of colonies. Examples for important breeding traits in honeybees are seasonal honey yield, gentleness, or resistance against the parasite \emph{Varroa destructor} \citep{buchler24}. It is generally not recorded how much a single bee contributed towards the honey yield or if an individual worker was aggressive. Instead, these traits are attributed to the colony as a whole. Honeybee colonies consist of a single queen and several thousands of worker bees, which are daughters of the queen. This means that (at least in the ideal type) all workers of a colony are sisters.
Worker bees are generally infertile, making the queen the only egg-laying individual of the colony.

At first glance, it may seem reasonable to assume that it is the worker bees that are mainly responsible for economically interesting traits. Workers collect the nectar, workers sting (or exhibit gentle behavior) and workers perform defense strategies against parasites. However, this assumption is too short-sighted as several studies have shown a strong influence of the queen on many traits \citep{bienefeld90, brascamp16, hoppe20}. However, the queens and workers contribute in different ways. For example, by her egg-laying frequency, the queen can influence the number of workers in the colony and thus affect honey yield because more workers can collect more nectar. Also, by pheromone release, the queen can orchestrate worker behavior and thereby also have an influence on behavior traits like gentleness \citep{gervan05}.

Male offspring of a queen are called drones. While drones do not play a (known) role in the performance regarding breeding traits, their purpose lies in reproduction. Shortly after hatching, a new queen leaves the hive to perform a nuptial flight, during which she mates with multiple drones from other colonies of the broader vicinity. The queen stores the drones' semen in her spermatheca and uses it for the remainder of her life (typically a few years) to fertilize eggs. Fertilized eggs develop into female bees, i.\,e. mainly workers. If a female larva is fed with a specific diet, it can also develop into a daughter queen. This means that genetically, there is no difference between queens and workers. Drones, however, develop from unfertilized eggs and are therefore haploid, in contrast to the diploid workers and queens.

The mating flights of queens pose a tough challenge to honeybee breeders. Typically, it is not observable, where the specific drones a queen mates with come from. By this, no paternal pedigree information is available and there is no guarantee that the mating partners provide desirable genetic properties. In many instances, honeybee breeding therefore works only with selection on the maternal side \citep{andonov19,pernal12, bigio14the}. However, there are strategies to gain at least a certain degree of control over the paternal inheritance path. The two most common of these strategies are isolated mating stations and instrumental insemination. Computer simulation studies have shown that both these strategies lead to much greater genetic response than breeding strategies that rely on free mating \citep{plate19the, du21a, du23}.

Isolated mating stations are established in geographically secluded areas, where one can (more or less) guarantee the absence of (unwanted) honeybee hives. There, a number of colonies with favorable genes is placed for drone production. The queens heading these colonies are often called DPQs, short for \emph{drone producing queens}. When a virgin queen is brought to such a place for her nuptial flight, the only drones she can mate with are those produced by the DPQs. For a daughter of a thus mated queen, it can be concluded that the father drones comes from one of the DPQs, while the particular origin remains unknown. It is, however, a common practice to let all DPQs of a mating station be sisters, i.\,e. daughters of a single queen. By doing so, all drones of the mating station share a common grand-dam, which for historic reasons is called the 4a-queen of the mating station \citep{uzunov22initiation, druml23}. 4a-queens are usually selected with great rigor, to ensure that all queens mating on a mating station will be equipped with excellent genetic material.

Instrumental insemination provides the breeder with even greater control over the fertilization process, because the drones can be chosen individually. It is possible to use a single drone for the insemination of a queen \citep{harbo99the}. Then, for all daughters of a thus inseminated queen it is known from which specific drone they inherited their paternal genes. However, the amount of sperm produced by a single drone is insufficient to let the queen develop full-sized colonies and typically single drone inseminated colonies do not survive their first winter. Instead, a strategy that comes with fewer problems is to inseminate queens with several drones from the same colony \citep{du24the}. By that, for an offspring queen it is still unclear who her father drone is but the dam of the drones is known (and not just the grand-dam as in the mating station case).

Simulation studies have shown that breeding schemes with instrumental insemination often generate higher genetic gain than breeding schemes with isolated mating stations. However, they also come with an increased risk of inbreeding \citep{du23}.

\subsection{Quantitative genetics}

These peculiarities in the reproductive biology of honeybees require a number of adaptations in the general quantitative genetic theory of breeding. These adaptations will be explained in the following.

\subsubsection{Breeding values}

The breeding value of an individual in the infinitesimal model is usually defined as the sum of infinitely many infinitesimally small allele effects \citep{lynch98}. However, because in honeybees, most traits are commonly affected by the queen and the worker group in different ways, the same allele may have different effects on the trait depending on whether it is expressed in a queen or in a worker. Thus, each allele is equipped with two allele effects -- a queen effect and a worker effect -- and consequently each individual bee $B$ has two (true) breeding values (per trait), namely the queen effect breeding value $u_B^{\text{queen eff.}}$ and the worker effect breeding value $u_B^{\text{worker eff.}}$ \citep{bienefeld90}.

\begin{Def}
    The \emph{total} (true) breeding value of a bee $B$ is defined as
        \[u_B:=u_B^{\text{queen eff.}}+u_B^{\text{worker eff.}}.\]
\end{Def}

\begin{Rmk}
    For the remainder of this text, the individual queen effect and worker effect breeding values will not play a role and all breeding values are meant to be total breeding values.
\end{Rmk}

For groups of bees of the same ploidy (so no mixed groups with queens and drones), it makes sense to define the breeding value of a group, very much like we have defined estimated breeding values of groups of individuals earlier.

\begin{Def}\label{def::mathcalb}
    Let $\mathcal B$ be an all-female or all-male finite group of bees. Then the breeding value of $\mathcal B$ is defined as
    \begin{equation}
        u_{\mathcal B} := \frac1{|\mathcal B|}\sum_{B\in\mathcal B}u_{B}.
    \end{equation}
\end{Def}

\begin{Rmk}
    In this way, it is possible to define the breeding value of a worker group $\mathcal W$ or of the group $\mathcal M$ of DPQs on a mating station. In the literature \citep{brascamp14methods, du21shortterm}, such breeding values of groups are often equipped with a bar to indicate that they are calculated as averages (i.\,e. $\bar u_{\mathcal B}$ instead of $u_{\mathcal B}$). We abstain from this practice to yield simpler notation. Instead, we remind the reader that a (lower) index in calligraphic font usually means that averages are taken (cf. Remark~\ref{rmk::nobars}\,\ref{item::nobars}). Some notation using a bar will occur much later in this manuscript (Notation~\ref{not::bar}).
\end{Rmk}

Also like earlier, we may extend the notion of breeding values to finite sets of groups of bees:

\begin{Def} \label{def::mathfrakb}
    Let $\mathfrak B=\left\{\mathcal B_1,...,\mathcal B_{|\mathfrak B|}\right\}$ be a finite set of groups of bees so that all bees in $\tilde{\mathcal B}:=\bigcup_{\mathcal B\in\mathfrak B}\mathcal B$ have the same sex. Then the breeding value of $\mathfrak B$ is defined as
    \begin{equation}
        u_{\mathfrak B}=\frac1{|\mathfrak B|}\sum_{B\in\mathfrak B}u_{\mathcal B}.
    \end{equation}
\end{Def}

\begin{Rmk}
    For example, it is possible to interpret a colony, consisting of a queen $Q$ and her worker group $\mathcal W$, as the two-elemented set
        \[\mathfrak C=\{\{Q\},\mathcal W\}.\]
    With this definition, we have
        \[u_{\mathfrak C}=\frac{u_Q+u_{\mathcal W}}{2}.\]
\end{Rmk}

\begin{Rmk}\label{rmk::inherit}
    We should note an important difference between queens and drones. Whenever a queen passes on her own genes to an offspring, she passes half of her alleles. Thus, the offspring is expected to inherit half of the queen's breeding value, disturbed by some Mendelian sampling with expectation zero. If, however, a drone passes on his genes, he will give all of his alleles, so the passed breeding value is precisely the drone's own breeding value, without any Mendelian sampling.
\end{Rmk}

The following rules of inheritance for honeybee breeding values can essentially also be found in \citep{du21shortterm} and \citep{kistler21}. They are straightforward consequences of the general rules of the infinitesimal model and Remark~\ref{rmk::inherit}

\begin{Lem}\label{lem::bvinherits}
    \begin{enumerate}[label = (\roman*)]
        \item \label{item::rrr} If a drone $D$ is the son of a queen $Q$, then
            \[u_D=\frac12u_Q+\phi_{Q,D},\]
        where $\phi_{Q,D}$ is a random normal variable with $\mathbb E\left[\phi_Q\right]=0$. In particular,
        \begin{equation}\label{eq::fqi}
            \mathbb E\left[u_D|u_Q\right]=\frac12u_Q.
        \end{equation}

        \item \label{item::qinherit} If a queen $R$ is the daughter of a queen $Q$ and a drone $D$, its breeding value is
            \[u_R=\frac12u_Q+u_D+\phi_{Q,R}\]
        with $\phi_{Q,R}$ as $\phi_{Q,D}$ in~\ref{item::rrr} and
        \begin{equation}
            \mathbb E\left[u_R|u_Q,u_D\right]=\frac12u_Q+u_D.
        \end{equation}
    \end{enumerate}
\end{Lem}

\begin{center}
    \begin{tikzpicture}
        \path (0,0) node (i) {(i)}
            node[below right =0cm and 5mm of i, queen] (Q1) {}
                node[right = 0cm of Q1] {$Q$}
            node[below = 2cm of Q1.center, drone] (D1) {}
                node[right = 0cm of D1] {$D,\quad\mathbb E\left[u_D\right]=\frac12u_Q$};

        \draw[inheritance] (Q1) -- (D1);

        \path node[right = 5cm of i] (ii) {(ii)}
            node[below right =0cm and 5mm of ii, queen] (Q2) {}
                node[left = 0cm of Q2] {$Q$}
            node[right = 2cm of Q2.center, drone] (D2) {}
                node[right = 0cm of D2] {$D$}
            node[below right = 2cm and 1 cm of Q2.center, queen] (R) {}
                node[right = 0cm of R] {$R,\quad\mathbb E\left[u_R\right]=\frac12u_Q+u_D$};

        \draw[inheritance] (Q2) -- (R);
        \draw[inheritance] (D2) -- (R);
    \end{tikzpicture}
\end{center}

Usually, one does not have information about the particular father drone of a queen but only has some information (or assumptions) on the group of drones that her dam mated with. If we assume that all drones of the group have the same chance to be the father drone, we obtain the following

\begin{Lem}\label{lem::bvrepq}
    If a queen $R$ is the daughter of a queen $Q$ who mated with a group of drones $\mathcal D$, its expected breeding value is
    \begin{equation}\label{eq::tqi}
        \mathbb E\left[u_R|u_Q,u_{\mathcal D}\right]=\frac12u_Q+u_{\mathcal D}.
    \end{equation}
\end{Lem}

\begin{center}
    \begin{tikzpicture}
        \path (0,0) node[queen] (Q1) {}
                node[above right = 0cm and 0cm of Q1] {$Q$}
            node[right = 3cm of Q1.center, group] (D1) {}
                -- (D1) pic {drones}
                node[right = 0cm of D1] {$\mathcal D$}
            node[below right = 2cm and 1.5 cm of Q1.center, queen] (R) {}
                node[right = 0cm of R]{$R,\quad\mathbb E\left[u_R\right]=\frac12u_Q+u_D$};

        \draw[inheritance] (Q1) -- (R);
        \draw[mating] (D1) -- (Q1)
            node[gene pass description] {mate};
    \end{tikzpicture}
\end{center}

In general, a group of drones $\mathcal D$ comes from a group $\mathcal M$ of queens rather than a single queen. Thus, we will also make use of the following Lemma, which assumes that all queens in $\mathcal M$ have an equal chance to contribute to $\mathcal D$.

\begin{Lem}\label{lem::bvdrgr}
    If a group of drones $\mathcal D$ was produced by a group $\mathcal M$ of queens, its expected breeding value is
    \begin{equation}\label{eq::bvdrgr}
        \mathbb E\left[u_\mathcal D|u_{\mathcal M}\right]=\frac12u_{\mathcal M}.
    \end{equation}
\end{Lem}

\begin{center}
    \begin{tikzpicture}
        \path (0,0) node[group] (Q1) {}
                -- (Q1.center) pic {queens}
                node[right = 0cm of Q1] {$\mathcal M$}
            node[below = 2cm of Q1.center, group] (D1) {}
                -- (D1) pic {drones}
                node[right = 0cm of D1] {$\mathcal D,\quad\mathbb E\left[u_\mathcal D\right]=\frac12u_{\mathcal M}$};

        \draw[inheritance] (Q1) -- (D1);
    \end{tikzpicture}
\end{center}

If we look at the inherited breeding value of a worker group $\mathcal W$ of a queen $Q$ who mated with a group of drones $\mathcal D$, we see that for each individual worker $W\in \mathcal W$ with drone father $D\in\mathcal D$, we have by Lemma~\ref{lem::bvinherits}\,\ref{item::qinherit}
    \[u_{W}=\frac12u_Q+u_{D}+\phi_{Q,W}.\]
Assuming that all drones contributed equally to the worker group and that the worker group is infinitely large, by the central limit theorem, this gives

\begin{Lem}\label{lem::bvwork}
    The breeding value of the worker group $\mathcal W$ of a queen $Q$ who mated with a group of drones $\mathcal D$ is
    \begin{equation}\label{eq::lqi}
        u_{\mathcal W}=\frac12u_Q+u_{\mathcal D}.
    \end{equation}
\end{Lem}

\begin{center}
    \begin{tikzpicture}
        \path (0,0) node[queen] (Q1) {}
                node[above right = 0cm and 0cm of Q1] {$Q$}
            node[right = 3cm of Q1.center, group] (D1) {}
                -- (D1) pic {drones}
                node[right = 0cm of D1] {$\mathcal D$}
            node[below right = 2cm and 1.5 cm of Q1.center, worker group] (W) {}
                node[right = 0cm of W]{$\mathcal W,\quad u_{\mathcal W}=\frac12u_Q+u_D$};

        \draw[inheritance] (Q1) -- (W);
        \draw[mating] (D1) -- (Q1)
            node[gene pass description] {mate};
    \end{tikzpicture}
\end{center}

\begin{Rmk}\label{rmk::rpl}
    \begin{enumerate}[label = (\roman*)]
        \item Note that there is no Mendelian sampling in the inheritance of breeding values to worker groups.

        \item \label{item::rplii} Furthermore, note, that the breeding value of a worker group is precisely the expected breeding value of a daughter queen.
    \end{enumerate}
\end{Rmk}

In reality, as for any breeding animal, we do not have access to the (true) breeding values of honeybees. But we may estimate them via the BLUP procedure described e.\,g. in \citep{bienefeld07, brascamp14methods}, assigning an estimated breeding value $\hat u_B$ to each bee $B$ of a population. The estimated breeding value $\hat u_{\mathcal B}$ of a (same-sex) group $\mathcal B$ of bees is defined as in Definition~\ref{def::mathcalb}.

A key property of BLUP-estimated breeding values is that they are unbiased (That is what the U in BLUP stands for). This means that the expectations of estimated breeding values and true breeding values coincide. Thus, Equations~\ref{eq::fqi} to~\ref{eq::lqi} still hold if we replace true breeding values $u$ by estimated breeding values $\hat u$. In particular, by combination of Equations~\ref{eq::tqi} and~\ref{eq::lqi}:

\begin{Lem}\label{lem::expdaughter}
    If we have a queen $Q$ with worker group $\mathcal W$, then for a daughter $R$ of $Q$, we have
        \[\mathbb E\left[\hat u_R\right]=\hat u_{\mathcal W}.\]
\end{Lem}

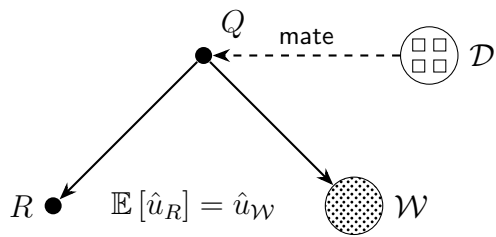
\begin{center}
    \begin{tikzpicture}
        \path (0,0) node[queen] (Q1) {}
                node[above right = 0cm and 0cm of Q1] {$Q$}
            node[right = 3cm of Q1.center, group] (D1) {}
                -- (D1) pic {drones}
                node[right = 0cm of D1] {$\mathcal D$}
            node[below right = 2cm and 2cm of Q1.center, worker group] (W) {}
                node[right = 0cm of W]{$\mathcal W$}
            node[below left = 2cm and 2 of Q1.center, queen] (R) {}
                node[left = 0cm of R]{$R$};

        \draw[inheritance] (Q1) -- (W);
        \draw[inheritance] (Q1) -- (R);
        \draw[mating] (D1) -- (Q1)
            node[gene pass description] {mate};

        \path (W) -- (R) node[midway]{$\mathbb E\left[\hat u_R\right]=\hat u_{\mathcal W}$};
    \end{tikzpicture}
\end{center}

\subsubsection{Kinships}

The concept of kinship can also be transported from diploid species to different ploidy levels. In polyploid species, different types of kinship coefficients can be defined \citep{gallais03}, but the most straightforward definition is the following:

\begin{Def}\label{def::polykin}
    Consider two individuals $I$ and $J$ with ploidies $p_I, p_J\in\mathbb N$. We fix a locus and draw randomly one of the $p_I$ alleles of $I$ and one of the $p_J$ alleles of $J$. As in the diploid case, the kinship between $I$ and $J$ is then defined as the probability to end up with ibd alleles.
\end{Def}

\begin{Rmk}
    \begin{enumerate}[label = (\roman*)]
        \item By Definition~\ref{def::polykin}, we can also consider kinships between queens and drones or between drones and drones.

        \item The concept of inbreeding for polyploids is more intricate than for diploids \citep{gallais03, kerr12}. Considering honeybees, for queens and workers the standard definition applies because they are diploid. The inbreeding coefficient $f_Q$ of a queen $Q$ is the probability of her two alleles at a random locus to be ibd (Definition~\ref{def::inbr}). For drones, it does not make sense to speak of inbreeding because they are hemizygous at every locus.
    \end{enumerate}
\end{Rmk}

As in diploid species, we may extend Definition~\ref{def::polykin} to define the kinship between two groups of bees.

\begin{Def}\label{def::grpkinbee}
    For two finite groups $\mathcal G$ and $\mathcal H$ of bees, we fix a locus and sample one of the alleles that are assembled at this locus in $\mathcal G$ and one of the alleles that are assembled at this locus in $\mathcal H$. The kinship $k_{\mathcal G,\mathcal H}$ between $\mathcal G$ and $\mathcal H$ is defined as the probability of the two sampled alleles to be ibd.
\end{Def}

\begin{Rmk}
    Typical groups of bees that are commonly used in quantitative genetic theory of honeybees are the group $\mathcal W$ of workers in a colony, the group $\mathcal D$ of drones that mated with a queen, and the group $\mathcal M$ of DPQs on an isolated mating station.
\end{Rmk}

We now consider the case that a queen $Q$ mated with a group $\mathcal D$ of drones and produced two disjoint groups of daughters (queens or workers), $\mathcal G$ and $\mathcal H$.

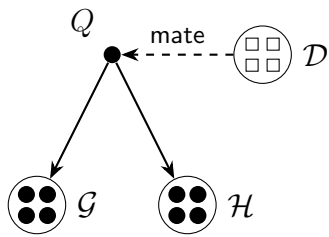
\begin{center}
    \begin{tikzpicture}
        \path (0,0) node[queen] (Q) {}
                node[above left = 0cm and 0cm of Q] {$Q$}
            node[right = 2cm of Q.center, group] (D) {}
                -- (D) pic {drones}
                node[right = 0cm of D] {$\mathcal D$}
            node[below left = 2cm and 1cm of Q.center, group] (G) {}
                -- (G) pic {queens}
                node[right = 0cm of G] {$\mathcal G$}
            node[below right = 2cm and 1cm of Q.center, group] (H) {}
                -- (H) pic {queens}
                node[right = 0cm of H] {$\mathcal H$};

        \draw[inheritance] (Q) -- (G);
        \draw[inheritance] (Q) -- (H);
        \draw[mating] (D) -- (Q)
            node[gene pass description] {mate};
    \end{tikzpicture}
\end{center}

Assume that we know the kinships $k_{Q,Q}$, $k_{Q,\mathcal D}$, and $k_{\mathcal D,\mathcal D}$, as well as the group sizes $|\mathcal G|$, and $|\mathcal H|$. What are the seven remaining kinships $k_{Q,\mathcal G}$, $k_{Q,\mathcal H}$, $k_{\mathcal D,\mathcal G}$, $k_{\mathcal D,\mathcal H}$, $k_{\mathcal G,\mathcal G}$, $k_{\mathcal G,\mathcal H}$, and $k_{\mathcal H,\mathcal H}$?

\begin{Lem}\label{lem::qkinsh}
    We have
    \begin{align}
        k_{Q,\mathcal G} = k_{Q,\mathcal H} &= \frac{k_{Q,Q}+k_{Q,\mathcal D}}{2}, \label{eq::ein}\\
        k_{\mathcal D,\mathcal G} = k_{\mathcal D,\mathcal H} &= \frac{k_{\mathcal D,\mathcal D}+k_{Q,\mathcal D}}{2}, \label{eq::zwe}\\
        k_{\mathcal G,\mathcal G} &= \frac{2+(|\mathcal G|-1)k_{Q,Q}+2\cdot|\mathcal G|\cdot k_{Q,\mathcal D}
                                        +(|\mathcal G|-1)k_{\mathcal D,\mathcal D}}{4\cdot|\mathcal G|},  \label {eq::dre}\\
        k_{\mathcal H,\mathcal H} &= \frac{2+(|\mathcal H|-1)k_{Q,Q}+2\cdot|\mathcal H|\cdot k_{Q,\mathcal D}
                                        +(|\mathcal H|-1)k_{\mathcal D,\mathcal D}}{4\cdot|\mathcal H|}, \label{eq::vie}\\
        k_{\mathcal G,\mathcal H} &= \frac14k_{Q,Q}+\frac12k_{Q,\mathcal D}+\frac14k_{\mathcal D,\mathcal D}. \label{eq::fuen}
    \end{align}
\end{Lem}

\begin{proof}
    \begin{enumerate}[label = (\roman*)]
        \item We first show Equation~\ref{eq::ein}, i.\,e.
            \[k_{Q,\mathcal G} = k_{Q,\mathcal H} = \frac{k_{Q,Q}+k_{Q,\mathcal D}}{2}.\]
        We fix a locus and draw a random allele $A^{\mathcal G}$ at this locus from the $2\cdot|\mathcal G|$ alleles of group $\mathcal G$. We further draw randomly one of the two alleles of $Q$ at the same locus and call it $A^Q$. With probability $\frac12$, allele $A^{\mathcal G}$ was inherited from $Q$ and with probability $\frac12$ it was inherited from one of the drones in $\mathcal D$. In the former case, $A^{\mathcal G}$ and $A^Q$ are two independent picks from $Q$'s two alleles and the probability of them to be ibd is $k_{Q,Q}$. In the latter case, $A^{\mathcal G}$ turns out to be a random allele picked from the drone group $\mathcal D$, while $A^Q$ is still a random allele from $Q$, so the chance for them to be ibd is $k_{Q,\mathcal D}$. This yields the assertion. The derivation for $k_{Q,\mathcal H}$ follows in complete analogy.

        \item We now show Equation~\ref{eq::zwe}, i.\,e.
            \[k_{\mathcal D,\mathcal G} = k_{\mathcal D,\mathcal H} = \frac{k_{\mathcal D,\mathcal D}+k_{Q,\mathcal D}}{2}.\]
        We fix a locus and draw a random allele $A^{\mathcal G}$ at this locus from the $2\cdot|\mathcal G|$ alleles of group $\mathcal G$. We further draw randomly one of the $N^{\mathcal D}$ alleles of $\mathcal D$ at the same locus and call it $A^{\mathcal D}$. With probability $\frac12$, allele $A^{\mathcal G}$ was inherited from one of the drones in $\mathcal D$ and with probability $\frac12$ it was inherited from $Q$. In the former case, $A^{\mathcal G}$ and $A^{\mathcal D}$ are two independent picks from $\mathcal D$'s alleles and the probability of them to be ibd is $k_{\mathcal D,\mathcal D}$. In the latter case, $A^{\mathcal G}$ turns out to be a random allele picked from $Q$, while $A^{\mathcal D}$ is still a random allele from $\mathcal D$, so the chance for them to be ibd is $k_{Q,\mathcal D}$. This yields the assertion. The derivation for $k_{\mathcal D,\mathcal H}$ follows in complete analogy.

        \item We now show Equation~\ref{eq::dre}, i.\,e.
            \[k_{\mathcal G,\mathcal G}=\frac{2+(|\mathcal G|-1)k_{Q,Q}+2\cdot|\mathcal G|\cdot k_{Q,\mathcal D}+(|\mathcal G|-1)k_{\mathcal D,\mathcal D}}{4\cdot|\mathcal G|}.\]
        We fix a locus and pick randomly (with replacement) two alleles from the $2\cdot|\mathcal G|$ alleles of $\mathcal G$. With probability $\frac{1}{2\cdot|\mathcal G|}$, we picked the very same allele twice, and the drawn alleles are surely ibd. With the remaining probability of $1-\frac1{2\cdot|\mathcal G|}$, we picked two different alleles, meaning that we are in a situation of drawing without replacement. In this case, the following consideration applies. Since $|\mathcal G|$ of the $2\cdot|\mathcal G|$ genes in $\mathcal G$ come from $Q$, the probability that the first drawn allele $A^{1}$ comes from $Q$ is $\frac{|\mathcal G|}{2\cdot|\mathcal G|}=\frac12$. If $A^1$ came indeed from $Q$, there are $2\cdot|\mathcal G|-1$ alleles left, $|\mathcal G|-1$ of which come from $Q$, so the probability that the second allele $A^2$ comes again from $Q$ is $\frac{|\mathcal G|-1}{2\cdot|\mathcal G|-1}$. If this is also the case, $A^1$ and $A^2$ are two random picks of alleles from $Q$ and their probability to be ibd is $k_{Q,Q}$. With similar considerations, the probability that the two alleles (which are not the very same allele) come from $Q$ and $\mathcal D$, respectively, is $\frac{|\mathcal G|}{2\cdot|\mathcal G|-1}$ and in that case their probability to be ibd is $k_{Q,\mathcal D}$. Finally, the probability that two alleles (not the very same) are both inherited from $\mathcal D$ and are ibd is $\frac{1}{2}\cdot\frac{|\mathcal G|-1}{2\cdot|\mathcal G|-1}k_{\mathcal D,\mathcal D}$. Putting all this together, we end up with
        \begin{align*}
            k_{\mathcal G,\mathcal G}
            =&\ \frac1{2\cdot|\mathcal G|}\\
                &\ \ + \frac{2|\mathcal G|-1}{2|\mathcal G|}\left(\frac12\cdot\frac{|\mathcal G|-1}{2|\mathcal G|-1}k_{Q,Q}
                + \frac{|\mathcal G|}{2|\mathcal G|-1}k_{Q,\mathcal D}
                + \frac12\cdot\frac{|\mathcal G|-1}{2|\mathcal G|-1}k_{\mathcal D,\mathcal D}\right)\\
            =&\ \frac{2+(|\mathcal G|-1)k_{Q,Q}+2\cdot|\mathcal G|k_{Q,\mathcal D}+(|\mathcal G|-1)k_{\mathcal D,\mathcal D}}
                {4\cdot|\mathcal G|}.
        \end{align*}
        The assertion for $k_{\mathcal H,\mathcal H}$ (Equation~\ref{eq::vie}) follows by replacing the variable $\mathcal G$ with $\mathcal H$ in all places.

        \item Finally, we show Equation~\ref{eq::fuen}, i.\,e.
            \[K_{\mathcal G,\mathcal H}=\frac14k_{Q,Q}+\frac12k_{Q,\mathcal D}+\frac14k_{\mathcal D,\mathcal D}.\]
        We fix a locus and draw an allele $A^{\mathcal G}$ from $\mathcal G$ and an allele $A^\mathcal H$ from $\mathcal H$. Since $\mathcal G$ and $\mathcal H$ are disjoint, these two draws are independent. Both alleles come with equal probability of $\frac12$ from $Q$ or from $\mathcal D$. Consequently, the probability that both alleles come from $Q$ is $\frac14$ and in this case they are ibd with probability $k_{Q,Q}$. With probability $\frac12$, one allele was inherited from $Q$ and one from $\mathcal D$, in which case their probability to be ibd is $k_{Q,\mathcal D}$. Finally, there is a chance of $\frac14$ that both alleles came from $\mathcal D$ and the ibd-probability then is $k_{\mathcal D,\mathcal D}$. In total, this gives us the assertion.
    \end{enumerate}
\end{proof}

\begin{Rmk}\label{rmk::withcol}
    \begin{enumerate}[label = (\roman*)]
        \item \label{item::withcoli} From Equation~\ref{eq::dre}, we can see directly that for a single daughter $R$ of $Q$ (i.\,e. $\mathcal G=\{R\}$, $|\mathcal G|=1$):
            \[k_{R,R}=\frac{1+k_{Q,\mathcal D}}{2},\]
        and consequently
            \[f_R=k_{Q,\mathcal D}.\]

        \item \label{item::withcolii} If, in contrast, we consider $\mathcal G=\mathcal W$ to be the worker group of $Q$, the cardinality $|\mathcal G|=|\mathcal W|$ becomes very large and we may approximate from Equation~\ref{eq::dre}
        \begin{align*}
            k_{\mathcal W,\mathcal W} &= \lim_{|\mathcal W|\to\infty}\frac{2+(|\mathcal W|-1)k_{Q,Q}+2\cdot|\mathcal W|\cdot k_{Q,\mathcal D}+(|\mathcal W|-1)k_{\mathcal D,\mathcal D}}{4\cdot|\mathcal W|}\\
            &=\frac14k_{Q,Q}+\frac12k_{Q,\mathcal D}+\frac14k_{\mathcal D,\mathcal D}.
        \end{align*}

        \item In combination with Equation~\ref{eq::fuen}, this shows that the kinship between two (non-identical) sister queens is the same as the kinship between one of the sisters and the worker group of their dam and as the kinship of the dam's worker group with itself.

        
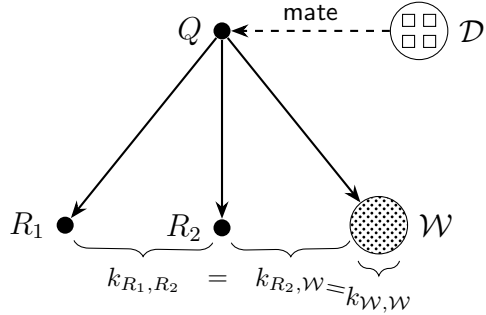
\begin{center}
            \begin{tikzpicture}
                \path (0,0) node[queen] (Q) {}
                        node[left = 0cm of Q] {$Q$}
                    node[right = 2.5cm of Q, group] (D) {}
                        node[right = 0cm of D] {$\mathcal D$}
                        -- (D) pic {drones}
                    node[below left = 2.5cm and 2cm of Q, queen] (R1) {}
                        node[left = 0cm of R1] {$R_1$}
                    node[below = 2.5cm of Q, queen] (R2) {}
                        node[left = 0cm of R2] {$R_2$}
                    node[below right = 2.5cm and 2cm of Q, worker group] (W) {}
                        node[right = 0cm of W] {$\mathcal W$};

                \draw[inheritance] (Q) -- (R1);
                \draw[inheritance] (Q) -- (R2);
                \draw[inheritance] (Q) -- (W);
                \draw[mating] (D) -- (Q)
                    node[gene pass description] {mate};
                \draw[relationship] (R1) -- (R2)
                    node[relationship description, font = \footnotesize] (a) {$k_{R_1,R_2}$};
                \draw[relationship] (R2) -- (W)
                    node[relationship description, font = \footnotesize] (b) {$k_{R_2,\mathcal W}$};
                \draw[relationship] (W.south west) -- (W.south east)
                    node[relationship description, font = \footnotesize] (c) {$k_{\mathcal W,\mathcal W}$};

                \path (a) -- (b)
                        node[midway, sloped, font = \footnotesize] {$=$}
                    (b) -- (c)
                        node[midway, sloped, font = \footnotesize] {$=$};
            \end{tikzpicture}
        \end{center}

        \item In the literature, kinship calculations between sister queens often perform a case distinction, if the sisters came from the same drone or not \citep{bienefeld89, brascamp14methods}. Working with $k_{\mathcal D,\mathcal D}$ as a given parameter, this case distinction is not necessary at this point.
    \end{enumerate}
\end{Rmk}

We want to calculate one further relevant property of kinships between groups of queens. In the situation of Lemma~\ref{lem::qkinsh}, let $\mathcal I$ be a further group of queens so that none of the queens in $\mathcal G$ and $\mathcal H$ is a direct ancestor of any of the queens in $\mathcal I$.

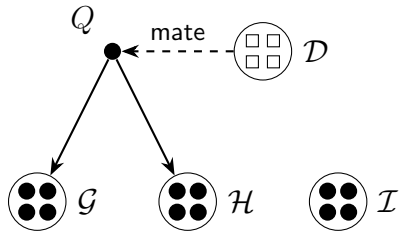
\begin{center}
    \begin{tikzpicture}
        \path (0,0) node[queen] (Q) {}
                node[above left = 0cm and 0cm of Q] {$Q$}
            node[right = 2cm of Q.center, group] (D) {}
                -- (D) pic {drones}
                node[right = 0cm of D] {$\mathcal D$}
            node[below left = 2cm and 1cm of Q.center, group] (G) {}
                -- (G) pic {queens}
                node[right = 0cm of G] {$\mathcal G$}
            node[below right = 2cm and 1cm of Q.center, group] (H) {}
                -- (H) pic {queens}
                node[right = 0cm of H] {$\mathcal H$}
            node[right = 2cm of H.center, group] (I) {}
                -- (I) pic {queens}
                node[right = 0cm of I] {$\mathcal I$};

        \draw[inheritance] (Q) -- (G);
        \draw[inheritance] (Q) -- (H);
        \draw[mating] (D) -- (Q)
            node[gene pass description] {mate};
    \end{tikzpicture}
\end{center}

\begin{Lem}\label{lem::kinbet}
    In this situation, we have
    \[k_{\mathcal G,\mathcal I}=k_{\mathcal H,\mathcal I}.\]
\end{Lem}

\begin{proof}
    We fix a locus and draw an allele from $\mathcal G$. This allele comes with equal probability from $Q$ or from $\mathcal D$. We can thus deduce
        \[k_{\mathcal G,\mathcal I}=\frac{k_{Q,\mathcal I}+k_{\mathcal D,\mathcal I}}{2}.\]
    By the same argument, we also have
        \[k_{\mathcal H,\mathcal I}=\frac{k_{Q,\mathcal I}+k_{\mathcal D,\mathcal I}}{2},\]
    and the assertion follows.
\end{proof}

\begin{Rmk}
    It is instructive to clarify for oneself, at which point this argument needs the fact that $\mathcal G$ and $\mathcal H$ do not contain ancestors of queens in $\mathcal I$.
\end{Rmk}

One important corollary to Lemma~\ref{lem::kinbet} is the following:

\begin{Cor}\label{cor::imp}
    Let $\mathcal G_1$ and $\mathcal G_2$ be two (non-identical) groups of sister queens so that no queen in $\mathcal G_1$ is an ancestor of a queen in $\mathcal G_2$ and vice versa; let $Q_1$ and $Q_2$ be their (possibly identical!) respective dams and $\mathcal W_1$ and $\mathcal W_2$ be the respective worker groups of $Q_1$ and $Q_2$. Then
        \[k_{\mathcal G_1,\mathcal G_2}=k_{\mathcal W_1,\mathcal W_2}.\]
\end{Cor}

\begin{center}
    \begin{tikzpicture}
        \path (0,0) node[queen] (Q1) {}
                node[left = 0cm of Q1] {$Q_1$}
            node[right = 4cm of Q1.center, queen] (Q2) {}
                node[left = 0cm of Q2] {$Q_2$}
            node[below left = 2.5cm and 1cm of Q1.center, worker group] (W1) {}
                node[left = 0cm of W1] {$\mathcal W_1$}
            node[below right = 2.5cm and 1cm of Q2.center, worker group] (W2) {}
                node[right = 0cm of W2] {$\mathcal W_2$}
            node[below right = 2.4cm and 1cm of Q1.center, group] (G1) {}
                -- (G1.center) pic{queens}
                node[left = 0cm of G1] {$\mathcal G_1$}
            node[below left = 2.4cm and 1cm of Q2.center, group] (G2) {}
                -- (G2.center) pic{queens}
                node[right = 0cm of G2] {$\mathcal G_2$};

        \draw[inheritance] (Q1) -- (W1);
        \draw[inheritance] (Q1) -- (G1);
        \draw[inheritance] (Q2) -- (W2);
        \draw[inheritance] (Q2) -- (G2);
        \draw[relationship] (G1) -- (G2);
        \draw[relationship] (W1) -- (W2)
            node[relationship description, font = \footnotesize]{$k_{\mathcal G_1,\mathcal G_2}=k_{\mathcal W_1,\mathcal W_2}$};
    \end{tikzpicture}
\end{center}

So far, we mainly calculated kinships between (groups of) queens. Next, we want to consider some cases of kinships including drones. Assume that we have two disjoint finite groups of drones $\mathcal D_1$ and $\mathcal D_2$ that were produced by two (not necessarily disjoint) groups of queens, $\mathcal G_1$ and $\mathcal G_2$, respectively.

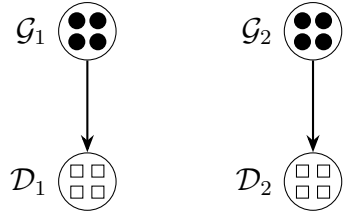
\begin{center}
    \begin{tikzpicture}
        \path (0,0) node[group] (G2) {}
                node[left = 0cm of G2] {$\mathcal G_2$}
                -- (G2.center) pic{queens}
            node[left = 3cm of G2.center, group] (G1) {}
                node[left = 0cm of G1] {$\mathcal G_1$}
                -- (G1) pic{queens}
            node[below = 2cm of G2.center, group] (D2) {}
                node[left = 0cm of D2] {$\mathcal D_2$}
                -- (D2) pic{drones}
            node[below = 2cm of G1.center, group] (D1) {}
                node[left = 0cm of D1] {$\mathcal D_1$}
                -- (D1) pic{drones};

        \draw[inheritance] (G1) -- (D1);
        \draw[inheritance] (G2) -- (D2);
    \end{tikzpicture}
\end{center}

For a drone $D\in\mathcal D_i$ ($i\in\{1,2\}$), we assume that any queen $Q\in\mathcal G_i$ has an equal chance to be his dam.
The following lemma describes the relevant kinships in this situation.

\begin{Lem}\label{lem::reldron}
    For $i,j\in\{1,2\}$, we have
    \begin{align}
        k_{\mathcal G_i,\mathcal D_j} &= k_{\mathcal G_i,\mathcal G_j}, \label{eq::kgd}\\
        k_{\mathcal D_i,\mathcal D_i} &= \frac1{|\mathcal D_i|}
                                         \left(1+\left(|\mathcal D_i|-1\right)k_{\mathcal G_i,\mathcal G_i}\right), \label{eq::skd}\\
        k_{\mathcal D_1,\mathcal D_2} &= k_{\mathcal G_1,\mathcal G_2}. \label{eq::kdd}
    \end{align}
\end{Lem}

\begin{proof}
    \begin{enumerate}[label = (\roman*)]
        \item We first show Equation~\ref{eq::kgd}, i.\,e.
            \[k_{\mathcal G_i,\mathcal D_j} = k_{\mathcal G_i,\mathcal G_j}.\]
        We fix a locus and draw an allele $A^{\mathcal G}$ from $\mathcal G_i$ and an allele $A^{\mathcal D}$ from $\mathcal D_j$. Since drones only have dams but no sires, the allele $A^{\mathcal D}$ must come from a queen in $\mathcal G_j$ and by our assumption all queens of $\mathcal G_j$ have equal probability to be the source of $A^{\mathcal D}$. Thus, $A^{\mathcal D}$ really is a randomly drawn allele from $\mathcal G_j$ and the probability that $A^{\mathcal G}$ and $A^{\mathcal D}$ are ibd is the probability that $A^{\mathcal G}$ is ibd with a randomly drawn allele from $\mathcal G_j$. From this, the assertion follows.

        \item We now show Equation~\ref{eq::skd}, i.\,e.
            \[k_{\mathcal D_i,\mathcal D_i}
                = \frac1{|\mathcal D_i|}\left(1+\left(|\mathcal D_i|-1\right)k_{\mathcal G_i,\mathcal G_i}\right).\]
        We fix a locus and draw (with replacement) two of the $|\mathcal D_i|$ alleles of $\mathcal D_i$. With probability $\frac1{|\mathcal D_i|}$, we picked the same allele which is surely ibd to itself. With the complementary probability of $\frac{|\mathcal D_i|-1}{|\mathcal D_i|}$, we picked to different alleles from $\mathcal D_i$ which are then just two random picks of alleles from $\mathcal G_i$ and have a probability of $k_{\mathcal G_i,\mathcal G_i}$ to be ibd. The assertion follows.

        \item Finally, we show Equation~\ref{eq::kdd}, i.\,e.
            \[k_{\mathcal D_1,\mathcal D_2} = k_{\mathcal G_1,\mathcal G_2}.\]
        We fix a locus and pick an allele $A^1$ from $\mathcal D_1$ and an allele $A^2$ from $\mathcal D_2$. Then, as in the proof of Equation~\ref{eq::kgd}, $A^1$ can be interpreted as a randomly drawn allele from $\mathcal G_1$ and $A^2$ as a randomly drawn allele from $\mathcal G_2$, and the assertion follows.
    \end{enumerate}
\end{proof}

\subsubsection{Worker groups or replacement queens}\label{sec::wgrq}

The approach to model a honeybee colony as consisting of two separate entities, namely a queen an a worker group, each with their own breeding values goes back to the early days of breeding value estimation in this species \citep{bienefeld90, bienefeld07} and has turned out very practical. It is, however, worth to take a step back and ask oneself what \emph{value} a worker group can have for breeding, since all the workers are sterile. So, while the workers have genetic properties, they are not able to pass them on to future generations. The reason why the estimated breeding value of a worker group is still of interest is that it is the expectation for the breeding value of a replacement queen \citep{bienefeld07, brascamp19a}. We had noted this fact in Remark~\ref{rmk::rpl}\,\ref{item::rplii}.

So, for the purposes of this manuscript, we could also imagine colonies to not consist of a queen $Q$ and a worker group $\mathcal W$ but of a queen $Q$ and a potential replacement queen $R$. So every queen is equipped with an imaginary replacement queen which -- by its imaginary nature -- will never have offspring.

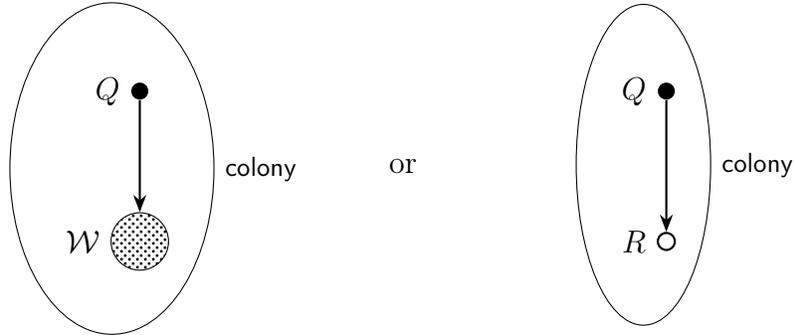
\begin{center}
    \begin{tikzpicture}
        \path (0,0) node[queen] (Q1) {}
                node[left = 0cm of Q1] (q1name) {$Q$}
            node[below = 2cm of Q1.center, worker group] (W1) {}
                node[left = 0cm of W1] (w1name) {$\mathcal W$}
            node[right = 7cm of Q1.center, queen] (Q2) {}
                node[left = 0cm of Q2] (q2name) {$Q$}
            node[below = 2cm of Q2.center, replacement queen] (R2) {}
                node[left = 0cm of R2] (r2name) {$R$}
            coordinate[below = 1cm of Q1.center] (C1)
            coordinate[below = 1cm of Q2.center] (C2)
            -- (C1) -- (C2)
                node[midway] {or};

        \begin{scope}[on background layer]
            \path node [fit=(Q1) (q1name) (W1) (w1name), ellipse, inner sep = 2mm, draw] (Col1) {}
                node [fit=(Q2) (q2name) (R2) (r2name), ellipse, inner sep = 2mm, draw] (Col2) {}
                node [right = 0mm of Col1, font = \footnotesize\sffamily] {colony}
                node [right = 0mm of Col2, font = \footnotesize\sffamily] {colony};
        \end{scope}

        \draw[inheritance] (Q1) -- (W1);
        \draw[inheritance] (Q2) -- (R2);
    \end{tikzpicture}
\end{center}

The difference becomes apparent when one looks at kinships. At first glance, one might think that there is no difference. Neither worker groups nor replacement queens have offspring, so we regularly find ourselves in the situation of Lemma~\ref{lem::kinbet} and any other bee or group of bees has the same kinship with $\mathcal W$ as with $R$. The difference lies in the self-kinship. As becomes apparent from Remark~\ref{rmk::withcol}\,\ref{item::withcoli} and~\ref{item::withcolii}, the kinships $k_{R,R}$ and $k_{\mathcal W,\mathcal W}$ will generally differ.

When developing a theory of OCS for honeybees, we will get to a point where we have to restrict average kinships to acceptable levels. As average kinships also contain self-kinships, the question arises which value is more important, $k_{\mathcal W,\mathcal W}$ or $k_{R,R}$.

The replacement queen self-kinship $k_{R,R}$ signifies the inbreeding of individual worker bees. This is of particular practical relevance because in honeybees there is a specific form of inbreeding depression. There is a locus within the honeybee genome, called \emph{csd}-locus, and worker or queen bees can only develop if they are heterozygous at this locus. Otherwise, they start a development into diploid drones but are soon removed from the colony \citep{woyke65}. Evidently, highly inbred workers have a high chance to become homozygous at the \emph{csd}-locus, which leads to holes in the brood pattern and weakened colonies \citep{bruckner78, zayed05}.

In contrast, $k_{\mathcal W,\mathcal W}$ describes the relatedness between different workers. This measure bears some importance, too, because several studies have shown that colonies with less related worker bees show greater overall vitality \citep{mattila07, tarpy13}. However, this vitality boost is not inheritable and therefore has little significance for breeding \citep{uzunov22the, du24comparison}. Furthermore, $k_{\mathcal W,\mathcal W}$ is mostly dependent on the strategy of mating control and not so much on the inbreeding development in the population. Therefore, $k_{R,R}$ turns out to be the more relevant value for our purposes.

However, at some points we will still need worker group kinships in our derivations. Note, for example, that by Corollary~\ref{cor::imp} the kinship between two newly hatched sister queens can be calculated as the kinship of their dam's worker group to itself but not as the self-kinship of a replacement queen of the dam. We will thus imagine honeybee colonies to consist of three components: a queen $Q$, a worker group $\mathcal W$ and a potential replacement queen $R$. With the exception of self-kinships, the information provided by $\mathcal W$ and $R$ is redundant.

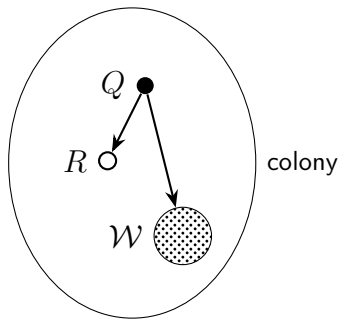
\begin{center}
    \begin{tikzpicture}
        \path (0,0) node[queen] (Q) {}
                node[left = 0cm of Q] (qname) {$Q$}
            node[below right = 2cm and 0.5cm of Q.center, worker group] (W) {}
                node[left = 0cm of W] (wname) {$\mathcal W$}
            node[below left = 1cm and 0.5cm of Q.center, replacement queen] (R) {}
                node[left = 0cm of R] (rname) {$R$};

        \begin{scope}[on background layer]
            \path node [fit=(Q) (qname) (W) (wname) (R) (rname), ellipse, inner sep = 1mm, draw] (Col) {}
                node [right = 0mm of Col, font = \footnotesize\sffamily] {colony};
        \end{scope}

        \draw[inheritance] (Q) -- (W);
        \draw[inheritance] (Q) -- (R);
    \end{tikzpicture}
\end{center}

\section{Optimum Contribution Selection for honeybees}\label{sec::ocshb}

We turn to the development of a theory of OCS for honeybees.

\begin{Not}
    \begin{enumerate}[label = (\roman*)]
        \item At each time $t\in\mathbb N$, we consider a population $\mathcal P_t$ of honeybees, consisting of $N_t$ colonies. As each colony consists of one queen $Q$, one worker group $\mathcal W$, and one replacement queen $R$, we model $\mathcal P_t$ as
            \[\mathcal P_t = \mathcal Q_t\sqcup\mathfrak W_t\sqcup\mathcal R_t,\]
        where $\mathcal Q_t$ comprises the queens alive at time $t$, and $\mathfrak W_t$ and $\mathcal R_t$, comprise the corresponding worker groups and replacement queens, respectively.

        \item When developing the theory of OCS for diploid species, we had spoken of \emph{individuals} $I\in\mathcal P_t$. In the following theory for honeybees, $\mathcal P_t$ also contains worker groups which are not individuals but collective units. Accordingly, we will henceforth speak of \emph{entities} $E\in\mathcal P_t$. When the type of an entity (queen, worker group or replacement queen) is clear, we will name them as such.
    \end{enumerate}
\end{Not}

\begin{Rmk}\label{rmk::nodr}
    \begin{enumerate}[label = (\roman*)]
        \item \label{item::nodri} Besides the queen $Q$, worker group $\mathcal W$, and replacement queen~$R$, another important group of bees that are uniquely associated with a colony is the group $\mathcal D$ of drones that $Q$ mated with. It is, thus, perceivable to add another component $\mathfrak D_t$ to $\mathcal P_t$, comprising all the groups of drones that mated with queens in $\mathcal Q_t$. However, drones are generally not seen as proper breeding entities but rather as \emph{flying gametes} \citep{mackensen67} and are thus rarely included in population analyses. Furthermore, note that letting $\mathcal P_t$ contain bees of mixed ploidy would prevent direct applications of Definitions~\ref{def::mathcalb} and~\ref{def::mathfrakb} to~$\mathcal P_t$.

        \item \label{item::nodrii} Nevertheless, we will sometimes need to consider the set of drones $\mathcal D$ that a queen $Q\in\mathcal Q_t$ mated with. When we include groups of drones in figures, we will henceforth use a gray hue to indicate that they are not counted as a part of $\mathcal P_t$.

        \item In the theory of OCS for other species, it is usually assumed that there cannot be parent-offspring relations within one generation \citep{wellmann19key}. In honeybees, this assumption is violated in the sense, that worker groups and replacement queens are considered as belonging to the same generation as their dam queens.

        
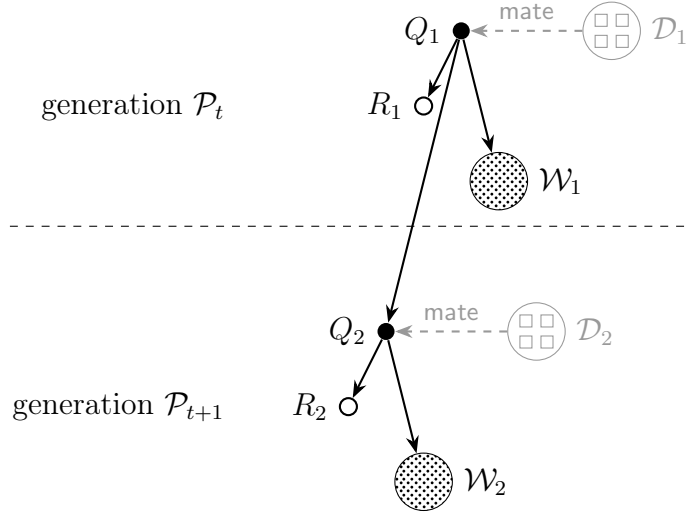
\begin{center}
            \begin{tikzpicture}
                \path (0,0) coordinate (Pt0)
                        node[left = 0mm of Pt0] {generation $\mathcal P_t$}
                    coordinate[below = 4cm of Pt0] (Pt1)
                        node[left = 0mm of Pt1] {generation $\mathcal P_{t+1}$}
                    (Pt0) -- (Pt1)
                        coordinate[pos = 0.4] (M)
                    coordinate[right = 3cm of Pt0] (C0)
                    node[above = 1cm of C0, queen] (Q0) {}
                        node[left = 0cm of Q0] {$Q_1$}
                    node[right = 2cm of Q0.center, not considered, group] (D0) {}
                        node[right = 0cm of D0, not considered] {$\mathcal D_1$}
                        -- (D0) pic[not considered] {drones}
                    node[below right = 2cm and 0.5cm of Q0.center, worker group] (W0) {}
                        node[right = 0cm of W0] {$\mathcal W_1$}
                    node[below left = 1cm and 0.5cm of Q0.center, replacement queen] (R0) {}
                        node[left = 0cm of R0] {$R_1$}
                    coordinate[right = 2cm of Pt1] (C1)
                    node[above = 1cm of C1, queen] (Q1) {}
                        node[left = 0cm of Q1] {$Q_2$}
                    node[right = 2cm of Q1.center, not considered, group] (D1) {}
                        node[right = 0cm of D1, not considered] {$\mathcal D_2$}
                        -- (D1) pic[not considered] {drones}
                    node[below right = 2cm and 0.5cm of Q1.center, worker group] (W1) {}
                        node[right = 0cm of W1] {$\mathcal W_2$}
                    node[below left = 1cm and 0.5cm of Q1.center, replacement queen] (R1) {}
                        node[left = 0cm of R1] {$R_2$};

                \draw[mating, not considered] (D0) -- (Q0)
                    node[gene pass description, not considered] {mate};
                \draw[mating, not considered] (D1) -- (Q1)
                    node[gene pass description, not considered] {mate};
                \draw[inheritance] (Q0) -- (W0);
                \draw[inheritance] (Q0) -- (R0);
                \draw[inheritance] (Q1) -- (W1);
                \draw[inheritance] (Q1) -- (R1);
                \draw[inheritance] (Q0) -- (Q1);

                \draw[dashed] (M) ++ (-3cm,0cm) -- ++(9cm,0cm);
            \end{tikzpicture}
        \end{center}
    \end{enumerate}
\end{Rmk}

\begin{Not}
    \begin{enumerate}[label = (\roman*)]
        \item For a queen $Q\in\mathcal Q_t$, we denote her unique worker group $\mathcal W\in\mathfrak W_t$ by $\mathcal W(Q)$ and her unique replacement queen by $R(Q)\in\mathcal R_t$.

        \item Similarly, for a worker group $\mathcal W\in\mathfrak W_t$, we denote its unique queen $Q \in\mathcal Q_t$ by $Q(\mathcal W)$ and the corresponding replacement queen by $R(\mathcal W)\in\mathcal R_t$.

        \item Finally, for a replacement queen $R\in\mathcal R_t$, we denote her unique dam queen $Q \in\mathcal Q_t$ by $Q(R)$ and the corresponding worker group by $\mathcal W(R)\in\mathfrak W_t$.
    \end{enumerate}
\end{Not}

Like in the derivation of OCS for diploid species we have to predict average breeding values and kinships at time $t+1$ from the data of time $t$. Some of the necessary derivations hold for any closed honeybee population, whereas others depend on the way in which mating control is organized. In Section~\ref{sec::gendev}, we will develop the general part of the theory as far as possible. Then, we will complete the theory by adding the missing parts that depend on the mode of mating control. We will consider \emph{single colony inseminations} in Section~\ref{sec::sci}, \emph{isolated mating stations} in Section~\ref{sec::ims} and the combination of both strategies in Section~\ref{sec::mix}. All these sections are subdivided into
\begin{enumerate}[label = (\roman*)]
    \item the analysis of estimated breeding values, and
    \item the analysis of kinships.
\end{enumerate}

\subsection{General derivations}\label{sec::gendev}
\subsubsection{Breeding value analysis}

\begin{Not}
    Each queen, worker group, and replacement queen is equipped with an estimated breeding value, giving rise to a vector of estimated breeding values
        \[\hat{\mathbf u}_t\in\mathbb R^{\mathcal P_t},\]
    which, under the isomorphism $\mathbb R^{\mathcal P_t}\cong\mathbb R^{\mathcal Q_t}\oplus\mathbb R^{\mathfrak W_t}\oplus\mathbb R^{\mathcal R_t}$ may also be interpreted as
        \[\hat{\mathbf u}_t^{\mathcal Q}\oplus\hat{\mathbf u}_t^{\mathfrak W}\oplus\hat{\mathbf u}_t^{\mathcal R}\in\mathbb R^{\mathcal Q_t}\oplus\mathbb R^{\mathfrak W_t}\oplus\mathbb R^{\mathcal R_t}.\]
\end{Not}

\begin{Rmk}
    \begin{enumerate}[label = (\roman*)]
        \item \label{item::ident} As explained in Remark~\ref{rmk::rpl}\,\ref{item::rplii} and Section~\ref{sec::wgrq}, we have
            \[\hat{\mathbf u}_t^{\mathfrak W}=\hat{\mathbf u}_t^{\mathcal R}.\]
        Or, to put it more precisely, $\hat{\mathbf u}_t^{\mathfrak W}$ is the image of $\hat{\mathbf u}_t^{\mathcal R}$ under the isomorphism $\mathbb R^{\mathcal R_t}\cong\mathbb R^{\mathfrak W_t}$ induced by the bijection $R\mapsto\mathcal W(R)$.

        \item Implicit identifications of the vector spaces $\mathbb R^{\mathcal Q_t}$, $\mathbb R^{\mathfrak W_t}$, and $\mathbb R^{\mathcal R_t}$, like when writing $\hat{\mathbf u}_t^{\mathfrak W}=\hat{\mathbf u}_t^{\mathcal R}$ in~\ref{item::ident}, will occur frequently in the remainder of this manuscript.
    \end{enumerate}
\end{Rmk}

Since all three sets $\mathcal Q_t$, $\mathfrak W_t$, and $\mathcal R_t$ have the same cardinality,
    \[\left|\mathcal Q_t\right|=\left|\mathfrak W_t\right|=\left|\mathcal R_t\right|=N_t,\]
we have
    \[\hat u_{\mathcal P_t} = \frac{\hat u_{\mathcal Q_t}+\hat u_{\mathfrak W_t}+\hat u_{\mathcal R_t}}{3}
                            = \frac{\hat u_{\mathcal Q_t}+2\hat u_{\mathcal R_t}}{3}.\]
However, this value is of little significance. By modeling colonies to have both a worker group and a replacement queen, we artificially upweigh the replacement queen's breeding value. Instead, we want to weigh the breeding values of queens and replacement queens equally in the average breeding value.

\begin{Def}
    We therefore define the \emph{reduced generation} in which colonies consist only of queens and replacement queens
        \[\mathcal P^{\ast}_t:=\mathcal Q_t\sqcup\mathcal R_t\]
    and focus our interest on the average breeding value
    \begin{equation}\label{eq::upast}
        \hat u_{\mathcal P^{\ast}_t}=\frac{\hat u_{\mathcal Q_t}+\hat u_{\mathcal R_t}}{2}.
    \end{equation}

    
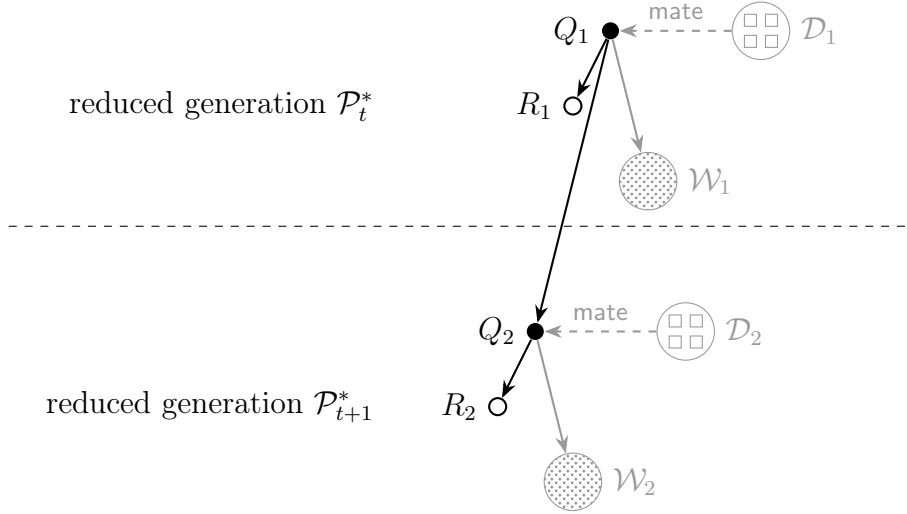
\begin{center}
        \begin{tikzpicture}
            \path (0,0) coordinate (Pt0)
                    node[left = 0mm of Pt0] {reduced generation $\mathcal P^{\ast}_t$}
                coordinate[below = 4cm of Pt0] (Pt1)
                    node[left = 0mm of Pt1] {reduced generation $\mathcal P^{\ast}_{t+1}$}
                (Pt0) -- (Pt1)
                    coordinate[pos = 0.4] (M)
                coordinate[right = 3cm of Pt0] (C0)
                node[above = 1cm of C0, queen] (Q0) {}
                    node[left = 0cm of Q0] {$Q_1$}
                node[right = 2cm of Q0.center, not considered, group] (D0) {}
                    node[right = 0cm of D0, not considered] {$\mathcal D_1$}
                    -- (D0) pic[not considered] {drones}
                node[below right = 2cm and 0.5cm of Q0.center, not considered, worker group, pattern color = black!40] (W0) {}
                    node[right = 0cm of W0, not considered] {$\mathcal W_1$}
                node[below left = 1cm and 0.5cm of Q0.center, replacement queen] (R0) {}
                    node[left = 0cm of R0] {$R_1$}
                coordinate[right = 2cm of Pt1] (C1)
                node[above = 1cm of C1, queen] (Q1) {}
                    node[left = 0cm of Q1] {$Q_2$}
                node[right = 2cm of Q1.center, not considered, group] (D1) {}
                    node[right = 0cm of D1, not considered] {$\mathcal D_2$}
                    -- (D1) pic[not considered] {drones}
                node[below right = 2cm and 0.5cm of Q1.center, not considered, worker group, pattern color = black!40] (W1) {}
                    node[right = 0cm of W1, not considered] {$\mathcal W_2$}
                node[below left = 1cm and 0.5cm of Q1.center, replacement queen] (R1) {}
                    node[left = 0cm of R1] {$R_2$};

            \draw[mating, not considered] (D0) -- (Q0)
                node[gene pass description, not considered] {mate};
            \draw[mating, not considered] (D1) -- (Q1)
                node[gene pass description, not considered] {mate};
            \draw[inheritance, not considered] (Q0) -- (W0);
            \draw[inheritance] (Q0) -- (R0);
            \draw[inheritance, not considered] (Q1) -- (W1);
            \draw[inheritance] (Q1) -- (R1);
            \draw[inheritance] (Q0) -- (Q1);

            \draw[dashed] (M) ++ (-5cm,0cm) -- ++(12cm,0cm);
        \end{tikzpicture}
    \end{center}
\end{Def}

\begin{Not}
    Each of the vector spaces $\mathbb R^{\mathcal Q_t}$, $\mathbb R^{\mathfrak W_t}$, $\mathbb R^{\mathcal R_t}$, $\mathbb R^{\mathcal P^{\ast}_t}$, and $\mathbb R^{\mathcal P_t}$, contains a vector which has ones as all entries. In analogy with Notation~\ref{not::einsfm}\,\ref{item::einsfm}, we could denote these vectors by $\mathbf 1_t^{\mathcal Q}$, $\mathbf 1_t^{\mathfrak W}$, etc. However, we opt for simpler (yet slightly ambiguous) notation and denote all of these vectors simply by $\mathbf 1_t$.
\end{Not}

\begin{Rmk}
    For each reduced generation $\mathcal P^{\ast}_t$, we have
    \begin{equation} \label{eq::uqcurr}
        \hat u_{\mathcal Q_t}=\frac1{N_t}\mathbf 1_t^{\top}\hat{\mathbf u}_t^{\mathcal Q}
    \end{equation}
    and
    \begin{equation} \label{eq::urcurr}
        \hat u_{\mathcal R_t}=\frac1{N_t}\mathbf 1_t^{\top}\hat{\mathbf u}_t^{\mathcal R}.
    \end{equation}
    Inserting these equations in Equation~\ref{eq::upast} yields
    \begin{equation} \label{upcurr}
        \hat u_{\mathcal P^{\ast}_t}=\frac1{2N_t}\mathbf 1_t^{\top}\hat{\mathbf u}_t^{\mathcal Q}
            + \frac1{2N_t}\mathbf 1_t^{\top}\hat{\mathbf u}_t^{\mathcal R}.
    \end{equation}
\end{Rmk}

In order to maximize $\hat u_{\mathcal P^{\ast}_t}$ over time, we need to calculate the expected average breeding values for the next reduced generation $\mathcal P_{t+1}^{\ast}$, i.\,e. the value $\mathbb E\bigl[\hat u_{\mathcal P_{t+1}^{\ast}}\bigr]$.

\begin{Rmk}
    Since Equation~\ref{eq::upast} also holds for reduced generation $\mathcal P_{t+1}^{\ast}$, we have
    \begin{equation}\label{eq::eupast}
        \mathbb E\bigl[\hat u_{\mathcal P_{t+1}^{\ast}}\bigr] = \frac12\mathbb E\left[\hat u_{\mathcal Q_{t+1}}\right] + \frac12\mathbb E\left[\hat u_{\mathcal R_{t+1}}\right].
    \end{equation}
    Thereby, the task of calculating $\mathbb E\bigl[\hat u_{\mathcal P_{t+1}^{\ast}}\bigr]$ is broken down to calculating the two values $\mathbb E\left[\hat u_{\mathcal Q_{t+1}}\right]$ and $\mathbb E\left[\hat u_{\mathcal R_{t+1}}\right]$.
\end{Rmk}

We break the task down further:

\begin{Not}
    \begin{enumerate}[label = (\roman*)]
        \item As in the case of diploids with overlapping generations (Notation~\ref{not::nsclassic}), we subdivide reduced generation $\mathcal P_{t+1}^{\ast}$ (consisting of queens and replacement queens) into the newly created entities
            \[\mathcal N_{t+1}=\mathcal P^{\ast}_{t+1}\backslash\mathcal P^{\ast}_t,\]
        and those that survived from the previous generation,
            \[\mathcal S_{t+1}=\mathcal P^{\ast}_{t+1}\cap\mathcal P^{\ast}_t,\]
        so that
            \[\mathcal P^{\ast}_{t+1}=\mathcal N_{t+1}\sqcup\mathcal S_{t+1}.\]

        \item Accordingly, the total number of colonies at time $t+1$ is the sum of the newly created ones and the older surviving ones,
            \[N_{t+1}=N_{t+1}^{\mathcal N}+N_{t+1}^{\mathcal S}.\]

        \item With the already existing subdivision of $\mathcal P^{\ast}_{t+1}$ into queens and replacement queens, this separates $\mathcal P^{\ast}_{t+1}$ into four disjoint classes:
            \[\mathcal P_{t+1}=\mathcal N\mathcal Q_{t+1}\sqcup\mathcal N\mathcal R_{t+1}\sqcup\mathcal S\mathcal Q_{t+1}\sqcup\mathcal S\mathcal R_{t+1},\]
        where
        \begin{align*}
            \mathcal N\mathcal Q_{t+1} &:= \mathcal N_{t+1}\cap\mathcal Q_{t+1},\\
            \mathcal N\mathcal R_{t+1} &:= \mathcal N_{t+1}\cap\mathcal R_{t+1},\\
            \mathcal S\mathcal Q_{t+1} &:= \mathcal S_{t+1}\cap\mathcal Q_{t+1},\\
            \mathcal S\mathcal R_{t+1} &:= \mathcal S_{t+1}\cap\mathcal R_{t+1}.
        \end{align*}
    \end{enumerate}
\end{Not}

\begin{Lem}\label{lem::genred}
    We have
    \begin{align}
        \mathbb E\left[\hat{u}_{\mathcal Q_{t+1}}\right] &=
            \frac{N_{t+1}^{\mathcal N}\mathbb E\left[\hat{u}_{\mathcal N\mathcal Q_{t+1}}\right]
                + N_{t+1}^{\mathcal S}\mathbb E\left[\hat{u}_{\mathcal S\mathcal Q_{t+1}}\right]}{N_{t+1}},\label{eq::fstt}\\
        \mathbb E\left[\hat{u}_{\mathcal R_{t+1}}\right] &=
            \frac{N_{t+1}^{\mathcal N}\mathbb E\left[\hat{u}_{\mathcal N\mathcal R_{t+1}}\right]
                + N_{t+1}^{\mathcal S}\mathbb E\left[\hat{u}_{\mathcal S\mathcal R_{t+1}}\right]}{N_{t+1}},\label{eq::sstt}\\
        \mathbb E\bigl[\hat{u}_{\mathcal P^{\ast}_{t+1}}\bigr] &=
            \frac{N_{t+1}^{\mathcal N}\mathbb E\left[\hat{u}_{\mathcal N\mathcal Q_{t+1}}\right]
                + N_{t+1}^{\mathcal S}\mathbb E\left[\hat{u}_{\mathcal S\mathcal Q_{t+1}}\right]
                + N_{t+1}^{\mathcal N}\mathbb E\left[\hat{u}_{\mathcal N\mathcal R_{t+1}}\right]
                + N_{t+1}^{\mathcal S}\mathbb E\left[\hat{u}_{\mathcal S\mathcal R_{t+1}}\right]}{2N_{t+1}}.\label{eq::tstt}
    \end{align}
\end{Lem}

\begin{proof}
    Equations~\ref{eq::fstt} and~\ref{eq::sstt} are immediate consequences of the partition
        \[\mathcal P_{t+1}^{\ast}=\mathcal N_{t+1}\sqcup\mathcal S_{t+1}.\]
    Equation~\ref{eq::tstt} follows by inserting Equations~\ref{eq::fstt} and~\ref{eq::sstt} into Equation~\ref{eq::eupast}.
\end{proof}

\begin{Rmk}\label{rmk::brokenu}
    By Lemma~\ref{lem::genred}, the task of calculating $\mathbb E\bigl[\hat{u}_{\mathcal P^{\ast}_{t+1}}\bigr]$ is equivalent to calculating the four values $\mathbb E\left[\hat{u}_{\mathcal N\mathcal Q_{t+1}}\right]$, $\mathbb E\left[\hat{u}_{\mathcal S\mathcal Q_{t+1}}\right]$, $\mathbb E\left[\hat{u}_{\mathcal N\mathcal R_{t+1}}\right]$, and $\mathbb E\left[\hat{u}_{\mathcal S\mathcal R_{t+1}}\right]$.
\end{Rmk}

We leave the task to calculate these four expectations open for now. It is tackled for mating control via single colony insemination in Section~\ref{sec::scibv}, for mating control via isolated mating stations in Section~\ref{sec::imsbv} and for the mixed strategy in Section~\ref{sec::mixbv}.

Instead, we turn our attention to the analysis of kinships.

\subsubsection{Kinship analysis}

\begin{Not}
    By the partition of $\mathcal P_t = \mathcal Q_t\sqcup\mathfrak W_t\sqcup\mathcal R_t$ into queens, worker groups, and replacement queens, the matrix $\mathbf K_t$ of kinships in generation $\mathcal P_t$ is subdivided into nine equal-sized blocks,
        \[\mathbf K_t =
            \begin{pmatrix}
                \mathbf K_t^{\mathcal Q\mathcal Q}  & \mathbf K_t^{\mathcal Q\mathfrak W}  & \mathbf K_t^{\mathcal Q\mathcal R} \\
                \mathbf K_t^{\mathfrak W\mathcal Q} & \mathbf K_t^{\mathfrak W\mathfrak W} & \mathbf K_t^{\mathfrak W\mathcal R}\\
                \mathbf K_t^{\mathcal R\mathcal Q}  & \mathbf K_t^{\mathcal R\mathfrak W}  & \mathbf K_t^{\mathcal R\mathcal R}
            \end{pmatrix}.\]
    Here, $\mathbf K_t^{\mathcal Q\mathcal Q}$ contains the kinships among queens, $\mathbf K_t^{\mathfrak W\mathfrak W}$ contains the kinships among worker groups and $\mathbf K_t^{\mathcal R\mathcal R}$ contains the kinships among replacement queens. The other blocks contain the kinships between the different categories.
\end{Not}

\begin{Rmk}\label{rmk::onk}
    \begin{enumerate}[label = (\roman*)]
        \item Because kinships are symmetric, so are $\mathbf K_t^{\mathcal Q\mathcal Q}$, $\mathbf K_t^{\mathfrak W\mathfrak W}$, and $\mathbf K_t^{\mathcal R\mathcal R}$.

        \item Moreover, we have
        \begin{align*}
            \mathbf K_t^{\mathfrak W\mathcal Q} &= \left(\mathbf K_t^{\mathcal Q\mathfrak W}\right)^{\top},\\
            \mathbf K_t^{\mathcal R\mathcal Q}  &= \left(\mathbf K_t^{\mathcal Q\mathcal R}\right)^{\top}, \\
            \mathbf K_t^{\mathcal R\mathfrak W} &= \left(\mathbf K_t^{\mathfrak W\mathcal R}\right)^{\top}.
        \end{align*}

        \item \label{item::onkdiag} As explained in Section~\ref{sec::wgrq}, kinships to replacement queens are generally the same as kinships to worker groups, with the only exception of self-kinships. We therefore have
            \[\mathbf K_t^{\mathcal Q\mathfrak W}=\mathbf K_t^{\mathcal Q\mathcal R}\]
        and
            \[\mathbf K_t^{\mathfrak W\mathcal R}=\mathbf K_t^{\mathfrak W\mathfrak W},\]
        whereas $\mathbf K_t^{\mathfrak W\mathfrak W}$ and $\mathbf K_t^{\mathcal R\mathcal R}$ only differ on the diagonal.
    \end{enumerate}
\end{Rmk}

We remind ourselves of the short discussion in Section~\ref{sec::wgrq} to conclude that we are mainly interested in the development of the average kinships $k_{\mathcal P_t^{\ast},\mathcal P_t^{\ast}}$ in the reduced population $\mathcal P^{\ast}_t=\mathcal Q_t\sqcup\mathcal R_t$

\begin{Rmk}
    \begin{enumerate}[label = (\roman*)]
        \item Since all blocks of matrix $\mathbf K_t$ are of equal size $N_t\times N_t$, the average kinship in the reduced generation $\mathcal P^{\ast}_t$ is
        \begin{align*}
            k_{\mathcal P^{\ast}_t,\mathcal P^{\ast}_t}
                &= \frac14k_{\mathcal Q_t,\mathcal Q_t}+\frac12k_{\mathcal Q_t,\mathcal R_t}+\frac14k_{\mathcal R_t,\mathcal R_t}  \\
                &= \frac1{4N_t^2}\left(\mathbf 1_t^{\top}\mathbf{K}_t^{\mathcal Q\mathcal Q}\mathbf 1_t
                    + 2\cdot\mathbf 1_t^{\top}\mathbf{K}_t^{\mathcal Q\mathcal R}\mathbf 1_t
                    + \mathbf 1_t^{\top}\mathbf{K}_t^{\mathcal R\mathcal R}\mathbf 1_t\right).
        \end{align*}

        \item Evidently, this identity also holds in the next generation, i.\,e.
        \begin{equation}\label{eq::kng}
            k_{\mathcal P^{\ast}_{t+1},\mathcal P^{\ast}_{t+1}}
                = \frac14k_{\mathcal Q_{t+1},\mathcal Q_{t+1}}+\frac12k_{\mathcal Q_{t+1},\mathcal R_{t+1}}+\frac14k_{\mathcal R_{t+1},\mathcal R_{t+1}}.
        \end{equation}
    \end{enumerate}
\end{Rmk}

\begin{Rmk}\label{rmk::brokenk}
    \begin{enumerate}[label = (\roman*)]
        \item \label{item::teni} By Equation~\ref{eq::kng}, we can determine $k_{\mathcal P_{t+1}^{\ast},\mathcal P_{t+1}^{\ast}}$ if we know the three average kinships $k_{\mathcal Q_{t+1},\mathcal Q_{t+1}}$, $k_{\mathcal Q_{t+1},\mathcal R_{t+1}}$, and $k_{\mathcal R_{t+1},\mathcal R_{t+1}}$.

        \item \label{item::tenii} These three values can in turn be calculated as weighted averages of kinships between newly created and surviving entities:
        \begin{align}
            k_{\mathcal Q_{t+1},\mathcal Q_{t+1}} &=
                \left(\frac{N_{t+1}^{\mathcal N}}{N_{t+1}}\right)^2k_{\mathcal N\mathcal Q_{t+1}, \mathcal N\mathcal Q_{t+1}}
                    + \frac{2N_{t+1}^{\mathcal N}N_{t+1}^{\mathcal S}}{N_{t+1}^2}k_{\mathcal N\mathcal Q_{t+1},\mathcal S\mathcal Q_{t+1}}\nonumber\\
                    &\hspace{1cm}+ \left(\frac{N_{t+1}^{\mathcal S}}{N_{t+1}}\right)^2k_{\mathcal S\mathcal Q_{t+1}, \mathcal S\mathcal Q_{t+1}},\\
            k_{\mathcal Q_{t+1},\mathcal R_{t+1}} &=
                \left(\frac{N_{t+1}^{\mathcal N}}{N_{t+1}}\right)^2 k_{\mathcal N\mathcal Q_{t+1}, \mathcal N\mathcal R_{t+1}}
                    + \frac{N_{t+1}^{\mathcal N}N_{t+1}^{\mathcal S}}{N_{t+1}^2}\left(k_{\mathcal N\mathcal Q_{t+1},\mathcal S\mathcal R_{t+1}}
                        + k_{\mathcal S\mathcal Q_{t+1},\mathcal N\mathcal R_{t+1}}\right)\nonumber\\
                    &\hspace{1cm}+ \left(\frac{N_{t+1}^{\mathcal S}}{N_{t+1}}\right)^2 k_{\mathcal S\mathcal Q_{t+1}, \mathcal S\mathcal R_{t+1}},\\
            k_{\mathcal R_{t+1},\mathcal R_{t+1}} &=
                \left(\frac{N_{t+1}^{\mathcal N}}{N_{t+1}}\right)^2k_{\mathcal N\mathcal R_{t+1}, \mathcal N\mathcal R_{t+1}}
                    +\frac{2N_{t+1}^{\mathcal N}N_{t+1}^{\mathcal S}}{N_{t+1}^2}k_{\mathcal N\mathcal R_{t+1},\mathcal S\mathcal R_{t+1}}\nonumber\\
                    &\hspace{1cm} +\left(\frac{N_{t+1}^{\mathcal S}}{N_{t+1}}\right)^2k_{\mathcal S\mathcal R_{t+1}, \mathcal S\mathcal R_{t+1}}.
        \end{align}

        \item \label{item::brokenkiii} Following~\ref{item::teni} and~\ref{item::tenii}, what we need to do is to calculate the ten different average kinships $k_{\mathcal N\mathcal Q_{t+1}, \mathcal N\mathcal Q_{t+1}}$, $k_{\mathcal N\mathcal Q_{t+1}, \mathcal N\mathcal R_{t+1}}$, $k_{\mathcal N\mathcal Q_{t+1}, \mathcal S\mathcal Q_{t+1}}$, $k_{\mathcal N\mathcal Q_{t+1}, \mathcal S\mathcal R_{t+1}}$, $k_{\mathcal N\mathcal R_{t+1}, \mathcal N\mathcal R_{t+1}}$, $k_{\mathcal N\mathcal R_{t+1}, \mathcal S\mathcal Q_{t+1}}$, $k_{\mathcal N\mathcal R_{t+1}, \mathcal S\mathcal R_{t+1}}$, $k_{\mathcal S\mathcal Q_{t+1}, \mathcal S\mathcal Q_{t+1}}$, $k_{\mathcal S\mathcal Q_{t+1}, \mathcal S\mathcal R_{t+1}}$, and $k_{\mathcal S\mathcal R_{t+1}, \mathcal S\mathcal R_{t+1}}$.
    \end{enumerate}
\end{Rmk}

As in the breeding value analysis, we leave the task to actually compute these ten values open for now. It is tackled for mating control via single colony insemination in Section~\ref{sec::scikin}, for mating control via isolated mating stations in Section~\ref{sec::imskin} and for the mixed strategy in Section~\ref{sec::mixkin}.

The distinction of the different mating strategies starts now with the treatment of single colony insemination.

\subsection{Single colony insemination}\label{sec::sci}

If young queens are always inseminated with drones from a single colony, any queen $Q\in\mathcal Q_t$ can contribute to the genetic setup of the next generation $\mathcal P_{t+1}$ in three different ways.
\begin{enumerate}[label = (\roman*)]
    \item $Q$ can produce new daughter queens in generation $\mathcal Q_{t+1}$. We call this type of contribution the \emph{dam path}.
    \item Drones produced by $Q$ can be used to fertilize new queens in $\mathcal Q_{t+1}$. In the nomenclature of \url{www.beebreed.eu}, a queen in the role of $Q$ is called "1b-queen" \citep{uzunov23}, which is why we call this the \emph{1b-path}.
    \item $Q$ can survive and still be alive a time $t+1$. This is the \emph{survival path}.
\end{enumerate}

\begin{Rmk}
    In case of the survival path, also the worker group $\mathcal W(Q)$ and the replacement queen $R(Q)$ contribute to the next generation.

    
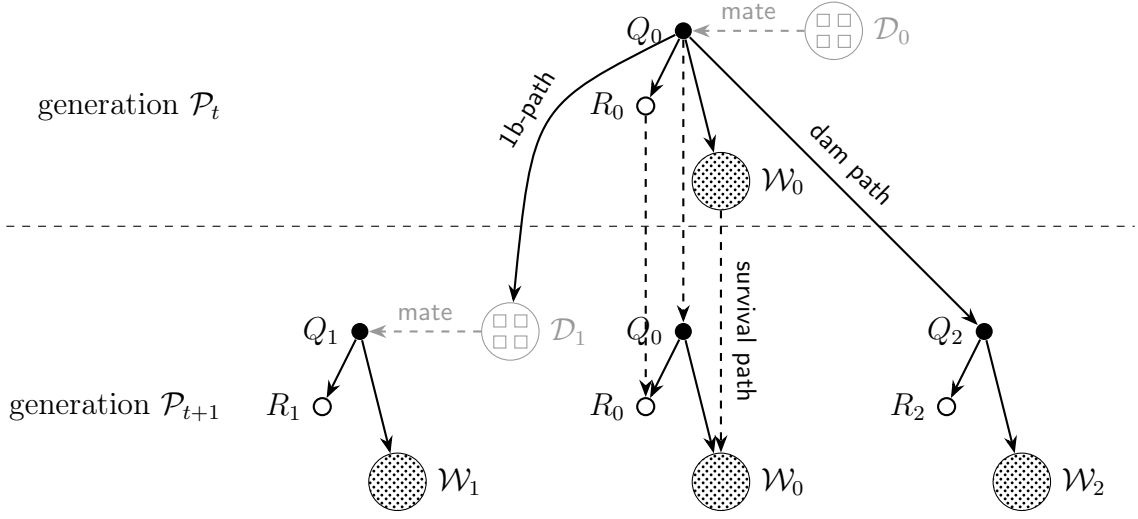
\begin{center}
        \begin{tikzpicture}
            \path (0,0) coordinate (Pt0)
                    node[left = 0mm of Pt0] {generation $\mathcal P_t$}
                coordinate[below = 4cm of Pt0] (Pt1)
                    node[left = 0mm of Pt1] {generation $\mathcal P_{t+1}$}
                (Pt0) -- (Pt1)
                    coordinate[pos = 0.4] (M)
                coordinate[right = 6cm of Pt0] (C0)
                node[above = 1cm of C0, queen] (Q0) {}
                    node[left = 0cm of Q0] {$Q_0$}
                node[right = 2cm of Q0.center, not considered, group] (D0) {}
                    node[right = 0cm of D0, not considered] {$\mathcal D_0$}
                    -- (D0) pic[not considered] {drones}
                node[below right = 2cm and 0.5cm of Q0.center, worker group] (W0) {}
                    node[right = 0cm of W0] {$\mathcal W_0$}
                node[below left = 1cm and 0.5cm of Q0.center, replacement queen] (R0) {}
                    node[left = 0cm of R0] {$R_0$}

                coordinate[right = 1.7cm of Pt1] (C1)
                node[above = 1cm of C1, queen] (Q1) {}
                    node[left = 0cm of Q1] {$Q_1$}
                node[right = 2cm of Q1.center, not considered, group] (D1) {}
                    node[right = 0cm of D1, not considered] {$\mathcal D_1$}
                    -- (D1) pic[not considered] {drones}
                node[below right = 2cm and 0.5cm of Q1.center, worker group] (W1) {}
                    node[right = 0cm of W1] {$\mathcal W_1$}
                node[below left = 1cm and 0.5cm of Q1.center, replacement queen] (R1) {}
                    node[left = 0cm of R1] {$R_1$}

                coordinate[right = 6cm of Pt1] (C2)
                node[above = 1cm of C2, queen] (Q2) {}
                    node[left = 0cm of Q2] {$Q_0$}
                node[below right = 2cm and 0.5cm of Q2.center, worker group] (W2) {}
                    node[right = 0cm of W2] {$\mathcal W_0$}
                node[below left = 1cm and 0.5cm of Q2.center, replacement queen] (R2) {}
                    node[left = 0cm of R2] {$R_0$}

                coordinate[right = 10cm of Pt1] (C3)
                node[above = 1cm of C3, queen] (Q3) {}
                    node[left = 0cm of Q3] {$Q_2$}
                node[below right = 2cm and 0.5cm of Q3.center, worker group] (W3) {}
                    node[right = 0cm of W3] {$\mathcal W_2$}
                node[below left = 1cm and 0.5cm of Q3.center, replacement queen] (R3) {}
                    node[left = 0cm of R3] {$R_2$};

            \draw[mating, not considered] (D0) -- (Q0)
                node[gene pass description, not considered] {mate};
            \draw[mating, not considered] (D1) -- (Q1)
                node[gene pass description, not considered] {mate};
            \draw[inheritance] (Q0) -- (W0);
            \draw[inheritance] (Q0) -- (R0);
            \draw[inheritance] (Q1) -- (W1);
            \draw[inheritance] (Q1) -- (R1);
            \draw[inheritance] (Q2) -- (W2);
            \draw[inheritance] (Q2) -- (R2);
            \draw[inheritance] (Q3) -- (W3);
            \draw[inheritance] (Q3) -- (R3);
            \draw[inheritance] (Q0) .. controls ++(-2,-1) .. (D1)
                node[gene pass description] {1b-path};
            \draw[inheritance] (Q0) -- (Q3)
                node[gene pass description] {dam path};
            \draw[survival] (W0) -- (W2)
                node[gene pass description] {survival path};
            \draw[survival] (Q0) -- (Q2);
            \draw[survival] (R0) -- (R2);

            \draw[dashed] (M) ++ (-3cm,0cm) -- ++(15cm,0cm);
        \end{tikzpicture}
    \end{center}
\end{Rmk}

\begin{Rmk}
    \begin{enumerate}[label = (\roman*)]
        \item In the theory for diploid species we had equipped each individual $I\in\mathcal P_t$ with a genetic contribution $c_{I,t}\in[0,1]$ towards the next generation (Notation~\ref{not::nsclassic}\,\ref{item::nsclassicc}). For honeybees, only the queens $Q\in\mathcal Q_t$ are equipped with such values, because worker groups and replacement queens do not pass on their genes to further generations.

        \item For a queen $Q\in\mathcal Q_t$, instead of a single value $c_{Q,t}$, we will need two separate values for $Q$'s contributions via the dam path and via the 1b-path.
    \end{enumerate}
\end{Rmk}

\begin{Not}\label{not::dcbc}
    \begin{enumerate}[label = (\roman*)]
        \item \label{item::dc} Each queen $NQ\in\mathcal N\mathcal Q_{t+1}$ that newly hatches at time $t+1$ has a dam $Q\in\mathcal Q_t$. For each queen $Q\in\mathcal Q_t$, we let $dc_{Q,t}$ be the fraction of queens in $\mathcal N\mathcal Q_{t+1}$ for which $Q$ serves as the \emph{dam}.

        \item Furthermore, each new queen $NQ\in\mathcal N\mathcal Q_{t+1}$ is inseminated with drones from a queen $Q\in\mathcal Q_t$, which thus serves as \emph{1b-queen}. We denote the fraction of queens in $\mathcal N\mathcal Q_{t+1}$ that were inseminated with drones from $Q\in\mathcal Q_t$ by $bc_{Q,t}$.

        \item \label{item::dcbcvec} This gives rise to two vectors $\mathbf {dc}_{t}, \mathbf {bc}_t\in\mathbb R_{\geq0}^{\mathcal Q_t}$ of contributions to the (newly generated entities in the) next generation via the dam path and 1b-path, respectively.
    \end{enumerate}
\end{Not}

\begin{Rmk}\label{rmk::dbone}
    Because all newly created queens in $\mathcal N\mathcal Q_{t+1}$ need to have a dam $\mathcal Q\in\mathcal Q_t$ and mate with drones from a 1b-queen $S\in\mathcal Q_t$, we have
        \[\mathbf 1_t^{\top}\mathbf {dc}_{t}=1\]
    and
        \[\mathbf 1_t^{\top}\mathbf {bc}_{t}=1.\]
\end{Rmk}

\begin{Not}\label{not::survi}
    As in Notation~\ref{not::basc}\,\ref{item::basci}, for each entity $E\in\mathcal P_t$, we denote the binary survival information by
        \[s_{E,t}=\begin{cases}1,&\text{if } E\in\mathcal P_{t+1}\\0,& \text{otherwise}\end{cases}.\]
    This gives rise to survival vectors $\mathbf s_t^{\mathcal Q}\in\mathbb R^{\mathcal Q_t}$, $\mathbf s_t^{\mathfrak W}\in\mathbb R^{\mathfrak W_t}$ and $\mathbf s_t^{\mathcal R}\in\mathbb R^{\mathcal R_t}$. But since a colony, consisting of queen worker group and replacement queen, dies or survives as a whole, all these three vectors are essentially the same (up to the canonical isomorphisms). For easier notation, we thus simply write
        \[\mathbf s_t:=\mathbf s_t^{\mathcal Q}=\mathbf s_t^{\mathfrak W}=\mathbf s_t^{\mathcal R}.\]
\end{Not}

\subsubsection{Breeding value development}\label{sec::scibv}

By Remark~\ref{rmk::brokenu}, we need to calculate the four expectations $\mathbb E\left[\hat{u}_{\mathcal N\mathcal Q_{t+1}}\right]$, $\mathbb E\left[\hat{u}_{\mathcal S\mathcal Q_{t+1}}\right]$, $\mathbb E\left[\hat{u}_{\mathcal N\mathcal R_{t+1}}\right]$, and $\mathbb E\left[\hat{u}_{\mathcal S\mathcal R_{t+1}}\right]$ in order to deduce the desired value of $\mathbb E\bigl[\hat{u}_{\mathcal P^{\ast}_{t+1}}\bigr]$. We are now equipped with the necessary tools to do so.

\begin{Lem}\label{lem::alltheus}
    We have
    \begin{align}
        \mathbb E\left[\hat{u}_{\mathcal N\mathcal Q_{t+1}}\right]
            &= \mathbf {dc}_{t}^{\top}\hat{\mathbf u}_{t}^{\mathcal R}, \label{eq::unq}\\
        \mathbb E\left[\hat{u}_{\mathcal S\mathcal Q_{t+1}}\right]
            &= \frac1{N_{t+1}^{\mathcal S}}\mathbf s_{t}^{\top}\hat{\mathbf u}_{t}^{\mathcal Q}, \label{eq::usq}\\
        \mathbb E\left[\hat u_{\mathcal N\mathcal R_{t+1}}\right]
            &= \frac12\mathbf {dc}_{t}^{\top}\hat{\mathbf u}_{t}^{\mathcal R}
                + \frac12\mathbf {bc}_{t}^{\top}\hat{\mathbf u}_{t}^{\mathcal Q}, \label{eq::unr}\\
        \mathbb E\left[\hat{u}_{\mathcal S\mathcal R_{t+1}}\right]
            &= \frac1{N_{t+1}^{\mathcal S}}\mathbf s_{t}^{\top}\hat{\mathbf u}_{t}^{\mathcal R}. \label{eq::usr}
    \end{align}
\end{Lem}

\begin{proof}
    \begin{enumerate}[label = (\roman*)]
        \item We start by Equations~\ref{eq::unq} and~\ref{eq::usq}, i.\,e.
            \[\mathbb E\left[\hat{u}_{\mathcal N\mathcal Q_{t+1}}\right]
                = \mathbf {dc}_{t}^{\top}\hat{\mathbf u}_{t}^{\mathcal R}\quad\text{and}\quad
            \mathbb E\left[\hat{u}_{\mathcal S\mathcal Q_{t+1}}\right]
                = \frac1{N_{t+1}^{\mathcal S}}\mathbf s_{t}^{\top}\hat{\mathbf u}_{t}^{\mathcal Q}.\]

        \begin{center}
            \begin{tikzpicture}
                \path (0,0) coordinate (Pt0)
                        node[left = 0mm of Pt0] {generation $\mathcal P_t$}
                    coordinate[below = 2.5cm of Pt0] (Pt1)
                        node[left = 0mm of Pt1] {generation $\mathcal P_{t+1}$}
                    (Pt0) -- (Pt1)
                        coordinate[pos = 0.4] (M)

                    coordinate[right = 2cm of Pt0] (C0)
                    node[above = 0.5cm of C0, queen] (Q) {}
                        node[left = 0cm of Q] {$Q$}
                    node[right = 2cm of Q.center, not considered, group] (DQ) {}
                        node[right = 0cm of DQ, not considered] {$\mathcal D$}
                        -- (DQ) pic[not considered] {drones}
                    node[below left = 1cm and 0.5cm of Q.center, replacement queen] (R) {}
                        node[left = 0cm of R] {$R(Q)$}

                    coordinate[right = 7.5cm of Pt0] (C1)
                    node[above = 0.5cm of C1, queen] (SQ0) {}
                        node[left = 0cm of SQ0] {$SQ$}

                    coordinate[right = 2cm of Pt1] (C2)
                    node[above = 0.5cm of C2, queen] (NQ) {}
                        node[left = 0cm of NQ] {$NQ$}
                        node[below = 0cm of NQ, font = \footnotesize] {$\mathbb E\left[\hat u_{NQ,t+1}\right]=\hat u_{R(Q),t}$}

                    coordinate[right = 7.5cm of Pt1] (C3)
                    node[above = 0.5cm of C3, queen] (SQ1) {}
                        node[left = 0cm of SQ1] {$SQ$}
                        node[below = 0cm of SQ1, font = \footnotesize] {$\mathbb E\left[\hat u_{SQ,t+1}\right]=\hat u_{SQ,t}$};

                \draw[mating, not considered] (DQ) -- (Q)
                    node[gene pass description, not considered] {mate};
            \draw[inheritance] (Q) -- (R);
            \draw[inheritance] (Q) -- (NQ)
                node[gene pass description] {dam path};
            \draw[survival] (SQ0) -- (SQ1)
                node[gene pass description] {survival path};

            \draw[dashed] (M) ++ (-3cm,0cm) -- ++(13.7cm,0cm);
            \end{tikzpicture}
        \end{center}

        The expected breeding value of a new queen $NQ\in\mathcal N\mathcal Q_{t+1}$ with dam $Q\in\mathcal Q_t$ is precisely the estimated breeding value of $Q$'s replacement queen $R(Q)$ (Lemma~\ref{lem::expdaughter} in combination with Section~\ref{sec::wgrq}). By this consideration, and the fact that breeding values are inherited proportionally to the contribution to the next generation, we obtain indeed Equation~\ref{eq::unq}.

        The expected average estimated breeding value $\mathbb E\left[\hat{u}_{\mathcal S\mathcal Q_{t+1}}\right]$ among the survivor queens is calculated precisely as in the diploid case -- all queens $SQ\in\mathcal Q_t$ that survive (i.\,e. with $SQ\in\mathcal S\mathcal Q_{t+1}$) contribute to equal parts with their respective own breeding values. This is what is described by Equation~\ref{eq::usq}.

        \item We then show Equations~\ref{eq::unr} and~\ref{eq::usr}, i.\,e.
            \[\mathbb E\left[\hat u_{\mathcal N\mathcal R_{t+1}}\right]
                = \frac12\mathbf {dc}_{t}^{\top}\hat{\mathbf u}_{t}^{\mathcal R}
                    + \frac12\mathbf {bc}_{t}^{\top}\hat{\mathbf u}_{t}^{\mathcal Q}\quad\text{and}\quad
            \mathbb E\left[\hat{u}_{\mathcal S\mathcal R_{t+1}}\right]
                = \frac1{N_{t+1}^{\mathcal S}}\mathbf s_{t}^{\top}\hat{\mathbf u}_{t}^{\mathcal R}.\]

        \begin{center}
            \begin{tikzpicture}
                \path (0,0) coordinate (Pt0)
                        node[left = 0mm of Pt0] {generation $\mathcal P_t$}
                    coordinate[below = 2.5cm of Pt0] (Pt1)
                        node[left = 0mm of Pt1] {generation $\mathcal P_{t+1}$}
                    (Pt0) -- (Pt1)
                        coordinate[pos = 0.4] (M)

                    coordinate[right = 4.5cm of Pt0] (C0)
                    node[above = 0.5cm of C0, queen] (Q) {}
                        node[left = 0cm of Q] {$Q$}

                    coordinate[right = 9cm of Pt0] (C1)
                    node[above = 0.5cm of C1, queen] (SQ0) {}
                        node[right = 0cm of SQ0] {$Q(SR)$}
                    node[below left = 1cm and 0.5cm of SQ0.center, replacement queen] (SR0) {}
                        node[left = 0cm of SR0] {$SR$}

                    coordinate[right = 2.5cm of Pt1] (C2)
                    node[above = 0.5cm of C2, queen] (NQ) {}
                        node[left = 0cm of NQ] {$Q(NR)$}
                    node[right = 2cm of NQ.center, not considered, group] (D) {}
                        node[right = 0cm of D, not considered] {$\mathcal D$}
                        -- (D) pic[not considered] {drones}
                    node[below left = 1cm and 0.5cm of NQ.center, replacement queen] (NR) {}
                        node[left = 0cm of NR] {$NR$}
                        node[below = 0cm of NR, font = \footnotesize, align = left] {$\mathbb E\left[\hat u_{NR,t+1}\right]=\frac12\mathbb E\left[\hat u_{Q(NR),t+1}\right]
                                                                                                                                 + \mathbb E\left[\hat u_{\mathcal D,t+1}\right]$\\
                                                                                     $\phantom{\mathbb E\left[\hat u_{NR,t+1}\right]}=\frac12\mathbb E\left[\hat u_{Q(NR),t+1}\right]
                                                                                                                                 + \frac12\hat u_{Q,t}$}

                    coordinate[right = 9cm of Pt1] (C3)
                    node[above = 0.5cm of C3, queen] (SQ1) {}
                        node[right = 0cm of SQ1] {$Q(SR)$}
                    node[below left = 1cm and 0.5cm of SQ1.center, replacement queen] (SR1) {}
                        node[left = 0cm of SR1] {$SR$}
                        node[below = 0cm of SR1, font = \footnotesize] {$\mathbb E\left[\hat u_{SR,t+1}\right]=\hat u_{SR,t}$};

                \draw[mating, not considered] (D) -- (NQ)
                    node[gene pass description, not considered] {mate};
                \draw[inheritance] (NQ) -- (NR);
                \draw[inheritance] (SQ0) -- (SR0);
                \draw[inheritance] (SQ1) -- (SR1);
                \draw[inheritance] (Q) -- (D)
                    node[gene pass description] {1b-path};
                \draw[survival] (SR0.south) ++(0.001pt,0pt) -- (SR1)
                    node[gene pass description] {survival path};

            \draw[dashed] (M) ++ (-3cm,0cm) -- ++(13.7cm,0cm);
            \end{tikzpicture}
        \end{center}

        The expected breeding value of a new replacement queen $NR\in\mathcal N\mathcal R_{t+1}$ is half the breeding value of its queen $Q(NR)\in\mathcal Q_{t+1}$ plus the breeding value of the drone group $\mathcal D$ that $Q(NR)$ mated with (Lemma~\ref{lem::bvrepq}). But the expected breeding value of $\mathcal D$ is half the breeding value of the queen $Q\in\mathcal Q_t$ that produced the drones (Lemma~\ref{lem::bvinherits}\,\ref{item::rrr}). The relative frequencies with which queens in $\mathcal Q_t$ occur as drone producers are given by the vector $\mathbf {bc}_{t}\in\mathbb R^{\mathcal Q_t}$. This leads to
            \[\mathbb E\left[\hat u_{\mathcal N\mathcal R_{t+1}}\right] = \frac12\mathbb E\left[\hat{u}_{\mathcal N\mathcal Q_{t+1}}\right]
                + \frac12\mathbf {bc}_{t}^{\top}\hat{\mathbf u}_{t}^{\mathcal Q}.\]
        Inserting Equation~\ref{eq::unq} yields the assertion for $\mathbb E\left[\hat u_{\mathcal N\mathcal R_{t+1}}\right]$.

        Lastly, Equation~\ref{eq::usr} holds with the exact same argument as for Equation~\ref{eq::usq}.
    \end{enumerate}
\end{proof}

By inserting the results of Lemma~\ref{lem::alltheus} into Lemma~\ref{lem::genred}, we obtain the desired formula for $\mathbb E\bigl[\hat u_{\mathcal P_{t+1}^{\ast}}\bigr]$:

\begin{Thm}\label{thm::ut}
    We have
    \begin{align}
        \mathbb E\left[\hat u_{\mathcal Q_{t+1}}\right] &=
            \frac{N_{t+1}^{\mathcal N}}{N_{t+1}}\mathbf {dc}_{t}^{\top}\hat{\mathbf u}_{t}^{\mathcal R}+\frac1{N_{t+1}}\mathbf s_t^{\top}\hat{\mathbf u}_{t}^{\mathcal Q}, \label{eq::utq}\\
        \mathbb E\left[\hat{u}_{\mathcal R_{t+1}}\right] &=
            \frac{N_{t+1}^{\mathcal N}}{2N_{t+1}}\mathbf {dc}_{t}^{\top}\hat{\mathbf u}_{t}^{\mathcal R}+\frac{N_{t+1}^{\mathcal N}}{2N_{t+1}}\mathbf {bc}_{t}^{\top}\hat{\mathbf u}_{t}^{\mathcal Q}+\frac1{N_{t+1}}\mathbf s_t^{\top}\hat{\mathbf u}_{t}^{\mathcal R}, \label{eq::utr}\\
        \mathbb E\bigl[\hat{u}_{\mathcal P^{\ast}_{t+1}}\bigr] &=
            \frac{3N_{t+1}^{\mathcal N}}{4N_{t+1}}\mathbf {dc}_{t}^{\top}\hat{\mathbf u}_{t}^{\mathcal R}
                + \frac{N_{t+1}^{\mathcal N}}{4N_{t+1}}\mathbf {bc}_{t}^{\top}\hat{\mathbf u}_{t}^{\mathcal Q}
                + \frac1{2N_{t+1}}\mathbf s_t^{\top}\left(\hat{\mathbf u}_{t}^{\mathcal R}+\hat{\mathbf u}_{t}^{\mathcal Q}\right).
    \end{align}
\end{Thm}

\subsubsection{Kinship development}\label{sec::scikin}

By Remark~\ref{rmk::brokenk}\,\ref{item::brokenkiii}, we need to calculate $k_{\mathcal N\mathcal Q_{t+1}, \mathcal N\mathcal Q_{t+1}}$, $k_{\mathcal N\mathcal Q_{t+1}, \mathcal N\mathcal R_{t+1}}$, $k_{\mathcal N\mathcal Q_{t+1}, \mathcal S\mathcal Q_{t+1}}$, $k_{\mathcal N\mathcal Q_{t+1}, \mathcal S\mathcal R_{t+1}}$, $k_{\mathcal N\mathcal R_{t+1}, \mathcal N\mathcal R_{t+1}}$, $k_{\mathcal N\mathcal R_{t+1}, \mathcal S\mathcal Q_{t+1}}$, $k_{\mathcal N\mathcal R_{t+1}, \mathcal S\mathcal R_{t+1}}$, $k_{\mathcal S\mathcal Q_{t+1}, \mathcal S\mathcal Q_{t+1}}$, $k_{\mathcal S\mathcal Q_{t+1}, \mathcal S\mathcal R_{t+1}}$, and $k_{\mathcal S\mathcal R_{t+1}, \mathcal S\mathcal R_{t+1}}$ in order to obtain the average genetic kinship in the next reduced generation, $k_{\mathcal P_{t+1}^{\ast},\mathcal P_{t+1}^{\ast}}$. This is what we will do in this section.

\begin{Lem}\label{lem::allthek}
    We have
    \begin{align}
        k_{\mathcal N\mathcal Q_{t+1},\mathcal N\mathcal Q_{t+1}}
            &= \mathbf {dc}_t^{\top}\mathbf K_t^{\mathfrak W\mathfrak W}\mathbf {dc}_t
                + \frac1{N_{t+1}^{\mathcal N}}\mathbf {dc}_t^{\top}\mathrm{diag}\left(\mathbf K_t^{\mathcal R\mathcal R}\right)
                - \frac1{N_{t+1}^{\mathcal N}}\mathbf {dc}_t^{\top}\mathrm{diag}\left(\mathbf K_t^{\mathfrak W\mathfrak W}\right), \label{eq::knqnq}\\
        k_{\mathcal N\mathcal Q_{t+1},\mathcal S\mathcal Q_{t+1}}
            &= \frac1{N_{t+1}^{\mathcal S}}\mathbf {dc}_{t}^{\top}\mathbf K_t^{\mathcal R\mathcal Q}\mathbf{s}_t, \label{eq::knqsq}\\
        k_{\mathcal S\mathcal Q_{t+1},\mathcal S\mathcal Q_{t+1}}
            &= \frac1{\left(N_{t+1}^{\mathcal S}\right)^2}\mathbf {s}_{t}^{\top}\mathbf K_t^{\mathcal Q\mathcal Q}\mathbf{s}_t,
                \label{eq::ksqsq}\\
        k_{\mathcal N\mathcal Q_{t+1},\mathcal N\mathcal R_{t+1}}
            &= \frac12\mathbf {dc}_t^{\top}\mathbf K_t^{\mathfrak W\mathfrak W}\mathbf {dc}_t
                + \frac12\mathbf {dc}_t^{\top}\mathbf K_t^{\mathcal R\mathcal Q}\mathbf {bc}_t \nonumber\\
                &\hspace{1cm}+ \frac1{2N_{t+1}^{\mathcal N}}\mathbf {dc}_t^{\top}\mathrm{diag}\left(\mathbf K_t^{\mathcal R\mathcal R}\right)
                - \frac1{2N_{t+1}^{\mathcal N}}\mathbf {dc}_t^{\top}\mathrm{diag}\left(\mathbf K_t^{\mathfrak W\mathfrak W}\right),  \label{eq::knqnr}\\
        k_{\mathcal N\mathcal Q_{t+1},\mathcal S\mathcal R_{t+1}}
            &= \frac1{N_{t+1}^{\mathcal S}}\mathbf {dc}_{t}^{\top}\mathbf K_t^{\mathfrak W\mathfrak W}\mathbf{s}_t, \label{eq::knqsr}\\
        k_{\mathcal S\mathcal Q_{t+1},\mathcal N\mathcal R_{t+1}}
            &= \frac1{2N_{t+1}^{\mathcal S}}\mathbf {dc}_{t}^{\top}\mathbf K_t^{\mathcal R\mathcal Q}\mathbf{s}_t
                + \frac1{2N_{t+1}^{\mathcal S}}\mathbf {bc}_t^{\top}\mathbf K_t^{\mathcal Q\mathcal Q}\mathbf s_t, \label{eq::ksqnr}\\
        k_{\mathcal S\mathcal Q_{t+1},\mathcal S\mathcal R_{t+1}}
            &= \frac1{\left(N_{t+1}^{\mathcal S}\right)^2}\mathbf s_t^{\top}\mathbf K_t^{\mathcal Q\mathcal R}\mathbf s_t,
                \label{eq::ksqsr}\\
        k_{\mathcal N\mathcal R_{t+1},\mathcal N\mathcal R_{t+1}}
            &=\frac14\mathbf {dc}_t^{\top}\mathbf K_{t}^{\mathfrak W\mathfrak W}\mathbf{dc}_t
                + \frac12\mathbf {bc}_t^{\top}\mathbf K_{t}^{\mathcal Q\mathcal R}\mathbf{dc}_t
                + \frac14\mathbf {bc}_t^{\top}\mathbf K_{t}^{\mathcal Q\mathcal Q}\mathbf{bc}_t\nonumber\\
            &\hspace{1cm}- \frac1{4N_{t+1}^{\mathcal N}}\mathbf{dc}_t^{\top}
                \mathrm{diag}\left(\mathbf K_t^{\mathfrak W\mathfrak W}\right)
                - \frac1{4N_{t+1}^{\mathcal N}}\mathbf{bc}_t^{\top}\mathrm{diag}\left(\mathbf K_t^{\mathcal Q\mathcal Q}\right)
                + \frac1{2N_{t+1}^{\mathcal N}}, \label{eq::knrnr}\\
        k_{\mathcal N\mathcal R_{t+1},\mathcal S\mathcal R_{t+1}}
            &= \frac1{2N_{t+1}^{\mathcal S}}\mathbf {dc}_{t}^{\top}\mathbf K_t^{\mathfrak W\mathfrak W}\mathbf{s}_t
                    +\frac1{2N_{t+1}^{\mathcal S}}\mathbf {bc}_t^{\top}\mathbf K_{t}^{\mathcal Q\mathcal R}\mathbf s_t,
                \label{eq::knrsr}\\
        k_{\mathcal S\mathcal R_{t+1},\mathcal S\mathcal R_{t+1}}
            &= \frac1{\left(N_{t+1}^{\mathcal S}\right)^2}\mathbf s_t^{\top}\mathbf K_t^{\mathcal R\mathcal R}\mathbf s_t.
                \label{eq::ksrsr}
    \end{align}
\end{Lem}

\begin{proof}
    \begin{enumerate}[label = (\roman*)]
        \item \label{item::knqnq} We show Equation~\ref{eq::knqnq}, i.\,e.
            \[k_{\mathcal N\mathcal Q_{t+1},\mathcal N\mathcal Q_{t+1}}
                = \mathbf {dc}_t^{\top}\mathbf K_t^{\mathfrak W\mathfrak W}\mathbf {dc}_t
                    + \frac1{N_{t+1}^{\mathcal N}}\mathbf {dc}_t^{\top}\mathrm{diag}\left(\mathbf K_t^{\mathcal R\mathcal R}\right)
                    - \frac1{N_{t+1}^{\mathcal N}}\mathbf {dc}_t^{\top}\mathrm{diag}\left(\mathbf K_t^{\mathfrak W\mathfrak W}\right).\]
        \begin{center}
            \begin{tikzpicture}
                \path (0,0) coordinate (Pt0)
                        node[left = 0mm of Pt0] {generation $\mathcal P_t$}
                    coordinate[below = 4cm of Pt0] (Pt1)
                        node[left = 0mm of Pt1] {generation $\mathcal P_{t+1}$}
                    (Pt0) -- (Pt1)
                        coordinate[pos = 0.4] (M)

                    coordinate[right = 2.5cm of Pt0] (C0)
                    node[above = 1cm of C0, queen] (Q0) {}
                        node[left = 0cm of Q0] {$Q_1$}
                    node[below right = 2cm and 0.5cm of Q0.center, worker group] (W0) {}
                        node[right = 0cm of W0] {$\mathcal W(Q_1)$}
                    node[below left = 1cm and 0.5cm of Q0.center, replacement queen] (R0) {}
                        node[left = 0cm of R0] {$R(Q_1)$}

                    coordinate[right = 6cm of Pt0] (C1)
                    node[above = 1cm of C1, queen] (Q1) {}
                        node[left = 0cm of Q1] {$Q_2$}
                    node[below right = 2cm and 0.5cm of Q1.center, worker group] (W1) {}
                        node[right = 0cm of W1] {$\mathcal W(Q_2)$}
                    node[below left = 1cm and 0.5cm of Q1.center, replacement queen] (R1) {}
                        node[left = 0cm of R1] {$R(Q_2)$}

                    coordinate[right = 9.5cm of Pt0] (C2)
                    node[above = 1cm of C2, queen] (Q2) {}
                        node[left = 0cm of Q2] {$Q$}
                    node[below right = 2cm and 0.5cm of Q2.center, worker group] (W2) {}
                    node[below left = 1cm and 0.5cm of Q2.center, replacement queen] (R2) {}
                        node[left = 0cm of R2] {$R(Q)$}

                    coordinate[right = 2cm of Pt1] (C3)
                    node[above = 1cm of C3, queen] (Q3) {}
                        node[left = 0cm of Q3] {$NQ_1$}

                    coordinate[right = 5.5cm of Pt1] (C4)
                    node[above = 1cm of C4, queen] (Q4) {}
                        node[left = 0cm of Q4] {$NQ_2$}

                    coordinate[right = 9cm of Pt1] (C5)
                    node[above = 1cm of C5, queen] (Q5) {}
                        node[left = 0cm of Q5] {$NQ$};

                    \draw[inheritance] (Q0) -- (R0);
                    \draw[inheritance] (Q0) -- (W0);
                    \draw[inheritance] (Q1) -- (R1);
                    \draw[inheritance] (Q1) -- (W1);
                    \draw[inheritance] (Q2) -- (R2);
                    \draw[inheritance] (Q2) -- (W2);
                    \draw[inheritance] (Q0) -- (Q3);
                    \draw[inheritance] (Q1) -- (Q4);
                    \draw[inheritance] (Q2) -- (Q5);
                    \draw[dotted, thick] (W0) -- (Q3);
                    \draw[dotted, thick] (W1) -- (Q4);
                    \draw[dotted, thick] (R2) -- (Q5);
                    \draw[relationship] (W0) -- (W1);
                    \draw[relationship] (R2.west) -- (R2.east);
                    \draw[relationship] (Q3) -- (Q4)
                        node[relationship description] {$k_{NQ_1,NQ_2}=k_{W(Q_1),W(Q_2)}$};
                    \draw[relationship] (Q5.west) -- (Q5.east)
                        node[relationship description] {$k_{NQ,NQ}=k_{R(Q),R(Q)}$};

                \draw[dashed] (M) ++ (-3cm,0cm) -- ++(13.9cm,0cm);
            \end{tikzpicture}
        \end{center}

        Let $NQ_1,NQ_2\in\mathcal N\mathcal Q_{t+1}$ be two non-identical newly hatched queens and let $Q_1,Q_2\in\mathcal Q_t$ be their respective dam queens (possibly identical). Then neither of $NQ_1$ and $NQ_2$ is an ancestor of the other and thus by Corollary~\ref{cor::imp},
            \[k_{NQ_1,NQ_2}=k_{\mathcal W(Q_1),\mathcal W(Q_2)}.\]
        (Note that here we actually need worker groups, not replacement queens, because otherwise the statement is not true in case $NQ_1$ and $NQ_2$ are siblings, i.\,e. $Q_1=Q_2$.)
        Under the (wrong) assumption that $k_{NQ_1,NQ_2}=k_{\mathcal W(Q_1),\mathcal W(Q_2)}$ also holds for $NQ_1=NQ_2$ this would result in an average kinship between the newly hatched queens in $\mathcal N\mathcal Q_{t+1}$ of $\mathbf {dc}_t^{\top}\mathbf K_t^{\mathfrak W\mathfrak W}\mathbf {dc}_t$. However, the kinship of a newly hatched queen $NQ\in\mathcal N\mathcal Q_{t+1}$ with dam $Q\in\mathcal Q_t$ to itself is not $k_{\mathcal W(Q),\mathcal W(Q)}$, but
            \[k_{NQ,NQ}=k_{R(Q),R(Q)}.\]
        If we look at all possible kinships $k_{NQ_1,NQ_2}$ between newly hatched queens, a fraction of $\frac1{N_t^{\mathcal N}}$ of them are self-kinships of the form $k_{NQ,NQ}$. For these, we have to add the correction terms $k_{R(Q),R(Q)}-k_{\mathcal W(Q),\mathcal W(Q)}$ multiplied with the frequency $dc_{Q,t}$ with which $Q\in\mathcal Q_t$ occurs as a dam. So, in total, we have
            \[k_{\mathcal N\mathcal Q_{t+1},\mathcal N\mathcal Q_{t+1}}
                = \mathbf {dc}_t^{\top}\mathbf K_t^{\mathfrak W\mathfrak W}\mathbf {dc}_t
                    + \frac1{N_t^{\mathcal N}}\sum_{Q\in\mathcal Q_t}dc_{Q,t}\left(k_{R(Q),R(Q)}-k_{\mathcal W(Q),\mathcal W(Q)}\right).\]
        and the assertion follows.

        \item \label{item::knqsq} We show Equation~\ref{eq::knqsq}, i.\,e.
            \[k_{\mathcal N\mathcal Q_{t+1},\mathcal S\mathcal Q_{t+1}}
                =\frac1{N_{t+1}^{\mathcal S}}\mathbf {dc}_{t}^{\top}\mathbf K_t^{\mathcal R\mathcal Q}\mathbf{s}_t.\]

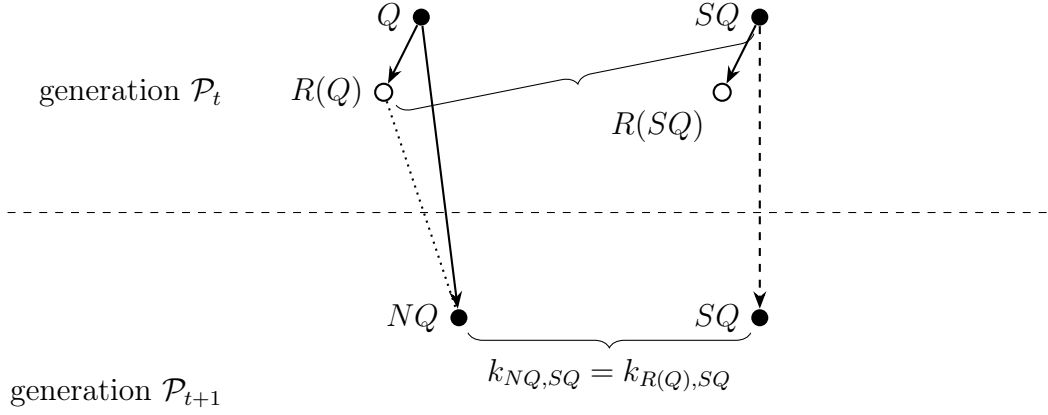
        \begin{center}
            \begin{tikzpicture}
                \path (0,0) coordinate (Pt0)
                        node[left = 0mm of Pt0] {generation $\mathcal P_t$}
                    coordinate[below = 4cm of Pt0] (Pt1)
                        node[left = 0mm of Pt1] {generation $\mathcal P_{t+1}$}
                    (Pt0) -- (Pt1)
                        coordinate[pos = 0.4] (M)

                    coordinate[right = 2.5cm of Pt0] (C0)
                    node[above = 1cm of C0, queen] (Q0) {}
                        node[left = 0cm of Q0] {$Q$}
                    node[below left = 1cm and 0.5cm of Q0.center, replacement queen] (R0) {}
                        node[left = 0cm of R0] {$R(Q)$}

                    coordinate[right = 7cm of Pt0] (C1)
                    node[above = 1cm of C1, queen] (Q1) {}
                        node[left = 0cm of Q1] {$SQ$}
                    node[below left = 1cm and 0.5cm of Q1.center, replacement queen] (R1) {}
                        node[below left = 0cm and 0cm of R1] {$R(SQ)$}

                    coordinate[right = 3cm of Pt1] (C3)
                    node[above = 1cm of C3, queen] (Q3) {}
                        node[left = 0cm of Q3] {$NQ$}

                    coordinate[right = 7cm of Pt1] (C4)
                    node[above = 1cm of C4, queen] (Q4) {}
                        node[left = 0cm of Q4] {$SQ$};

                    \draw[inheritance] (Q0) -- (R0);
                    \draw[inheritance] (Q1) -- (R1);
                    \draw[inheritance] (Q0) -- (Q3);
                    \draw[survival] (Q1) -- (Q4);
                    \draw[dotted, thick] (R0) -- (Q3);
                    \draw[relationship] (R0) -- (Q1);
                    \draw[relationship] (Q3) -- (Q4)
                        node[relationship description] {$k_{NQ,SQ}=k_{R(Q),SQ}$};

                \draw[dashed] (M) ++ (-3cm,0cm) -- ++(13.9cm,0cm);
            \end{tikzpicture}
        \end{center}

        Let $NQ\in\mathcal N\mathcal Q_{t+1}$ be a newly hatched queen with dam $Q\in\mathcal Q_t$ and let $SQ\in\mathcal S\mathcal Q_{t+1}\subseteq\mathcal Q_t$ be a survivor queen. Then the younger queen $NQ$ cannot be an ancestor of $SQ$ and neither is the replacement queen $R(Q)$ of $Q$. Thus, by Lemma~\ref{lem::kinbet}, we have
            \[k_{NQ,SQ}=k_{R(Q),SQ}.\]
        The frequency with which a specific queen $Q\in\mathcal Q_t$ occurs as a dam of a queen $NQ\in\mathcal N\mathcal Q_{t+1}$ is $dc_{Q,t}$ and the frequency with which it is identical with a survivor queen $SQ\in\mathcal S\mathcal Q_{t+1}$ is $\frac1{N_{t+1}^{\mathcal S}}s_{Q,t}$. From this, we conclude the assertion.

        \item \label{item::ksqsq} We show Equation~\ref{eq::ksqsq}, i.\,e.
            \[k_{\mathcal S\mathcal Q_{t+1},\mathcal S\mathcal Q_{t+1}} =
                \frac1{\left(N_{t+1}^{\mathcal S}\right)^2}\mathbf s_t^{\top}\mathbf K_t^{\mathcal Q\mathcal Q}\mathbf{s}_t.\]

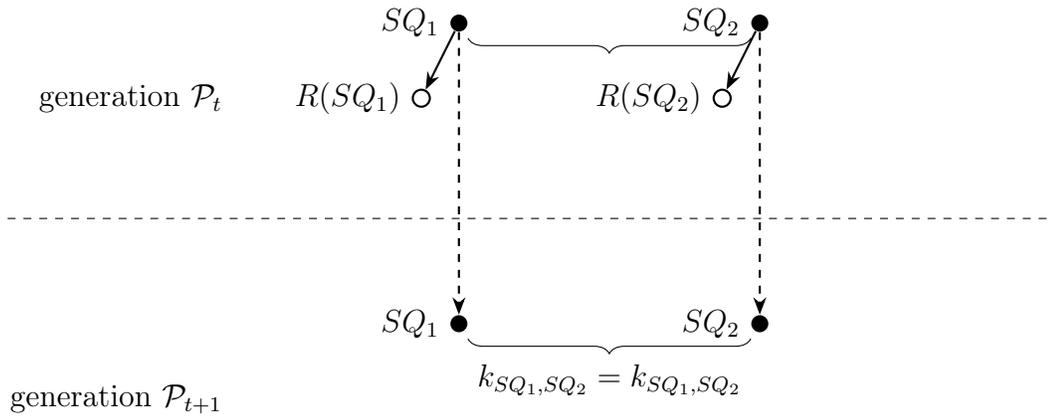
        \begin{center}
            \begin{tikzpicture}
                \path (0,0) coordinate (Pt0)
                        node[left = 0mm of Pt0] {generation $\mathcal P_t$}
                    coordinate[below = 4cm of Pt0] (Pt1)
                        node[left = 0mm of Pt1] {generation $\mathcal P_{t+1}$}
                    (Pt0) -- (Pt1)
                        coordinate[pos = 0.4] (M)

                    coordinate[right = 3cm of Pt0] (C0)
                    node[above = 1cm of C0, queen] (Q0) {}
                        node[left = 0cm of Q0] {$SQ_1$}
                    node[below left = 1cm and 0.5cm of Q0.center, replacement queen] (R0) {}
                        node[left = 0cm of R0] {$R(SQ_1)$}

                    coordinate[right = 7cm of Pt0] (C1)
                    node[above = 1cm of C1, queen] (Q1) {}
                        node[left = 0cm of Q1] {$SQ_2$}
                    node[below left = 1cm and 0.5cm of Q1.center, replacement queen] (R1) {}
                        node[left = 0cm of R1] {$R(SQ_2)$}

                    coordinate[right = 3cm of Pt1] (C3)
                    node[above = 1cm of C3, queen] (Q3) {}
                        node[left = 0cm of Q3] {$SQ_1$}

                    coordinate[right = 7cm of Pt1] (C4)
                    node[above = 1cm of C4, queen] (Q4) {}
                        node[left = 0cm of Q4] {$SQ_2$};

                    \draw[survival] (Q0) -- (Q3);
                    \draw[survival] (Q1) -- (Q4);
                    \draw[inheritance] (Q0) -- (R0);
                    \draw[inheritance] (Q1) -- (R1);
                    \draw[relationship] (Q0) -- (Q1);
                    \draw[relationship] (Q3) -- (Q4)
                        node[relationship description] {$k_{SQ_1,SQ_2}=k_{SQ_1,SQ_2}$};

                \draw[dashed] (M) ++ (-3cm,0cm) -- ++(13.9cm,0cm);
            \end{tikzpicture}
        \end{center}

        Just as in part~\ref{item::olkiii} of the proof to Lemma~\ref{lem::olk}, this follows from the fact that kinships between surviving colonies do not change over time.

        \item \label{item::knqnr} We show Equation~\ref{eq::knqnr}, i.\,e.
        \begin{align*}
            k_{\mathcal N\mathcal Q_{t+1},\mathcal N\mathcal R_{t+1}}
                &= \frac12\mathbf {dc}_t^{\top}\mathbf K_t^{\mathfrak W\mathfrak W}\mathbf {dc}_t
                    + \frac12\mathbf {dc}_t^{\top}\mathbf K_t^{\mathcal R\mathcal Q}\mathbf {bc}_t \nonumber\\
                    &\hspace{1cm}+ \frac1{2N_{t+1}^{\mathcal N}}\mathbf {dc}_t^{\top}\mathrm{diag}\left(\mathbf K_t^{\mathcal R\mathcal R}\right)
                    - \frac1{2N_{t+1}^{\mathcal N}}\mathbf {dc}_t^{\top}\mathrm{diag}\left(\mathbf K_t^{\mathfrak W\mathfrak W}\right).
        \end{align*}

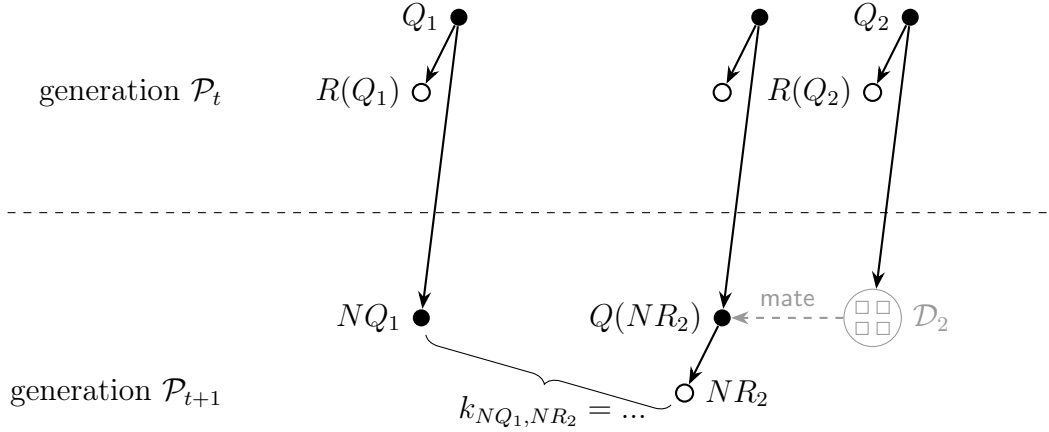
        \begin{center}
            \begin{tikzpicture}
                \path (0,0) coordinate (Pt0)
                        node[left = 0mm of Pt0] {generation $\mathcal P_t$}
                    coordinate[below = 4cm of Pt0] (Pt1)
                        node[left = 0mm of Pt1] {generation $\mathcal P_{t+1}$}
                    (Pt0) -- (Pt1)
                        coordinate[pos = 0.4] (M)

                    coordinate[right = 3cm of Pt0] (C0)
                    node[above = 1cm of C0, queen] (Q0) {}
                        node[left = 0cm of Q0] {$Q_1$}
                    node[below left = 1cm and 0.5cm of Q0.center, replacement queen] (R0) {}
                        node[left = 0cm of R0] {$R(Q_1)$}

                    coordinate[right = 7cm of Pt0] (C1)
                    node[above = 1cm of C1, queen] (Q1) {}
                        node[left = 0cm of Q1] {}
                    node[below left = 1cm and 0.5cm of Q1.center, replacement queen] (R1) {}
                        node[left = 0cm of R1] {}

                    coordinate[right = 2.5cm of Pt1] (C3)
                    node[above = 1cm of C3, queen] (Q3) {}
                        node[left = 0cm of Q3] {$NQ_1$}

                    coordinate[right = 6.5cm of Pt1] (C4)
                    node[above = 1cm of C4, queen] (Q4) {}
                        node[left = 0cm of Q4] {$Q(NR_2)$}
                    node[right = 2cm of Q4.center, not considered, group] (D) {}
                        node[right = 0cm of D, not considered] {$\mathcal D_2$}
                        -- (D) pic[not considered] {drones}
                    node[below left = 1cm and 0.5cm of Q4.center, replacement queen] (R4) {}
                        node[right = 0cm of R4] {$NR_2$}

                    coordinate[right = 9cm of Pt0] (C5)
                    node[above = 1cm of C5, queen] (Q5) {}
                        node[left = 0cm of Q5] {$Q_2$}
                    node[below left = 1cm and 0.5cm of Q5.center, replacement queen] (R5) {}
                        node[left = 0cm of R5] {$R(Q_2)$};

                    \draw[inheritance] (Q0) -- (R0);
                    \draw[inheritance] (Q1) -- (R1);
                    \draw[inheritance] (Q4) -- (R4);
                    \draw[inheritance] (Q5) -- (R5);
                    \draw[inheritance] (Q0) -- (Q3);
                    \draw[inheritance] (Q1) -- (Q4);
                    \draw[inheritance] (Q5) -- (D);
                    \draw[mating, not considered] (D) -- (Q4)
                        node[gene pass description, not considered] {mate};
                    \draw[relationship] (Q3) -- (R4)
                        node[relationship description] {$k_{NQ_1,NR_2}=...$};

                \draw[dashed] (M) ++ (-3cm,0cm) -- ++(13.9cm,0cm);
            \end{tikzpicture}
        \end{center}

        Let $NQ_1\in\mathcal N\mathcal Q_{t+1}$ and $NR_2\in\mathcal N\mathcal R_{t+1}$. We fix a locus and draw an allele $A^{1}$ from $NQ_1$ and an allele $A^{2}$ from $NR_2$. The latter allele comes with probability $\frac12$ from $NR_2$'s  dam $Q(NR_2)$ and with probability $\frac12$ from the group $\mathcal D_2$ of drones that $Q(NR_2)$ mated with. Let $Q_2\in\mathcal Q_t$ be the dam of these drones. Then,
            \[k_{NQ_1,NR_2}=\frac12k_{NQ_1,Q(NR_2)}+\frac12k_{NQ_1,\mathcal D_2}\]
        and thus by Equation~\ref{eq::kgd} (Lemma~\ref{lem::reldron})
        \begin{equation} \label{eq::helper}
            k_{NQ_1,NR_2}=\frac12k_{NQ_1,Q(NR_2)}+\frac12k_{NQ_1,Q_2}.
        \end{equation}
        Let $Q_1\in\mathcal Q_t$ be the dam of $NQ_1$. Since $NQ_1\in\mathcal N\mathcal Q_{t+1}$ cannot be an ancestor of $Q_2\in\mathcal Q_t$, we have by Lemma~\ref{lem::kinbet}
            \[k_{NQ_1,Q_2}=k_{R(Q_1),Q_2}.\]
        and thus by inserting this into Equation~\ref{eq::helper}
            \[k_{NQ_1,NR_2}=\frac12k_{NQ_1,Q(NR_2)}+\frac12k_{R(Q_1),Q_2}.\]
        If now, we take averages over all choices of $NQ_1$ and $NR_2$, any given queen $Q\in\mathcal Q_t$ will occur in the role of $Q_1$ with frequency $dc_{Q,t}$ and in the role of $Q_2$ with frequency $bc_{Q,t}$. By that, we have
        \begin{align*}
            k_{\mathcal N\mathcal Q_{t+1},\mathcal N\mathcal R_{t+1}}
                &= \frac12k_{\mathcal N\mathcal Q_{t+1},\mathcal N\mathcal Q_{t+1}}
                    +\frac12\mathbf {dc}_t^{\top}\mathbf K_t^{\mathcal R\mathcal Q}\mathbf {bc}_t
        \end{align*}
        and by inserting Equation~\ref{eq::knqnq} for $k_{\mathcal N\mathcal Q_{t+1},\mathcal N\mathcal Q_{t+1}}$ (shown in~\ref{item::knqnq}), the assertion follows.

        \item \label{item::knqsr} We show Equation~\ref{eq::knqsr}, i.\,e.
            \[k_{\mathcal N\mathcal Q_{t+1},\mathcal S\mathcal R_{t+1}} =
                \frac1{N_{t+1}^{\mathcal S}}\mathbf {dc}_{t}^{\top}\mathbf K_t^{\mathfrak W\mathfrak W}\mathbf{s}_t.\]

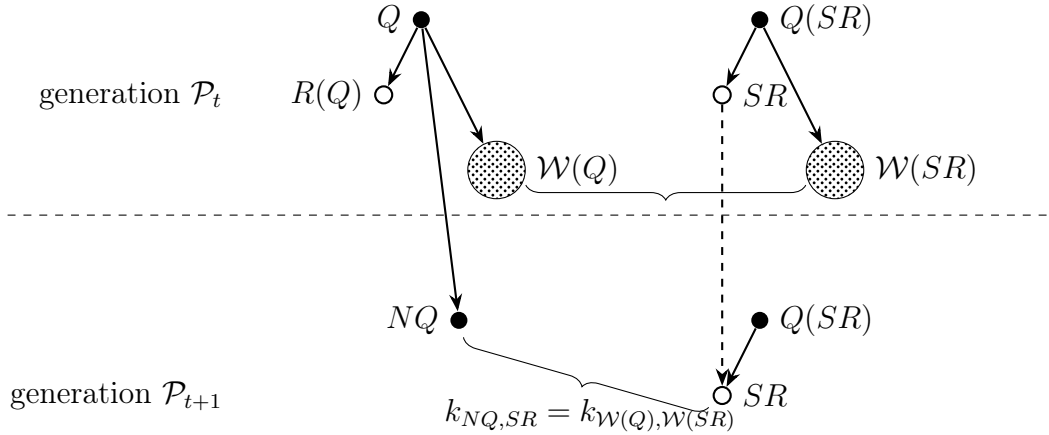
        \begin{center}
            \begin{tikzpicture}
                \path (0,0) coordinate (Pt0)
                        node[left = 0mm of Pt0] {generation $\mathcal P_t$}
                    coordinate[below = 4cm of Pt0] (Pt1)
                        node[left = 0mm of Pt1] {generation $\mathcal P_{t+1}$}
                    (Pt0) -- (Pt1)
                        coordinate[pos = 0.4] (M)

                    coordinate[right = 2.5cm of Pt0] (C0)
                    node[above = 1cm of C0, queen] (Q0) {}
                        node[left = 0cm of Q0] {$Q$}
                    node[below right = 2cm and 1cm of Q0.center, worker group] (W0) {}
                        node[right = 0cm of W0] {$\mathcal W(Q)$}
                    node[below left = 1cm and 0.5cm of Q0.center, replacement queen] (R0) {}
                        node[left = 0cm of R0] {$R(Q)$}

                    coordinate[right = 7cm of Pt0] (C1)
                    node[above = 1cm of C1, queen] (Q1) {}
                        node[right = 0cm of Q1] {$Q(SR)$}
                    node[below right = 2cm and 1cm of Q1.center, worker group] (W1) {}
                        node[right = 0cm of W1] {$\mathcal W(SR)$}
                    node[below left = 1cm and 0.5cm of Q1.center, replacement queen] (R1) {}
                        node[right = 0cm of R1] {$SR$}

                    coordinate[right = 3cm of Pt1] (C3)
                    node[above = 1cm of C3, queen] (Q3) {}
                        node[left = 0cm of Q3] {$NQ$}

                    coordinate[right = 7cm of Pt1] (C4)
                    node[above = 1cm of C4, queen] (Q4) {}
                        node[right = 0cm of Q4] {$Q(SR)$}
                    node[below left = 1cm and 0.5cm of Q4.center, replacement queen] (R4) {}
                        node[right = 0cm of R4] {$SR$};

                    \draw[inheritance] (Q0) -- (R0);
                    \draw[inheritance] (Q1) -- (R1);
                    \draw[inheritance] (Q0) -- (W0);
                    \draw[inheritance] (Q1) -- (W1);
                    \draw[inheritance] (Q4) -- (R4);
                    \draw[inheritance] (Q0) -- (Q3);
                    \draw[survival] (R1) -- (R4);
                    \draw[relationship] (W0) -- (W1);
                    \draw[relationship] (Q3) -- (R4)
                        node[relationship description] {$k_{NQ,SR}=k_{\mathcal W(Q),\mathcal W(SR)}$};

                \draw[dashed] (M) ++ (-3cm,0cm) -- ++(13.9cm,0cm);
            \end{tikzpicture}
        \end{center}

        Let $NQ\in\mathcal N\mathcal Q_{t+1}$ be a newly hatched queen with dam $Q\in\mathcal Q_t$ and let $SR\in\mathcal S\mathcal R_{t+1}\subseteq\mathcal R_t$ be a survivor replacement queen. Then $NQ$ cannot be identical with $SR$ and therefore, by Corollary~\ref{cor::imp}, we have
            \[k_{NQ,SR}=k_{\mathcal W(Q),\mathcal W(SR)}.\]
        Note, that we need to work with worker groups instead of replacement queens to cover the case $Q=Q(SR)$.
        The frequency with which a specific queen $Q\in\mathcal Q_t$ occurs as a dam of a queen $NQ\in\mathcal N\mathcal Q_{t+1}$ is $dc_{Q,t}$ and the frequency with which a replacement queen $R\in\mathcal R_t$ is identical with a survivor replacement queen $SR\in\mathcal S\mathcal R_{t+1}$ is $\frac1{N_{t+1}^{\mathcal S}}s_{R,t}$. From this, we conclude the claimed identity.

        \item \label{item::ksqnr} We show Equation~\ref{eq::ksqnr}, i.\,e.
            \[k_{\mathcal S\mathcal Q_{t+1},\mathcal N\mathcal R_{t+1}} =
                \frac1{2N_{t+1}^{\mathcal S}}\mathbf {dc}_{t}^{\top}\mathbf K_t^{\mathcal R\mathcal Q}\mathbf{s}_t^
                    +\frac1{2N_{t+1}^{\mathcal S}}\mathbf {bc}_t^{\top}\mathbf K_t^{\mathcal Q\mathcal Q}\mathbf s_t.\]

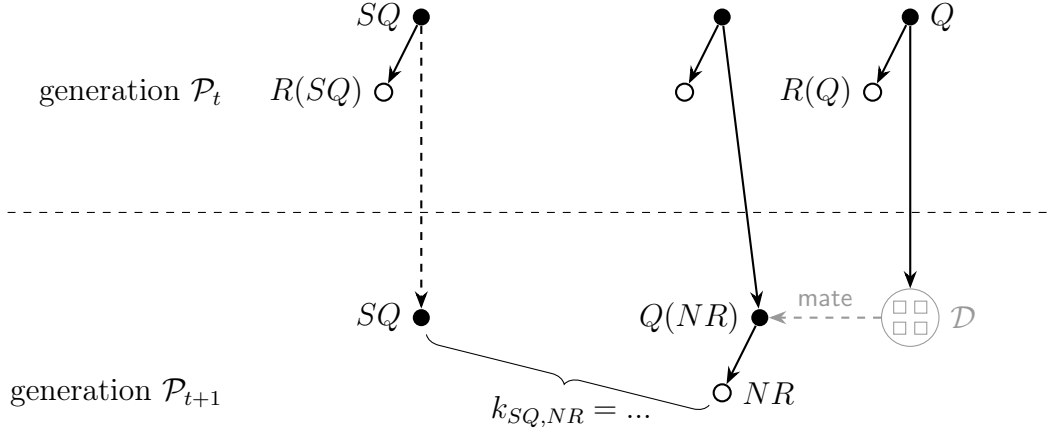
        \begin{center}
            \begin{tikzpicture}
                \path (0,0) coordinate (Pt0)
                        node[left = 0mm of Pt0] {generation $\mathcal P_t$}
                    coordinate[below = 4cm of Pt0] (Pt1)
                        node[left = 0mm of Pt1] {generation $\mathcal P_{t+1}$}
                    (Pt0) -- (Pt1)
                        coordinate[pos = 0.4] (M)

                    coordinate[right = 2.5cm of Pt0] (C0)
                    node[above = 1cm of C0, queen] (Q0) {}
                        node[left = 0cm of Q0] {$SQ$}
                    node[below left = 1cm and 0.5cm of Q0.center, replacement queen] (R0) {}
                        node[left = 0cm of R0] {$R(SQ)$}

                    coordinate[right = 6.5cm of Pt0] (C1)
                    node[above = 1cm of C1, queen] (Q1) {}
                        node[right = 0cm of Q1] {}
                    node[below left = 1cm and 0.5cm of Q1.center, replacement queen] (R1) {}
                        node[left = 0cm of R1] {}

                    coordinate[right = 9cm of Pt0] (C2)
                    node[above = 1cm of C2, queen] (Q2) {}
                        node[right = 0cm of Q2] {$Q$}
                    node[below left = 1cm and 0.5cm of Q2.center, replacement queen] (R2) {}
                        node[left = 0cm of R2] {$R(Q)$}

                    coordinate[right = 2.5cm of Pt1] (C3)
                    node[above = 1cm of C3, queen] (Q3) {}
                        node[left = 0cm of Q3] {$SQ$}

                    coordinate[right = 7cm of Pt1] (C4)
                    node[above = 1cm of C4, queen] (Q4) {}
                        node[left = 0cm of Q4] {$Q(NR)$}
                    node[right = 2cm of Q4.center, not considered, group] (D) {}
                        node[right = 0cm of D, not considered] {$\mathcal D$}
                        -- (D) pic[not considered] {drones}
                    node[below left = 1cm and 0.5cm of Q4.center, replacement queen] (R4) {}
                        node[right = 0cm of R4] {$NR$};

                    \draw[inheritance] (Q0) -- (R0);
                    \draw[inheritance] (Q1) -- (R1);
                    \draw[inheritance] (Q1) -- (Q4);
                    \draw[inheritance] (Q2) -- (R2);
                    \draw[inheritance] (Q4) -- (R4);
                    \draw[inheritance] (Q2) -- (D);
                    \draw[survival] (Q0) -- (Q3);
                    \draw[mating, not considered] (D) -- (Q4)
                        node[gene pass description, not considered]{mate};
                    \draw[relationship] (Q3) -- (R4)
                        node[relationship description] {$k_{SQ,NR}=...$};

                \draw[dashed] (M) ++ (-3cm,0cm) -- ++(13.9cm,0cm);
            \end{tikzpicture}
        \end{center}

        Let $SQ\in\mathcal S\mathcal Q_{t+1}\subseteq\mathcal Q_t$ be a survivor queen and let $NR\in\mathcal N\mathcal R_{t+1}$ be a newly hatched replacement queen whose dam $Q(NR)\in\mathcal N\mathcal Q_{t+1}$ mated with a group $\mathcal D$ of drones. Let $Q\in\mathcal Q_t$ be the dam of $\mathcal D$. Then
        \begin{align*}
            k_{SQ,NR} &= \frac12k_{SQ,Q(NR)}+\frac12k_{SQ,\mathcal D}\\
                      &= \frac12k_{SQ,Q(NR)}+\frac12k_{SQ,Q}.
        \end{align*}
        Taking averages, we obtain by the usual arguments
        \begin{align*}
            k_{\mathcal S\mathcal Q_{t+1},\mathcal N\mathcal R_{t+1}}
                &=\frac12k_{\mathcal S\mathcal Q_{t+1},\mathcal N\mathcal Q_{t+1}}
                    +\frac1{2N_{t+1}^{\mathcal S}}\mathbf {bc}_t^{\top}\mathbf K_t^{\mathcal Q\mathcal Q}\mathbf s_t.
        \end{align*}
        The assertion follows by inserting Equation~\ref{eq::knqsq} (shown in~\ref{item::knqsq}).

        \item \label{item::ksqsr} We show Equation~\ref{eq::ksqsr}, i.\,e.
            \[k_{\mathcal S\mathcal Q_{t+1},\mathcal S\mathcal R_{t+1}} =
                \frac1{\left(N_{t+1}^{\mathcal S}\right)^2}\mathbf s_t^{\top}\mathbf K_t^{\mathcal Q\mathcal R}\mathbf s_t.\]

        \begin{center}
            \begin{tikzpicture}
                \path (0,0) coordinate (Pt0)
                        node[left = 0mm of Pt0] {generation $\mathcal P_t$}
                    coordinate[below = 4cm of Pt0] (Pt1)
                        node[left = 0mm of Pt1] {generation $\mathcal P_{t+1}$}
                    (Pt0) -- (Pt1)
                        coordinate[pos = 0.4] (M)

                    coordinate[right = 3cm of Pt0] (C0)
                    node[above = 1cm of C0, queen] (Q0) {}
                        node[left = 0cm of Q0] {$SQ_1$}
                    node[below left = 1cm and 0.5cm of Q0.center, replacement queen] (R0) {}
                        node[left = 0cm of R0] {$R(SQ_1)$}

                    coordinate[right = 7cm of Pt0] (C1)
                    node[above = 1cm of C1, queen] (Q1) {}
                        node[right = 0cm of Q1] {$Q(SR_2)$}
                    node[below left = 1cm and 0.5cm of Q1.center, replacement queen] (R1) {}
                        node[right = 0cm of R1] {$SR_2$}

                    coordinate[right = 3cm of Pt1] (C3)
                    node[above = 1cm of C3, queen] (Q3) {}
                        node[left = 0cm of Q3] {$SQ_1$}

                    coordinate[right = 7cm of Pt1] (C4)
                    node[above = 1cm of C4, queen] (Q4) {}
                        node[right = 0cm of Q4] {$Q(SR_2)$}
                    node[below left = 1cm and 0.5cm of Q4.center, replacement queen] (R4) {}
                        node[right = 0cm of R4] {$SR_2$};

                    \draw[survival] (Q0) -- (Q3);
                    \draw[survival] (R1) -- (R4);
                    \draw[inheritance] (Q0) -- (R0);
                    \draw[inheritance] (Q1) -- (R1);
                    \draw[inheritance] (Q4) -- (R4);
                    \draw[relationship] (Q0) -- (R1);
                    \draw[relationship] (Q3) -- (R4)
                        node[relationship description] {$k_{SQ_1,SR_2}=k_{SQ_1,SR_2}$};

                \draw[dashed] (M) ++ (-3cm,0cm) -- ++(13.9cm,0cm);
            \end{tikzpicture}
        \end{center}

        Just as in ~\ref{item::ksqsq} (i.\,e. proof of Equation~\ref{eq::ksqsq}), this follows from the fact that kinships between surviving colonies do not change over time.

        \item \label{item::knrnr} We show Equation~\ref{eq::knrnr}, i.\,e.
        \begin{align*}
            k_{\mathcal N\mathcal R_{t+1},\mathcal N\mathcal R_{t+1}}
                &=\frac14\mathbf {dc}_t^{\top}\mathbf K_{t}^{\mathfrak W\mathfrak W}\mathbf{dc}_t
                    + \frac12\mathbf {bc}_t^{\top}\mathbf K_{t}^{\mathcal Q\mathcal R}\mathbf{dc}_t
                    + \frac14\mathbf {bc}_t^{\top}\mathbf K_{t}^{\mathcal Q\mathcal Q}\mathbf{bc}_t\\
                &\hspace{1cm}- \frac1{4N_{t+1}^{\mathcal N}}\mathbf{dc}_t^{\top}
                    \mathrm{diag}\left(\mathbf K_t^{\mathfrak W\mathfrak W}\right)
                    - \frac1{4N_{t+1}^{\mathcal N}}\mathbf{bc}_t^{\top}\mathrm{diag}\left(\mathbf K_t^{\mathcal Q\mathcal Q}\right)
                    + \frac1{2N_{t+1}^{\mathcal N}}.
        \end{align*}

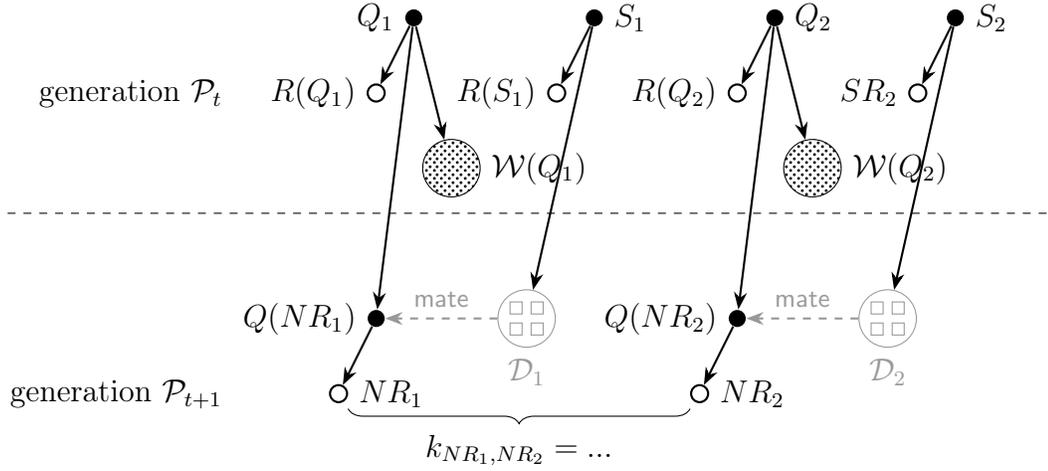
        \begin{center}
            \begin{tikzpicture}
                \path (0,0) coordinate (Pt0)
                        node[left = 0mm of Pt0] {generation $\mathcal P_t$}
                    coordinate[below = 4cm of Pt0] (Pt1)
                        node[left = 0mm of Pt1] {generation $\mathcal P_{t+1}$}
                    (Pt0) -- (Pt1)
                        coordinate[pos = 0.4] (M)

                    coordinate[right = 2.4cm of Pt0] (C0)
                    node[above = 1cm of C0, queen] (Q0) {}
                        node[left = 0cm of Q0] {$Q_1$}
                    node[below right = 2cm and 0.5cm of Q0.center, worker group] (W0) {}
                        node[right = 0cm of W0] {$\mathcal W(Q_1)$}
                    node[below left = 1cm and 0.5cm of Q0.center, replacement queen] (R0) {}
                        node[left = 0cm of R0] {$R(Q_1)$}

                    coordinate[right = 4.8cm of Pt0] (C1)
                    node[above = 1cm of C1, queen] (Q1) {}
                        node[right = 0cm of Q1] {$S_1$}
                    node[below left = 1cm and 0.5cm of Q1.center, replacement queen] (R1) {}
                        node[left = 0cm of R1] {$R(S_1)$}

                    coordinate[right = 7.2cm of Pt0] (C2)
                    node[above = 1cm of C2, queen] (Q2) {}
                        node[right = 0cm of Q2] {$Q_2$}
                    node[below right = 2cm and 0.5cm of Q2.center, worker group] (W2) {}
                        node[right = 0cm of W2] {$\mathcal W(Q_2)$}
                    node[below left = 1cm and 0.5cm of Q2.center, replacement queen] (R2) {}
                        node[left = 0cm of R2] {$R(Q_2)$}

                    coordinate[right = 9.6cm of Pt0] (C3)
                    node[above = 1cm of C3, queen] (Q3) {}
                        node[right = 0cm of Q3] {$S_2$}
                    node[below left = 1cm and 0.5cm of Q3.center, replacement queen] (R3) {}
                        node[left = 0cm of R3] {$SR_2$}

                    coordinate[right = 1.9cm of Pt1] (C4)
                    node[above = 1cm of C4, queen] (Q4) {}
                        node[left = 0cm of Q4] {$Q(NR_1)$}
                    node[right = 2cm of Q4.center, not considered, group] (D1) {}
                        node[below = 0cm of D1, not considered] {$\mathcal D_1$}
                        -- (D1) pic[not considered] {drones}
                    node[below left = 1cm and 0.5cm of Q4.center, replacement queen] (R4) {}
                        node[right = 0cm of R4] {$NR_1$}

                    coordinate[right = 6.7cm of Pt1] (C5)
                    node[above = 1cm of C5, queen] (Q5) {}
                        node[left = 0cm of Q5] {$Q(NR_2)$}
                    node[right = 2cm of Q5.center, not considered, group] (D2) {}
                        node[below = 0cm of D2, not considered] {$\mathcal D_2$}
                        -- (D2) pic[not considered] {drones}
                    node[below left = 1cm and 0.5cm of Q5.center, replacement queen] (R5) {}
                        node[right = 0cm of R5] {$NR_2$};

                    \draw[inheritance] (Q0) -- (R0);
                    \draw[inheritance] (Q0) -- (W0);
                    \draw[inheritance] (Q1) -- (R1);
                    \draw[inheritance] (Q2) -- (R2);
                    \draw[inheritance] (Q2) -- (W2);
                    \draw[inheritance] (Q3) -- (R3);
                    \draw[inheritance] (Q4) -- (R4);
                    \draw[inheritance] (Q5) -- (R5);
                    \draw[inheritance] (Q0) -- (Q4);
                    \draw[inheritance] (Q1) -- (D1);
                    \draw[inheritance] (Q2) -- (Q5);
                    \draw[inheritance] (Q3) -- (D2);
                    \draw[mating, not considered] (D1) -- (Q4)
                        node[gene pass description, not considered] {mate};
                    \draw[mating, not considered] (D2) -- (Q5)
                        node[gene pass description, not considered] {mate};
                    \draw[relationship] (R4) -- (R5)
                        node[relationship description] {$k_{NR_1,NR_2}=...$};

                \draw[dashed] (M) ++ (-3cm,0cm) -- ++(13.9cm,0cm);
            \end{tikzpicture}
        \end{center}

        Let $NR_1,NR_2\in\mathcal N\mathcal R_{t+1}$ be two non-identical (!) replacement queens with dams $Q(NR_1)$ and $Q(NR_2)\in\mathcal N\mathcal Q_{t+1}$. Let $\mathcal D_1,\mathcal D_2$ be the respective groups of drones that $Q(NR_1)$ and $Q(NR_2)$ mated with. Then by the standard argument that any allele drawn from $NR_i$ with $i\in\{1,2\}$ comes with equal probability either from $NQ_i$ or from $\mathcal D_i$, we have
            \[k_{NR_1,NR_2}=\frac14k_{Q(NR_1),Q(NR_2)}+\frac14k_{Q(NR_1),\mathcal D_2}
                + \frac14k_{\mathcal D_1,Q(NR_2)}+\frac14k_{\mathcal D_1,\mathcal D_2}.\]
        Let $Q_1,Q_2\in\mathcal Q_t$ be the respective dams of $Q(NR_1)$ and $Q(NR_2)$ and let $S_1,S_2\in\mathcal Q_t$ be the respective dams of $\mathcal D_1$ and $\mathcal D_2$. Then by the replacements according to Corollary~\ref{cor::imp} and Lemma~\ref{lem::reldron}, we have
            \[k_{NR_1,NR_2}=\frac14k_{\mathcal W(Q_1),\mathcal W(Q_2)}+\frac14k_{R(Q_1),S_2}+\frac14k_{S_1,R(Q_2)}+\frac14k_{S_1,S_2}.\]
        Note that $NR_1\neq NR_2$ implies $Q(NR_1)\neq Q(NR_2)$, so that Corollary~\ref{cor::imp} can be applied.
        Note furthermore that similar to the proof of Equation~\ref{eq::knqnq} in~\ref{item::knqnq}, we need to resort to worker groups in order to cover the case $Q_1=Q_2$ correctly.
        The frequencies with which a given queen $Q\in\mathcal Q_t$ occurs in the roles of $Q_1,Q_2,S_1$ and $S_2$ when taking averages are $dc_{Q,t},dc_{Q,t},bc_{Q,t}$, and $bc_{Q,t}$, respectively. From this we deduce the approximation
            \[k_{\mathcal N\mathcal R_{t+1},\mathcal N\mathcal R_{t+1}} \approx
                \frac14\mathbf {dc}_t^{\top}\mathbf K_{t}^{\mathfrak W\mathfrak W}\mathbf{dc}_t
                    + \frac12\mathbf {bc}_t^{\top}\mathbf K_{t}^{\mathcal Q\mathcal R}\mathbf{dc}_t
                    + \frac14\mathbf {bc}_t^{\top}\mathbf K_{t}^{\mathcal Q\mathcal Q}\mathbf{bc}_t.\]
        This approximation would be an equality if the kinship of $NR_1$ to herself was $\frac14k_{\mathcal W(Q_1),\mathcal W(Q_1)}+\frac12k_{R(Q_1),S_1}+\frac14k_{S_1,S_1}$, which is not the case.
        Instead, we have by Remark~\ref{rmk::withcol}\,\ref{item::withcoli} in combination with Equation~\ref{eq::kgd} (Lemma~\ref{lem::reldron}) and Lemma~\ref{lem::kinbet}
        \begin{align*}
            k_{NR_1,NR_1} &= \frac12+\frac12k_{Q(NR_1),\mathcal D_1}\\
                          &= \frac12+\frac12k_{R(Q_1),S_1}.
        \end{align*}
        So, for each replacement queen $NR_1\in\mathcal N\mathcal R_{t+1}$, we have to add the correction term
        \begin{align*}
            k_{NR_1,NR_1} -& \left(\frac14k_{\mathcal W(Q_1),\mathcal W(Q_1)}+\frac12k_{R(Q_1),S_1}+\frac14k_{S_1,S_1}\right)\\
                &= \frac12-\frac14k_{\mathcal W(Q_1),\mathcal W(Q_1)}-\frac14k_{S_1,S_1}
        \end{align*}
        A queen $Q\in\mathcal Q_t$ occurs with frequency $dc_{Q,t}$ in the role of $Q_1$ and with frequency $bc_{Q,t}$ in the role of $S_1$.Thus, the term that needs to be added to the approximation is
            \[\frac1{2N_{t+1}^{\mathcal N}}
                - \frac1{4N_{t+1}^{\mathcal N}}\mathbf{dc}_t^{\top}\mathrm{diag}\left(\mathbf K_t^{\mathfrak W\mathfrak W}\right)
                - \frac1{4N_{t+1}^{\mathcal N}}\mathbf{bc}_t^{\top}\mathrm{diag}\left(\mathbf K_t^{\mathcal Q\mathcal Q}\right),\]
        and we end up at the claimed identity.

        \item \label{item::knrsr} We show Equation~\ref{eq::knrsr}, i.\,e.
            \[k_{\mathcal N\mathcal R_{t+1},\mathcal S\mathcal R_{t+1}}=
                \frac1{2N_{t+1}^{\mathcal S}}\mathbf {dc}_{t}^{\top}\mathbf K_t^{\mathfrak W\mathfrak W}\mathbf{s}_t
                    +\frac1{2N_{t+1}^{\mathcal S}}\mathbf {bc}_t^{\top}\mathbf K_{t}^{\mathcal Q\mathcal R}\mathbf s_t.\]

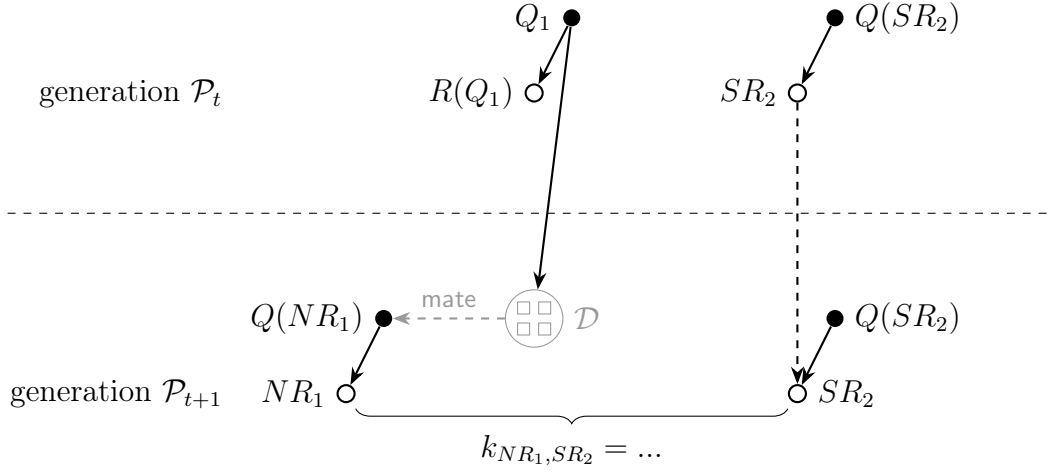
        \begin{center}
            \begin{tikzpicture}
                \path (0,0) coordinate (Pt0)
                        node[left = 0mm of Pt0] {generation $\mathcal P_t$}
                    coordinate[below = 4cm of Pt0] (Pt1)
                        node[left = 0mm of Pt1] {generation $\mathcal P_{t+1}$}
                    (Pt0) -- (Pt1)
                        coordinate[pos = 0.4] (M)

                    coordinate[right = 4.5cm of Pt0] (C0)
                    node[above = 1cm of C0, queen] (Q0) {}
                        node[left = 0cm of Q0] {$Q_1$}
                    node[below left = 1cm and 0.5cm of Q0.center, replacement queen] (R0) {}
                        node[left = 0cm of R0] {$R(Q_1)$}

                    coordinate[right = 8cm of Pt0] (C1)
                    node[above = 1cm of C1, queen] (Q1) {}
                        node[right = 0cm of Q1] {$Q(SR_2)$}
                    node[below left = 1cm and 0.5cm of Q1.center, replacement queen] (R1) {}
                        node[left = 0cm of R1] {$SR_2$}

                    coordinate[right = 2cm of Pt1] (C3)
                    node[above = 1cm of C3, queen] (Q3) {}
                        node[left = 0cm of Q3] {$Q(NR_1)$}
                    node[right = 2cm of Q3.center, not considered, group] (D) {}
                        node[right = 0cm of D, not considered] {$\mathcal D$}
                        -- (D) pic[not considered] {drones}
                    node[below left = 1cm and 0.5cm of Q3.center, replacement queen] (R3) {}
                        node[left = 0cm of R3] {$NR_1$}

                    coordinate[right = 8cm of Pt1] (C4)
                    node[above = 1cm of C4, queen] (Q4) {}
                        node[right = 0cm of Q4] {$Q(SR_2)$}
                    node[below left = 1cm and 0.5cm of Q4.center, replacement queen] (R4) {}
                        node[right = 0cm of R4] {$SR_2$};

                    \draw[inheritance] (Q0) -- (R0);
                    \draw[inheritance] (Q1) -- (R1);
                    \draw[inheritance] (Q3) -- (R3);
                    \draw[inheritance] (Q4) -- (R4);
                    \draw[inheritance] (Q0) -- (D);
                    \draw[survival] (R1) -- (R4);
                    \draw[mating, not considered] (D) -- (Q3)
                        node[gene pass description, not considered]{mate};
                    \draw[relationship] (R3) -- (R4)
                        node[relationship description] {$k_{NR_1,SR_2}=...$};

                \draw[dashed] (M) ++ (-3cm,0cm) -- ++(13.9cm,0cm);
            \end{tikzpicture}
        \end{center}

        Let $NR_1\in\mathcal N\mathcal R_{t+1}$ and $SR_2\in\mathcal S\mathcal R_{t+1}\subseteq\mathcal R_t$ be two replacement queens. Let $\mathcal D$ be the group of drones that mated with $NR_1$'s dam $Q(NR_1)$ and let $Q_1\in\mathcal Q_t$ be the dam of the drones in $\mathcal D$. Then $NR_1$ is not an ancestor of $SR_2$ and thus by the standard arguments
        \begin{align*}
            k_{NR_1,SR_2} &= \frac12k_{Q(NR_1),SR_2} + \frac12k_{\mathcal D,SR_2}\\
                          &= \frac12k_{Q(NR_1),SR_2} + \frac12k_{Q_1,SR_2}.
        \end{align*}
        When taking averages, a queen $Q\in\mathcal Q_t$ will occur in the role of $Q_1$ with frequency $bc_{Q,t}$ and a replacement queen $R\in\mathcal R_t$ will occur in the role of $SR_2$ with frequency $\frac1{N_{t+1}^{\mathcal S}}s_{R,t}$. This yields
            \[k_{\mathcal N\mathcal R_{t+1},\mathcal S\mathcal R_{t+1}} =
                \frac12k_{\mathcal N\mathcal Q_{t+1},\mathcal S\mathcal R_{t+1}}
                    +\frac1{2N_{t+1}^{\mathcal S}}\mathbf {bc}_t^{\top}\mathbf K_{t}^{\mathcal Q\mathcal R}\mathbf s_t\]
        The assertion follows by inserting Equation~\ref{eq::knqsr} (shown in~\ref{item::knqsr}).

        \item We show
            \[k_{\mathcal S\mathcal R_{t+1},\mathcal S\mathcal R_{t+1}}=\frac1{\left(N_{t+1}^{\mathcal S}\right)^2}\left(\mathbf {s}_{t}^{\mathcal R}\right)^{\top}\mathbf K_t^{\mathcal R\mathcal R}\mathbf{s}_t^{\mathcal R}.\]

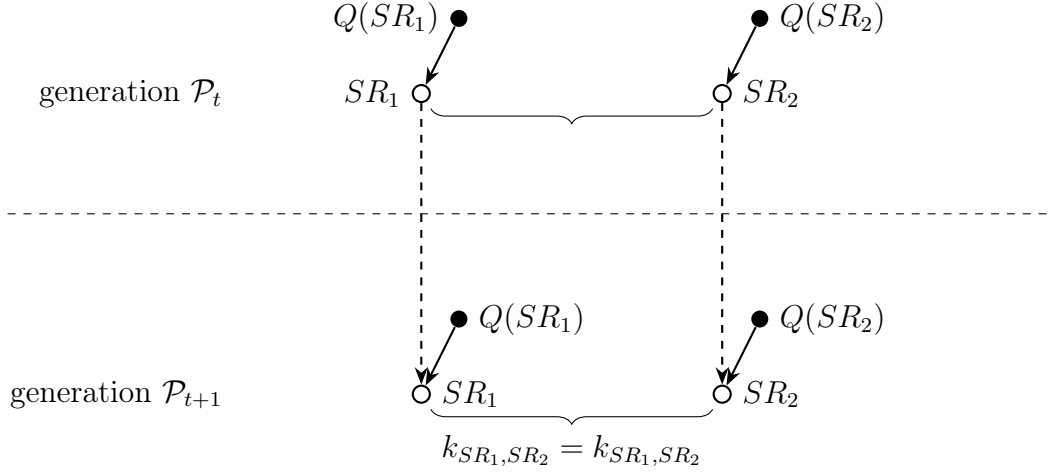
        \begin{center}
            \begin{tikzpicture}
                \path (0,0) coordinate (Pt0)
                        node[left = 0mm of Pt0] {generation $\mathcal P_t$}
                    coordinate[below = 4cm of Pt0] (Pt1)
                        node[left = 0mm of Pt1] {generation $\mathcal P_{t+1}$}
                    (Pt0) -- (Pt1)
                        coordinate[pos = 0.4] (M)

                    coordinate[right = 3cm of Pt0] (C0)
                    node[above = 1cm of C0, queen] (Q0) {}
                        node[left = 0cm of Q0] {$Q(SR_1)$}
                    node[below left = 1cm and 0.5cm of Q0.center, replacement queen] (R0) {}
                        node[left = 0cm of R0] {$SR_1$}

                    coordinate[right = 7cm of Pt0] (C1)
                    node[above = 1cm of C1, queen] (Q1) {}
                        node[right = 0cm of Q1] {$Q(SR_2)$}
                    node[below left = 1cm and 0.5cm of Q1.center, replacement queen] (R1) {}
                        node[right = 0cm of R1] {$SR_2$}

                    coordinate[right = 3cm of Pt1] (C3)
                    node[above = 1cm of C3, queen] (Q3) {}
                        node[right = 0cm of Q3] {$Q(SR_1)$}
                    node[below left = 1cm and 0.5cm of Q3.center, replacement queen] (R3) {}
                        node[right = 0cm of R3] {$SR_1$}

                    coordinate[right = 7cm of Pt1] (C4)
                    node[above = 1cm of C4, queen] (Q4) {}
                        node[right = 0cm of Q4] {$Q(SR_2)$}
                    node[below left = 1cm and 0.5cm of Q4.center, replacement queen] (R4) {}
                        node[right = 0cm of R4] {$SR_2$};

                    \draw[survival] (R0) -- (R3);
                    \draw[survival] (R1) -- (R4);
                    \draw[inheritance] (Q0) -- (R0);
                    \draw[inheritance] (Q1) -- (R1);
                    \draw[inheritance] (Q3) -- (R3);
                    \draw[inheritance] (Q4) -- (R4);
                    \draw[relationship] (R0) -- (R1);
                    \draw[relationship] (R3) -- (R4)
                        node[relationship description] {$k_{SR_1,SR_2}=k_{SR_1,SR_2}$};

                \draw[dashed] (M) ++ (-3cm,0cm) -- ++(13.9cm,0cm);
            \end{tikzpicture}
        \end{center}

        Just as in Equation~\ref{eq::ksqsq} (part~\ref{item::ksqsq} of this proof) and in Equation~\ref{eq::ksqsr} (part~\ref{item::ksqsr} of this proof), this follows from the fact that kinships between surviving colonies do not change over time.
    \end{enumerate}
\end{proof}

With all these terms calculated, we insert them into the equations of Remark~\ref{rmk::brokenk}\,\ref{item::tenii}:

\begin{Lem}
    We have
    \begin{align}
        k_{\mathcal Q_{t+1},\mathcal Q_{t+1}}
            &= \left(\frac{N_{t+1}^{\mathcal N}}{N_{t+1}}\right)^2\mathbf{dc}_t^{\top}\mathbf{K}_t^{\mathfrak W\mathfrak W}\mathbf{dc}_t
                + \frac{N_{t+1}^{\mathcal N}}{N_{t+1}^2}\mathbf {dc}_t^{\top}\mathrm{diag}\left(\mathbf K_t^{\mathcal R\mathcal R}\right)
                - \frac{N_{t+1}^{\mathcal N}}{N_{t+1}^2}\mathbf {dc}_t^{\top}\mathrm{diag}\left(\mathbf K_t^{\mathfrak W\mathfrak W}\right) \nonumber\\
            &\hspace{1cm} +\frac{2N_{t+1}^{\mathcal N}}{N_{t+1}^2}\mathbf{dc}_t^{\top}\mathbf K_t^{\mathcal R\mathcal Q}\mathbf s_t
                + \frac1{N_{t+1}^2}\mathbf s_t^{\top}\mathbf K_t^{\mathcal Q\mathcal Q}\mathbf s_t,\\
        k_{\mathcal Q_{t+1},\mathcal R_{t+1}}
            &= \frac{\left(N_{t+1}^{\mathcal N}\right)^2}{2N_{t+1}^2}
                    \mathbf{dc}_t^{\top}\mathbf K_t^{\mathfrak W\mathfrak W}\mathbf{dc}_t
                + \frac{\left(N_{t+1}^{\mathcal N}\right)^2}{2N_{t+1}^2}
                    \mathbf{dc}_t^{\top}\mathbf K_t^{\mathcal R\mathcal Q}\mathbf{bc}_t \nonumber\\
            &\hspace{1cm}+ \frac{N_{t+1}^{\mathcal N}}{2N_{t+1}^2}
                    \mathbf{1}_t^{\top}\left(\mathbf K_t^{\mathcal R\mathcal R}-\mathbf K_t^{\mathfrak W\mathfrak W}\right)\mathbf{dc}_t
                + \frac{N_{t+1}^{\mathcal N}}{N_{t+1}^2}
                    \mathbf{dc}_t^{\top}\mathbf K_t^{\mathfrak W\mathfrak W}\mathbf{s}_t \nonumber\\
            &\hspace{1cm}+ \frac{N_{t+1}^{\mathcal N}}{2N_{t+1}^2}\mathbf {dc}_t^{\top}\mathbf K_t^{\mathcal R\mathcal Q}\mathbf s_t
                + \frac{N_{t+1}^{\mathcal N}}{2N_{t+1}^2}\mathbf {bc}_t^{\top}\mathbf K_t^{\mathcal Q\mathcal Q}\mathbf s_t
                + \frac1{N_{t+1}^2}\mathbf s_t^{\top}\mathbf K_t^{\mathcal Q\mathcal R}\mathbf s_t,\\
        k_{\mathcal R_{t+1},\mathcal R_{t+1}}
            &= \frac{\left(N_{t+1}^{\mathcal N}\right)^2}{4N_{t+1}^2}
                    \mathbf{dc}_t^{\top}\mathbf K_t^{\mathfrak W\mathfrak W}\mathbf{dc}_t
                + \frac{\left(N_{t+1}^{\mathcal N}\right)^2}{2N_{t+1}^2}
                    \mathbf{bc}_t^{\top}\mathbf K_t^{\mathcal Q\mathcal R}\mathbf{dc}_t
                + \frac{\left(N_{t+1}^{\mathcal N}\right)^2}{4N_{t+1}^2}
                    \mathbf{bc}_t^{\top}\mathbf K_t^{\mathcal Q\mathcal Q}\mathbf{bc}_t \nonumber\\
            &\hspace{1cm} - \frac{N_{t+1}^{\mathcal N}}{4N_{t+1}^2}
                    \mathbf{dc}_t^{\top}\mathrm{diag}\left(K_t^{\mathfrak W\mathfrak W}\right)
                - \frac{N_{t+1}^{\mathcal N}}{4N_{t+1}^2}
                    \mathbf{bc}_t^{\top}\mathrm{diag}\left(K_t^{\mathcal Q\mathcal Q}\right) \nonumber\\
            &\hspace{1cm} + \frac{N_{t+1}^{\mathcal N}}{N_{t+1}^2}\mathbf{dc}_t^{\top}\mathbf{K}_t^{\mathfrak W\mathfrak W}\mathbf s_t
                + \frac{N_{t+1}^{\mathcal N}}{N_{t+1}^2}\mathbf{bc}_t^{\top}\mathbf{K}_t^{\mathcal Q\mathcal R}\mathbf s_t
                + \frac{N_{t+1}^{\mathcal N}}{2N_{t+1}^2}
                + \frac1{N_{t+1}^2}\mathbf s_t\mathbf K_t^{\mathcal R\mathcal R}\mathbf s_t.
    \end{align}
\end{Lem}

These equations can then be inserted into Equation~\ref{eq::kng} to obtain $k_{\mathcal P_{t+1}^{\ast},\mathcal P_{t+1}^{\ast}}$.

\begin{Thm}\label{thm::burner}
    We have
    \begin{align}
        k_{\mathcal P^{\ast}_{t+1},\mathcal P^{\ast}_{t+1}}
            &= \left(\frac{N_{t+1}^{\mathcal N}}{4N_{t+1}}\right)^2\left(
                9\cdot\mathbf{dc}_t^{\top}\mathbf K_t^{\mathfrak W\mathfrak W}\mathbf {dc}_t
                    + 6\cdot\mathbf{dc}_t^{\top}\mathbf K_t^{\mathcal R\mathcal Q}\mathbf {bc}_t
                    + \mathbf{bc}_t^{\top}\mathbf K_t^{\mathcal Q\mathcal Q}\mathbf {bc}_t
                \right) \nonumber \\
            &\hspace{0.5cm} + \frac{N_{t+1}^{\mathcal N}}{\left(4N_{t+1}\right)^2}\left(
                8\cdot\mathbf{dc}_t^{\top}\mathrm{diag}\left(\mathbf K_t^{\mathcal R\mathcal R}\right)
                    - 9\cdot\mathbf{dc}_t^{\top}\mathrm{diag}\left(\mathbf K_t^{\mathfrak W\mathfrak W}\right)\right. \nonumber \\
                    &\hspace{9cm}\left.- \mathbf{bc}_t^{\top}\mathrm{diag}\left(\mathbf K_t^{\mathcal Q\mathcal Q}\right)
                \right) \nonumber \\
            &\hspace{0.5cm} + \frac{N_{t+1}^{\mathcal N}}{\left(2N_{t+1}\right)^2}\left(
                3\cdot\mathbf{dc}_t^{\top}\mathbf K_t^{\mathcal R\mathcal Q}\mathbf s_t
                    + 3\cdot\mathbf{dc}_t^{\top}\mathbf K_t^{\mathfrak W\mathfrak W}\mathbf s_t
                    + \mathbf{bc}_t^{\top}\mathbf K_t^{\mathcal Q\mathcal R}\mathbf s_t
                    + \mathbf{bc}_t^{\top}\mathbf K_t^{\mathcal Q\mathcal Q}\mathbf s_t
                \right) \nonumber\\
            &\hspace{0.5cm} + \frac1{4N_{t+1}^2}\mathbf s_t^{\top}\left(
                    \mathbf K_t^{\mathcal Q\mathcal Q} + 2\cdot\mathbf K_t^{\mathcal Q\mathcal R} + \mathbf K_t^{\mathcal R\mathcal R}
                \right)\mathbf s_t
                + \frac{N^{\mathcal N}_{t+1}}{8N_{t+1}^2}.
    \end{align}
\end{Thm}

With all that, we are able to formulate the task of OCS for a honeybee population with single colony inseminations.

\begin{Task}\label{task::sci}
    Given a generation $\mathcal P_t=\mathcal Q_t\sqcup\mathfrak W_t\sqcup\mathcal R_t$ of honeybee colonies, and
        \begin{itemize}
            \item vectors $\hat{\mathbf u}_t^{\mathcal Q}\in\mathbb R^{\mathcal Q_t}$ and $\hat{\mathbf u}_t^{\mathcal R}=\hat{\mathbf u}_t^{\mathfrak W}\in\mathbb R^{\mathcal R_t}\cong\mathbb R^{\mathfrak W_t}$ of estimated breeding values,

            \item a survival vector $\mathbf s_t\in\{0,1\}^{\mathcal Q_t}(\cong\{0,1\}^{\mathcal R_t}\cong\{0,1\}^{\mathfrak W_t})$,

            \item a symmetric and positive definite kinship matrix $\mathbf K_t\in\mathbb R^{\mathcal P_t\times\mathcal P_t}$ that falls into the blocks
                \[\mathbf K_t =
                    \begin{pmatrix}
                    \mathbf K_t^{\mathcal Q\mathcal Q}  & \mathbf K_t^{\mathcal Q\mathfrak W}  & \mathbf K_t^{\mathcal Q\mathcal R}\\
                    \mathbf K_t^{\mathfrak W\mathcal Q} & \mathbf K_t^{\mathfrak W\mathfrak W} & \mathbf K_t^{\mathfrak W\mathcal R}\\
                    \mathbf K_t^{\mathcal R\mathcal Q}  & \mathbf K_t^{\mathcal R\mathfrak W}  & \mathbf K_t^{\mathcal R\mathcal R}
                    \end{pmatrix}.\]
            and fulfills the properties listed in Remark~\ref{rmk::onk},

            \item the required number of newly created colonies of the next generation, $N_{t+1}^{\mathcal N}$,

            \item and a maximum acceptable kinship level $k_{t+1}^{\ast}$,
        \end{itemize}
        let $N_{t+1}:=N_{t+1}^{\mathcal N}+\mathbf 1_t^{\top}\mathbf s_t$ and maximize the function
        \begin{align*}
            \mathbb E\bigl[\hat u_{\mathcal P^{\ast}_{t+1}}\bigr]:
                \mathbb R_{\geq0}^{\mathcal Q_t}\oplus\mathbb R_{\geq0}^{\mathcal Q_t} &\to \mathbb R,\\
                \mathbf {dc}_{t}\oplus\mathbf {bc}_{t} &\mapsto
                \frac{3N_{t+1}^{\mathcal N}}{4N_{t+1}}\mathbf {dc}_{t}^{\top}\hat{\mathbf u}_{t}^{\mathcal R}
                + \frac{N_{t+1}^{\mathcal N}}{4N_{t+1}}\mathbf {bc}_{t}^{\top}\hat{\mathbf u}_{t}^{\mathcal Q}
                + \frac1{2N_{t+1}}\mathbf s_t^{\top}\left(\hat{\mathbf u}_{t}^{\mathcal R}+\hat{\mathbf u}_{t}^{\mathcal Q}\right)
        \end{align*}

        under the constraints
            \[\mathbf 1_t^{\top}\mathbf {dc}_t=1,\]
            \[\mathbf 1_t^{\top}\mathbf {bc}_t=1,\]
        and
            \[k_{\mathcal P^{\ast}_{t+1},\mathcal P^{\ast}_{t+1}} \leq k_{t+1}^{\ast},\]
        where $k_{\mathcal P^{\ast}_{t+1},\mathcal P^{\ast}_{t+1}}$ denotes the term described in Theorem~\ref{thm::burner}.
\end{Task}

\subsection{Isolated mating stations}\label{sec::ims}

We next consider the case that young queens are not instrumentally inseminated but mate on isolated mating stations. We further assume that the group $\mathcal M$ of DPQs on a mating station consists of sisters, i.\,e. shares a single dam.

Under these circumstances, any queen $Q\in\mathcal Q_t$ can still contribute to the genetic setup of the next generation $\mathcal P_{t+1}$ via the dam and survival paths just like in the case of single colony inseminations (Section~\ref{sec::sci}). However, contribution via the 1b-path is replaced by the following pathway:

\begin{Def}
    If a queen $Q\in\mathcal Q_t$ produces the group $\mathcal M$ of DPQs on a mating station she is by the nomenclature of \url{www.beebreed.eu} called a "4a-queen" \citep{uzunov23, druml23}. New queens in $\mathcal N\mathcal Q_{t+1}$ may then mate with drones produced by $\mathcal M$, and thereby $Q$ makes a genetic contribution to the next generation $\mathcal P_{t+1}$. We call this pathway the \emph{4a-path}.
\end{Def}

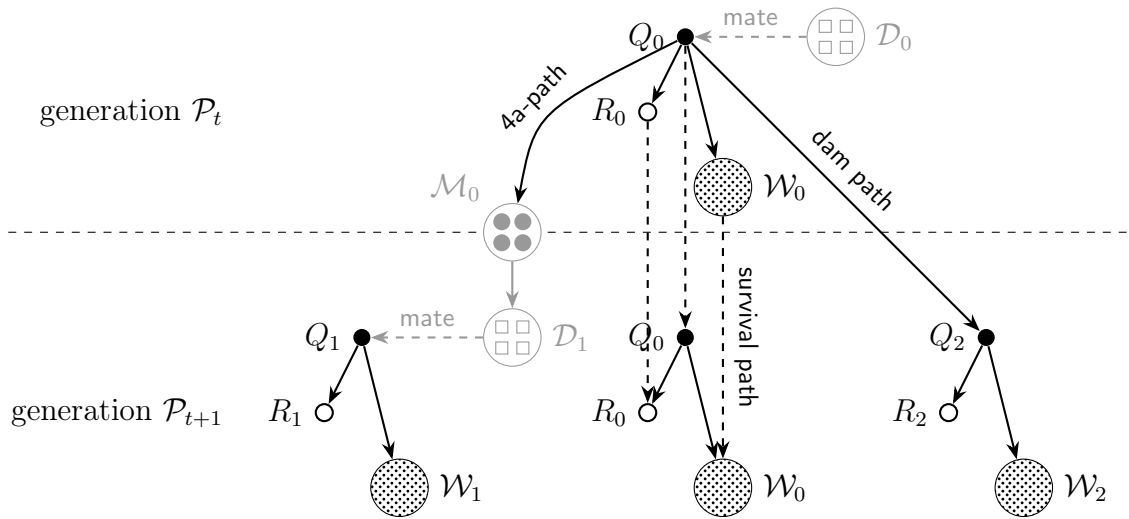
\begin{center}
    \begin{tikzpicture}
        \path (0,0) coordinate (Pt0)
                node[left = 0mm of Pt0] {generation $\mathcal P_t$}
            coordinate[below = 4cm of Pt0] (Pt1)
                node[left = 0mm of Pt1] {generation $\mathcal P_{t+1}$}
            (Pt0) -- (Pt1)
                coordinate[pos = 0.4] (M)
            coordinate[right = 6cm of Pt0] (C0)
            node[above = 1cm of C0, queen] (Q0) {}
                node[left = 0cm of Q0] {$Q_0$}
            node[right = 2cm of Q0.center, not considered, group] (D0) {}
                node[right = 0cm of D0, not considered] {$\mathcal D_0$}
                -- (D0) pic[not considered] {drones}
            node[below right = 2cm and 0.5cm of Q0.center, worker group] (W0) {}
                node[right = 0cm of W0] {$\mathcal W_0$}
            node[below left = 1cm and 0.5cm of Q0.center, replacement queen] (R0) {}
                node[left = 0cm of R0] {$R_0$}
             coordinate[right = 1.7cm of Pt1] (C1)
            node[above = 1cm of C1, queen] (Q1) {}
                node[left = 0cm of Q1] {$Q_1$}
            node[right = 2cm of Q1.center, not considered, group] (D1) {}
                node[right = 0cm of D1, not considered] {$\mathcal D_1$}
                -- (D1) pic[not considered] {drones}
            node[above = 1.4cm of D1.center, group base] (M1prov) {}
                -- (M1prov.center) node[not considered, group] (M1) {}
                node[above left = 0cm and 0cm of M1, not considered] {$\mathcal M_0$}
                -- (M1.center) pic[not considered] {queens}
            node[below right = 2cm and 0.5cm of Q1.center, worker group] (W1) {}
                node[right = 0cm of W1] {$\mathcal W_1$}
            node[below left = 1cm and 0.5cm of Q1.center, replacement queen] (R1) {}
                node[left = 0cm of R1] {$R_1$}

            coordinate[right = 6cm of Pt1] (C2)
            node[above = 1cm of C2, queen] (Q2) {}
                node[left = 0cm of Q2] {$Q_0$}
            node[below right = 2cm and 0.5cm of Q2.center, worker group] (W2) {}
                node[right = 0cm of W2] {$\mathcal W_0$}
            node[below left = 1cm and 0.5cm of Q2.center, replacement queen] (R2) {}
                node[left = 0cm of R2] {$R_0$}

            coordinate[right = 10cm of Pt1] (C3)
            node[above = 1cm of C3, queen] (Q3) {}
                node[left = 0cm of Q3] {$Q_2$}
            node[below right = 2cm and 0.5cm of Q3.center, worker group] (W3) {}
                node[right = 0cm of W3] {$\mathcal W_2$}
            node[below left = 1cm and 0.5cm of Q3.center, replacement queen] (R3) {}
                node[left = 0cm of R3] {$R_2$};

        \draw[mating, not considered] (D0) -- (Q0)
            node[gene pass description, not considered] {mate};
        \draw[mating, not considered] (D1) -- (Q1)
            node[gene pass description, not considered] {mate};
        \draw[inheritance] (Q0) -- (W0);
        \draw[inheritance] (Q0) -- (R0);
        \draw[inheritance] (Q1) -- (W1);
        \draw[inheritance] (Q1) -- (R1);
        \draw[inheritance] (Q2) -- (W2);
        \draw[inheritance] (Q2) -- (R2);
        \draw[inheritance] (Q3) -- (W3);
        \draw[inheritance] (Q3) -- (R3);
        \draw[inheritance] (Q0) .. controls ++(-2,-1) .. (M1)
            node[gene pass description] {4a-path};
        \draw[inheritance, not considered] (M1) -- (D1);
        \draw[inheritance] (Q0) -- (Q3)
            node[gene pass description] {dam path};
        \draw[survival] (W0) -- (W2)
            node[gene pass description] {survival path};
        \draw[survival] (Q0) -- (Q2);
        \draw[survival] (R0) -- (R2);

        \begin{scope}[on background layer]
            \draw[dashed] (M) ++ (-3cm,0cm) -- ++(15cm,0cm);
        \end{scope}
    \end{tikzpicture}
\end{center}

\begin{Rmk}
    The introduction of groups $\mathcal M$ of DPQs raises the question if and how these should be integrated in the overall population $\mathcal P$.
    \begin{enumerate}[label = (\roman*)]
        \item Counting individual DPQs as elements of $\mathcal Q_t$ would surely be a bad idea. Because the DPQs on a mating station are sisters, this would introduce many close relationships and increase $k_{\mathcal Q_t,\mathcal Q_t}$. In particular, it would appear better to have small groups $\mathcal M$ of DPQs on mating stations than to have large groups. In reality, however, there are virtually no downsides of having mating stations comprise large numbers of DPQs. In contrast, higher numbers of DPQs result in a greater drone density and therefore greater mating success rates \citep{tiesler16, uzunov22evaluation}. (See however \citep{neumann99queen} who did not find an influence on the number of DPQs on a mating station on mating success.)

        \item More promising appears the idea to count groups $\mathcal M$ of DPQs on a mating station as separate entities. This idea goes back to \citet{bienefeld89}, who coined the term \emph{pseudo sires} for such groups. This would lead to a set $\mathfrak M_t$ of pseudo sires for each generation. However, also this approach comes with problems. If a queen $Q\in\mathcal Q_t$ produces a pseudo sire $\mathcal M$ which in turn produces the drones to mate a new queen $NQ\in\mathcal N\mathcal Q_{t+1}\subseteq\mathcal Q_{t+1}$, should $\mathcal M$ then be counted towards generation $\mathcal P_t$ or towards $\mathcal P_{t+1}$? In either case, we would be confronted with (unwanted) parent-offspring relations within one generation without the offspring being a genetic dead end (like worker groups or replacement queens).

        \item Instead, it turns out most practical to consider (groups of) DPQs as outside of the considered population, just as we do not consider groups of drones as part of the population either (Remark~\ref{rmk::nodr}\,\ref{item::nodri}). This approach is justified by the fact that analyses of developments in simulated \citep{plate20} and real \citep{hoppe20} honeybee populations generally focus on breeding queens and not DPQ. Furthermore, since DPQs are generally not phenotyped \citep{basso24}, the estimated breeding values of pseudo sires are only weighted averages of the estimated breeding values of their relatives, without any intrinsic information.

        \item By the choice of not considering DPQs as part of the population, we can leave the question open, which generation should be associated with a pseudo sire $\mathcal M$. In diagrams, we will depict pseudo sires as gray (like drone groups, cf. Remark~\ref{rmk::nodr}\,\ref{item::nodrii}) and place them right on the border between two generations.
    \end{enumerate}
\end{Rmk}

We need to add one further assumption concerning the design of breeding schemes with isolated mating station.

\begin{Rmk}\label{rmk::msassu}
    \begin{enumerate}[label = (\roman*)]
        \item \label{item::msassui} In theory it is possible that a queen $Q$ provides the DPQs for multiple mating stations at the same time $t$. The database \url{www.beebreed.eu} reveals that this is indeed sometimes the case in reality. For our considerations, however, we exclude this possibility. Each queen $Q\in\mathcal Q_t$ may serve as the 4a-queen of at most one mating station at time $t$. Serving as 4a-queen to mating stations at different times than $t$ is allowed.

        We clarify this by figures: The following situation is not allowed:

        
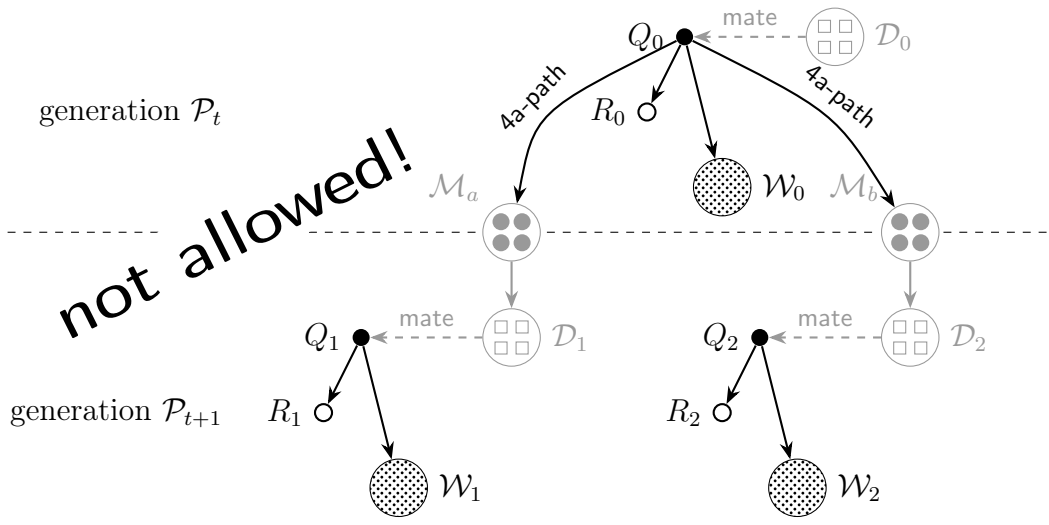
\begin{center}
            \begin{tikzpicture}
                \path (0,0) coordinate (Pt0)
                        node[left = 0mm of Pt0] {generation $\mathcal P_t$}
                    coordinate[below = 4cm of Pt0] (Pt1)
                        node[left = 0mm of Pt1] {generation $\mathcal P_{t+1}$}
                    (Pt0) -- (Pt1)
                        coordinate[pos = 0.4] (M)
                    coordinate[right = 6cm of Pt0] (C0)
                    node[above = 1cm of C0, queen] (Q0) {}
                        node[left = 0cm of Q0] {$Q_0$}
                    node[right = 2cm of Q0.center, not considered, group] (D0) {}
                        node[right = 0cm of D0, not considered] {$\mathcal D_0$}
                        -- (D0) pic[not considered] {drones}
                    node[below right = 2cm and 0.5cm of Q0.center, worker group] (W0) {}
                        node[right = 0cm of W0] {$\mathcal W_0$}
                    node[below left = 1cm and 0.5cm of Q0.center, replacement queen] (R0) {}
                        node[left = 0cm of R0] {$R_0$}

                    coordinate[right = 1.7cm of Pt1] (C1)
                    node[above = 1cm of C1, queen] (Q1) {}
                        node[left = 0cm of Q1] {$Q_1$}
                    node[right = 2cm of Q1.center, not considered, group] (D1) {}
                        node[right = 0cm of D1, not considered] {$\mathcal D_1$}
                        -- (D1) pic[not considered] {drones}
                    node[above = 1.4cm of D1.center, group base] (M1prov) {}
                        -- (M1prov.center) node[not considered, group] (M1) {}
                        node[above left = 0cm and 0cm of M1, not considered] {$\mathcal M_a$}
                        -- (M1.center) pic[not considered] {queens}
                    node[below right = 2cm and 0.5cm of Q1.center, worker group] (W1) {}
                        node[right = 0cm of W1] {$\mathcal W_1$}
                    node[below left = 1cm and 0.5cm of Q1.center, replacement queen] (R1) {}
                        node[left = 0cm of R1] {$R_1$}

                    coordinate[right = 7cm of Pt1] (C2)
                    node[above = 1cm of C2, queen] (Q2) {}
                        node[left = 0cm of Q2] {$Q_2$}
                    node[right = 2cm of Q2.center, not considered, group] (D2) {}
                        node[right = 0cm of D2, not considered] {$\mathcal D_2$}
                        -- (D2) pic[not considered] {drones}
                    node[above = 1.4cm of D2.center, group base] (M2prov) {}
                        -- (M2prov.center) node[not considered, group] (M2) {}
                        node[above left = 0cm and 0cm of M2, not considered] {$\mathcal M_b$}
                        -- (M2.center) pic[not considered] {queens}
                    node[below right = 2cm and 0.5cm of Q2.center, worker group] (W2) {}
                        node[right = 0cm of W2] {$\mathcal W_2$}
                    node[below left = 1cm and 0.5cm of Q2.center, replacement queen] (R2) {}
                        node[left = 0cm of R2] {$R_2$};

                \draw[mating, not considered] (D0) -- (Q0)
                    node[gene pass description, not considered] {mate};
                \draw[mating, not considered] (D1) -- (Q1)
                    node[gene pass description, not considered] {mate};
                \draw[mating, not considered] (D2) -- (Q2)
                    node[gene pass description, not considered] {mate};
                \draw[inheritance] (Q0) -- (W0);
                \draw[inheritance] (Q0) -- (R0);
                \draw[inheritance] (Q1) -- (W1);
                \draw[inheritance] (Q1) -- (R1);
                \draw[inheritance] (Q2) -- (W2);
                \draw[inheritance] (Q2) -- (R2);
                \draw[inheritance] (Q0) .. controls ++(-2,-1) .. (M1)
                    node[gene pass description] {4a-path};
                \draw[inheritance] (Q0) .. controls ++(2,-1) .. (M2)
                    node[gene pass description] {4a-path};
                \draw[inheritance, not considered] (M1) -- (D1);
                \draw[inheritance, not considered] (M2) -- (D2);

                \begin{scope}[on background layer]
                    \draw[dashed] (M) ++ (-3cm,0cm) -- ++(13.9cm,0cm);
                \end{scope}

                \path (M) ++ (-1,-0.5) -- ++(2,1) node[fill = white, font = \huge\sffamily, midway,sloped] {\textbf{not allowed!}};

            \end{tikzpicture}
        \end{center}

        The following situation is allowed:

        
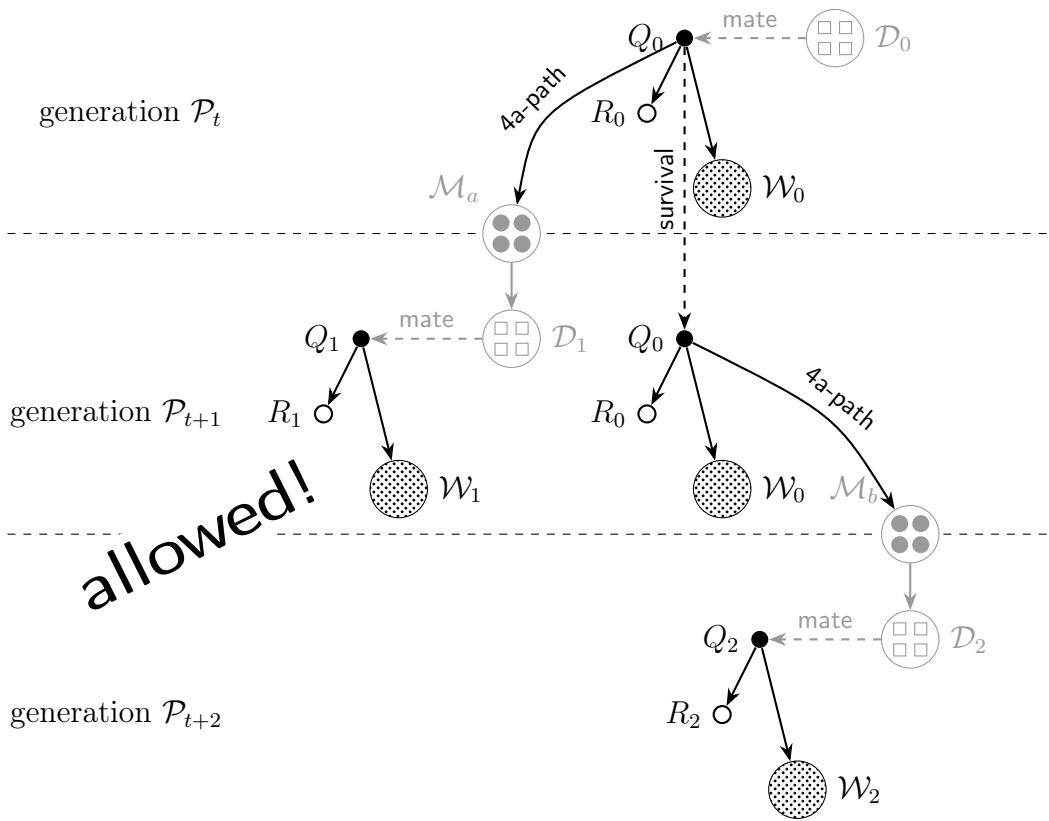
\begin{center}
            \begin{tikzpicture}
                \path (0,0) coordinate (Pt0)
                        node[left = 0mm of Pt0] {generation $\mathcal P_t$}
                    coordinate[below = 4cm of Pt0] (Pt1)
                        node[left = 0mm of Pt1] {generation $\mathcal P_{t+1}$}
                    (Pt0) -- (Pt1)
                        coordinate[pos = 0.4] (M)
                    coordinate[below = 4cm of Pt1] (Pt2)
                        node[left = 0mm of Pt2] {generation $\mathcal P_{t+2}$}
                    (Pt1) -- (Pt2)
                        coordinate[pos = 0.4] (N)
                    coordinate[right = 6cm of Pt0] (C0)
                    node[above = 1cm of C0, queen] (Q0) {}
                        node[left = 0cm of Q0] {$Q_0$}
                    node[right = 2cm of Q0.center, not considered, group] (D0) {}
                        node[right = 0cm of D0, not considered] {$\mathcal D_0$}
                        -- (D0) pic[not considered] {drones}
                    node[below right = 2cm and 0.5cm of Q0.center, worker group] (W0) {}
                        node[right = 0cm of W0] {$\mathcal W_0$}
                    node[below left = 1cm and 0.5cm of Q0.center, replacement queen] (R0) {}
                        node[left = 0cm of R0] {$R_0$}

                    coordinate[right = 6cm of Pt1] (C00)
                    node[above = 1cm of C00, queen] (Q00) {}
                        node[left = 0cm of Q00] {$Q_0$}
                    node[below right = 2cm and 0.5cm of Q00.center, worker group] (W00) {}
                        node[right = 0cm of W00] {$\mathcal W_0$}
                    node[below left = 1cm and 0.5cm of Q00.center, replacement queen] (R00) {}
                        node[left = 0cm of R00] {$R_0$}

                    coordinate[right = 1.7cm of Pt1] (C1)
                    node[above = 1cm of C1, queen] (Q1) {}
                        node[left = 0cm of Q1] {$Q_1$}
                    node[right = 2cm of Q1.center, not considered, group] (D1) {}
                        node[right = 0cm of D1, not considered] {$\mathcal D_1$}
                        -- (D1) pic[not considered] {drones}
                    node[above = 1.4cm of D1.center, group base] (M1prov) {}
                        -- (M1prov.center) node[not considered, group] (M1) {}
                        node[above left = 0cm and 0cm of M1, not considered] {$\mathcal M_a$}
                        -- (M1.center) pic[not considered] {queens}
                    node[below right = 2cm and 0.5cm of Q1.center, worker group] (W1) {}
                        node[right = 0cm of W1] {$\mathcal W_1$}
                    node[below left = 1cm and 0.5cm of Q1.center, replacement queen] (R1) {}
                        node[left = 0cm of R1] {$R_1$}

                    coordinate[right = 7cm of Pt2] (C2)
                    node[above = 1cm of C2, queen] (Q2) {}
                        node[left = 0cm of Q2] {$Q_2$}
                    node[right = 2cm of Q2.center, not considered, group] (D2) {}
                        node[right = 0cm of D2, not considered] {$\mathcal D_2$}
                        -- (D2) pic[not considered] {drones}
                    node[above = 1.4cm of D2.center, group base] (M2prov) {}
                        -- (M2prov.center) node[not considered, group] (M2) {}
                        node[above left = 0cm and 0cm of M2, not considered] {$\mathcal M_b$}
                        -- (M2.center) pic[not considered] {queens}
                    node[below right = 2cm and 0.5cm of Q2.center, worker group] (W2) {}
                        node[right = 0cm of W2] {$\mathcal W_2$}
                    node[below left = 1cm and 0.5cm of Q2.center, replacement queen] (R2) {}
                        node[left = 0cm of R2] {$R_2$};

                \draw[mating, not considered] (D0) -- (Q0)
                    node[gene pass description, not considered] {mate};
                \draw[mating, not considered] (D1) -- (Q1)
                    node[gene pass description, not considered] {mate};
                \draw[mating, not considered] (D2) -- (Q2)
                    node[gene pass description, not considered] {mate};
                \draw[inheritance] (Q0) -- (W0);
                \draw[inheritance] (Q0) -- (R0);
                \draw[inheritance] (Q00) -- (W00);
                \draw[inheritance] (Q00) -- (R00);
                \draw[inheritance] (Q1) -- (W1);
                \draw[inheritance] (Q1) -- (R1);
                \draw[inheritance] (Q2) -- (W2);
                \draw[inheritance] (Q2) -- (R2);
                \draw[survival] (Q0) -- (Q00)
                    node[gene pass description, rotate = 180] {survival};
                \draw[inheritance] (Q0) .. controls ++(-2,-1) .. (M1)
                    node[gene pass description] {4a-path};
                \draw[inheritance] (Q00) .. controls ++(2,-1) .. (M2)
                    node[gene pass description] {4a-path};
                \draw[inheritance, not considered] (M1) -- (D1);
                \draw[inheritance, not considered] (M2) -- (D2);

                \begin{scope}[on background layer]
                    \draw[dashed] (M) ++ (-3cm,0cm) -- ++(13.9cm,0cm);
                    \draw[dashed] (N) ++ (-3cm,0cm) -- ++(13.9cm,0cm);
                \end{scope}

                \path (N) ++ (-1.5,-0.5) -- ++(2,1) node[fill = white, font = \huge\sffamily, midway,sloped] {\textbf{allowed!}};

            \end{tikzpicture}
        \end{center}

        \item Of course, a mating station will in general be frequented by multiple young queens, so the following situation is not only allowed but frequent.

        
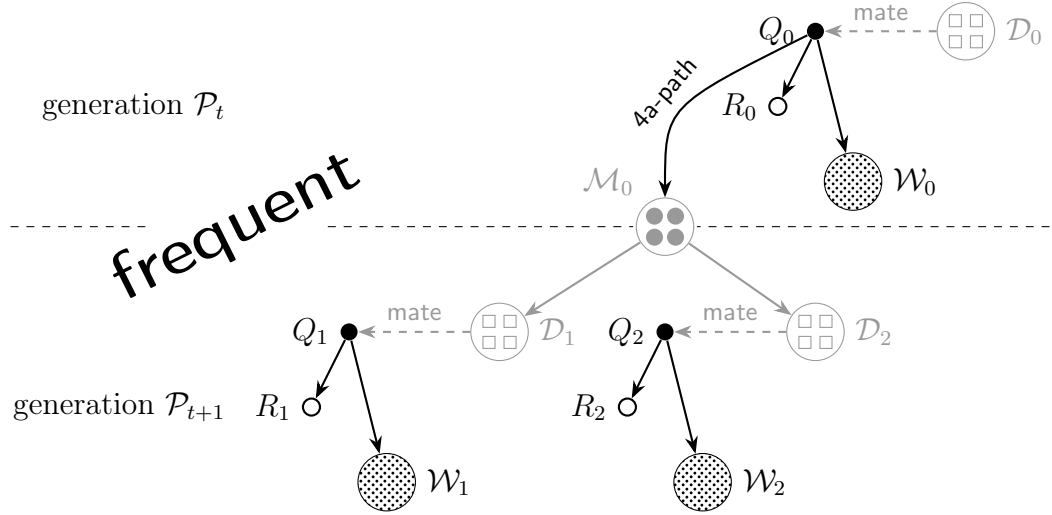
\begin{center}
            \begin{tikzpicture}
                \path (0,0) coordinate (Pt0)
                        node[left = 0mm of Pt0] {generation $\mathcal P_t$}
                    coordinate[below = 4cm of Pt0] (Pt1)
                        node[left = 0mm of Pt1] {generation $\mathcal P_{t+1}$}
                    (Pt0) -- (Pt1)
                        coordinate[pos = 0.4] (M)

                    coordinate[right = 7.7cm of Pt0] (C0)
                    node[above = 1cm of C0, queen] (Q0) {}
                        node[left = 0cm of Q0] {$Q_0$}
                    node[right = 2cm of Q0.center, not considered, group] (D0) {}
                        node[right = 0cm of D0, not considered] {$\mathcal D_0$}
                        -- (D0) pic[not considered] {drones}
                    node[below right = 2cm and 0.5cm of Q0.center, worker group] (W0) {}
                        node[right = 0cm of W0] {$\mathcal W_0$}
                    node[below left = 1.6cm and 2cm of C0, group base] (M1prov) {}
                        -- (M1prov.center) node[not considered, group] (M1) {}
                        node[above left = 0cm and 0cm of M1, not considered] {$\mathcal M_0$}
                        -- (M1.center) pic[not considered] {queens}
                    node[below left = 1cm and 0.5cm of Q0.center, replacement queen] (R0) {}
                        node[left = 0cm of R0] {$R_0$}

                    coordinate[right = 1.5cm of Pt1] (C1)
                    node[above = 1cm of C1, queen] (Q1) {}
                        node[left = 0cm of Q1] {$Q_1$}
                    node[right = 2cm of Q1.center, not considered, group] (D1) {}
                        node[right = 0cm of D1, not considered] {$\mathcal D_1$}
                        -- (D1) pic[not considered] {drones}
                    node[below right = 2cm and 0.5cm of Q1.center, worker group] (W1) {}
                        node[right = 0cm of W1] {$\mathcal W_1$}
                    node[below left = 1cm and 0.5cm of Q1.center, replacement queen] (R1) {}
                        node[left = 0cm of R1] {$R_1$}

                    coordinate[right = 5.7cm of Pt1] (C2)
                    node[above = 1cm of C2, queen] (Q2) {}
                        node[left = 0cm of Q2] {$Q_2$}
                    node[right = 2cm of Q2.center, not considered, group] (D2) {}
                        node[right = 0cm of D2, not considered] {$\mathcal D_2$}
                        -- (D2) pic[not considered] {drones}
                    node[below right = 2cm and 0.5cm of Q2.center, worker group] (W2) {}
                        node[right = 0cm of W2] {$\mathcal W_2$}
                    node[below left = 1cm and 0.5cm of Q2.center, replacement queen] (R2) {}
                        node[left = 0cm of R2] {$R_2$};

                \draw[mating, not considered] (D0) -- (Q0)
                    node[gene pass description, not considered] {mate};
                \draw[mating, not considered] (D1) -- (Q1)
                    node[gene pass description, not considered] {mate};
                \draw[mating, not considered] (D2) -- (Q2)
                    node[gene pass description, not considered] {mate};
                \draw[inheritance] (Q0) -- (W0);
                \draw[inheritance] (Q0) -- (R0);
                \draw[inheritance] (Q1) -- (W1);
                \draw[inheritance] (Q1) -- (R1);
                \draw[inheritance] (Q2) -- (W2);
                \draw[inheritance] (Q2) -- (R2);
                \draw[inheritance] (Q0) .. controls ++(-2,-1) .. (M1)
                    node[gene pass description] {4a-path};
                \draw[inheritance, not considered] (M1) -- (D1);
                \draw[inheritance, not considered] (M1) -- (D2);

                \begin{scope}[on background layer]
                    \draw[dashed] (M) ++ (-3cm,0cm) -- ++(13.9cm,0cm);
                \end{scope}

                \path (M) ++ (-1,-0.5) -- ++(2,1) node[fill = white, font = \huge\sffamily, midway,sloped] {\textbf{frequent}};

            \end{tikzpicture}
        \end{center}

        \item We assume that no queen $Q\in\mathcal Q_t$ is part of a mating station $\mathcal M$ and that any two non-identical mating stations are disjoint (also across generations), meaning that no DPQ can be used on more than one mating station or in more than one season.

    \end{enumerate}
\end{Rmk}

\begin{Not}
    Because we do not consider groups $\mathcal M$ of DPQs on a mating station as part of the population, it makes no difference for our purposes if a queen $Q\in\mathcal Q_t$ does not serve as the 4a-queen of any mating station or if she serves as the 4a-queen of a mating station that is not frequented by any queen in $\mathcal N\mathcal Q_{t+1}$. We may therefore assume that every queen $Q\in\mathcal Q_t$ serves as the 4a-queen of a mating station. By Remark~\ref{rmk::msassu}\,\ref{item::msassui}, this mating station is unique per season and we denote it by $\mathcal M_{Q,t}$.
\end{Not}

\begin{Rmk}
    \begin{enumerate}[label = (\roman*)]
        \item Inheritance via the dam path is unchanged in comparison to breeding schemes with single colony insemination. Thus, for a queen $Q\in\mathcal Q_t$, the relative contribution via the dam path, $dc_{Q,t}$ is defined as in Notation~\ref{not::dcbc}\,\ref{item::dc}. Consequently, also the vector $\mathbf {dc}_t\in\mathbb R^{\mathcal Q_t}$ is defined as in Notation~\ref{not::dcbc}\,\ref{item::dcbcvec}.

        \item Furthermore, also the survival path is unchanged, so for each queen $Q\in\mathcal Q_t$, we have the binary survival information $s_{Q,t}$ which results in a survival vector $\mathbf s_t\in\{0,1\}^{\mathcal Q_t}\cong\{0,1\}^{\mathfrak W_t}\cong \{0,1\}^{\mathcal R_t}$ just as introduced in Notation~\ref{not::survi}.

        \item We thus still have
            \[\mathbf 1_t^{\top}\mathbf {dc}_t=1\]
        and
            \[\mathbf 1_t^{\top}\mathbf {s}_t=N^{\mathcal S}_{t+1}.\]
    \end{enumerate}
\end{Rmk}

\begin{Not}
    However, we need to introduce a value for the contribution via the 4a-path. For a queen $Q\in\mathcal Q_t$ we denote by $ac_{Q,t}$ the fraction of queens $NQ\in\mathcal N\mathcal Q_{t+1}$ that went to the mating station $\mathcal M_{Q,t}$ for their mating flights. The resulting vector is denoted $\mathbf {ac}_t\in\mathbb R^{\mathcal Q_t}$.
\end{Not}

\begin{Rmk}\label{rmk::aone}
    Because all newly created queens in $\mathcal N\mathcal Q_{t+1}$ need to visit a mating station for their nuptial flights, we have
        \[\mathbf 1_t^{\top}\mathbf {ac}_{t}=1.\]
\end{Rmk}

\subsubsection{Breeding value development}\label{sec::imsbv}

As for breeding schemes with single colony insemination, we need to calculate the four expectations $\mathbb E\left[\hat{u}_{\mathcal N\mathcal Q_{t+1}}\right]$, $\mathbb E\left[\hat{u}_{\mathcal S\mathcal Q_{t+1}}\right]$, $\mathbb E\left[\hat{u}_{\mathcal N\mathcal R_{t+1}}\right]$, and $\mathbb E\left[\hat{u}_{\mathcal S\mathcal R_{t+1}}\right]$ in order to deduce the desired value of $\mathbb E\bigl[\hat{u}_{\mathcal P^{\ast}_{t+1}}\bigr]$ by means of Remark~\ref{rmk::brokenu}.

\begin{Lem}\label{lem::alltheusms}
    We have
    \begin{align}
        \mathbb E\left[\hat{u}_{\mathcal N\mathcal Q_{t+1}}\right]
            &= \mathbf {dc}_{t}^{\top}\hat{\mathbf u}_{t}^{\mathcal R}, \label{eq::unqms}\\
        \mathbb E\left[\hat{u}_{\mathcal S\mathcal Q_{t+1}}\right]
            &= \frac1{N_{t+1}^{\mathcal S}}\mathbf s_{t}^{\top}\hat{\mathbf u}_{t}^{\mathcal Q}, \label{eq::usqms}\\
        \mathbb E\left[\hat u_{\mathcal N\mathcal R_{t+1}}\right]
            &= \frac12\mathbf {dc}_{t}^{\top}\hat{\mathbf u}_{t}^{\mathcal R}
                + \frac12\mathbf {ac}_{t}^{\top}\hat{\mathbf u}_{t}^{\mathcal R}, \label{eq::unrms}\\
        \mathbb E\left[\hat{u}_{\mathcal S\mathcal R_{t+1}}\right]
            &= \frac1{N_{t+1}^{\mathcal S}}\mathbf s_{t}^{\top}\hat{\mathbf u}_{t}^{\mathcal R}. \label{eq::usrms}
    \end{align}
\end{Lem}

\begin{proof}
    Equations~\ref{eq::unqms}, \ref{eq::usqms}, and~\ref{eq::usrms} only involve contributions via the dam and survival paths. Therefore, they hold with the exact same arguments as in the case of single colony insemination (Lemma~\ref{lem::alltheus}).
    We show Equation~\ref{eq::unrms}, i.\,e.
        \[\mathbb E\left[\hat u_{\mathcal N\mathcal R_{t+1}}\right]
            = \frac12\mathbf {dc}_{t}^{\top}\hat{\mathbf u}_{t}^{\mathcal R}
                + \frac12\mathbf {ac}_{t}^{\top}\hat{\mathbf u}_{t}^{\mathcal R}.\]
    \begin{center}
        \begin{tikzpicture}
            \path (0,0) coordinate (Pt0)
                    node[left = 0mm of Pt0] {generation $\mathcal P_t$}
                coordinate[below = 4cm of Pt0] (Pt1)
                    node[left = 0mm of Pt1] {generation $\mathcal P_{t+1}$}
                (Pt0) -- (Pt1)
                    coordinate[pos = 0.4] (M)

                coordinate[right = 5.5cm of Pt0] (C0)
                node[above = 0.5cm of C0, queen] (Q) {}
                    node[left = 0cm of Q] {$Q$}
                node[below left = 1cm and 0.5cm of Q.center, replacement queen] (R) {}
                    node[left = 0cm of R] {$R(Q)$}
                node[below = 1.6cm of C0, group base] (M1prov) {}
                    -- (M1prov.center) node[not considered, group] (M1) {}
                    node[below right = 0cm and 0cm of M1, not considered] {$\mathcal M_{Q,t}$}
                    -- (M1.center) pic[not considered] {queens}

                coordinate[right = 3.5cm of Pt1] (C2)
                node[above = 0.5cm of C2, queen] (NQ) {}
                    node[left = 0cm of NQ] {$Q(NR)$}
                node[right = 2cm of NQ.center, not considered, group] (D) {}
                    node[right = 0cm of D, not considered] {$\mathcal D$}
                    -- (D) pic[not considered] {drones}
                node[below left = 1cm and 0.5cm of NQ.center, replacement queen] (NR) {}
                    node[left = 0cm of NR] {$NR$}
                    node[below = 0cm of NR, font = \footnotesize, align = left]
                        {$\mathbb E\left[\hat u_{NR,t+1}\right]=\frac12\mathbb E\left[\hat u_{Q(NR),t+1}\right]
                            + \mathbb E\left[\hat u_{\mathcal D,t+1}\right]$\\
                        $\phantom{\mathbb E\left[\hat u_{NR,t+1}\right]} = \frac12\mathbb E\left[\hat u_{Q(NR),t+1}\right]
                            + \frac12 \mathbb E\left[\hat u_{\mathcal M_{Q,t},t+1}\right]$\\
                        $\phantom{\mathbb E\left[\hat u_{NR,t+1}\right]} = \frac12\mathbb E\left[\hat u_{Q(NR),t+1}\right]
                            + \frac12 \hat u_{R(Q),t}$};

            \draw[mating, not considered] (D) -- (NQ)
                node[gene pass description, not considered] {mate};
            \draw[inheritance] (Q) -- (R);
            \draw[inheritance] (NQ) -- (NR);
            \draw[inheritance] (Q) -- (M1)
                node[gene pass description] {4a-path};
            \draw[inheritance, not considered] (M1) -- (D);

            \begin{scope}[on background layer]
                \draw[dashed] (M) ++ (-3cm,0cm) -- ++(13.9cm,0cm);
            \end{scope}

        \end{tikzpicture}
    \end{center}

    The expected breeding value of a new replacement queen $NR\in\mathcal N\mathcal R_{t+1}$ is half the breeding value of its queen $Q(NR)\in\mathcal Q_{t+1}$ plus the breeding value of the drone group $\mathcal D$ that $Q(NR)$ mated with (Lemma~\ref{lem::bvrepq}). But the expected breeding value of $\mathcal D$ is half the expected breeding value of the group of queens $\mathcal M_{Q,t}$ that produced the drones (Lemma~\ref{lem::bvdrgr}). Let $Q\in\mathcal Q_t$ be the 4a-queen of $\mathcal M_{Q,t}$. Then, the expected breeding value of $\mathcal M_{Q,t}$ is the estimated breeding value $\hat u_{R(Q),t}$ of $Q$'s replacement queen (Lemma~\ref{lem::expdaughter} in combination with Section~\ref{sec::wgrq}). The relative frequencies with which queens in $\mathcal Q_t$ occur as 4a-queens are given by the vector $\mathbf {ac}_{t}\in\mathbb R^{\mathcal Q_t}$. This leads to
        \[\mathbb E\left[\hat u_{\mathcal N\mathcal R_{t+1}}\right] = \frac12\mathbb E\left[\hat{u}_{\mathcal N\mathcal Q_{t+1}}\right]
            + \frac12\mathbf {ac}_{t}^{\top}\hat{\mathbf u}_{t}^{\mathcal R}.\]
    Inserting Equation~\ref{eq::unqms} yields the assertion for $\mathbb E\left[\hat u_{\mathcal N\mathcal R_{t+1}}\right]$.
\end{proof}

By inserting the results of Lemma~\ref{lem::alltheusms} into Lemma~\ref{lem::genred}, we obtain the desired formula for $\mathbb E\bigl[\hat u_{\mathcal P_{t+1}^{\ast}}\bigr]$ in case of mating on isolated mating stations:

\begin{Thm}\label{thm::utms}
    We have
    \begin{align}
        \mathbb E\left[\hat u_{\mathcal Q_{t+1}}\right] &=
            \frac{N_{t+1}^{\mathcal N}}{N_{t+1}}\mathbf {dc}_{t}^{\top}\hat{\mathbf u}_{t}^{\mathcal R}+\frac1{N_{t+1}}\mathbf s_t^{\top}\hat{\mathbf u}_{t}^{\mathcal Q}, \label{eq::utqms}\\
        \mathbb E\left[\hat{u}_{\mathcal R_{t+1}}\right] &=
            \frac{N_{t+1}^{\mathcal N}}{2N_{t+1}}\mathbf {dc}_{t}^{\top}\hat{\mathbf u}_{t}^{\mathcal R}+\frac{N_{t+1}^{\mathcal N}}{2N_{t+1}}\mathbf {ac}_{t}^{\top}\hat{\mathbf u}_{t}^{\mathcal R}+\frac1{N_{t+1}}\mathbf s_t^{\top}\hat{\mathbf u}_{t}^{\mathcal R}, \label{eq::utrms}\\
        \mathbb E\bigl[\hat{u}_{\mathcal P^{\ast}_{t+1}}\bigr] &=
            \frac{3N_{t+1}^{\mathcal N}}{4N_{t+1}}\mathbf {dc}_{t}^{\top}\hat{\mathbf u}_{t}^{\mathcal R}
                + \frac{N_{t+1}^{\mathcal N}}{4N_{t+1}}\mathbf {ac}_{t}^{\top}\hat{\mathbf u}_{t}^{\mathcal R}
                + \frac1{2N_{t+1}}\mathbf s_t^{\top}\left(\hat{\mathbf u}_{t}^{\mathcal R}+\hat{\mathbf u}_{t}^{\mathcal Q}\right).
    \end{align}
\end{Thm}

\subsubsection{Kinship development}\label{sec::imskin}

As in the case of single colony inseminations, we need to calculate $k_{\mathcal N\mathcal Q_{t+1}, \mathcal N\mathcal Q_{t+1}}$, $k_{\mathcal N\mathcal Q_{t+1}, \mathcal N\mathcal R_{t+1}}$, $k_{\mathcal N\mathcal Q_{t+1}, \mathcal S\mathcal Q_{t+1}}$, $k_{\mathcal N\mathcal Q_{t+1}, \mathcal S\mathcal R_{t+1}}$, $k_{\mathcal N\mathcal R_{t+1}, \mathcal N\mathcal R_{t+1}}$, $k_{\mathcal N\mathcal R_{t+1}, \mathcal S\mathcal Q_{t+1}}$, $k_{\mathcal N\mathcal R_{t+1}, \mathcal S\mathcal R_{t+1}}$, $k_{\mathcal S\mathcal Q_{t+1}, \mathcal S\mathcal Q_{t+1}}$, $k_{\mathcal S\mathcal Q_{t+1}, \mathcal S\mathcal R_{t+1}}$, and $k_{\mathcal S\mathcal R_{t+1}, \mathcal S\mathcal R_{t+1}}$ in order to obtain the average genetic kinship in the next reduced generation, $k_{\mathcal P_{t+1}^{\ast},\mathcal P_{t+1}^{\ast}}$ (Remark~\ref{rmk::brokenk}\,\ref{item::brokenkiii}). This is what we will do in this section.

\begin{Lem}\label{lem::allthekms}
    We have
    \begin{align}
        k_{\mathcal N\mathcal Q_{t+1},\mathcal N\mathcal Q_{t+1}}
            &= \mathbf {dc}_t^{\top}\mathbf K_t^{\mathfrak W\mathfrak W}\mathbf {dc}_t
                + \frac1{N_{t+1}^{\mathcal N}}\mathbf {dc}_t^{\top}\mathrm{diag}\left(\mathbf K_t^{\mathcal R\mathcal R}\right)
                - \frac1{N_{t+1}^{\mathcal N}}\mathbf {dc}_t^{\top}\mathrm{diag}\left(\mathbf K_t^{\mathfrak W\mathfrak W}\right), \label{eq::knqnqms}\\
        k_{\mathcal N\mathcal Q_{t+1},\mathcal S\mathcal Q_{t+1}}
            &= \frac1{N_{t+1}^{\mathcal S}}\mathbf {dc}_{t}^{\top}\mathbf K_t^{\mathcal R\mathcal Q}\mathbf{s}_t, \label{eq::knqsqms}\\
        k_{\mathcal S\mathcal Q_{t+1},\mathcal S\mathcal Q_{t+1}}
            &= \frac1{\left(N_{t+1}^{\mathcal S}\right)^2}\mathbf {s}_{t}^{\top}\mathbf K_t^{\mathcal Q\mathcal Q}\mathbf{s}_t,
                \label{eq::ksqsqms}\\
        k_{\mathcal N\mathcal Q_{t+1},\mathcal N\mathcal R_{t+1}}
            &= \frac12\mathbf {dc}_t^{\top}\mathbf K_t^{\mathfrak W\mathfrak W}\mathbf {dc}_t
                + \frac12\mathbf {dc}_t^{\top}\mathbf K_t^{\mathfrak W\mathfrak W}\mathbf {ac}_t \nonumber \\
                &\hspace{1cm} + \frac1{2N_{t+1}^{\mathcal N}}\mathbf {dc}_t^{\top}\mathrm{diag}\left(\mathbf K_t^{\mathcal R\mathcal R}\right)
                - \frac1{2N_{t+1}^{\mathcal N}}\mathbf {dc}_t^{\top}\mathrm{diag}\left(\mathbf K_t^{\mathfrak W\mathfrak W}\right), \label{eq::knqnrms}\\
        k_{\mathcal N\mathcal Q_{t+1},\mathcal S\mathcal R_{t+1}}
            &= \frac1{N_{t+1}^{\mathcal S}}\mathbf {dc}_{t}^{\top}\mathbf K_t^{\mathfrak W\mathfrak W}\mathbf{s}_t, \label{eq::knqsrms}\\
        k_{\mathcal S\mathcal Q_{t+1},\mathcal N\mathcal R_{t+1}}
            &= \frac1{2N_{t+1}^{\mathcal S}}\mathbf {dc}_{t}^{\top}\mathbf K_t^{\mathcal R\mathcal Q}\mathbf{s}_t
                + \frac1{2N_{t+1}^{\mathcal S}}\mathbf {ac}_t^{\top}\mathbf K_t^{\mathcal R\mathcal Q}\mathbf s_t, \label{eq::ksqnrms}\\
        k_{\mathcal S\mathcal Q_{t+1},\mathcal S\mathcal R_{t+1}}
            &= \frac1{\left(N_{t+1}^{\mathcal S}\right)^2}\mathbf s_t^{\top}\mathbf K_t^{\mathcal Q\mathcal R}\mathbf s_t,
                \label{eq::ksqsrms}\\
        k_{\mathcal N\mathcal R_{t+1},\mathcal N\mathcal R_{t+1}}
            &=\frac14\mathbf {dc}_t^{\top}\mathbf K_{t}^{\mathfrak W\mathfrak W}\mathbf{dc}_t
                + \frac12\mathbf {ac}_t^{\top}\mathbf K_{t}^{\mathfrak W\mathfrak W}\mathbf{dc}_t
                + \frac14\mathbf {ac}_t^{\top}\mathbf K_{t}^{\mathfrak W\mathfrak W}\mathbf{ac}_t\nonumber\\
            &\hspace{1cm}- \frac1{4N_{t+1}^{\mathcal N}}\mathbf{dc}_t^{\top}
                \mathrm{diag}\left(\mathbf K_t^{\mathfrak W\mathfrak W}\right)
                - \frac1{4N_{t+1}^{\mathcal N}}\mathbf{ac}_t^{\top}\mathrm{diag}\left(\mathbf K_t^{\mathfrak W\mathfrak W}\right)
                + \frac1{2N_{t+1}^{\mathcal N}}, \label{eq::knrnrms}\\
        k_{\mathcal N\mathcal R_{t+1},\mathcal S\mathcal R_{t+1}}
            &= \frac1{2N_{t+1}^{\mathcal S}}\mathbf {dc}_{t}^{\top}\mathbf K_t^{\mathfrak W\mathfrak W}\mathbf{s}_t
                    +\frac1{2N_{t+1}^{\mathcal S}}\mathbf {ac}_t^{\top}\mathbf K_{t}^{\mathfrak W\mathfrak W}\mathbf s_t,
                \label{eq::knrsrms}\\
        k_{\mathcal S\mathcal R_{t+1},\mathcal S\mathcal R_{t+1}}
            &= \frac1{\left(N_{t+1}^{\mathcal S}\right)^2}\mathbf s_t^{\top}\mathbf K_t^{\mathcal R\mathcal R}\mathbf s_t.
                \label{eq::ksrsrms}
    \end{align}
\end{Lem}

\begin{proof}
    \begin{enumerate}[label = (\roman*)]
        \item Equations~\ref{eq::knqnqms}, \ref{eq::knqsqms}, \ref{eq::ksqsqms}, \ref{eq::knqsrms}, \ref{eq::ksqsrms}, and~\ref{eq::ksrsrms} only concern the dam and survival paths and therefore hold with the same arguments as the corresponding equations in Lemma~\ref{lem::allthek} (Equations~\ref{eq::knqnq}, \ref{eq::knqsq}, \ref{eq::ksqsq}, \ref{eq::knqsr}, \ref{eq::ksqsr}, and~\ref{eq::ksrsr}, respectively). We will show the remaining four identities.

        \item We show Equation~\ref{eq::knqnrms}, i.\,e.
        \begin{align*}
            k_{\mathcal N\mathcal Q_{t+1},\mathcal N\mathcal R_{t+1}}
                &= \frac12\mathbf {dc}_t^{\top}\mathbf K_t^{\mathfrak W\mathfrak W}\mathbf {dc}_t
                    + \frac12\mathbf {dc}_t^{\top}\mathbf K_t^{\mathfrak W\mathfrak W}\mathbf {ac}_t \nonumber \\
                    &\hspace{1cm} + \frac1{2N_{t+1}^{\mathcal N}}\mathbf {dc}_t^{\top}\mathrm{diag}\left(\mathbf K_t^{\mathcal R\mathcal R}\right)
                    - \frac1{2N_{t+1}^{\mathcal N}}\mathbf {dc}_t^{\top}\mathrm{diag}\left(\mathbf K_t^{\mathfrak W\mathfrak W}\right).
        \end{align*}

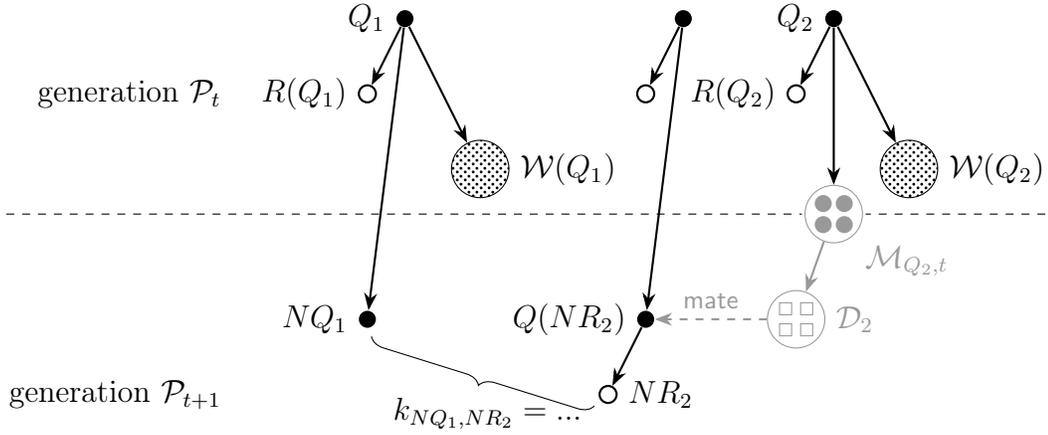
        \begin{center}
            \begin{tikzpicture}
                \path (0,0) coordinate (Pt0)
                        node[left = 0mm of Pt0] {generation $\mathcal P_t$}
                    coordinate[below = 4cm of Pt0] (Pt1)
                        node[left = 0mm of Pt1] {generation $\mathcal P_{t+1}$}
                    (Pt0) -- (Pt1)
                        coordinate[pos = 0.4] (M)

                    coordinate[right = 2.3cm of Pt0] (C0)
                    node[above = 1cm of C0, queen] (Q0) {}
                        node[left = 0cm of Q0] {$Q_1$}
                    node[below right = 2cm and 1cm of Q0.center, worker group] (W0) {}
                        node[right = 0cm of W0] {$\mathcal W(Q_1)$}
                    node[below left = 1cm and 0.5cm of Q0.center, replacement queen] (R0) {}
                        node[left = 0cm of R0] {$R(Q_1)$}

                    coordinate[right = 6cm of Pt0] (C1)
                    node[above = 1cm of C1, queen] (Q1) {}
                        node[left = 0cm of Q1] {}
                    node[below left = 1cm and 0.5cm of Q1.center, replacement queen] (R1) {}
                        node[left = 0cm of R1] {}

                    coordinate[right = 1.8cm of Pt1] (C3)
                    node[above = 1cm of C3, queen] (Q3) {}
                        node[left = 0cm of Q3] {$NQ_1$}

                    coordinate[right = 5.5cm of Pt1] (C4)
                    node[above = 1cm of C4, queen] (Q4) {}
                        node[left = 0cm of Q4] {$Q(NR_2)$}
                    node[right = 2cm of Q4.center, not considered, group] (D) {}
                        node[right = 0cm of D, not considered] {$\mathcal D_2$}
                        -- (D) pic[not considered] {drones}
                    node[below left = 1cm and 0.5cm of Q4.center, replacement queen] (R4) {}
                        node[right = 0cm of R4] {$NR_2$}

                    coordinate[right = 8cm of Pt0] (C5)
                    node[above = 1cm of C5, queen] (Q5) {}
                        node[left = 0cm of Q5] {$Q_2$}
                    node[below right = 2cm and 1cm of Q5.center, worker group] (W5) {}
                        node[right = 0cm of W5] {$\mathcal W(Q_2)$}
                    node[below left = 1cm and 0.5cm of Q5.center, replacement queen] (R5) {}
                        node[left = 0cm of R5] {$R(Q_2)$}
                    node[below = 1.6cm of C5, group base] (M5prov) {}
                        -- (M5prov.center) node[not considered, group] (M5) {}
                        node[below right = 0cm and 0cm of M5, not considered] {$\mathcal M_{Q_2,t}$}
                        -- (M5.center) pic[not considered] {queens};

                \draw[inheritance] (Q0) -- (R0);
                \draw[inheritance] (Q0) -- (W0);
                \draw[inheritance] (Q1) -- (R1);
                \draw[inheritance] (Q4) -- (R4);
                \draw[inheritance] (Q5) -- (R5);
                \draw[inheritance] (Q5) -- (W5);
                \draw[inheritance] (Q0) -- (Q3);
                \draw[inheritance] (Q1) -- (Q4);
                \draw[inheritance] (Q5) -- (M5);
                \draw[inheritance, not considered] (M5) -- (D);
                \draw[mating, not considered] (D) -- (Q4)
                    node[gene pass description, not considered] {mate};
                \draw[relationship] (Q3) -- (R4)
                    node[relationship description] {$k_{NQ_1,NR_2}=...$};

                \begin{scope}[on background layer]
                    \draw[dashed] (M) ++ (-3cm,0cm) -- ++(13.9cm,0cm);
                \end{scope}
            \end{tikzpicture}
        \end{center}

        Let $NQ_1\in\mathcal N\mathcal Q_{t+1}$ and $NR_2\in\mathcal N\mathcal R_{t+1}$. We fix a locus and draw an allele $A^{1}$ from $NQ_1$ and an allele $A^{2}$ from $NR_2$. The latter allele comes with probability $\frac12$ from $NR_2$'s  dam $Q(NR_2)$ and with probability $\frac12$ from the group $\mathcal D_2$ of drones that $Q(NR_2)$ mated with. Let $Q_2\in\mathcal Q_t$ be the 4a-queen of these drones and $\mathcal M_{Q_2,t}$ the group of DPQs on the corresponding mating station. Then,
            \[k_{NQ_1,NR_2}=\frac12k_{NQ_1,Q(NR_2)}+\frac12k_{NQ_1,\mathcal D_2}\]
        and thus by Equation~\ref{eq::kgd} (Lemma~\ref{lem::reldron})
        \begin{equation} \label{eq::helperms}
            k_{NQ_1,NR_2}=\frac12k_{NQ_1,Q(NR_2)}+\frac12k_{NQ_1,\mathcal M_{Q_2,t}}.
        \end{equation}
        Let $Q_1\in\mathcal Q_t$ be the dam of $NQ_1$. Since $NQ_1\in\mathcal N\mathcal Q_{t+1}$ cannot be an ancestor of $\mathcal M_{Q_2,t}$ and vice versa, we have by Corollary~\ref{cor::imp}
            \[k_{NQ_1,\mathcal M_{Q_2,t}}=k_{\mathcal W(Q_1),\mathcal W(Q_2)}.\]
        Note that once again we need to work with worker groups in order to cover the case $Q_1=Q_2$.
        By inserting this into Equation~\ref{eq::helperms}, we obtain
            \[k_{NQ_1,NR_2}=\frac12k_{NQ_1,Q(NR_2)}+\frac12k_{\mathcal W(Q_1),\mathcal W(Q_2)}.\]
        If now, we take averages over all choices of $NQ_1$ and $NR_2$, any given queen $Q\in\mathcal Q_t$ will occur in the role of $Q_1$ with frequency $dc_{Q,t}$ and in the role of $Q_2$ with frequency $ac_{Q,t}$. By that, we have
        \begin{align*}
            k_{\mathcal N\mathcal Q_{t+1},\mathcal N\mathcal R_{t+1}}
                &= \frac12k_{\mathcal N\mathcal Q_{t+1},\mathcal N\mathcal Q_{t+1}}
                    +\frac12\mathbf {dc}_t^{\top}\mathbf K_t^{\mathfrak W\mathfrak W}\mathbf {ac}_t
        \end{align*}
        and by inserting Equation~\ref{eq::knqnqms} for $k_{\mathcal N\mathcal Q_{t+1},\mathcal N\mathcal Q_{t+1}}$ (shown in part~\ref{item::knqnq} of the proof of Lemma~\ref{lem::allthek}), the assertion follows.

        \item \label{item::ksqnrms} We show Equation~\ref{eq::ksqnrms}, i.\,e.
            \[k_{\mathcal S\mathcal Q_{t+1},\mathcal N\mathcal R_{t+1}} =
                \frac1{2N_{t+1}^{\mathcal S}}\mathbf {dc}_{t}^{\top}\mathbf K_t^{\mathcal R\mathcal Q}\mathbf{s}_t^
                    +\frac1{2N_{t+1}^{\mathcal S}}\mathbf {ac}_t^{\top}\mathbf K_t^{\mathcal R\mathcal Q}\mathbf s_t.\]

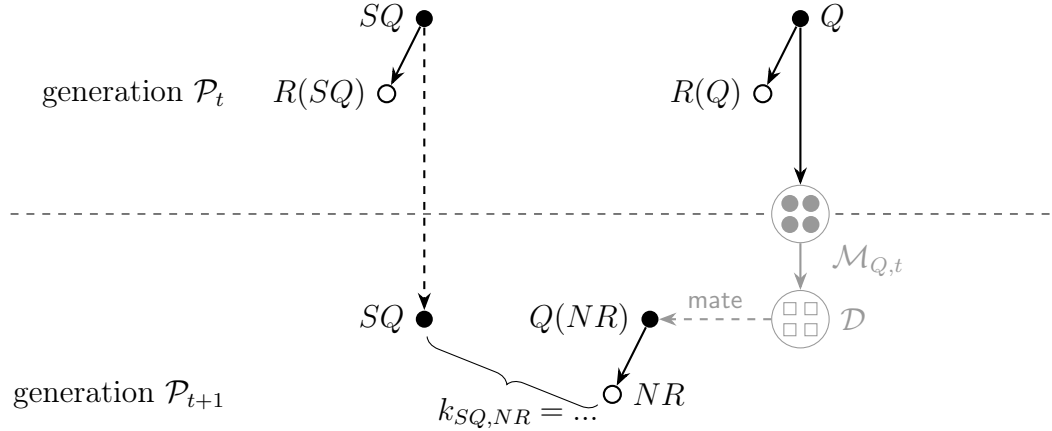
        \begin{center}
            \begin{tikzpicture}
                \path (0,0) coordinate (Pt0)
                        node[left = 0mm of Pt0] {generation $\mathcal P_t$}
                    coordinate[below = 4cm of Pt0] (Pt1)
                        node[left = 0mm of Pt1] {generation $\mathcal P_{t+1}$}
                    (Pt0) -- (Pt1)
                        coordinate[pos = 0.4] (M)

                    coordinate[right = 2.5cm of Pt0] (C0)
                    node[above = 1cm of C0, queen] (Q0) {}
                        node[left = 0cm of Q0] {$SQ$}
                    node[below left = 1cm and 0.5cm of Q0.center, replacement queen] (R0) {}
                        node[left = 0cm of R0] {$R(SQ)$}

                    coordinate[right = 7.5cm of Pt0] (C2)
                    node[above = 1cm of C2, queen] (Q2) {}
                        node[right = 0cm of Q2] {$Q$}
                    node[below = 2.6cm of Q2.center, group base] (M2base) {}
                    node[below = 2.6cm of Q2.center, group, not considered] (M2) {}
                        node[below right = 0cm and 0cm of M2, not considered] {$\mathcal M_{Q,t}$}
                        -- (M2.center) pic[not considered]{queens}
                    node[below left = 1cm and 0.5cm of Q2.center, replacement queen] (R2) {}
                        node[left = 0cm of R2] {$R(Q)$}

                    coordinate[right = 2.5cm of Pt1] (C3)
                    node[above = 1cm of C3, queen] (Q3) {}
                        node[left = 0cm of Q3] {$SQ$}

                    coordinate[right = 5.5cm of Pt1] (C4)
                    node[above = 1cm of C4, queen] (Q4) {}
                        node[left = 0cm of Q4] {$Q(NR)$}
                    node[right = 2cm of Q4.center, not considered, group] (D) {}
                        node[right = 0cm of D, not considered] {$\mathcal D$}
                        -- (D) pic[not considered] {drones}
                    node[below left = 1cm and 0.5cm of Q4.center, replacement queen] (R4) {}
                        node[right = 0cm of R4] {$NR$};

                    \draw[inheritance] (Q0) -- (R0);
                    \draw[inheritance] (Q2) -- (R2);
                    \draw[inheritance] (Q4) -- (R4);
                    \draw[inheritance] (Q2) -- (M2);
                    \draw[inheritance, not considered] (M2) -- (D);
                    \draw[survival] (Q0) -- (Q3);
                    \draw[mating, not considered] (D) -- (Q4)
                        node[gene pass description, not considered]{mate};
                    \draw[relationship] (Q3) -- (R4)
                        node[relationship description] {$k_{SQ,NR}=...$};

                \begin{scope}[on background layer]
                    \draw[dashed] (M) ++ (-3cm,0cm) -- ++(13.9cm,0cm);
                \end{scope}

            \end{tikzpicture}
        \end{center}

        Let $SQ\in\mathcal S\mathcal Q_{t+1}\subseteq\mathcal Q_t$ be a survivor queen and let $NR\in\mathcal N\mathcal R_{t+1}$ be a newly hatched replacement queen whose dam $Q(NR)\in\mathcal N\mathcal Q_{t+1}$ mated with a group $\mathcal D$ of drones on a mating station $\mathcal M_{Q,t}$ with 4a-queen $Q\in\mathcal Q_t$. Then
        \begin{align}
            k_{SQ,NR} &= \frac12k_{SQ,Q(NR)}+\frac12k_{SQ,\mathcal D}\nonumber\\
                      &= \frac12k_{SQ,Q(NR)}+\frac12k_{SQ,\mathcal M_{Q,t}}. \label{eq::hhh}
        \end{align}
        By construction, neither a DPQ from $\mathcal M_{Q,t}$ nor the replacement queen $R(Q)$ can be ancestors of $SQ$. Thus, by Lemma~\ref{lem::kinbet},
            \[k_{SQ,\mathcal M_{Q,t}} = k_{SQ,R(Q)}\]
        and insertion into Equation~\ref{eq::hhh} yields
            \[k_{SQ,NR} = \frac12k_{SQ,Q(NR)}+\frac12k_{SQ,R(Q)},\]
        Taking averages, we obtain by the usual arguments
        \begin{align*}
            k_{\mathcal S\mathcal Q_{t+1},\mathcal N\mathcal R_{t+1}}
                &=\frac12k_{\mathcal S\mathcal Q_{t+1},\mathcal N\mathcal Q_{t+1}}
                    +\frac1{2N_{t+1}^{\mathcal S}}\mathbf {ac}_t^{\top}\mathbf K_t^{\mathcal R\mathcal Q}\mathbf s_t.
        \end{align*}
        The assertion follows by inserting Equation~\ref{eq::knqsqms} (shown in part~\ref{item::knqsq} of the proof of Lemma~\ref{lem::allthek}).

        \item \label{item::knrnrms} We show Equation~\ref{eq::knrnrms}, i.\,e.
        \begin{align*}
            k_{\mathcal N\mathcal R_{t+1},\mathcal N\mathcal R_{t+1}}
                &=\frac14\mathbf {dc}_t^{\top}\mathbf K_{t}^{\mathfrak W\mathfrak W}\mathbf{dc}_t
                    + \frac12\mathbf {ac}_t^{\top}\mathbf K_{t}^{\mathfrak W\mathfrak W}\mathbf{dc}_t
                    + \frac14\mathbf {ac}_t^{\top}\mathbf K_{t}^{\mathfrak W\mathfrak W}\mathbf{ac}_t\\
                &\hspace{0.9cm}- \frac1{4N_{t+1}^{\mathcal N}}\mathbf{dc}_t^{\top}
                    \mathrm{diag}\left(\mathbf K_t^{\mathfrak W\mathfrak W}\right)
                    - \frac1{4N_{t+1}^{\mathcal N}}\mathbf{ac}_t^{\top}\mathrm{diag}\left(\mathbf K_t^{\mathfrak W\mathfrak W}\right)
                    + \frac1{2N_{t+1}^{\mathcal N}}.
        \end{align*}

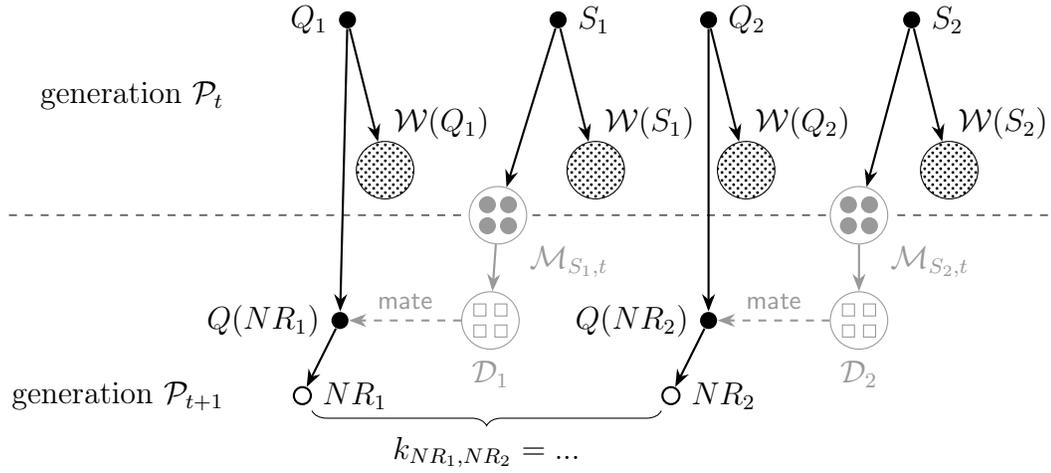
        \begin{center}
            \begin{tikzpicture}
                \path (0,0) coordinate (Pt0)
                        node[left = 0mm of Pt0] {generation $\mathcal P_t$}
                    coordinate[below = 4cm of Pt0] (Pt1)
                        node[left = 0mm of Pt1] {generation $\mathcal P_{t+1}$}
                    (Pt0) -- (Pt1)
                        coordinate[pos = 0.4] (M)

                    coordinate[right = 1.5cm of Pt0] (C0)
                    node[above = 1cm of C0, queen] (Q0) {}
                        node[left = 0cm of Q0] {$Q_1$}
                    node[below right = 2cm and 0.5cm of Q0.center, worker group] (W0) {}
                        node[above right = 0cm and -3mm of W0] {$\mathcal W(Q_1)$}

                    coordinate[right = 4.3cm of Pt0] (C1)
                    node[above = 1cm of C1, queen] (Q1) {}
                        node[right = 0cm of Q1] {$S_1$}
                    node[below left = 2.6cm and 8mm of Q1.center, group base] (M1base) {}
                    node[below left = 2.6cm and 8mm of Q1.center, group, not considered] (M1) {}
                        node[below right = 0cm and 0cm of M1, not considered] {$\mathcal M_{S_1,t}$}
                        -- (M1.center) pic[not considered]{queens}
                    node[below right = 2cm and 0.5cm of Q1.center, worker group] (W1) {}
                        node[above right = 0cm and -3mm of W1] {$\mathcal W(S_1)$}

                    coordinate[right = 6.3cm of Pt0] (C2)
                    node[above = 1cm of C2, queen] (Q2) {}
                        node[right = 0cm of Q2] {$Q_2$}
                    node[below right = 2cm and 0.5cm of Q2.center, worker group] (W2) {}
                        node[above right = 0cm and -3mm of W2] {$\mathcal W(Q_2)$}

                    coordinate[right = 9cm of Pt0] (C3)
                    node[above = 1cm of C3, queen] (Q3) {}
                        node[right = 0cm of Q3] {$S_2$}
                    node[below left = 2.6cm and 7mm of Q3.center, group base] (M3base) {}
                    node[below left = 2.6cm and 7mm of Q3.center, group, not considered] (M3) {}
                        node[below right = 0cm and 0cm of M3, not considered] {$\mathcal M_{S_2,t}$}
                        -- (M3.center) pic[not considered]{queens}
                    node[below right = 2cm and 0.5cm of Q3.center, worker group] (W3) {}
                        node[above right = 0cm and -3mm of W3] {$\mathcal W(S_2)$}

                    coordinate[right = 1.4cm of Pt1] (C4)
                    node[above = 1cm of C4, queen] (Q4) {}
                        node[left = 0cm of Q4] {$Q(NR_1)$}
                    node[right = 2cm of Q4.center, not considered, group] (D1) {}
                        node[below = 0cm of D1, not considered] {$\mathcal D_1$}
                        -- (D1) pic[not considered] {drones}
                    node[below left = 1cm and 0.5cm of Q4.center, replacement queen] (R4) {}
                        node[right = 0cm of R4] {$NR_1$}

                    coordinate[right = 6.3cm of Pt1] (C5)
                    node[above = 1cm of C5, queen] (Q5) {}
                        node[left = 0cm of Q5] {$Q(NR_2)$}
                    node[right = 2cm of Q5.center, not considered, group] (D2) {}
                        node[below = 0cm of D2, not considered] {$\mathcal D_2$}
                        -- (D2) pic[not considered] {drones}
                    node[below left = 1cm and 0.5cm of Q5.center, replacement queen] (R5) {}
                        node[right = 0cm of R5] {$NR_2$};

                \draw[inheritance] (Q0) -- (W0);
                \draw[inheritance] (Q1) -- (W1);
                \draw[inheritance] (Q2) -- (W2);
                \draw[inheritance] (Q3) -- (W3);
                \draw[inheritance] (Q4) -- (R4);
                \draw[inheritance] (Q5) -- (R5);
                \draw[inheritance] (Q0) -- (Q4);
                \draw[inheritance] (Q1) -- (M1);
                \draw[inheritance] (Q2) -- (Q5);
                \draw[inheritance] (Q3) -- (M3);
                \draw[inheritance, not considered] (M1) -- (D1);
                \draw[inheritance, not considered] (M3) -- (D2);
                \draw[mating, not considered] (D1) -- (Q4)
                    node[gene pass description, not considered] {mate};
                \draw[mating, not considered] (D2) -- (Q5)
                    node[gene pass description, not considered] {mate};
                \draw[relationship] (R4) -- (R5)
                    node[relationship description] {$k_{NR_1,NR_2}=...$};

                \begin{scope}[on background layer]
                    \draw[dashed] (M) ++ (-3cm,0cm) -- ++(13.9cm,0cm);
                \end{scope}
            \end{tikzpicture}
        \end{center}

        Let $NR_1,NR_2\in\mathcal N\mathcal R_{t+1}$ be two non-identical (!) replacement queens with dams $Q(NR_1)$ and $Q(NR_2)\in\mathcal N\mathcal Q_{t+1}$. Let $\mathcal M_{S_1,t},\mathcal M_{S_2,t}$ be the respective mating stations (with 4a-queens $S_1,S_2\in\mathcal Q_t$) that $Q(NR_1)$ and $Q(NR_2)$ mated on. Then by the standard argument that any allele drawn from $NR_i$ with $i\in\{1,2\}$ comes with equal probability either from $NQ_i$ or (via the drones) from $\mathcal M_{S_i,t}$, we have
            \[k_{NR_1,NR_2}=\frac14k_{Q(NR_1),Q(NR_2)}+\frac14k_{Q(NR_1),\mathcal M_{S_2,t}}
                + \frac14k_{\mathcal M_{S_1,t},Q(NR_2)}+\frac14k_{\mathcal M_{S_1,t},\mathcal M_{S_2,t}}.\]
        Let $Q_1,Q_2\in\mathcal Q_t$ be the respective dams of $Q(NR_1)$ and $Q(NR_2)$. Then by the replacements according to Corollary~\ref{cor::imp} and Lemma~\ref{lem::reldron}, we have
            \[k_{NR_1,NR_2}=\frac14k_{\mathcal W(Q_1),\mathcal W(Q_2)}+\frac14k_{\mathcal W(Q_1),\mathcal W(S_2)}+\frac14k_{\mathcal W(S_1),\mathcal W(Q_2)}+\frac14k_{\mathcal W(S_1),\mathcal W(S_2)}.\]
        Note that $NR_1\neq NR_2$ implies $Q(NR_1)\neq Q(NR_2)$, so that Corollary~\ref{cor::imp} can be applied.
        Note furthermore that similar to the proof of Equation~\ref{eq::knrnr} (in part~\ref{item::knqnq} of the proof of Lemma~\ref{lem::allthek}), we need to resort to worker groups in order to cover the case $Q_1=Q_2$ correctly.
        The frequencies with which a given queen $Q\in\mathcal Q_t$ occurs in the roles of $Q_1,Q_2,S_1$ and $S_2$ when taking averages are $dc_{Q,t},dc_{Q,t},ac_{Q,t}$, and $ac_{Q,t}$, respectively. From this we deduce the approximation
            \[k_{\mathcal N\mathcal R_{t+1},\mathcal N\mathcal R_{t+1}} \approx
                \frac14\mathbf {dc}_t^{\top}\mathbf K_{t}^{\mathfrak W\mathfrak W}\mathbf{dc}_t
                    + \frac12\mathbf {ac}_t^{\top}\mathbf K_{t}^{\mathfrak W\mathfrak W}\mathbf{dc}_t
                    + \frac14\mathbf {ac}_t^{\top}\mathbf K_{t}^{\mathfrak W\mathfrak W}\mathbf{ac}_t.\]
        This approximation would be an equality if the kinship of $NR_1$ to herself was $\frac14k_{\mathcal W(Q_1),\mathcal W(Q_1)}+\frac12k_{\mathcal W(Q_1),\mathcal W(S_1)}+\frac14k_{\mathcal W(S_1),\mathcal W(S_1)}$, which is not the case.
        Instead, we have by Remark~\ref{rmk::withcol}\,\ref{item::withcoli} in combination with Equation~\ref{eq::kgd} (Lemma~\ref{lem::reldron}) and Corollary~\ref{cor::imp}
        \begin{align*}
            k_{NR_1,NR_1} &= \frac12+\frac12k_{Q(NR_1),\mathcal M_{S_1,t}}\\
                          &= \frac12+\frac12k_{\mathcal W(Q_1),\mathcal W(S_1)}.
        \end{align*}
        So, for each replacement queen $NR_1\in\mathcal N\mathcal R_{t+1}$, we have to add the correction term
        \begin{align*}
            k_{NR_1,NR_1} -& \left(\frac14k_{\mathcal W(Q_1),\mathcal W(Q_1)}+\frac12k_{\mathcal W(Q_1),\mathcal W(S_1)}+\frac14k_{\mathcal W(S_1),\mathcal W(S_1)}\right)\\
                &= \frac12-\frac14k_{\mathcal W(Q_1),\mathcal W(Q_1)}-\frac14k_{\mathcal W(S_1),\mathcal W(S_1)}
        \end{align*}
        A queen $Q\in\mathcal Q_t$ occurs with frequency $dc_{Q,t}$ in the role of $Q_1$ and with frequency $ac_{Q,t}$ in the role of $S_1$.Thus, the term that needs to be added to the approximation is
            \[\frac1{2N_{t+1}^{\mathcal N}}
                - \frac1{4N_{t+1}^{\mathcal N}}\mathbf{dc}_t^{\top}\mathrm{diag}\left(\mathbf K_t^{\mathfrak W\mathfrak W}\right)
                - \frac1{4N_{t+1}^{\mathcal N}}\mathbf{ac}_t^{\top}\mathrm{diag}\left(\mathbf K_t^{\mathfrak W\mathfrak W}\right),\]
        and we end up at the claimed identity.

        \item \label{item::knrsrms} We show Equation~\ref{eq::knrsrms}, i.\,e.
            \[k_{\mathcal N\mathcal R_{t+1},\mathcal S\mathcal R_{t+1}}=
                \frac1{2N_{t+1}^{\mathcal S}}\mathbf {dc}_{t}^{\top}\mathbf K_t^{\mathfrak W\mathfrak W}\mathbf{s}_t
                    +\frac1{2N_{t+1}^{\mathcal S}}\mathbf {ac}_t^{\top}\mathbf K_{t}^{\mathfrak W\mathfrak W}\mathbf s_t.\]

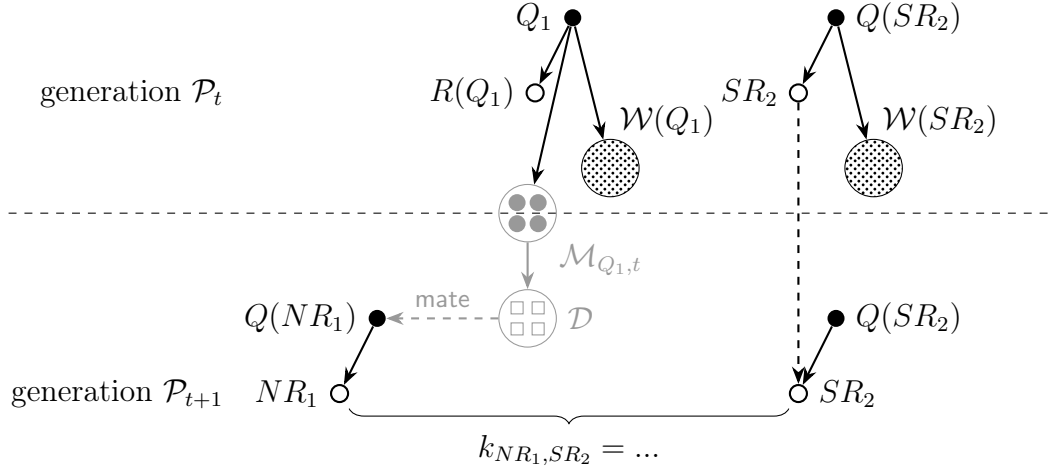
        \begin{center}
            \begin{tikzpicture}
                \path (0,0) coordinate (Pt0)
                        node[left = 0mm of Pt0] {generation $\mathcal P_t$}
                    coordinate[below = 4cm of Pt0] (Pt1)
                        node[left = 0mm of Pt1] {generation $\mathcal P_{t+1}$}
                    (Pt0) -- (Pt1)
                        coordinate[pos = 0.4] (M)

                    coordinate[right = 4.5cm of Pt0] (C0)
                    node[above = 1cm of C0, queen] (Q0) {}
                        node[left = 0cm of Q0] {$Q_1$}
                    node[below right = 2cm and 0.5cm of Q0.center, worker group] (W0) {}
                        node[above right = 0cm and -3mm of W0] {$\mathcal W(Q_1)$}
                    node[below left = 2.6cm and 6mm of Q0.center, group base] (M0base) {}
                    node[below left = 2.6cm and 6mm of Q0.center, group, not considered] (M0) {}
                        node[below right = 0cm and 0cm of M0, not considered] {$\mathcal M_{Q_1,t}$}
                        -- (M0.center) pic[not considered]{queens}
                    node[below left = 1cm and 0.5cm of Q0.center, replacement queen] (R0) {}
                        node[left = 0cm of R0] {$R(Q_1)$}

                    coordinate[right = 8cm of Pt0] (C1)
                    node[above = 1cm of C1, queen] (Q1) {}
                        node[right = 0cm of Q1] {$Q(SR_2)$}
                    node[below right = 2cm and 0.5cm of Q1.center, worker group] (W1) {}
                        node[above right = 0cm and -3mm of W1] {$\mathcal W(SR_2)$}
                    node[below left = 1cm and 0.5cm of Q1.center, replacement queen] (R1) {}
                        node[left = 0cm of R1] {$SR_2$}

                    coordinate[right = 1.9cm of Pt1] (C3)
                    node[above = 1cm of C3, queen] (Q3) {}
                        node[left = 0cm of Q3] {$Q(NR_1)$}
                    node[right = 2cm of Q3.center, not considered, group] (D) {}
                        node[right = 0cm of D, not considered] {$\mathcal D$}
                        -- (D) pic[not considered] {drones}
                    node[below left = 1cm and 0.5cm of Q3.center, replacement queen] (R3) {}
                        node[left = 0cm of R3] {$NR_1$}

                    coordinate[right = 8cm of Pt1] (C4)
                    node[above = 1cm of C4, queen] (Q4) {}
                        node[right = 0cm of Q4] {$Q(SR_2)$}
                    node[below left = 1cm and 0.5cm of Q4.center, replacement queen] (R4) {}
                        node[right = 0cm of R4] {$SR_2$};

                    \draw[inheritance] (Q0) -- (R0);
                    \draw[inheritance] (Q1) -- (R1);
                    \draw[inheritance] (Q0) -- (W0);
                    \draw[inheritance] (Q1) -- (W1);
                    \draw[inheritance] (Q3) -- (R3);
                    \draw[inheritance] (Q4) -- (R4);
                    \draw[inheritance] (Q0) -- (M0);
                    \draw[inheritance, not considered] (M0) -- (D);
                    \draw[survival] (R1) -- (R4);
                    \draw[mating, not considered] (D) -- (Q3)
                        node[gene pass description, not considered]{mate};
                    \draw[relationship] (R3) -- (R4)
                        node[relationship description] {$k_{NR_1,SR_2}=...$};

                \draw[dashed] (M) ++ (-3cm,0cm) -- ++(13.9cm,0cm);
            \end{tikzpicture}
        \end{center}

        Let $NR_1\in\mathcal N\mathcal R_{t+1}$ and $SR_2\in\mathcal S\mathcal R_{t+1}\subseteq\mathcal R_t$ be two replacement queens and let $\mathcal M_{Q_1,t}$ (with 4a-queen $Q_1\in\mathcal Q_t$) be the mating station on which $NR_1$'s dam $Q(NR_1)$ mated. Then $NR_1$ is not an ancestor of $SR_2$ and thus by the standard arguments (including Corollary~\ref{cor::imp}),
        \begin{align*}
            k_{NR_1,SR_2} &= \frac12k_{Q(NR_1),SR_2} + \frac12k_{\mathcal M_{Q_1,t},SR_2}\\
                          &= \frac12k_{Q(NR_1),SR_2} + \frac12k_{\mathcal W(Q_1),\mathcal W(SR_2)}.
        \end{align*}
        When taking averages, a queen $Q\in\mathcal Q_t$ will occur in the role of $Q_1$ with frequency $ac_{Q,t}$ and a replacement queen $R\in\mathcal R_t$ will occur in the role of $SR_2$ with frequency $\frac1{N_{t+1}^{\mathcal S}}s_{R,t}$. This yields
            \[k_{\mathcal N\mathcal R_{t+1},\mathcal S\mathcal R_{t+1}} =
                \frac12k_{\mathcal N\mathcal Q_{t+1},\mathcal S\mathcal R_{t+1}}
                    +\frac1{2N_{t+1}^{\mathcal S}}\mathbf {ac}_t^{\top}\mathbf K_{t}^{\mathfrak W\mathfrak W}\mathbf s_t\]
        The assertion follows by inserting Equation~\ref{eq::knqsrms} (shown in part~\ref{item::knqsr} of the proof of Lemma~\ref{lem::allthek}).

    \end{enumerate}
\end{proof}

With all these terms calculated, we once more insert them into the equations of Remark~\ref{rmk::brokenk}\,\ref{item::tenii}:

\begin{Lem}
    We have
    \begin{align}
        k_{\mathcal Q_{t+1},\mathcal Q_{t+1}}
            &= \left(\frac{N_{t+1}^{\mathcal N}}{N_{t+1}}\right)^2\mathbf{dc}_t^{\top}\mathbf{K}_t^{\mathfrak W\mathfrak W}\mathbf{dc}_t
                + \frac{N_{t+1}^{\mathcal N}}{N_{t+1}^2}\mathbf {dc}_t^{\top}\mathrm{diag}\left(\mathbf K_t^{\mathcal R\mathcal R}\right)
                - \frac{N_{t+1}^{\mathcal N}}{N_{t+1}^2}\mathbf {dc}_t^{\top}\mathrm{diag}\left(\mathbf K_t^{\mathfrak W\mathfrak W}\right) \nonumber\\
            &\hspace{1cm} +\frac{2N_{t+1}^{\mathcal N}}{N_{t+1}^2}\mathbf{dc}_t^{\top}\mathbf K_t^{\mathcal R\mathcal Q}\mathbf s_t
                + \frac1{N_{t+1}^2}\mathbf s_t^{\top}\mathbf K_t^{\mathcal Q\mathcal Q}\mathbf s_t,\\
        k_{\mathcal Q_{t+1},\mathcal R_{t+1}}
            &= \frac{\left(N_{t+1}^{\mathcal N}\right)^2}{2N_{t+1}^2}
                    \mathbf{dc}_t^{\top}\mathbf K_t^{\mathfrak W\mathfrak W}\mathbf{dc}_t
                + \frac{\left(N_{t+1}^{\mathcal N}\right)^2}{2N_{t+1}^2}
                    \mathbf{dc}_t^{\top}\mathbf K_t^{\mathfrak W\mathfrak W}\mathbf{ac}_t \nonumber\\
            &\hspace{1cm}+ \frac{N_{t+1}^{\mathcal N}}{2N_{t+1}^2}
                    \mathbf{1}_t^{\top}\left(\mathbf K_t^{\mathcal R\mathcal R}-\mathbf K_t^{\mathfrak W\mathfrak W}\right)\mathbf{dc}_t
                + \frac{N_{t+1}^{\mathcal N}}{N_{t+1}^2}
                    \mathbf{dc}_t^{\top}\mathbf K_t^{\mathfrak W\mathfrak W}\mathbf{s}_t \nonumber\\
            &\hspace{1cm}+ \frac{N_{t+1}^{\mathcal N}}{2N_{t+1}^2}\mathbf {dc}_t^{\top}\mathbf K_t^{\mathcal R\mathcal Q}\mathbf s_t
                + \frac{N_{t+1}^{\mathcal N}}{2N_{t+1}^2}\mathbf {ac}_t^{\top}\mathbf K_t^{\mathcal R\mathcal Q}\mathbf s_t
                + \frac1{N_{t+1}^2}\mathbf s_t^{\top}\mathbf K_t^{\mathcal Q\mathcal R}\mathbf s_t,\\
        k_{\mathcal R_{t+1},\mathcal R_{t+1}}
            &= \frac{\left(N_{t+1}^{\mathcal N}\right)^2}{4N_{t+1}^2}
                    \mathbf{dc}_t^{\top}\mathbf K_t^{\mathfrak W\mathfrak W}\mathbf{dc}_t
                + \frac{\left(N_{t+1}^{\mathcal N}\right)^2}{2N_{t+1}^2}
                    \mathbf{ac}_t^{\top}\mathbf K_t^{\mathfrak W\mathfrak W}\mathbf{dc}_t
                + \frac{\left(N_{t+1}^{\mathcal N}\right)^2}{4N_{t+1}^2}
                    \mathbf{ac}_t^{\top}\mathbf K_t^{\mathfrak W\mathfrak W}\mathbf{ac}_t \nonumber\\
            &\hspace{1cm} - \frac{N_{t+1}^{\mathcal N}}{4N_{t+1}^2}
                    \mathbf{dc}_t^{\top}\mathrm{diag}\left(K_t^{\mathfrak W\mathfrak W}\right)
                - \frac{N_{t+1}^{\mathcal N}}{4N_{t+1}^2}
                    \mathbf{ac}_t^{\top}\mathrm{diag}\left(K_t^{\mathfrak W\mathfrak W}\right) \nonumber\\
            &\hspace{1cm} + \frac{N_{t+1}^{\mathcal N}}{N_{t+1}^2}\mathbf{dc}_t^{\top}\mathbf{K}_t^{\mathfrak W\mathfrak W}\mathbf s_t
                + \frac{N_{t+1}^{\mathcal N}}{N_{t+1}^2}\mathbf{ac}_t^{\top}\mathbf{K}_t^{\mathfrak W\mathfrak W}\mathbf s_t
                + \frac{N_{t+1}^{\mathcal N}}{2N_{t+1}^2}
                + \frac1{N_{t+1}^2}\mathbf s_t\mathbf K_t^{\mathcal R\mathcal R}\mathbf s_t.
    \end{align}
\end{Lem}

These equations are then again inserted into Equation~\ref{eq::kng} to obtain $k_{\mathcal P_{t+1}^{\ast},\mathcal P_{t+1}^{\ast}}$.

\begin{Thm}\label{thm::burnerms}
    We have
    \begin{align}
        k_{\mathcal P^{\ast}_{t+1},\mathcal P^{\ast}_{t+1}}
            &= \left(\frac{N_{t+1}^{\mathcal N}}{4N_{t+1}}\right)^2\left(
                9\cdot\mathbf{dc}_t^{\top}\mathbf K_t^{\mathfrak W\mathfrak W}\mathbf {dc}_t
                    + 6\cdot\mathbf{dc}_t^{\top}\mathbf K_t^{\mathfrak W\mathfrak W}\mathbf {ac}_t
                    + \mathbf{ac}_t^{\top}\mathbf K_t^{\mathfrak W\mathfrak W}\mathbf {ac}_t
                \right) \nonumber \\
            &\hspace{0.5cm} + \frac{N_{t+1}^{\mathcal N}}{\left(4N_{t+1}\right)^2}\left(
                8\cdot\mathbf{dc}_t^{\top}\mathrm{diag}\left(\mathbf K_t^{\mathcal R\mathcal R}\right)
                    - 9\cdot\mathbf{dc}_t^{\top}\mathrm{diag}\left(\mathbf K_t^{\mathfrak W\mathfrak W}\right)\right. \nonumber \\
                    &\hspace{9cm}\left.- \mathbf{ac}_t^{\top}\mathrm{diag}\left(\mathbf K_t^{\mathfrak W\mathfrak W}\right)
                \right) \nonumber \\
            &\hspace{0.5cm} + \frac{N_{t+1}^{\mathcal N}}{\left(2N_{t+1}\right)^2}\left(
                3\cdot\mathbf{dc}_t^{\top}\mathbf K_t^{\mathcal R\mathcal Q}\mathbf s_t
                    + 3\cdot\mathbf{dc}_t^{\top}\mathbf K_t^{\mathfrak W\mathfrak W}\mathbf s_t
                    + \mathbf{ac}_t^{\top}\mathbf K_t^{\mathfrak W\mathfrak W}\mathbf s_t
                    + \mathbf{ac}_t^{\top}\mathbf K_t^{\mathcal R\mathcal Q}\mathbf s_t
                \right) \nonumber\\
            &\hspace{0.5cm} + \frac1{4N_{t+1}^2}\mathbf s_t^{\top}\left(
                    \mathbf K_t^{\mathcal Q\mathcal Q} + 2\cdot\mathbf K_t^{\mathcal Q\mathcal R} + \mathbf K_t^{\mathcal R\mathcal R}
                \right)\mathbf s_t
                + \frac{N^{\mathcal N}_{t+1}}{8N_{t+1}^2}.
    \end{align}
\end{Thm}

With all that, we are able to formulate the task of OCS for a honeybee population with mating on isolated mating stations.

\begin{Task}\label{task::ms}
    Given a generation $\mathcal P_t=\mathcal Q_t\sqcup\mathfrak W_t\sqcup\mathcal R_t$ of honeybee colonies, and
        \begin{itemize}
            \item vectors $\hat{\mathbf u}_t^{\mathcal Q}\in\mathbb R^{\mathcal Q_t}$ and $\hat{\mathbf u}_t^{\mathcal R}=\hat{\mathbf u}_t^{\mathfrak W}\in\mathbb R^{\mathcal R_t}\cong\mathbb R^{\mathfrak W_t}$ of estimated breeding values,

            \item a survival vector $\mathbf s_t\in\{0,1\}^{\mathcal Q_t}(\cong\{0,1\}^{\mathcal R_t}\cong\{0,1\}^{\mathfrak W_t})$,

            \item a symmetric and positive definite kinship matrix $\mathbf K_t\in\mathbb R^{\mathcal P_t\times\mathcal P_t}$ that falls into the blocks
                \[\mathbf K_t =
                    \begin{pmatrix}
                    \mathbf K_t^{\mathcal Q\mathcal Q}  & \mathbf K_t^{\mathcal Q\mathfrak W}  & \mathbf K_t^{\mathcal Q\mathcal R}\\
                    \mathbf K_t^{\mathfrak W\mathcal Q} & \mathbf K_t^{\mathfrak W\mathfrak W} & \mathbf K_t^{\mathfrak W\mathcal R}\\
                    \mathbf K_t^{\mathcal R\mathcal Q}  & \mathbf K_t^{\mathcal R\mathfrak W}  & \mathbf K_t^{\mathcal R\mathcal R}
                    \end{pmatrix}.\]
            and fulfills the properties listed in Remark~\ref{rmk::onk},

            \item the required number of newly created colonies of the next generation, $N_{t+1}^{\mathcal N}$,

            \item and a maximum acceptable kinship level $k_{t+1}^{\ast}$,
        \end{itemize}
        let $N_{t+1}:=N_{t+1}^{\mathcal N}+\mathbf 1_t^{\top}\mathbf s_t$ and maximize the function
        \begin{align*}
            \mathbb E\bigl[\hat u_{\mathcal P^{\ast}_{t+1}}\bigr]:
                \mathbb R_{\geq0}^{\mathcal Q_t}\oplus\mathbb R_{\geq0}^{\mathcal Q_t} &\to \mathbb R,\\
                \mathbf {dc}_{t}\oplus\mathbf {ac}_{t} &\mapsto
                \frac{3N_{t+1}^{\mathcal N}}{4N_{t+1}}\mathbf {dc}_{t}^{\top}\hat{\mathbf u}_{t}^{\mathcal R}
                + \frac{N_{t+1}^{\mathcal N}}{4N_{t+1}}\mathbf {ac}_{t}^{\top}\hat{\mathbf u}_{t}^{\mathcal R}
                + \frac1{2N_{t+1}}\mathbf s_t^{\top}\left(\hat{\mathbf u}_{t}^{\mathcal R}+\hat{\mathbf u}_{t}^{\mathcal Q}\right)
        \end{align*}

        under the constraints
            \[\mathbf 1_t^{\top}\mathbf {dc}_t=1,\]
            \[\mathbf 1_t^{\top}\mathbf {ac}_t=1,\]
        and
            \[k_{\mathcal P^{\ast}_{t+1},\mathcal P^{\ast}_{t+1}} \leq k_{t+1}^{\ast},\]
        where $k_{\mathcal P^{\ast}_{t+1},\mathcal P^{\ast}_{t+1}}$ denotes the term described in Theorem~\ref{thm::burnerms}.
\end{Task}

\subsection{Mixed strategies of mating control}\label{sec::mix}

Of course, breeding populations of honeybees can be heterogeneous and rely on more than one mode of mating control. Therefore, we finally look at a honeybee population in which both single colony insemination and isolated mating stations are used at the same time. 

Thus, each queen $Q\in\mathcal Q_t$ has four possibilities to contribute to the next generation, namely via the dam path, 1b-path, 4a-path and survival path. All these four paths have previously been considered, just not all at the same time. We can thus still use the notation from previous sections, in particular the vectors $\mathbf {dc}_t, \mathbf {bc}_t, \mathbf {ac}_t, \mathbf {s}_t\in\mathbb R^{\mathcal Q_t}$. 

\begin{Rmk}\label{rmk::abone}
    Every new queen $NQ\in\mathcal N\mathcal Q_{t+1}$ needs to either be inseminated or to mate on a mating station. Thus, instead of $\mathbf 1_t^{\top}\mathbf {bc}_t=1$ (Remark~\ref{rmk::dbone}) and $\mathbf 1_t^{\top}\mathbf {ac}_t=1$ (Remark~\ref{rmk::aone}), we now have
        \[\mathbf 1_t^{\top}\mathbf {bc}_t+\mathbf 1_t^{\top}\mathbf {ac}_t=1.\]
\end{Rmk}

\subsubsection{Breeding value development}\label{sec::mixbv}

The expected average breeding values $\mathbb E\left[\hat u_{\mathcal Q_{t+1}}\right]$, $\mathbb E\bigl[\hat u_{\mathcal R_{t+1}}\bigr]$, and $\mathbb E\bigl[\hat u_{\mathcal P^{\ast}_{t+1}}\bigr]$ in the case of mixed mating control strategies are derived just as in Sections~\ref{sec::scibv} and~\ref{sec::imsbv}. The respective contributions via the 1b-path and via the 4a-path in Theorems~\ref{thm::ut} and~\ref{thm::utms} have to be added. This yields

\begin{Thm}\label{thm::utmix}
    We have
    \begin{align}
        \mathbb E\left[\hat u_{\mathcal Q_{t+1}}\right] &=
            \frac{N_{t+1}^{\mathcal N}}{N_{t+1}}\mathbf {dc}_{t}^{\top}\hat{\mathbf u}_{t}^{\mathcal R}+\frac1{N_{t+1}}\mathbf s_t^{\top}\hat{\mathbf u}_{t}^{\mathcal Q}, \label{eq::utqmix}\\
        \mathbb E\left[\hat{u}_{\mathcal R_{t+1}}\right] &=
            \frac{N_{t+1}^{\mathcal N}}{2N_{t+1}}\mathbf {dc}_{t}^{\top}\hat{\mathbf u}_{t}^{\mathcal R}+\frac{N_{t+1}^{\mathcal N}}{2N_{t+1}}\mathbf {bc}_{t}^{\top}\hat{\mathbf u}_{t}^{\mathcal Q}+\frac{N_{t+1}^{\mathcal N}}{2N_{t+1}}\mathbf {ac}_{t}^{\top}\hat{\mathbf u}_{t}^{\mathcal R}+\frac1{N_{t+1}}\mathbf s_t^{\top}\hat{\mathbf u}_{t}^{\mathcal R}, \label{eq::utrmix}\\
        \mathbb E\bigl[\hat{u}_{\mathcal P^{\ast}_{t+1}}\bigr] &=
            \frac{3N_{t+1}^{\mathcal N}}{4N_{t+1}}\mathbf {dc}_{t}^{\top}\hat{\mathbf u}_{t}^{\mathcal R}
                + \frac{N_{t+1}^{\mathcal N}}{4N_{t+1}}\mathbf {bc}_{t}^{\top}\hat{\mathbf u}_{t}^{\mathcal Q}
                + \frac{N_{t+1}^{\mathcal N}}{4N_{t+1}}\mathbf {ac}_{t}^{\top}\hat{\mathbf u}_{t}^{\mathcal R}
                + \frac1{2N_{t+1}}\mathbf s_t^{\top}\left(\hat{\mathbf u}_{t}^{\mathcal R}+\hat{\mathbf u}_{t}^{\mathcal Q}\right).
    \end{align}
\end{Thm}

\subsubsection{Kinship development}\label{sec::mixkin}

Once again, we need to calculate the terms $k_{\mathcal N\mathcal Q_{t+1}, \mathcal N\mathcal Q_{t+1}}, ..., k_{\mathcal S\mathcal R_{t+1}, \mathcal S\mathcal R_{t+1}}$ in order to obtain the average genetic kinship in the next reduced generation, $k_{\mathcal P_{t+1}^{\ast},\mathcal P_{t+1}^{\ast}}$ (Remark~\ref{rmk::brokenk}\,\ref{item::brokenkiii}). The following Lemma gives the results.

\begin{Lem}\label{lem::allthekmix}
    We have
    \begin{align}
        k_{\mathcal N\mathcal Q_{t+1},\mathcal N\mathcal Q_{t+1}}
            &= \mathbf {dc}_t^{\top}\mathbf K_t^{\mathfrak W\mathfrak W}\mathbf {dc}_t
                + \frac1{N_{t+1}^{\mathcal N}}\mathbf {dc}_t^{\top}\mathrm{diag}\left(\mathbf K_t^{\mathcal R\mathcal R}\right)
                - \frac1{N_{t+1}^{\mathcal N}}\mathbf {dc}_t^{\top}\mathrm{diag}\left(\mathbf K_t^{\mathfrak W\mathfrak W}\right), \label{eq::knqnqmsix}\\
        k_{\mathcal N\mathcal Q_{t+1},\mathcal S\mathcal Q_{t+1}}
            &= \frac1{N_{t+1}^{\mathcal S}}\mathbf {dc}_{t}^{\top}\mathbf K_t^{\mathcal R\mathcal Q}\mathbf{s}_t, \label{eq::knqsqmix}\\
        k_{\mathcal S\mathcal Q_{t+1},\mathcal S\mathcal Q_{t+1}}
            &= \frac1{\left(N_{t+1}^{\mathcal S}\right)^2}\mathbf {s}_{t}^{\top}\mathbf K_t^{\mathcal Q\mathcal Q}\mathbf{s}_t,
                \label{eq::ksqsqmix}\\
        k_{\mathcal N\mathcal Q_{t+1},\mathcal N\mathcal R_{t+1}}
            &= \frac12\mathbf {dc}_t^{\top}\mathbf K_t^{\mathfrak W\mathfrak W}\mathbf {dc}_t
                + \frac12\mathbf {dc}_t^{\top}\mathbf K_t^{\mathcal R\mathcal Q}\mathbf {bc}_t
                + \frac12\mathbf {dc}_t^{\top}\mathbf K_t^{\mathfrak W\mathfrak W}\mathbf {ac}_t \nonumber \\
                &\hspace{1cm} + \frac1{2N_{t+1}^{\mathcal N}}\mathbf {dc}_t^{\top}\mathrm{diag}\left(\mathbf K_t^{\mathcal R\mathcal R}\right)
                - \frac1{2N_{t+1}^{\mathcal N}}\mathbf {dc}_t^{\top}\mathrm{diag}\left(\mathbf K_t^{\mathfrak W\mathfrak W}\right), \label{eq::knqnrmix}\\
        k_{\mathcal N\mathcal Q_{t+1},\mathcal S\mathcal R_{t+1}}
            &= \frac1{N_{t+1}^{\mathcal S}}\mathbf {dc}_{t}^{\top}\mathbf K_t^{\mathfrak W\mathfrak W}\mathbf{s}_t, \label{eq::knqsrmix}\\
        k_{\mathcal S\mathcal Q_{t+1},\mathcal N\mathcal R_{t+1}}
            &= \frac1{2N_{t+1}^{\mathcal S}}\mathbf {dc}_{t}^{\top}\mathbf K_t^{\mathcal R\mathcal Q}\mathbf{s}_t
                + \frac1{2N_{t+1}^{\mathcal S}}\mathbf {bc}_t^{\top}\mathbf K_t^{\mathcal Q\mathcal Q}\mathbf s_t
                + \frac1{2N_{t+1}^{\mathcal S}}\mathbf {ac}_t^{\top}\mathbf K_t^{\mathcal R\mathcal Q}\mathbf s_t, \label{eq::ksqnrmix}\\
        k_{\mathcal S\mathcal Q_{t+1},\mathcal S\mathcal R_{t+1}}
            &= \frac1{\left(N_{t+1}^{\mathcal S}\right)^2}\mathbf s_t^{\top}\mathbf K_t^{\mathcal Q\mathcal R}\mathbf s_t,
                \label{eq::ksqsrmix}\\
        k_{\mathcal N\mathcal R_{t+1},\mathcal N\mathcal R_{t+1}}
            &=\frac14\mathbf {dc}_t^{\top}\mathbf K_{t}^{\mathfrak W\mathfrak W}\mathbf{dc}_t
                + \frac12\mathbf {bc}_t^{\top}\mathbf K_{t}^{\mathcal Q\mathcal R}\mathbf{dc}_t
                + \frac12\mathbf {ac}_t^{\top}\mathbf K_{t}^{\mathfrak W\mathfrak W}\mathbf{dc}_t \nonumber\\
            &\hspace{1cm}+ \frac14\mathbf {bc}_t^{\top}\mathbf K_{t}^{\mathcal Q\mathcal Q}\mathbf{bc}_t
                + \frac12\mathbf {bc}_t^{\top}\mathbf K_{t}^{\mathcal Q\mathcal R}\mathbf{ac}_t
                + \frac14\mathbf {ac}_t^{\top}\mathbf K_{t}^{\mathfrak W\mathfrak W}\mathbf{ac}_t\nonumber\\
            &\hspace{1cm}- \frac1{4N_{t+1}^{\mathcal N}}\mathbf{dc}_t^{\top}
                \mathrm{diag}\left(\mathbf K_t^{\mathfrak W\mathfrak W}\right)
                - \frac1{4N_{t+1}^{\mathcal N}}\mathbf{bc}_t^{\top}\mathrm{diag}\left(\mathbf K_t^{\mathcal Q\mathcal Q}\right)\nonumber\\
            &\hspace{1cm}- \frac1{4N_{t+1}^{\mathcal N}}\mathbf{ac}_t^{\top}\mathrm{diag}\left(\mathbf K_t^{\mathfrak W\mathfrak W}\right)
                + \frac1{2N_{t+1}^{\mathcal N}}, \label{eq::knrnrmix}\\
        k_{\mathcal N\mathcal R_{t+1},\mathcal S\mathcal R_{t+1}}
            &= \frac1{2N_{t+1}^{\mathcal S}}\mathbf {dc}_{t}^{\top}\mathbf K_t^{\mathfrak W\mathfrak W}\mathbf{s}_t
                    +\frac1{2N_{t+1}^{\mathcal S}}\mathbf {bc}_t^{\top}\mathbf K_{t}^{\mathcal Q\mathcal R}\mathbf s_t
                    +\frac1{2N_{t+1}^{\mathcal S}}\mathbf {ac}_t^{\top}\mathbf K_{t}^{\mathfrak W\mathfrak W}\mathbf s_t,
                \label{eq::knrsrmix}\\
        k_{\mathcal S\mathcal R_{t+1},\mathcal S\mathcal R_{t+1}}
            &= \frac1{\left(N_{t+1}^{\mathcal S}\right)^2}\mathbf s_t^{\top}\mathbf K_t^{\mathcal R\mathcal R}\mathbf s_t.
                \label{eq::ksrsrmix}
    \end{align}
\end{Lem}

\begin{proof}
    The proof of this Lemma goes in full analogy with the proofs of Lemmas~\ref{lem::allthek} and~\ref{lem::allthekms}. Whenever the proof of an equation in Lemma~\ref{lem::allthek} considers an 1b-path, the proof of the corresponding equation in Lemma~\ref{lem::allthekms} considers an 1a path. In order to show the corresponding equation in the present Lemma~\ref{lem::allthekmix}, one simply has to make a case distinction, keeping in mind that now both paths are possible.
\end{proof}

We thus obtain
\begin{Thm}\label{thm::burnermix}
    We have
    \begin{align}
        k_{\mathcal P^{\ast}_{t+1},\mathcal P^{\ast}_{t+1}}
            &= \left(\frac{N_{t+1}^{\mathcal N}}{4N_{t+1}}\right)^2\left(
                9\cdot\mathbf{dc}_t^{\top}\mathbf K_t^{\mathfrak W\mathfrak W}\mathbf {dc}_t
                    + 6\cdot\mathbf{dc}_t^{\top}\mathbf K_t^{\mathcal R\mathcal Q}\mathbf {bc}_t
                    + 6\cdot\mathbf{dc}_t^{\top}\mathbf K_t^{\mathfrak W\mathfrak W}\mathbf {ac}_t\right.\nonumber\\
            &\hspace{4.5cm}\left.+ \mathbf{bc}_t^{\top}\mathbf K_t^{\mathcal Q\mathcal Q}\mathbf {bc}_t
                    + 2\mathbf{bc}_t^{\top}\mathbf K_t^{\mathcal Q\mathcal R}\mathbf {ac}_t
                    + \mathbf{ac}_t^{\top}\mathbf K_t^{\mathfrak W\mathfrak W}\mathbf {ac}_t
                \right) \nonumber \\
            &\hspace{0.5cm} + \frac{N_{t+1}^{\mathcal N}}{\left(4N_{t+1}\right)^2}\left(
                8\cdot\mathbf{dc}_t^{\top}\mathrm{diag}\left(\mathbf K_t^{\mathcal R\mathcal R}\right)
                    - 9\cdot\mathbf{dc}_t^{\top}\mathrm{diag}\left(\mathbf K_t^{\mathfrak W\mathfrak W}\right)\right. \nonumber \\
                    &\hspace{6cm}\left.- \mathbf{bc}_t^{\top}\mathrm{diag}\left(\mathbf K_t^{\mathcal Q\mathcal Q}\right)
                        - \mathbf{ac}_t^{\top}\mathrm{diag}\left(\mathbf K_t^{\mathfrak W\mathfrak W}\right)
                \right) \nonumber \\
            &\hspace{0.5cm} + \frac{N_{t+1}^{\mathcal N}}{\left(2N_{t+1}\right)^2}\left(
                3\cdot\mathbf{dc}_t^{\top}\mathbf K_t^{\mathcal R\mathcal Q}\mathbf s_t
                    + 3\cdot\mathbf{dc}_t^{\top}\mathbf K_t^{\mathfrak W\mathfrak W}\mathbf s_t
                    + \mathbf{bc}_t^{\top}\mathbf K_t^{\mathcal Q\mathcal Q}\mathbf s_t
                    + \mathbf{bc}_t^{\top}\mathbf K_t^{\mathcal Q\mathcal R}\mathbf s_t\right.\nonumber\\
                    &\hspace{8cm}\left.+ \mathbf{ac}_t^{\top}\mathbf K_t^{\mathfrak W\mathfrak W}\mathbf s_t
                    + \mathbf{ac}_t^{\top}\mathbf K_t^{\mathcal R\mathcal Q}\mathbf s_t
                \right) \nonumber\\
            &\hspace{0.5cm} + \frac1{4N_{t+1}^2}\mathbf s_t^{\top}\left(
                    \mathbf K_t^{\mathcal Q\mathcal Q} + 2\cdot\mathbf K_t^{\mathcal Q\mathcal R} + \mathbf K_t^{\mathcal R\mathcal R}
                \right)\mathbf s_t
                + \frac{N^{\mathcal N}_{t+1}}{8N_{t+1}^2}.
    \end{align}
\end{Thm}

Finally, we can formulate the task of OCS for a honeybee population with both instrumental insemination and mating on isolated mating stations.

\begin{Task}\label{task::mix}
    Given a generation $\mathcal P_t=\mathcal Q_t\sqcup\mathfrak W_t\sqcup\mathcal R_t$ of honeybee colonies, and
        \begin{itemize}
            \item vectors $\hat{\mathbf u}_t^{\mathcal Q}\in\mathbb R^{\mathcal Q_t}$ and $\hat{\mathbf u}_t^{\mathcal R}=\hat{\mathbf u}_t^{\mathfrak W}\in\mathbb R^{\mathcal R_t}\cong\mathbb R^{\mathfrak W_t}$ of estimated breeding values,

            \item a survival vector $\mathbf s_t\in\{0,1\}^{\mathcal Q_t}(\cong\{0,1\}^{\mathcal R_t}\cong\{0,1\}^{\mathfrak W_t})$,

            \item a symmetric and positive definite kinship matrix $\mathbf K_t\in\mathbb R^{\mathcal P_t\times\mathcal P_t}$ that falls into the blocks
                \[\mathbf K_t =
                    \begin{pmatrix}
                    \mathbf K_t^{\mathcal Q\mathcal Q}  & \mathbf K_t^{\mathcal Q\mathfrak W}  & \mathbf K_t^{\mathcal Q\mathcal R}\\
                    \mathbf K_t^{\mathfrak W\mathcal Q} & \mathbf K_t^{\mathfrak W\mathfrak W} & \mathbf K_t^{\mathfrak W\mathcal R}\\
                    \mathbf K_t^{\mathcal R\mathcal Q}  & \mathbf K_t^{\mathcal R\mathfrak W}  & \mathbf K_t^{\mathcal R\mathcal R}
                    \end{pmatrix}.\]
            and fulfills the properties listed in Remark~\ref{rmk::onk},

            \item the required number of newly created colonies of the next generation, $N_{t+1}^{\mathcal N}$,

            \item and a maximum acceptable kinship level $k_{t+1}^{\ast}$,
        \end{itemize}
        let $N_{t+1}:=N_{t+1}^{\mathcal N}+\mathbf 1_t^{\top}\mathbf s_t$ and maximize the function
        \begin{align*}
            \mathbb E\bigl[\hat u_{\mathcal P^{\ast}_{t+1}}\bigr]:
                \mathbb R_{\geq0}^{\mathcal Q_t}\oplus\mathbb R_{\geq0}^{\mathcal Q_t}\oplus\mathbb R_{\geq0}^{\mathcal Q_t}
                    &\to \mathbb R,\\
                \mathbf {dc}_{t}\oplus\mathbf {bc}_{t}\oplus\mathbf {ac}_{t} &\mapsto
                \frac{3N_{t+1}^{\mathcal N}}{4N_{t+1}}\mathbf {dc}_{t}^{\top}\hat{\mathbf u}_{t}^{\mathcal R}
                + \frac{N_{t+1}^{\mathcal N}}{4N_{t+1}}\mathbf {bc}_{t}^{\top}\hat{\mathbf u}_{t}^{\mathcal Q}\\
                &\hspace{3cm}+ \frac{N_{t+1}^{\mathcal N}}{4N_{t+1}}\mathbf {ac}_{t}^{\top}\hat{\mathbf u}_{t}^{\mathcal R}
                + \frac1{2N_{t+1}}\mathbf s_t^{\top}\left(\hat{\mathbf u}_{t}^{\mathcal R}+\hat{\mathbf u}_{t}^{\mathcal Q}\right)
        \end{align*}

        under the constraints
            \[\mathbf 1_t^{\top}\mathbf {dc}_t=1,\]
            \[\mathbf 1_t^{\top}\left(\mathbf {bc}_t+\mathbf {ac}_t\right)=1,\]
        and
            \[k_{\mathcal P^{\ast}_{t+1},\mathcal P^{\ast}_{t+1}} \leq k_{t+1}^{\ast},\]
        where $k_{\mathcal P^{\ast}_{t+1},\mathcal P^{\ast}_{t+1}}$ denotes the term described in Theorem~\ref{thm::burnermix}.
\end{Task}

\begin{Rmk}\label{rmk::convert}
    Note that Task~\ref{task::mix} turns into Task~\ref{task::sci} if one imposes the additional condition
        \[\mathbf 1_t^{\top}\mathbf {ac}_t=0\]
    and into Task~\ref{task::ms} if one imposes
        \[\mathbf 1_t^{\top}\mathbf {bc}_t=0.\]
    We may therefore see Tasks~\ref{task::sci} and~\ref{task::ms} as special cases of Task~\ref{task::mix}.
\end{Rmk}

\subsection{OCS with limited eligibility and other variations}

In all the concepts of OCS we have derived so far, it has been assumed that all members of a generation $\mathcal P_t$ (all members of $\mathcal Q_t$ in case of honeybees) are eligible for reproduction. In practice, this is not generally the case. In horses, for example, it is not recommended to foal a mare before she is three years old \citep{panzani07}. Thus, if one considers overlapping generations with time steps of one year, for the first three years of her life, a filly will be part of the population but her contribution to the next generation has to be zero.

In honeybees, one may want to impose a rule that queens have to undergo a complete performance test before they can be selected \citep{dib21}. Thus, a one-year-old queen will be part of the population but should not yet reproduce. Or, as another example, assume that a queen's mating failed in the sense that she likely mated with drones of the wrong subspecies. To avoid hybridization in the population, one would not want such a queen to reproduce anymore via the dam path or the 4a-path. However, it is perceivable to still use such a queen as a 1b-queen because the unsuccessful mating does not affect her drone production.

\begin{Not}\label{not::elig}
    \begin{enumerate}[label = (\roman*)]
        \item \label{item::eligone} In diploids, for each individual $I\in\mathcal P_t$ we introduce its \emph{eligibility} for reproduction as the binary value
            \[e_{I,t}=\begin{cases}
                          1,&\text{if }I\text{ can currently reproduce,}\\
                          0,&\text{otherwise}
                      \end{cases}.
            \]
        This give rise to the eligibility vector $\mathbf e_t\in\{0,1\}^{\mathcal P_t}$. We further introduce the non-eligibility vector $\mathbf {ne}_t$ as
            \[\mathbf {ne}_t=\mathbf 1_t-\mathbf e_t\in\{0,1\}^{\mathcal P_t}.\]

        \item \label{item::hbelig} In honeybees, a queen $Q\in\mathcal Q_t$ can reproduce via three paths (dam path, 1b-path, and 4a-path) and for all three paths, different eligibility criteria may be in place. Thus, in analogy to~\ref{item::eligone} we define separate different binary eligibility values $e^d_{Q,t},e^b_{Q,t},e^a_{Q,t}\in\{0.1\}$ for the dam path, 1b-path and 4a-path, respectively. Accordingly, we obtain three eligibility vectors $\mathbf e^d_t,\mathbf e^b_t,\mathbf e^a_t\in\{0,1\}^{\mathcal Q_t}$ and three non-eligibility vectors $\mathbf {ne}^d_t,\mathbf {ne}^b_t,\mathbf {ne}^a_t\in\{0,1\}^{\mathcal Q_t}$.

        \item In order to incorporate these limited eligibilities of individuals or queens for reproduction, one has to add the further condition
            \[\mathbf {ne}_t^{\top}\mathbf c_t=0\]
        in the case of diploids, and the conditions
            \[\left(\mathbf {ne}^d_t\right)^{\top}\mathbf {dc}_t=\left(\mathbf {ne}^b_t\right)^{\top}\mathbf {bc}_t=\left(\mathbf {ne}^a_t\right)^{\top}\mathbf {ac}_t=0\]
        in the honeybee case.
    \end{enumerate}
\end{Not}

\begin{Rmk}\label{rmk::sel4a}
    A special limitation of eligibility comes in case of isolated mating stations. Without eligibility restrictions, all queens of $\mathcal Q_t$ could in general serve as a 4a-queen, resulting in $N_t$ isolated mating stations. In practice, the number of maintained physical mating stations is generally a predefined number $N^{\mathfrak M}_t\ll N_t$. In addition to general (age-related) eligibility criteria, one would therefore like to impose a condition that at most $N^{\mathfrak M}_t$ entries of $\mathbf {ac_t}$ may be non-zero. However, this constraint turns out to be highly nonlinear (and non-quadratic) and much more complicated than the other constraints we imposed on the maximization problems in our tasks. Therefore, what one will have to do in practice is to preselect the $N^{\mathfrak M}_t$ 4a-queens out of $\mathcal Q_t$ and declare only these queens as eligible for the 4a-path. OCS will thus not tell which queens should be selected as 4a-queens but only how often the respective mating stations of otherwise selected 4a-queens should be frequented. The question how to select the 4a-queens deserves further investigation. Obvious possibilities are to choose the $N^{\mathfrak M}_t$ queens with the highest estimated breeding values, or to do a \emph{within-family selection} approach by avoiding to select sister queens as 4a-queens. Another potentially interesting possibility is to first solve the OCS task without restriction on the number of mating stations and then solve it again but declare only the $N^{\mathfrak M}_t$ queens eligible that were attributed the greatest contribution values in the unrestricted problem.
\end{Rmk}

All the tasks we have derived so far can be altered in multiple ways. For example, \citet{wellmann19optimum} follows a slightly different approach for diploids, predefining different numbers of offspring for different age\,$\times$\,sex-classes. In the remainder of this chapter, we want to present and motivate two noteworthy alternatives for honeybee-specific OCS.

\begin{Rmk}\label{rmk::alternatives}
    \begin{enumerate}[label = (\roman*)]
        \item In our approach of maximizing $\mathbb E\bigl[\hat u_{\mathcal P_{t+1}^{\ast}}\bigr]$, we weighed the average breeding values of queens (i.\,e. $\mathbb E\left[\hat u_{\mathcal Q_{t+1}}\right]$) and replacement queens (i.\,e. $\mathbb E\left[\hat u_{\mathcal R_{t+1}}\right]$) equally. If mating control is organized solely via single colony insemination, this appears justified. Queens can pass on their own breeding value via the 1b-path and the breeding value of their replacement queens via the dam path and both paths should be seen as equally important. If, however, mating control is organized with isolated mating stations, queens only pass on the breeding values of their replacement queens, both via the dam path and via the 4a-path. Thus, it seems reasonable to maximize $\mathbb E\left[\hat u_{\mathcal R_{t+1}}\right]$ instead of $\mathbb E\bigl[\hat u_{\mathcal P_{t+1}^{\ast}}\bigr]$. For mixed strategies, a weighted average between $\mathbb E\left[\hat u_{\mathcal R_{t+1}}\right]$ and $\mathbb E\bigl[\hat u_{\mathcal P_{t+1}^{\ast}}\bigr]$ could be chosen for maximization.

        A counterargument against this approach might be that also with mating stations, phenotypes are still influenced by (the queen effect of) the queen's breeding value and (the worker effect of) the worker group's breeding value.

        \item \label{item::altii} An alternative for the restrictions on average kinships is to replace the single condition
            \[k_{\mathcal P_{t+1}^{\ast},\mathcal P_{t+1}^{\ast}}\leq k_{t+1}^{\ast}\]
        by two separate conditions for the kinships between queens and replacement queens:
        \begin{align*}
            k_{\mathcal Q_{t+1},\mathcal Q_{t+1}}&\leq k_{t+1}^{\mathcal Q,\ast},\\
            k_{\mathcal R_{t+1},\mathcal R_{t+1}}&\leq k_{t+1}^{\mathcal R,\ast}.
        \end{align*}
        for acceptable kinship values $k_{t+1}^{\mathcal Q,\ast}$ and $k_{t+1}^{\mathcal R,\ast}$. By doing so, one drops limitations for the kinships between queens and replacement queens, the significance of which seems unclear.

        \item The OCS tasks corresponding to these alternatives can easily be formulated. We do not see ourselves able to give a definitive judgment on what is the best approach to follow. Likely, it is best to test the alternatives against each other in simulation studies and then opt for the variant with the most promising results.
    \end{enumerate}
\end{Rmk}

\section{Solving the tasks} \label{sec::solvetask}
\subsection{General form} \label{sec::genform}

Finally, we turn to the question of how to solve the different tasks. We start with the observation that all tasks we have introduced have the following form:

\begin{Task}\label{task::general}
    Given a dimension number $N$, and
    \begin{itemize}
        \item a vector $\tilde{\mathbf a}\in\mathbb R^N$ and a scalar $\tilde b\in\mathbb R$,

        \item a number $n\in\mathbb N$, a family of vectors $\tilde{\mathbf e}_i\in\mathbb R^N$ for $1\leq i\leq n$ and a family of scalars $\tilde d_i$ for $1\leq i \leq n$,

        \item a symmetric matrix $\tilde{\mathbf K}\in\mathbb R^{N\times N}$, a vector $\tilde{\mathbf m}\in\mathbb R^N$, and a scalar $\tilde k^{\ast}$,
    \end{itemize}
    maximize the function
    \begin{align*}
        \mathbb R_{\geq0}^{N}\to \mathbb R,\quad \tilde{\mathbf c} \mapsto \tilde{\mathbf a}^{\top}\tilde{\mathbf c} + \tilde b
    \end{align*}

    under the constraints
        \[\tilde{\mathbf e}_i^{\top}\tilde{\mathbf {c}}=\tilde d_i,\quad \text{for }1\leq i\leq n\]
    and
        \[\tilde{\mathbf c}^{\top}\tilde{\mathbf K}\tilde{\mathbf c}+\tilde{\mathbf m}^{\top}\tilde{\mathbf c}\leq\tilde k^{\ast}.\]
\end{Task}

\begin{Rmk}\label{rmk::concretetask}
    \begin{enumerate}[label = (\roman*)]
        \item Our formulation of Task~\ref{task::mix} for OCS with a mixed strategy of mating control can be brought in the general form of Task~\ref{task::general} by choosing $N=3N_t$ and letting $\tilde{\mathbf c}=\begin{pmatrix}\mathbf {dc}_t\\\mathbf {bc}_t\\\mathbf {ac}_t\\\end{pmatrix}\in\mathbb R^{3N_t}$ with the following choices of the remaining variables:
        \begin{align*}
            \tilde{\mathbf a} &= \frac{N^{\mathcal N}_{t+1}}{4N_{t+1}}
                                 \begin{pmatrix}3\hat{\mathbf u}_t^{\mathcal R}\\\hat{\mathbf u}_t^{\mathcal Q}\\\hat{\mathbf u}_t^{\mathcal R}\end{pmatrix}\in\mathbb R^{3N_{t}},\\
            \tilde{b} &= \frac1{2N_{t+1}}\mathbf s_t^{\top}\left(\hat{\mathbf u}_{t}^{\mathcal R}+\hat{\mathbf u}_{t}^{\mathcal Q}\right)\in \mathbb R,\\
            n &= 2,\\
            \tilde{\mathbf e}_1 &= \begin{pmatrix}\mathbf 1_t\\\mathbf 0_t\\\mathbf 0_t\end{pmatrix}\in\mathbb R^{3N_t},\quad 
                \tilde{\mathbf e}_2 = \begin{pmatrix}\mathbf 0_t\\\mathbf 1_t\\\mathbf 1_t\end{pmatrix}\in\mathbb R^{3N_t},\\
            \tilde d_1 &= \tilde d_2 = 1,\\
            \tilde{\mathbf K} &= \left(\frac{N^{\mathcal N}_{t+1}}{4N_{t+1}}\right)^2
                                   \begin{pmatrix}
                                       9\mathbf K_{t}^{\mathfrak W\mathfrak W} & 3\mathbf K_{t}^{\mathcal R\mathcal Q} & 3\mathbf K_{t}^{\mathfrak W\mathfrak W}\\
                                       3\mathbf K_{t}^{\mathcal Q\mathcal R} & \mathbf K_{t}^{\mathcal Q\mathcal Q} & \mathbf K_{t}^{\mathcal Q\mathcal R} \\
                                       3\mathbf K_{t}^{\mathfrak W\mathfrak W} & \mathbf K_{t}^{\mathcal R\mathcal Q} & \mathbf K_{t}^{\mathfrak W\mathfrak W}  
                                   \end{pmatrix} \in\mathbb R^{3N_t\times3N_t},\\
            \tilde{\mathbf m} &= \frac{N^{\mathcal N}_{t+1}}{\left(4N_{t+1}\right)^2}
                       \begin{pmatrix}
                           12\left(\mathbf K_t^{\mathcal R\mathcal Q}+\mathbf K_t^{\mathfrak W\mathfrak W}\right)\mathbf s_t+8\cdot\mathrm{diag}\left(\mathbf K_t^{\mathcal R\mathcal R}\right)
                                       - 9\cdot\mathrm{diag}\left(\mathbf K_t^{\mathfrak W\mathfrak W}\right)\\
                           4\left(\mathbf K_t^{\mathcal Q\mathcal Q}+\mathbf K_t^{\mathcal Q\mathcal R}\right)\mathbf s_t - \mathrm{diag}\left(\mathbf K_t^{\mathcal Q\mathcal Q}\right)\\
                           4\left(\mathbf K_t^{\mathcal R\mathcal Q}+\mathbf K_t^{\mathfrak W\mathfrak W}\right)\mathbf s_t - \mathrm{diag}\left(\mathbf K_t^{\mathfrak W\mathfrak W}\right)
                       \end{pmatrix}\\&\quad\in\mathbb R^{3N_t}\\
            \tilde k^{\ast} &= k_{t+1}^{\ast}
                                 - \frac1{4N_{t+1}^2}
                                     \mathbf s_t^{\top}\left(\mathbf K_t^{\mathcal Q\mathcal Q}+2\mathbf K_t^{\mathcal Q\mathcal R}+\mathbf K_t^{\mathcal R\mathcal R}\right)\mathbf s_t 
                                 - \frac{N^{\mathcal N}_{t+1}}{8N_{t+1}^2}\in\mathbb R.
        \end{align*}

        \item As noted in Remark~\ref{rmk::convert}, in case one relies on only one mode of mating control, one can simply increase $n$ by one and add the variables 
            \[\tilde{\mathbf e}_3=\begin{pmatrix}\mathbf 0_t\\\mathbf 1_t\\\mathbf 0_t\end{pmatrix}\quad\text{and}\quad\tilde d_3= 0\]
        or
            \[\tilde{\mathbf e}_3=\begin{pmatrix}\mathbf 0_t\\\mathbf 0_t\\\mathbf 1_t\end{pmatrix}\quad\text{and}\quad\tilde d_3= 0\]
        depending on which mating control strategy is followed. However, numerically more feasible appears to directly translate Tasks~\ref{task::sci} and~\ref{task::ms} into the form of Task\ref{task::general}.
        
        \item In case there are non-eligible queens for the different paths, one may once more increase $n$ and let
            \[\tilde{\mathbf e}_4=\begin{pmatrix}\mathbf {ne}_t^d\\\mathbf {ne}_t^b\\\mathbf {ne}_t^a\end{pmatrix}\quad\text{and}\quad\tilde d_4= 0,\]
        or, numerically smarter, one may let $\mathbf {dc}_t$, $\mathbf{bc}_t$, and $\mathbf{ac}_t$ only have entries for the respective eligible queens and restrict the other vectors and matrices accordingly.
        
        \item Note that in Task~\ref{task::general}, we did not demand the matrix $\tilde{\mathbf K}$ to be positive definite. Indeed, by our choice of $\tilde{\mathbf K}$, the matrix is only semi-definite. This follows from the fact that the matrix $\begin{pmatrix}9&3\\3&1\end{pmatrix}\in\mathbb R^{2\times2}$ is positive semi-definite and the matrix $\begin{pmatrix}
        9\mathbf K_{t}^{\mathfrak W\mathfrak W} & 3\mathbf K_{t}^{\mathfrak W\mathfrak W}\\3\mathbf K_{t}^{\mathfrak W\mathfrak W} & \mathbf K_{t}^{\mathfrak W\mathfrak W}\end{pmatrix}=\begin{pmatrix}9&3\\3&1\end{pmatrix}\otimes \mathbf K_{t}^{\mathfrak W\mathfrak W}$ is a sub-matrix of matrix $\tilde{\mathbf K}$.
        
        \item Further note that while Task~\ref{task::general} allows an arbitrary number $n$ of linear constraints, there is only one quadratic constraint, namely $\tilde{\mathbf c}^{\top}\tilde{\mathbf K}\tilde{\mathbf c}+\tilde{\mathbf m}^{\top}\tilde{\mathbf c}\leq\tilde k^{\ast}$. Thus, the alternative discussed in Remark~\ref{rmk::alternatives}\,\ref{item::altii} to separately restrict kinships among queens and replacement queens does not directly fall into the scope of the general formulation of Task~\ref{task::general}. 
    \end{enumerate}
\end{Rmk}

\subsection{Implementation} \label{sec::implement}

Several variations of OCS for other farm animals are bundled in the R package 'optiSel' \citep{wellmann19optimum}. But despite the remarkable flexibility of this package, it is not suitable to cover OCS for honeybees as it was derived here. The underlying package behind 'optiSel' is the package 'optiSolve' \citep{wellmann21}. This package allows in general to solve tasks in the form of Task~\ref{task::general}. We wrote the attached R script \texttt{honeybee\_ocs.r} using the package 'optiSolve' to implement an OCS for honeybees. 
The script can be run via the command 
\begin{lstlisting}
    Rscript --vanilla honeybee_ocs.r <arguments>
\end{lstlisting}
where \texttt{<arguments>} specifies the necessary arguments passed to the script. In general, \texttt{<arguments>} consists of up to nine components, named
\begin{itemize}
    \item \texttt{-{}-N\_N}, 
    \item \texttt{-{}-delta\_k},
    \item \texttt{-{}-curr\_gen},
    \item \texttt{-{}-K\_QQ},
    \item \texttt{-{}-K\_RR}, 
    \item \texttt{-{}-K\_QR},
    \item \texttt{-{}-K\_WW},
    \item \texttt{-{}-output\_numbers}, and 
    \item \texttt{-{}-output\_stats}.
\end{itemize}
A possible valid call of the script could thus look as follows: 
\begin{lstlisting}
    Rscript --vanilla honeybee_ocs.r            \
        --N_N 400                               \
        --delta_k 0.5                           \
        --curr_gen ./current_generation.tsv     \
        --K_QQ ./K_QQ.tsv                       \
        --K_RR ./K_RR.tsv                       \
        --K_QR ./K_QR.tsv                       \
        --K_WW ./K_WW.tsv                       \
        --output_numbers ./output_numbers.tsv   \
        --output_stats ./output_stats.tsv
\end{lstlisting}
Below, we will explain these nine arguments in detail.

\subsubsection{Input to \texttt{honeybee\_ocs.r}}

\paragraph{\texttt{-{}-N\_N}}
After \texttt{-{}-N\_N}, the desired value for $N^{\mathcal N}_{t+1}$ is specified, i.\,e. the number of queens that are to be newly produced for generation $\mathcal P_{t+1}$. The value has to be a positive integer. There is no default value, the script will produce an error message if this value is not provided.

\paragraph{\texttt{-{}-delta\_k}}
The value provided after \texttt{-{}-delta\_k} is used to determine the maximum allowed average kinship $k_{t+1}^{\ast}$ in the next reduced generation $\mathcal P_{t+1}^{\ast}$. However, it is not the value $k_{t+1}^{\ast}$ that is to be provided here, but percentage by which the panmictic index $\left(1-k_{\mathcal P_t^{\ast},\mathcal P_t^{\ast}}\right)$ may be reduced. Thus, if a value $\Delta k_{\mathcal P_t^{\ast},\mathcal P_t^{\ast}}$ is provided for \texttt{-{}-delta\_k}, the maximum allowed average kinship for the reduced population $\mathcal P_{t+1}^{\ast}$ is set to
\[k_{t+1}^{\ast}:=k_{\mathcal P_t^{\ast},\mathcal P_t^{\ast}} + \frac {\Delta k_{\mathcal P_t^{\ast},\mathcal P_t^{\ast}}}{100}\cdot\left(1-k_{\mathcal P_t^{\ast},\mathcal P_t^{\ast}}\right).\]
If no \texttt{-{}-delta\_k} is provided, the default value of $\Delta k_{\mathcal P_t^{\ast},\mathcal P_t^{\ast}}=1.0$ is used.

\begin{Rmk}\label{rmk::fao}
    The default value of 1.0 for \texttt{-{}-delta\_k} is motivated by the recommendation that the inbreeding rate should not exceed 1\% per generation \citep{fao13}. As noted in Remark~\ref{rmk::inbkin}\,\ref{item::inbkin}, the increases in inbreeding and average kinship typically show parallel behavior. Note, however, that the FAO recommendation considers discrete generations, so the change from $\mathcal P_t$ to $\mathcal P_{t+1}$ does not mean a \emph{generation} in the sense of the FAO if there are survivors. In that case, the desired inbreeding rate should still be divided by the average generation interval $L$ \citep{wellmann19optimum}. In classical honeybee breeding, where mating is organized on isolated mating stations, it is often assumed that only two-year-old queens are eligible for the dam path and only three-year-old queens are eligible for the 4a-path. This yields an average generation interval of $L=2.5$ years \citep{plate19comparison, uzunov22initiation, brascamp24}. In this situation, the value of $\Delta k_{\mathcal P_t^{\ast},\mathcal P_t^{\ast}}=\frac1{2.5}=0.4$ should be chosen as the value for \texttt{-{}-delta\_k} in order to comply with the FAO recommendation.
\end{Rmk}

\paragraph{\texttt{-{}-curr\_gen}}
After \texttt{-{}-curr\_gen}, a string is to be provided that contains the path to a text file containing information on the current generation $\mathcal P_t$. Not providing such a file will lead to an error. The file itself needs to be structured as follows: It consists of seven tab-separated columns. The first line contains the column headers which are 
\begin{itemize}
    \item \texttt{queen},
    \item \texttt{survival},
    \item \texttt{dam\_candidate}, 
    \item \texttt{one\_b\_candidate},
    \item \texttt{four\_a\_candidate},
    \item \texttt{u\_Q}, and
    \item \texttt{u\_R}.
\end{itemize}
Underneath the respective header, each column contains information about the queens in $\mathcal Q_t$, where each row corresponds to one queen $Q\in\mathcal Q_t$.
\begin{enumerate}[label = (\roman*)]
    \item The column \texttt{queen} needs to contain unique IDs (names) for all queens $Q\in\mathcal Q_t$.
    \item The column \texttt{survival} contains the values $s_{Q,t}$, i.\,e. the information whether $Q\in\mathcal Q_t$ survives to be an element of $\mathcal S\mathcal Q_{t+1}$ (cf. Notation~\ref{not::survi}). The values in this column can be taken either from $\{0,1\}$ or from $\{\mathtt{FALSE},\mathtt{TRUE}\}$.
    \item The columns \texttt{dam\_cand}, \texttt{one\_b\_cand}, and \texttt{four\_a\_cand} contain the values $e_{Q,t}^d$, $e_{Q,t}^b$, and $e_{Q,t}^a$, respectively, i.\,e. the information whether $Q\in\mathcal Q_t$ is eligible as dam, 1b-queen or 4a-queen (cf. Notation~\ref{not::elig}\,\ref{item::hbelig}). The values in these columns can be taken either from $\{0,1\}$ or from $\{\mathtt{FALSE},\mathtt{TRUE}\}$.
    \item The columns \texttt{u\_Q} and \texttt{u\_R} contain the estimated total breeding values $\hat{u}_{Q,t}$ of $Q\in\mathcal Q_t$ and $\hat{u}_{R(Q),t}$ of $Q$'s replacement queen $R(Q)\in\mathcal R_{t}$.
\end{enumerate}

\begin{Ex}
    The first lines of the file provided via \texttt{-{}-curr\_gen} could thus look as follows:
    \begin{lstlisting}
queen survival    dam_cand    four_a_cand one_b_cand  u_Q   u_R
Q_1   TRUE  TRUE  FALSE TRUE  1.0343      1.2297
Q_2   TRUE  FALSE FALSE FALSE 1.5210      2.0226
Q_3   FALSE FALSE TRUE  TRUE  2.5441      2.8841
Q_4   FALSE TRUE  FALSE FALSE 1.7779      1.4900      
    \end{lstlisting} 
\end{Ex}

\begin{Rmk}
    \begin{enumerate}[label = (\roman*)]
         \item The order in which the seven columns are provided is irrelevant. Listing additional columns is not harmful, they will simply be ignored. 
         
         \item The program will calculate optimum contributions according to the mixed strategy expounded in Section~\ref{sec::mix} with mating control via instrumental insemination and mating stations. If one wants calculations to be done according to Section~\ref{sec::sci} (only insemination), one simply has to ensure that all entries in column \texttt{four\_a\_cand} are \texttt{FALSE}. Accordingly, if calculations should be performed according to Section~\ref{sec::ims} (only mating stations), all entries in column \texttt{one\_b\_cand} have to be \texttt{FALSE}.
    \end{enumerate}
\end{Rmk}

\paragraph{\texttt{-{}-K\_QQ}, \texttt{-{}-K\_RR}, \texttt{-{}-K\_QR}, and \texttt{-{}-K\_WW}}
After these arguments, the paths to files need to be provided, which contain information on $\mathbf K_t^{\mathcal Q\mathcal Q}$, $\mathbf K_t^{\mathcal R\mathcal R}$, $\mathbf K_t^{\mathcal Q\mathcal R}$, and $\mathbf K_t^{\mathfrak W\mathfrak W}$, respectively. Not providing these files will result in an error. All four files have the same structure: They consist of a header line with the tab-separated IDs (names) of the queens $Q\in\mathcal Q_t$ as they are listed in column \texttt{queen} of the file provided under \texttt{-{}-curr\_gen}. This header line is followed by $N_t+1$ tab-separated columns. The first column also contains the IDs of the queens $Q\in\mathcal Q_t$ and serves as 'row names'. The following $N_t$ columns are associated with the headers in the first row. For two queens $Q_1, Q_2\in\mathcal Q_t$, the entry belonging to the row associated with $Q_1$ and the column associated with $Q_2$ in the file provided under
\begin{itemize}
     \item \texttt{-{}-K\_QQ} is $k_{Q_1,Q_2}$,
     \item \texttt{-{}-K\_RR} is $k_{R(Q_1),R(Q_2)}$,
     \item \texttt{-{}-K\_QR} is $k_{Q_1,R(Q_2)}$,
     \item \texttt{-{}-K\_WW} is $k_{\mathcal W(Q_1),\mathcal W(Q_2)}$.
\end{itemize}

\begin{Ex}
    The first lines and columns of the file provided under \texttt{-{}-K\_QQ} may thus look as follows
    \begin{lstlisting}
Q_1     Q_2     Q_3     Q_4
Q_1     0.5000  0.0000  0.0012  0.1944
Q_2     0.0000  0.5000  0.0000  0.0000
Q_3     0.0012  0.0000  0.5040  0.0012
Q_4     0.1944  0.0000  0.0012  0.5000      
    \end{lstlisting}
and similar for the other kinship matrices.
\end{Ex}

\begin{Rmk}
    \begin{enumerate}[label = (\roman*)]
        \item Do not enter a tab before the name of the first queen in the first row. The format of the table has to be such that the default behavior of R function \texttt{read.table} recognizes the entries of the first row as headers: 
        \begin{quote}
            `\textbf{header}` is set to `\textbf{TRUE}` if and only if the first row contains one fewer field than the number of columns \citep{readtable}.
        \end{quote}
    
        \item It is not mandatory, that lines and columns list the queens of $\mathcal Q_t$ in the same order.
        
        \item Be aware that the programs of \citet{bernstein18} and \citet{brascamp19software}, which can be used to derive the necessary matrices from honeybee pedigrees, are designed to calculate relationships rather than kinships. Thus, the results provided by these programs need to be divided by two (cf. Remark~\ref{rmk::kinsh}\,\ref{item::relation}).
    \end{enumerate}
\end{Rmk}

\paragraph{\texttt{-{}-output\_numbers} and \texttt{-{}-output\_stats}}
Here, paths to files can be provided to which the output is written. The file specified after \texttt{-{}-output\_numbers} will contain for each queen $Q\in\mathcal Q_t$, how often she should serve as a dam, 1b-queen and 4a-queen respectively. The file specified after \texttt{-{}-output\_stats} will contain statistics resulting from these contributions. These include expected average breeding values and kinships for the next generation. The detailed structure of the output files will be discussed in Section~\ref{sec::output}. If \texttt{-{}-output\_numbers} and \texttt{-{}-output\_stats} are not specified, the information will be written to the default files, \texttt{optimum\_contributions.tsv} and \texttt{stats.tsv}.

\subsubsection{Implementation details}

The input data is used to build a constrained optimization problem with the help of the function \texttt{cop} from the R package 'optiSolve' \citep{wellmann21}. The linear function to be maximized and the constraints passed to this function are calculated as specified in Remark~\ref{rmk::concretetask}. Afterwards, the function \texttt{solvecop} is called to solve the constrained optimization problem. If a solution is found, one then is equipped with vectors $\mathbf{dc}_t$, $\mathbf{bc}_t$, and $\mathbf{ac}_t$ of optimum contributions.

\begin{Rmk}
    The program may fail to find an optimum solution if $k_{t+1}^{\ast}$ is chosen too small.
\end{Rmk}

The vectors $\mathbf{dc}_t$, $\mathbf{bc}_t$, and $\mathbf{ac}_t\in\mathbb R^{\mathcal Q_t}$ specify the relative contributions of the queens in $\mathcal Q_t$ to the next generation. In order to obtain their absolute contributions, these vectors have to be multiplied with the total number $N^{\mathcal N}_{t+1}$ of newly created queens in generation $\mathcal P_{t+1}$. However, there is no guarantee that the resulting absolute contribution numbers are integers, so that the results need to be rounded in a way that the total numbers of contributions via (1.) the dam path and (2.) via the 1b-path and 4a-path combined remain precisely $N^{\mathcal N}_{t+1}$.

\begin{Not}
    For a number $r\in\mathbb R$, we denote by $\lfloor r \rfloor$ the largest integer that is not larger than $r$. We denote the remainder by $\langle r\rangle:=r-\lfloor r\rfloor$.
\end{Not}

\begin{Rmk}\label{rmk::round}
    There are several ways to achieve suitably rounded versions of $N_{t+1}^{\mathcal N}\mathbf{dc}_t$, $N_{t+1}^{\mathcal N}\mathbf{bc}_t$, and $N_{t+1}^{\mathcal N}\mathbf{ac}_t$.
    \begin{enumerate}[label = (\roman*)]
        \item The function \texttt{noffspring} of the package 'optiSel' \citep{wellmann19optimum} has two different options to calculate the absolute number of offspring of an individual $I$ from its relative contribution $c_{I,t}$. The different options are determined by whether the function parameter \texttt{random} is set to \texttt{TRUE} (default) or \texttt{FALSE}. If we translated this function to the honeybee setting, both options would first allow each queen $Q\in\mathcal Q_t$ to serve $\lfloor N_{t+1}^{\mathcal N}dc_{Q,t}\rfloor$ times as a dam, $\lfloor N_{t+1}^{\mathcal N}bc_{Q,t}\rfloor$ times as a 1b-queen, and $\lfloor N_{t+1}^{\mathcal N}ac_{Q,t}\rfloor$ times as a 4a-queen. Because all values have been rounded downwards, this will lead to total numbers of offspring $\leq N_{t+1}^{\mathcal N}$. The parameter \texttt{random} determines, how the remaining $\sum\limits_{Q\in\mathcal Q_t}\langle N_{t+1}^{\mathcal N}dc_{Q,t}\rangle$ usages as dams are to be distributed (and likewise for the usages as 1b-queens or 4a-queens).
        \begin{enumerate}
            \item With the default option \texttt{random=TRUE}, the remaining offspring are distributed to the queens randomly. Hereby, it is guaranteed that each queen may receive not more than one additional offspring and the probabilities for queens to be assigned an extra offspring are weighted by the remainder values $\langle N_{t+1}^{\mathcal N}dc_{Q,t}\rangle$. In our view, this procedure has two disadvantages. First, in view of transparency, repeatability and clearness of selection decisions, it appears unfavorable to include a random element in the selection process. Secondly, by its nature, the solver of the constrained optimization problem calculates a (typically very good) approximate solution. This means, however, that also a bad queen $Q\in\mathcal Q_t$ that clearly should have a contribution of $dc_{Q,t}=0$ may actually be assigned an 'optimum' contribution of, say, $\sim10^{-7}$. With the randomized approach, it is possible (albeit unlikely) that such a queen is suggested to produce one offspring.
            
            \item With the option \texttt{random=FALSE}, the queens are ranked by their remainders $\langle N_{t+1}^{\mathcal N}dc_{Q,t}\rangle$ and the queens with the highest remainder numbers are assigned an additional offspring each until the correct number of offspring is reached.
        \end{enumerate}
        \item In general, the problem of assigning (integer numbers of) offspring to different queens according to their (non-integer) fractions of optimum contributions is similar to the problem of assigning (integer numbers of) seats in a parliament to different parties according to their (non-integer) fractions of relative votes. For these apportionment problems, a number of competing procedures exist. While there is some theory on random apportionment \citep{grimmett04}, these procedures are generally deterministic. 
        \begin{enumerate}
            \item Several apportionment methods, like notably the Jefferson method (sometimes also named D'Hondt method), are biased in favor of greater parties \citep{balinski82}. For our purposes, this appears disadvantageous. The main purpose of OCS is to restrict average kinships. If a method tends to make large sister groups even larger, this will have detrimental effects.
            \item In the context of apportionment theory, the procedure implemented in the 'OptiSel' function \texttt{noffspring} with \texttt{random=FALSE} is called Hamilton's method \citep{balinski82} (in Germany, it is named after Hare/Niemeyer instead \citep{agricola17}). It is generally unbiased but can lead to a number of paradoxes. In particular, when queens are obviously unsuitable for reproduction because of low estimated breeding values, one may preclude them from the list of eligible queens and thereby speed up the optimization algorithm. These queens will then receive optimum contributions of 0. If one leaves them in the list of eligible queens, the algorithm may attribute very small 'optimum' contributions to these queens (say $\sim10^{-8}$). By Hamilton's method, these queens will still be assigned 0 offspring but the numbers of offspring of the other queens may depend on the choice whether or not the hopeless candidates have been included in the procedure.
            \item The method of Webster (also named after Sainte-Laguë) is also unbiased and avoids the aforementioned paradox \citep{balinski82}. Furthermore, there is a modification of this method (typically called \emph{modified Sainte-Laguë method}), that makes it harder for parties to win their first seat in parliament \citep{lijphart03}. While this property is usually sought in order to avoid fragmented parliaments, it is also useful in the context of honeybee breeding. It makes no sense to go through the effort of maintaining a mating station if that mating station is then supposed to be used by only a single queen.
        \end{enumerate}
    \end{enumerate}
    Based on these considerations, we decided to distribute the numbers of usages as dam, 1b-queen, or 4a-queen according to the modified Sainte-Laguë method. For this, we used the function \texttt{seats} of the R package 'electoral' \citep{albuja22}.
\end{Rmk}

\begin{Rmk}\label{rmk::roundcons}
    Because of the rounding procedures described above, the real contributions corresponding to the calculated numbers of usages as dam, 1b-queen or 4a-queen will differ slightly from the calculated optimum contributions. In small populations, this may result in slight violations of the restriction on the average kinships in the next generation, $k_{\mathcal P_{t+1}^{\ast},\mathcal P_{t+1}^{\ast}}$.
\end{Rmk}

\subsubsection{Output}\label{sec::output}
The script \texttt{honeybee\_ocs.r} creates two output files, which by default are named \texttt{optimum\_contributions.tsv} and \texttt{stats.tsv}.
\paragraph{\texttt{optimum\_contributions.tsv}} This file contains seven tab-separated columns with headers \texttt{queen}, \texttt{dc\_opt}, \texttt{n\_dam}, \texttt{bc\_opt}, \texttt{n\_1b}, \texttt{ac\_opt}, and \texttt{n\_4a}.
\begin{itemize}
    \item Column \texttt{queen} lists all queens of $\mathcal Q_t$.
    \item Columns \texttt{dc\_opt}, \texttt{bc\_opt}, and \texttt{ac\_opt} list the optimum (relative) contributions of the queens via the dam path, 1b-path, and 4a-path, respectively. Values are rounded to 5 decimal digits.
    \item Columns \texttt{n\_dam}, \texttt{n\_1b}, and \texttt{n\_4a} translate the optimum relative contributions into numbers of utilizations as dams, 1b-queens, and 4a-queens, respectively.
\end{itemize}

\paragraph{\texttt{stats.tsv}}\label{par::stats} This file consists of two rows, the first row containing headers, the second row containing the corresponding values. In total, there are 31 tab-separated two-elemented columns:
\begin{enumerate}[label = (\roman*)]
    \item Columns \texttt{u\_Q\_curr}, \texttt{u\_Q\_surv}, \texttt{u\_Q\_new}, and \texttt{u\_Q\_next} contain the values $\hat{u}_{\mathcal Q_t}$, $\mathbb E\left[\hat{u}_{\mathcal S\mathcal Q_{t+1}}\right]$, $\mathbb E\left[\hat{u}_{\mathcal N\mathcal Q_{t+1}}\right]$, and , $\mathbb E\left[\hat{u}_{\mathcal Q_{t+1}}\right]$, calculated according to Equations~\ref{eq::uqcurr}, \ref{eq::usq} ($=$\ref{eq::usqms}), \ref{eq::unq} ($=$~\ref{eq::unqms}), and \ref{eq::utq} ($=$~\ref{eq::utqms} $=$~\ref{eq::utqmix}), respectively. In the calculations, the vectors $\mathbf{dc}_t$, $\mathbf{bc}_t$, and $\mathbf{ac}_t$ are not chosen as the direct output of function \texttt{solvecop} but are adjusted according to the rounding procedure explained in Remark~\ref{rmk::round}. This also holds for all other columns.

    \item Columns \texttt{u\_R\_curr}, \texttt{u\_R\_surv}, \texttt{u\_R\_new}, and \texttt{u\_R\_next} contain the values $\hat{u}_{\mathcal R_t}$, $\mathbb E\left[\hat{u}_{\mathcal S\mathcal R_{t+1}}\right]$, $\mathbb E\left[\hat{u}_{\mathcal N\mathcal R_{t+1}}\right]$, and , $\mathbb E\left[\hat{u}_{\mathcal R_{t+1}}\right]$, calculated according to Equations~\ref{eq::urcurr}, \ref{eq::usr} ($=$~\ref{eq::usrms}), \ref{eq::unr} (or~\ref{eq::unrms}), and \ref{eq::utr} (or~\ref{eq::utrms} or~\ref{eq::utrmix}), respectively.

    \item Columns \texttt{u\_Pstar\_curr}, \texttt{u\_Pstar\_surv}, \texttt{u\_Pstar\_new}, and \texttt{u\_Pstar\_next}\\contain the values $\hat{u}_{\mathcal P^{\ast}_t}$, $\mathbb E\left[\hat{u}_{\mathcal S_{t+1}}\right]$, $\mathbb E\left[\hat{u}_{\mathcal N_{t+1}}\right]$, and , $\mathbb E\bigl[\hat{u}_{\mathcal P^{\ast}_{t+1}}\bigr]$.

    \item Columns \texttt{k\_QQ\_curr}, \texttt{k\_QQ\_surv}, \texttt{k\_QQ\_new}, and \texttt{k\_QQ\_next},\\\texttt{k\_RR\_curr}, \texttt{k\_RR\_surv}, \texttt{k\_RR\_new}, and \texttt{k\_RR\_next},\\ \texttt{k\_QR\_curr}, \texttt{k\_QR\_surv}, \texttt{k\_QR\_new}, and \texttt{k\_QR\_next},\\\texttt{k\_Pstar\_curr}, \texttt{k\_Pstar\_surv}, \texttt{k\_Pstar\_new}, and \texttt{k\_Pstar\_next}\\ contain the values of\\ 
    $k_{\mathcal Q_t,\mathcal Q_t}$, $k_{\mathcal S\mathcal Q_{t+1},\mathcal S\mathcal Q_{t+1}}$, $k_{\mathcal N\mathcal Q_{t+1},\mathcal N\mathcal Q_{t+1}}$, and $k_{\mathcal Q_{t+1},\mathcal Q_{t+1}}$,\\ 
    $k_{\mathcal R_t,\mathcal R_t}$, $k_{\mathcal S\mathcal R_{t+1},\mathcal S\mathcal R_{t+1}}$, $k_{\mathcal N\mathcal R_{t+1},\mathcal N\mathcal R_{t+1}}$, and $k_{\mathcal R_{t+1},\mathcal R_{t+1}}$,\\
    $k_{\mathcal Q_t,\mathcal R_t}$, $k_{\mathcal S\mathcal Q_{t+1},\mathcal S\mathcal R_{t+1}}$, $k_{\mathcal N\mathcal Q_{t+1},\mathcal N\mathcal R_{t+1}}$, and $k_{\mathcal Q_{t+1},\mathcal R_{t+1}}$,\\
    $k_{\mathcal P^{\ast}_t,\mathcal P^{\ast}_t}$, $k_{\mathcal S_{t+1},\mathcal S_{t+1}}$, $k_{\mathcal N_{t+1},\mathcal N_{t+1}}$, and $k_{\mathcal P^{\ast}_{t+1},\mathcal P^{\ast}_{t+1}}$,\\ respectively, according to the formulas in Lemma~\ref{lem::allthek} and Theorem~\ref{thm::burnermix}.

    \item Columns \texttt{n\_dam}, \texttt{n\_1b}, and \texttt{n\_4a} specify, how many queens have a non-zero contribution via the dam path, 1b-path and 4a-path respectively.
\end{enumerate}
All values are rounded to five decimal digits.

\section{Demonstration}\label{sec::demo}

We demonstrate two examples to illustrate how OCS in honeybees works. The first example is small: The subsequent generations $\mathcal P_t$ and $\mathcal P_{t+1}$ only comprise three colonies each. It is used, so that several of the underlying calculations can actually be reproduced with pen and paper. In the larger example, $\mathcal P_t$ and $\mathcal P_{t+1}$ comprise 1500 colonies each. OCS is performed for several such populations and compared to other selection strategies in terms of expected breeding value and average kinship development.

\subsection{Small example}

The input files for this example can be found in the folder \texttt{ocs\_small\_example}.
We consider a generation $\mathcal P_t$ consisting of three colonies whose queens, named $A$, $B$, and $C$, are all non-inbred. Queens $A$ and $B$ are siblings, their common dam was mated on an isolated mating station. Without detailed knowledge about the deeper pedigree and the composition of the mating station, it is typically assumed that the kinship between $A$ and $B$ in such a situation is approximately 0.2 (relationship 0.4) \citep{guichard20, bernstein23}. Queen $C$ is unrelated to both~$A$ and~$B$.

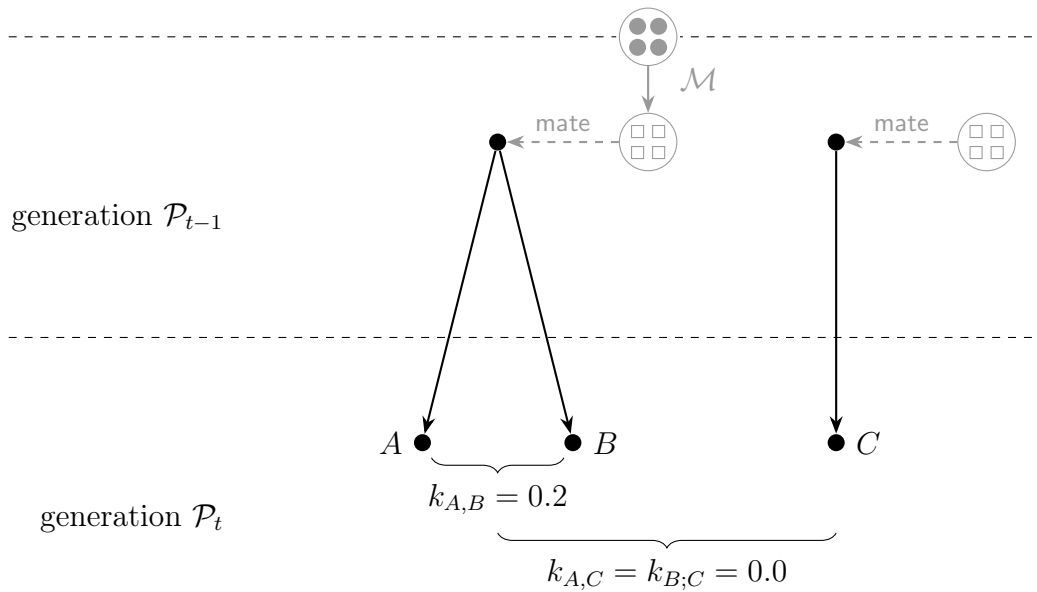
\begin{center}
    \begin{tikzpicture}
        \path (0,0) coordinate (M1)
            coordinate[below = 2.4cm of M1] (Ptm1)
                node[left = 0mm of Ptm1] {generation $\mathcal P_{t-1}$}
            coordinate[below = 4cm of Ptm1] (Pt0)
                node[left = 0mm of Pt0] {generation $\mathcal P_{t}$}
            (Ptm1) -- (Pt0)
                coordinate[pos = 0.4] (M2)

                coordinate[right = 2.5cm of Pt0] (A)
                node[above = 1cm of A, queen] (QA) {}
                node[left = 0cm of QA] {$A$}

                coordinate[right = 4.5cm of Pt0] (B)
                node[above = 1cm of B, queen] (QB) {}
                node[right = 0cm of QB] {$B$}
                
                coordinate[right = 8cm of Pt0] (C)
                node[above = 1cm of C, queen] (QC) {}
                node[right = 0cm of QC] {$C$}

                coordinate[right = 3.5cm of Ptm1] (D)
                node[above = 1cm of D, queen] (QD) {}
                    node[above right = 1.4cm and 2cm of QD.center, group base] (MDbase) {}
                    node[above right = 1.4cm and 2cm of QD.center, group, not considered] (MD) {}
                        node[below right = 0cm and 0cm of MD, not considered] {$\mathcal M$}
                        -- (MD.center) pic[not considered]{queens}
                    node[right = 2cm of QD.center, not considered, group] (DD) {}
                        node[right = 0cm of DD, not considered] {}
                        -- (DD) pic[not considered] {drones}

                coordinate[right = 8cm of Ptm1] (D2)
                node[above = 1cm of D2, queen] (QD2) {}
                    node[right = 2cm of QD2.center, not considered, group] (DD2) {}
                        node[right = 0cm of DD2, not considered] {}
                        -- (DD2) pic[not considered] {drones};

                    \draw[inheritance] (QD) -- (QA);
                    \draw[inheritance] (QD) -- (QB);
                    \draw[inheritance] (QD2) -- (QC);
                    \draw[inheritance, not considered] (MD) -- (DD);
                    \draw[mating, not considered] (DD) -- (QD)
                        node[gene pass description, not considered]{mate};
                    \draw[mating, not considered] (DD2) -- (QD2)
                        node[gene pass description, not considered]{mate};
                        
                    \draw[relationship] (QA) -- (QB)
                        node[relationship description] {$k_{A,B}=0.2$};  
                    \path (QA) -- (QB)
                        coordinate[midway] (QAB);
                    \draw[relationship] (QAB) ++(0,-1) -- ++(4.5,0)
                        node[relationship description] {$k_{A,C}=k_{B;C}=0.0$};                        
                \begin{scope}[on background layer]
                    \draw[dashed] (M1) ++ (-3cm,0cm) -- ++(13.9cm,0cm);
                    \draw[dashed] (M2) ++ (-3cm,0cm) -- ++(13.9cm,0cm);
                \end{scope}
            \end{tikzpicture}
        \end{center}

Accordingly, the file \texttt{ocs\_small\_example/K\_QQ.tsv}, containing $\mathbf K_{t}^{\mathcal Q\mathcal Q}$, looks as follows:
    \begin{lstlisting}
Queen_A Queen_B Queen_C
Queen_A 0.5     0.2     0.0
Queen_B 0.2     0.5     0.0
Queen_C 0.0     0.0     0.5      
    \end{lstlisting}

Furthermore, queens $A$ and $B$ were instrumentally inseminated with many drones from the same colony, while queen~$C$ was inseminated with drones from an entirely unrelated colony. 

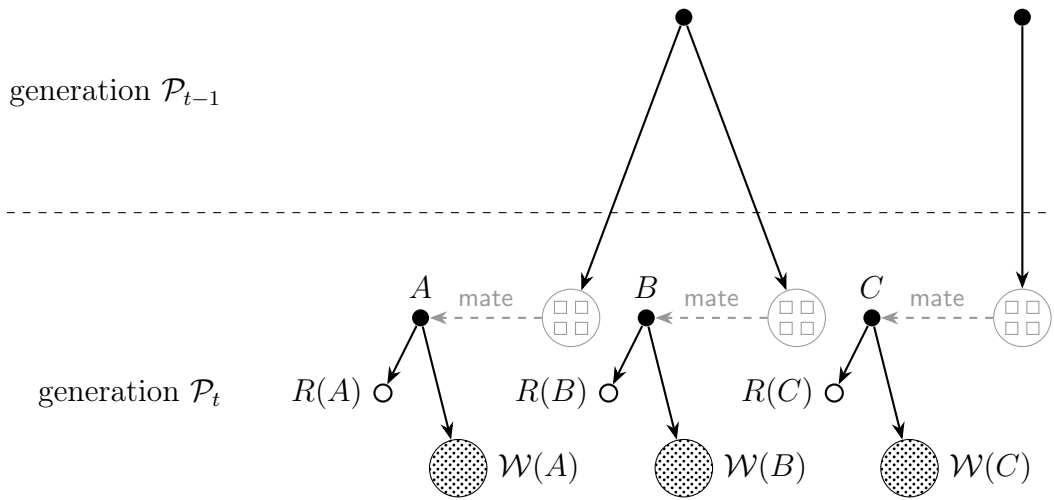
\begin{center}
    \begin{tikzpicture}
        \path (0,0) (Ptm1)
                node[left = 0mm of Ptm1] {generation $\mathcal P_{t-1}$}
            coordinate[below = 4cm of Ptm1] (Pt0)
                node[left = 0mm of Pt0] {generation $\mathcal P_{t}$}
            (Ptm1) -- (Pt0)
                coordinate[pos = 0.4] (M)

                coordinate[right = 2.5cm of Pt0] (A)
                node[above = 1cm of A, queen] (QA) {}
                node[above = 0cm of QA] {$A$}
                node[right = 2cm of QA.center, not considered, group] (DA) {}
                    node[right = 0cm of DA, not considered] {}
                    -- (DA) pic[not considered] {drones}
                node[below right = 2cm and 0.5cm of QA.center, worker group] (WA) {}
                    node[right = 0mm of WA] {$\mathcal W(A)$}
                node[below left = 1cm and 0.5cm of QA.center, replacement queen] (RA) {}
                    node[left = 0cm of RA] {$R(A)$}

                coordinate[right = 5.5cm of Pt0] (B)
                node[above = 1cm of B, queen] (QB) {}
                node[above = 0cm of QB] {$B$}
                node[right = 2cm of QB.center, not considered, group] (DB) {}
                    node[right = 0cm of DB, not considered] {}
                    -- (DB) pic[not considered] {drones}
                node[below right = 2cm and 0.5cm of QB.center, worker group] (WB) {}
                    node[right = 0mm of WB] {$\mathcal W(B)$}
                node[below left = 1cm and 0.5cm of QB.center, replacement queen] (RB) {}
                    node[left = 0cm of RB] {$R(B)$}
                
                coordinate[right = 8.5cm of Pt0] (C)
                node[above = 1cm of C, queen] (QC) {}
                node[above = 0cm of QC] {$C$}
                node[right = 2cm of QC.center, not considered, group] (DC) {}
                    node[right = 0cm of DC, not considered] {}
                    -- (DC) pic[not considered] {drones}
                node[below right = 2cm and 0.5cm of QC.center, worker group] (WC) {}
                    node[right = 0mm of WC] {$\mathcal W(C)$}
                node[below left = 1cm and 0.5cm of QC.center, replacement queen] (RC) {}
                    node[left = 0cm of RC] {$R(C)$}

                coordinate[right = 6cm of Ptm1] (D)
                node[above = 1cm of D, queen] (QD) {}
                coordinate[right = 10.5cm of Ptm1] (E)
                node[above = 1cm of E, queen] (QE) {};

                    \draw[inheritance] (QD) -- (DA);
                    \draw[inheritance] (QD) -- (DB);
                    \draw[inheritance] (QE) -- (DC);
                    \draw[inheritance] (QA) -- (RA);
                    \draw[inheritance] (QA) -- (WA);
                    \draw[inheritance] (QB) -- (RB);
                    \draw[inheritance] (QB) -- (WB);
                    \draw[inheritance] (QC) -- (RC);
                    \draw[inheritance] (QC) -- (WC);
                    \draw[mating, not considered] (DA) -- (QA)
                        node[gene pass description, not considered]{mate};
                    \draw[mating, not considered] (DB) -- (QB)
                        node[gene pass description, not considered]{mate};
                    \draw[mating, not considered] (DC) -- (QC)
                        node[gene pass description, not considered]{mate};
                        
                \begin{scope}[on background layer]
                    \draw[dashed] (M) ++ (-3cm,0cm) -- ++(13.9cm,0cm);
                \end{scope}
            \end{tikzpicture}
        \end{center}

Suitable files containing $\mathbf K_t^{\mathcal R\mathcal R}$, $\mathbf K_t^{\mathcal Q\mathcal R}$, and $\mathbf K_t^{\mathfrak W\mathfrak W}$ are then

\begin{itemize}
    \item \texttt{ocs\_small\_example/K\_RR.tsv},
        \begin{lstlisting}
Queen_A Queen_B Queen_C
Queen_A 0.5     0.175   0.0
Queen_B 0.175   0.5     0.0
Queen_C 0.0     0.0     0.5      
        \end{lstlisting}
        
    \item \texttt{ocs\_small\_example/K\_QR.tsv},
        \begin{lstlisting}
Queen_A Queen_B Queen_C
Queen_A 0.25    0.1     0.0
Queen_B 0.1     0.25    0.0
Queen_C 0.0     0.0     0.25      
        \end{lstlisting}
        
    \item and \texttt{ocs\_small\_example/K\_WW.tsv}
        \begin{lstlisting}
Queen_A Queen_B Queen_C
Queen_A 0.26    0.175   0.0
Queen_B 0.175   0.26    0.0
Queen_C 0.0     0.0     0.26      
        \end{lstlisting}
\end{itemize}

Of the three colonies in generation $\mathcal P_t$, only the one headed by queen $A$ survives to the next generation, whereas queens $B$ and $C$ die. All three queens are eligible as dams, 1b-queens and 4a-queens. The estimated breeding values of the queens and replacement queens are
\begin{align*}
    \hat{u}_{A,t} = 8.34, &\quad \hat{u}_{R(A),t} = 7.62,\\
    \hat{u}_{B,t} = 6.98, &\quad \hat{u}_{R(B),t} = 7.55,\\
    \hat{u}_{C,t} = 4.60, &\quad \hat{u}_{R(C),t} = 3.03.
\end{align*}

Accordingly, the file \texttt{ocs\_small\_example/curr\_gen.tsv} with the information on the current generation looks as follows:

\begin{lstlisting}
queen   survival  dam_cand  four_a_cand one_b_cand  u_Q   u_R
Queen_A TRUE      TRUE      TRUE        TRUE        8.34  7.62
Queen_B TRUE      TRUE      TRUE        TRUE        6.98  7.55
Queen_C FALSE     TRUE      TRUE        TRUE        4.60  3.03
\end{lstlisting} 
        
The next generation $\mathcal P_{t+1}$ shall again consist of three queens. Since only queen $A$ survives, this means that two new colonies must be generated, $N_{t+1}^{\mathcal N}=2$. We try and call \texttt{honeybee\_ocs.r} with different values for \texttt{-{}-delta\_k}, i.\,e. different maximum acceptable average kinship levels $k_{t+1}^{\ast}$.

\begin{Rmk}
    The average estimated breeding value in the reduced generation $\mathcal P_{t}^{\ast}$ is
        \[\hat{u}_{\mathcal P_t^{\ast}}\approx6.3533,\]
    the average kinship is
        \[k_{\mathcal P_t^{\ast},\mathcal P_t^{\ast}}\approx0.1569.\]
\end{Rmk}

\begin{Ex}
    \begin{enumerate}[label = (\roman*)]
        \item We first try and pass the value $\Delta k_{\mathcal P_t^{\ast},\mathcal P_t^{\ast}}=100$ to \texttt{-{}-delta\_k}, which results in $k_{t+1}^{\ast}=1$. By its nature as a probability (cf. Definition~\ref{def::grpkinbee}), the value for $k_{\mathcal P_{t+1}^{\ast},\mathcal P_{t+1}^{\ast}}$ can never exceed the value 1, so that, effectively, there is no restriction on the average kinship of the next partial generation $\mathcal P_{t+1}^{\ast}$.
    
        In the case of a diploid monoecious population with selfing, we had seen in Remark~\ref{rmk::noconstr} that the best strategy is to let only the individual with the highest estimated breeding value reproduce via selfing. Translating this to our example, we expect that queen $A$, who has the highest estimated breeding values, should be the dam of both newly generated queens and also be responsible for the drones to fertilize the new queens. If $A$ provides the drones via the 1b-path, she will pass her own estimated breeding value $\hat{u}_{A,t}$, if she provides drones via the 4a-path, she will pass the estimated breeding value $\hat{u}_{R(A),t}$ of her replacement queen. Since $\hat{u}_{A,t}=8.34>7.62=\hat{u}_{R(A),t}$, she should be used as a 1b-queen. We check by calling the script:
    
        \begin{lstlisting}
Rscript --vanilla honeybee_ocs.r               \
    --N_N 2                                    \
    --delta_k 100                              \
    --curr_gen ocs_small_example/curr_gen.tsv  \
    --K_QQ ocs_small_example/K_QQ.tsv          \
    --K_RR ocs_small_example/K_RR.tsv          \
    --K_QR ocs_small_example/K_QR.tsv          \
    --K_WW ocs_small_example/K_WW.tsv
        \end{lstlisting}

        Since no output files are specified by this call, we find the relevant information in the default files. A look into \texttt{optimum\_contributions.tsv} reveals
        \begin{lstlisting}
queen   dc_opt  n_dam   bc_opt  n_1b    ac_opt  n_4a
Queen_A 0.99999 2       0.99999 2       0       0
Queen_B 0.00001 0       0       0       0       0
Queen_C 0       0       0       0       0       0
        \end{lstlisting}
        Up to an error of order $10^{-5}$, the optimum contributions were indeed calculated correctly, resulting in two daughters of queen $A$ which are to be inseminated with drones from $A$'s colony.
    
        A look into \texttt{stats.tsv} reveals that the expected average breeding value of the next generation is $\mathbb E\left[\hat{u}_{\mathcal P_{t+1}^{\ast}}\right] = 7.86$, which means an improvement of 1.5067 units compared to $\hat{u}_{\mathcal P_{t}^{\ast}}$. The average coancestry in the next generation is $k_{\mathcal P_{t+1}^{\ast}, \mathcal P_{t+1}^{\ast}}=0.3328$. Coming from $k_{\mathcal P_t^{\ast}, \mathcal P_t^{\ast}}=0.1569$, this means an increase of 20.86\%.

        \item We lower the allowed percentage of increase in average kinship to $\Delta k_{\mathcal P_t^{\ast},\mathcal P_t^{\ast}}=13\%$, i.\,e., we call
        
        \begin{lstlisting}
Rscript --vanilla honeybee_ocs.r               \
    --N_N 2                                    \
    --delta_k 13                               \
    --curr_gen ocs_small_example/curr_gen.tsv  \
    --K_QQ ocs_small_example/K_QQ.tsv          \
    --K_RR ocs_small_example/K_RR.tsv          \
    --K_QR ocs_small_example/K_QR.tsv          \
    --K_WW ocs_small_example/K_WW.tsv
        \end{lstlisting}
        
        Looking at \texttt{optimum\_contributions.tsv} shows
        \begin{lstlisting}
queen   dc_opt  n_dam   bc_opt  n_1b    ac_opt  n_4a
Queen_A 0       0       0.96847 2       0.00001 0
Queen_B 1       2       0       0       0.00926 0
Queen_C 0       0       0.02225 0       0       0
        \end{lstlisting}
        
        Queen $A$ is no longer used as a dam. This is not surprising. Letting $A$ serve as both dam and sire as in the previous example leads to an average kinship $k_{\mathcal P_{t+1}^{\ast}, \mathcal P_{t+1}^{\ast}}$ that is no longer acceptable -- particularly because $A$ is also the only survivor. Looking at the replacement queens' breeding values (which are passed via the dam path), we see that $B$ is only marginally worse than $A$ ($\hat{u}_{R(A),t} = 7.62$ vs. $\hat{u}_{R(B),t} = 7.55$). Thus, by letting $B$ rather than $A$ serve as dam, not much is lost in terms of genetic progress. On the other hand, by letting the new queens be nieces rather than daughters of the surviving queen $A$, the average kinship can be lowered considerably. Furthermore, since $A$ still produces all the drones via the 1b-path, new queens are no longer inseminated with sperm from their own brothers but rather from their cousins. 
        
        Looking into \texttt{stats.tsv}, we find that $\mathbb E\left[\hat{u}_{\mathcal P_{t+1}^{\ast}}\right]$ is now 7.825, only marginally lower than in the unrestricted case of the previous example. The average coancestry in the next generation is $k_{\mathcal P_{t+1}^{\ast}, \mathcal P_{t+1}^{\ast}}=0.2686$. The resulting increase of 13.25\% slightly exceeds the 13\% we had allowed for. This is a consequence of the fact that the relative optimum contributions cannot fully be represented by the integer numbers of offspring (cf. Remark~\ref{rmk::roundcons}). In situations with realistic population sizes these violations become negligible.
        
        \item We further lower the value after \texttt{-{}-delta\_k} to $\Delta k_{\mathcal P_t^{\ast},\mathcal P_t^{\ast}}=9\%$ and call
        
        \begin{lstlisting}
Rscript --vanilla honeybee_ocs.r               \
    --N_N 2                                    \
    --delta_k 9                                \
    --curr_gen ocs_small_example/curr_gen.tsv  \
    --K_QQ ocs_small_example/K_QQ.tsv          \
    --K_RR ocs_small_example/K_RR.tsv          \
    --K_QR ocs_small_example/K_QR.tsv          \
    --K_WW ocs_small_example/K_WW.tsv
        \end{lstlisting}
        
        Now, \texttt{optimum\_contributions.tsv} looks as follows:
        \begin{lstlisting}
queen   dc_opt  n_dam   bc_opt  n_1b    ac_opt  n_4a
Queen_A 0.00001 0       0.50542 1       0.00002 0
Queen_B 0.99999 2       0       0       0.00003 0
Queen_C 0       0       0.49453 1       0       0
        \end{lstlisting}
        
        Letting the siblings $A$ and $B$ be solely responsible for the next generation $\mathcal P_{t+1}$ is no longer acceptable. Instead, also $C$ needs to be included in the reproduction strategy -- despite her markedly lower estimated breeding values.
        
        According to \texttt{stats.tsv}, $\mathbb E\bigl[\hat{u}_{\mathcal P_{t+1}^{\ast}}\bigr]$ is lowered to 7.5133 and $k_{\mathcal P_{t+1}^{\ast}, \mathcal P_{t+1}^{\ast}}=0.2325$, meaning an increase of average kinship by 8.97\%.
        
        \item With a further reduced allowed increase in average kinship of $\Delta k_{\mathcal P_t^{\ast},\mathcal P_t^{\ast}}=2.4\%$, \texttt{optimum\_contributions.tsv} shows:
        \begin{lstlisting}
queen   dc_opt  n_dam   bc_opt  n_1b    ac_opt  n_4a
Queen_A 0       0       0.00001 0       0.00001 0
Queen_B 0.59266 1       0.07734 0       0.1349  0
Queen_C 0.40734 1       0.67668 2       0.11106 0
        \end{lstlisting}
        
        Now, the focus is clearly to avoid kinships as much as possible. Despite the low estimated breeding value of her replacement queen ($\hat{u}_{R(C),t}=3.03$), queen~$C$ now also serves as a dam. As a result, $\mathbb E\bigl[\hat{u}_{\mathcal P_{t+1}^{\ast}}\bigr]$ now only amounts to 6.0717. The value for $k_{\mathcal P_{t+1}^{\ast}, \mathcal P_{t+1}^{\ast}}$ is 0.1757, meaning an increase of average kinship by 2.23\%.
        
        \item Finally, calling the program with \texttt{-{}-delta\_k 2} reveals that an increase in average kinship of only 2\% is not possible.
        
        \begin{lstlisting}
Error: No optimal solution could be found. Possibly, delta_k was chosen too small.
        \end{lstlisting}

    \end{enumerate}
\end{Ex}
 
\subsection{Larger example}

To see how OCS for honeybees works on a bigger scale, we used the simulation program BeeSim \citep{plate19comparison} to create 100 ``current generations'' of $N_t=1500$ queens each. From there, we applied one single generation of selection to compare different selection strategies, including OCS.
In the folder \texttt{ocs\_large\_example}, we provide the input files to run \texttt{honeybee\_ocs.r} on one of the 100 ``current generations''.

\subsubsection{Setting}

A population under selection was simulated with 100 repetitions under identical simulation parameters. It consisted of 500 colonies per year with phenotypes that were shaped by a queen effect genetic variance of $\sigma_{A,Q}^2=1$, a worker effect genetic variance of $\sigma_{A,W}^2=2$, a covariance between the genetic effects of $\sigma_{A,QW}=-0.75$, and a residual variance of $\sigma_E^2=4$. All newly created queens mated with 12 drones on one of 20 isolated mating stations consisting of 8 DPQs each. Each year, a BLUP breeding value estimation was performed and the 100 best two-year-old queens were selected as dams, producing five daughter queens each. Similarly, the 20 highest rated three-year-old queens were selected to serve as 4a-queen of a mating station. All populations were simulated for 15 years. Colonies born in years 13 to 15 were then chosen as an instance of a ``current generation'' $\mathcal P_t$.

Those colonies in $\mathcal P_t$ that were born in years 14 or 15 were considered to survive to the next generation $\mathcal P_{t+1}$. Queens born in year 14 were eligible as dams and queens born in year 13 or 14 were eligible as 1b-queens. The 4a-queens need to be chosen before performing OCS (Remark~\ref{rmk::sel4a}) and we selected them among the queens born in year 13 according to the estimated breeding values of their replacement queens. However, no two 4a-queens were allowed to have the same dam (within-family selection). 

From this data, we calculated estimated average breeding values and kinships for a next generation $\mathcal P_{t+1}$ (i.\,e. year 16) according to ten different selection strategies: We tested OCS according to Tasks~\ref{task::sci},~\ref{task::ms}, and~\ref{task::mix}, i.\,e. with instrumental insemination only, mating stations only and the combination of mating stations and insemination. We thereby allowed for a generational increase of $k_{\mathcal P_t^{\ast},\mathcal P_t^{\ast}}$ of 1\% or 0.5\% (corresponding to options \texttt{-{}-delta\_k 0.4} and \texttt{-{}-delta\_k 0.2} according to Remark~\ref{rmk::fao}).

\begin{Not}
    At times, we will write OCS-0.4 and OCS-0.2 to indicate the choice of \texttt{-{}-delta\_k}. 
\end{Not}

In addition to these OCS strategies, we also considered four classical selection strategies: across-family selection and within-family selection with mating either via insemination or on mating stations. In all classical strategies, 100 queens were selected as dams and were assigned five offspring each. In across-family selection strategies, the chosen dams were the two-year-old queens with the highest estimated breeding values of their replacement queens, in within-family selection strategies, the 100 queens were also selected based on the estimated breeding values of their replacement queens, but no two selected dams were allowed to have the same dam. When mating was organized via insemination, the 40 queens aged two or three with the highest estimated breeding values were chosen as 1b-queens -- either with the restriction that no two selected queens may share a common dam (within-family selection) or without such restrictions (across-family selection). Each 1b-queen was considered to be used equally often to inseminate newly generated queens. When mating was organized via mating stations, all 20 mating stations were considered to be frequented equally often.

\begin{Rmk}
    \begin{enumerate}[label = (\roman*)]
        \item The numbers of 100 dams and 40 1b-queens or 20 4a-queens per year have earlier been found optimal for populations of 500 queens per year under classical selection strategies\citep{plate20, du23}. 

        \item Selection according to the 10 strategies was not explicitly carried out in simulations. Instead, the resulting values were merely calculated according to the formulas named in the paragraph on \texttt{stats.tsv} in Section~\ref{par::stats}.
    \end{enumerate}
\end{Rmk}

\subsubsection{Results and Discussion}

\begin{Not}\label{not::bar}
    We equip variables with bars to indicate that they report averages over the 100 repetitions. For example, we had $\bar{k}_{\mathcal P_{t}^{\ast},\mathcal P_{t}^{\ast}}=0.0301$ and $\bar{\hat{u}}_{\mathcal P_t^{\ast}}=4.932$.
\end{Not}

Regarding the results, we mainly focus on the increases in average kinship, $\Delta k_{\mathcal P_t^{\ast},\mathcal P_t^{\ast}}=\frac{k_{\mathcal P_{t+1}^{\ast},\mathcal P_{t+1}^{\ast}}-k_{\mathcal P_t^{\ast},\mathcal P_t^{\ast}}}{1 - k_{\mathcal P_t^{\ast},\mathcal P_t^{\ast}}}$, the expected increases in average breeding values, $\mathbb E\left[\Delta \hat{u}_{\mathcal P_t^{\ast}}\right]=\mathbb E\bigl[\hat{u}_{\mathcal P_{t+1}^{\ast}}\bigr]-\hat{u}_{\mathcal P_t^{\ast}}$, and the numbers $N^{\mathrm{dam}}_t$, $N^{\mathrm{1b}}_t$, and $N^{\mathrm{4a}}_t$ of queens that were selected for the different purposes.

The following table gives a survey regarding the averages of these values.\\

\begin{tabularx}{\textwidth} {l >{\centering\arraybackslash}X >{\centering\arraybackslash}X >{\centering\arraybackslash}Xcc}
     strategy & $\bar{N}^{\mathrm{dam}}_t$ & $\bar{N}^{\mathrm{1b}}_t$ & $\bar{N}^{\mathrm{4a}}_t$ & $\overline{\Delta k}_{\mathcal P_t^{\ast},\mathcal P_t^{\ast}}$ & $\overline{\mathbb E\left[\Delta \hat{u}_{\mathcal P_t^{\ast}}\right]}$ \\
     \hline
     within-family sel., insemination    & 100  & 40   & 0    & 0.279 & 0.399 \\
     within-family sel., mating stations & 100  & 0    & 20   & 0.250 & 0.410 \\
     across-family sel., mating stations & 100  & 0    & 20   & 0.334 & 0.544 \\
     across-family sel., insemination    & 100  & 40   & 0    & 0.407 & 0.547 \\[1ex]
     OCS-0.2, insemination               & 17.0 & 14.0 & 0    & 0.200 & 0.646 \\
     OCS-0.2, mating stations            & 19.1 & 0    & 5.55 & 0.200 & 0.657 \\
     OCS-0.2, combination                & 17.6 & 5.72 & 3.69 & 0.200 & 0.668 \\[1ex]
     OCS-0.4, insemination               & 12.0 & 9.49 & 0    & 0.400 & 0.706 \\
     OCS-0.4, mating stations            & 12.0 & 0    & 4.08 & 0.400 & 0.725 \\
     OCS-0.4, combination                & 11.9 & 3.23 & 2.86 & 0.400 & 0.732 \\
    \hline
\end{tabularx}\\[2ex]

\begin{Rmk}\label{rmk::finrm}
    \begin{enumerate}[label = (\roman*)]
        \item The OCS strategies with $\Delta k_{\mathcal P_t^{\ast},\mathcal P_t^{\ast}}=0.2$ yielded higher genetic gain than all classical selection strategies with lower increases in average kinships. OCS-0.4 strategies yielded even higher genetic gain but also had higher average kinship rates than most classical strategies.
        
        \item Remarkably, on average only 17 to 19 dams were needed with strategy OCS-0.2 and also the number of 1b-queens and 4a-queens was drastically reduced in comparison to classical strategies. Partly this is made possible by allowing for large inbreeding coefficients which are balanced by particularly small kinships between entities in $\mathcal S_{t+1}$ and $\mathcal N_{t+1}$. If one does not trust in these small numbers of dams and sires, one may consider to add a further linear restrictions to the OCS task which puts upper limits on the values of $dc_{Q,t}$, $bc_{Q,t}$, and $ac_{Q,t}$. Thereby, one can restrict the maximum number of offspring per selected queen.
        
        \item The differences between pure insemination strategies and pure mating station strategies in terms of genetic progress are very small. At first glance, this is in contradiction with the results of \citet{du23} who found much higher genetic progress for instrumental insemination breeding schemes than for breeding with isolated mating stations. However, the differences in genetic gain between the strategies in \citep{du23} are particularly attributed to more accurately estimated breeding values due to more precise pedigrees. Such effects do not occur for a single round of selection based on identical estimated breeding values as in the example presented here.
        
        \item \label{item::finrmiv} Our results base on a single round of OCS for a population that was hitherto selected with a classical selection strategy. Population dynamics resulting from multiple years of OCS in honeybees cannot be inferred from our data. 
    \end{enumerate}
\end{Rmk}

The following figure provides a visual impression of the outcomes of the different selection strategies.
Each mark corresponds to the results from one of the ten selection strategies in one of the 100 populations.\\[1em]
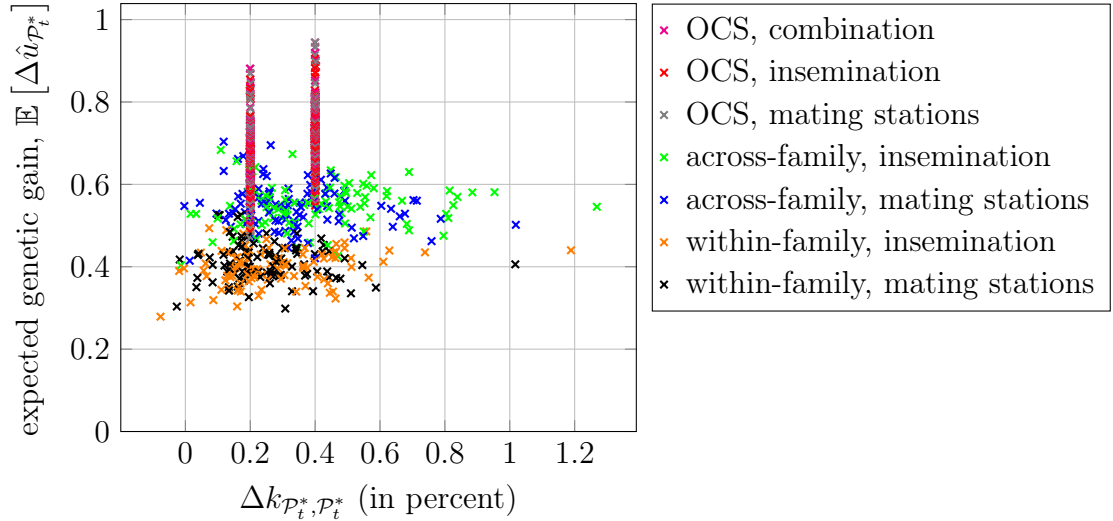
\begin{tikzpicture}
    \begin{axis}[
        xlabel=$\Delta k_{\mathcal P_t^{\ast},\mathcal P_t^{\ast}}$ (in percent),
        ylabel={expected genetic gain, $\mathbb E\left[\Delta \hat{u}_{\mathcal P_t^{\ast}}\right]$},
        ymin = 0, 
        xmin = -0.199, xmax = 1.39,
        grid = major,
        legend style={
            cells={anchor=west},
            legend pos=outer north east,
            nodes={inner sep=3pt,text depth=0.15em}
        },
    ]
        \addplot[
            scatter/classes={
                ocs_0={mark=x, thick, magenta},
                ocs_0_ii={mark=x, thick, red},
                ocs_0_ms={mark=x, thick, black!50},
                classic_ii={mark = x, thick, green},
                classic_ms={mark = x, thick, blue},
                classic_ii_wf={mark = x, thick, orange},
                classic_ms_wf={mark = x, thick, black},
                ocs_1_ms={mark = x, thick, black!50},
                ocs_1_ii={mark = x, thick, red},
                ocs_1={mark = x, thick, magenta}
            },
            scatter, only marks,
            scatter src=explicit symbolic,
        ]
            table [x=delta_k, y=genetic_progress, meta = strat] {large_example.tsv};
            \legend{{\ OCS, combination},
                    {\ OCS, insemination},
                    {\ OCS, mating stations},
                    {\ across-family, insemination},
                    {\ across-family, mating stations},
                    {\ within-family, insemination},
                    {\ within-family, mating stations}}
    \end{axis}
\end{tikzpicture}

\section{Conclusion}
In this manuscript, we have derived a suitable version of OCS to use in honeybee breeding. Particularly the larger simulation example gives hope for its practicability in practice. However, many theoretical questions remain still open and we have indicated them in the manuscript (e.\,g. Remarks~\ref{rmk::alternatives} and~\ref{rmk::finrm}\,\ref{item::finrmiv}). And of course, a practical application of OCS with real animals and real breeders generally comes with its own set of problems \citep{kohl17}. A we see it, we have opened a playground for much further research and hope that many researchers will frequent and enjoy it.
\newpage

\bibliographystyle{plainnat} 
\bibliography{lit.bib}      

\begin{thebibliography}{75}
\providecommand{\natexlab}[1]{#1}
\providecommand{\url}[1]{\texttt{#1}}
\expandafter\ifx\csname urlstyle\endcsname\relax
  \providecommand{\doi}[1]{doi: #1}\else
  \providecommand{\doi}{doi: \begingroup \urlstyle{rm}\Url}\fi

\bibitem[Agricola et~al.(2017)Agricola, Pukelsheim, and Horst]{agricola17}
I.~Agricola, F.~Pukelsheim, and F.~Horst.
\newblock {{Niemeyer und das Proportionalverfahren}}.
\newblock \emph{Math. Semesterber.}, 64:\penalty0 129--146, 2017.
\newblock \doi{10.1007/s00591-017-0201-8}.

\bibitem[Albuja(2022)]{albuja22}
J.~Albuja.
\newblock {{electoral: allocating seats methods and party system scores}},
  2022.
\newblock URL \url{https://CRAN.R-project.org/package=electoral}.

\bibitem[Andonov et~al.(2019)Andonov, Costa, Uzunov, Bergomi, Lourenco, and
  Misztal]{andonov19}
S.~Andonov, C.~Costa, A.~Uzunov, P.~Bergomi, D.~Lourenco, and I.~Misztal.
\newblock {{Modeling honey yield, defensive and swarming behaviors of Italian
  honey bees (\textit{Apis mellifera ligustica}) using linear threshold
  approaches}}.
\newblock \emph{BMC Genet.}, 20:\penalty0 78, 2019.
\newblock \doi{10.1186/s12863-019-0776-2}.

\bibitem[Balinski and Young(1982)]{balinski82}
M.~L. Balinski and H.~P. Young.
\newblock \emph{{{Fair representation: meeting the ideal of one man, one
  vote}}}.
\newblock Yale University Press, New Haven, CT, United States, 1982.

\bibitem[Basso et~al.(2024)Basso, Kistler, and Phocas]{basso24}
B.~Basso, T.~Kistler, and F.~Phocas.
\newblock {{Genetic parameters, trends, and inbreeding in a honeybee breeding
  program for royal jelly production and behavioral traits}}.
\newblock \emph{Apidologie}, 55\penalty0 (11), 2024.
\newblock \doi{10.1007/s13592-023-01055-3}.

\bibitem[Bernstein et~al.(2018)Bernstein, Plate, Hoppe, and
  Bienefeld]{bernstein18}
R.~Bernstein, M.~Plate, A.~Hoppe, and K.~Bienefeld.
\newblock {{Computing inbreeding coefficients and the inverse numerator
  relationship matrix in large populations of honey bees}}.
\newblock \emph{J. Anim. Breed. Genet.}, 135:\penalty0 323--332, 2018.
\newblock \doi{10.1111/jbg.12347}.

\bibitem[Bernstein et~al.(2023)Bernstein, Du, Du, Strauss, Hoppe, and
  Bienefeld]{bernstein23}
R.~Bernstein, M.~Du, Z.~G. Du, A.~S. Strauss, A.~Hoppe, and K.~Bienefeld.
\newblock {{First large-scale genomic prediction in the honey bee}}.
\newblock \emph{Heredity}, 130:\penalty0 320--328, 2023.
\newblock \doi{10.1038/s41437-023-00606-9}.

\bibitem[Bienefeld and Pirchner(1990)]{bienefeld90}
K.~Bienefeld and F.~Pirchner.
\newblock {{Heritabilities for several colony traits in the honeybee
  (\textit{Apis mellifera carnica})}}.
\newblock \emph{Apidologie}, 21:\penalty0 175--183, 1990.
\newblock \doi{10.1051/apido:19900302}.

\bibitem[Bienefeld et~al.(1989)Bienefeld, Reinhardt, and Pirchner]{bienefeld89}
K.~Bienefeld, F.~Reinhardt, and F.~Pirchner.
\newblock {{Inbreeding effects of queen and workers on colony traits in the
  honey bee}}.
\newblock \emph{Apidologie}, 20:\penalty0 439--450, 1989.
\newblock \doi{10.1051/apido:19890509}.

\bibitem[Bienefeld et~al.(2007)Bienefeld, Ehrhardt, and Reinhardt]{bienefeld07}
K.~Bienefeld, K.~Ehrhardt, and F.~Reinhardt.
\newblock {{Genetic evaluation in the honey bee considering queen and worker
  effects -- A BLUP-Animal Model approach}}.
\newblock \emph{Apidologie}, 38:\penalty0 77--85, 2007.
\newblock \doi{10.1051/apido:2006050}.

\bibitem[Bigio et~al.(2014)Bigio, Toufailia, Hughes, and Ratnieks]{bigio14the}
G.~Bigio, H.~Al Toufailia, W.~O.~H. Hughes, and F.~L.~W. Ratnieks.
\newblock {{The effect of one generation of controlled mating on the expression
  of hygienic behaviour in honey bees}}.
\newblock \emph{J. Apicult. Res.}, 53:\penalty0 563--8, 2014.
\newblock \doi{10.3896/IBRA.1.53.5.07}.

\bibitem[Brascamp and Bijma(2014)]{brascamp14methods}
E.~W. Brascamp and P.~Bijma.
\newblock {{Methods to estimate breeding values in honey bees}}.
\newblock \emph{Genet. Sel. Evol.}, 46:\penalty0 53, 2014.
\newblock \doi{10.1186/s12711-014-0053-9}.

\bibitem[Brascamp and Bijma(2019{\natexlab{a}})]{brascamp19a}
E.~W. Brascamp and P.~Bijma.
\newblock {{A note on genetic parameters and accuracy of estimated breeding
  values in honey bees}}.
\newblock \emph{Genet. Sel. Evol.}, 51:\penalty0 71, 2019{\natexlab{a}}.
\newblock \doi{10.1186/s12711-019-0510-6}.

\bibitem[Brascamp and Bijma(2019{\natexlab{b}})]{brascamp19software}
E.~W. Brascamp and P.~Bijma.
\newblock {{Software to facilitate estimation of genetic parameters and
  breeding values for honey bees}}.
\newblock In: \emph{Proceedings of the 46th International Apicultural Congress,
  Apimondia}. Montréal, Canada, 2019{\natexlab{b}}.

\bibitem[Brascamp et~al.(2016)Brascamp, Willam, Boigenzahn, Bijma, and
  Veerkamp]{brascamp16}
E.~W. Brascamp, A.~Willam, C.~Boigenzahn, P.~Bijma, and R.~F. Veerkamp.
\newblock {{Heritabilities and genetic correlations for honey yield,
  gentleness, calmness and swarming behaviour in Austrian honey bees}}.
\newblock \emph{Apidologie}, 47:\penalty0 739--748, 2016.
\newblock \doi{10.1007/s13592-016-0427-9}.

\bibitem[Brascamp et~al.(2024)Brascamp, Uzunov, Bijma, and Du]{brascamp24}
E.~W. Brascamp, A.~Uzunov, P.~Bijma, and M.~Du.
\newblock {{Genetics of selection in honeybees}}.
\newblock Wageningen, The Netherlands, 2024.
\newblock \doi{10.18174/677857}.

\bibitem[Bruckner et~al.(2023)Bruckner, Wilson, Aurell, Rennich, vanEngelsdorp,
  Steinhauer, and Williams]{bruckner23}
S.~Bruckner, M.~Wilson, D.~Aurell, K.~Rennich, D.~vanEngelsdorp, N.~Steinhauer,
  and G.~R. Williams.
\newblock {{A national survey of managed honey bee colony losses in the USA:
  results from the Bee Informed Partnership for 2017–18, 2018–19, and
  2019–20}}.
\newblock \emph{J. Apicult. Res.}, 62:\penalty0 439--443, 2023.
\newblock \doi{10.1080/00218839.2022.2158586}.

\bibitem[Brückner(1978)]{bruckner78}
D.~Brückner.
\newblock {{Why are there inbreeding effects in haplo-diploid systems?}}
\newblock \emph{Evolution}, 32:\penalty0 456--458, 1978.
\newblock \doi{10.1080/00218839.2022.2158586}.

\bibitem[Büchler et~al.(2024)Büchler, Andonov, Bernstein, Bienefeld, Costa,
  Du, Gabel, Given, Hatjina, Harpur, Hoppe, Kezic, Kovačić, Kryger, Mondet,
  Spivak, Uzunov, Wegener, and Wilde]{buchler24}
R.~Büchler, S.~Andonov, R.~Bernstein, K.~Bienefeld, C.~Costa, M.~Du, M.~Gabel,
  K.~Given, F.~Hatjina, B.~A. Harpur, A.~Hoppe, N.~Kezic, M.~Kovačić,
  P.~Kryger, F.~Mondet, M.~Spivak, A.~Uzunov, J.~Wegener, and J.~Wilde.
\newblock {{Standard methods for rearing and selection of \textit{Apis
  mellifera} queens 2.0}}.
\newblock \emph{J. Apicult. Res.}, 2024.
\newblock \doi{10.1080/00218839.2023.2295180}.

\bibitem[{Deutscher Imkerbund}(2021)]{dib21}
{Deutscher Imkerbund}.
\newblock {Richtlinien für das Zuchtwesen des Deutschen Imkerbundes (ZRL)},
  2021.

\bibitem[Druml et~al.(2023)Druml, Putz, Rubinigg, Kärcher, Neubauer, and
  Boigenzahn]{druml23}
T.~Druml, A.~Putz, M.~Rubinigg, M.~H. Kärcher, K.~Neubauer, and C.~Boigenzahn.
\newblock {{Founder gene pool composition and genealogical structure in two
  populations of Austrian Carniolan honey bees (\textit{Apis mellifera
  carnica}) as derived from pedigree analysis}}.
\newblock \emph{Apidologie}, 54:\penalty0 24, 2023.
\newblock \doi{10.1007/s13592-023-00999-w}.

\bibitem[Du et~al.(2021{\natexlab{a}})Du, Bernstein, Hoppe, and
  Bienefeld]{du21a}
M.~Du, R.~Bernstein, A.~Hoppe, and K.~Bienefeld.
\newblock {{A theoretical derivation of response to selection with and without
  controlled mating in honeybees}}.
\newblock \emph{Genet. Sel. Evol.}, 53:\penalty0 17, 2021{\natexlab{a}}.
\newblock \doi{10.1186/s12711-021-00606-5}.

\bibitem[Du et~al.(2021{\natexlab{b}})Du, Bernstein, Hoppe, and
  Bienefeld]{du21shortterm}
M.~Du, R.~Bernstein, A.~Hoppe, and K.~Bienefeld.
\newblock {{Short-term effects of controlled mating and selection on the
  genetic variance of honeybee populations}}.
\newblock \emph{Heredity}, 162:\penalty0 733--747, 2021{\natexlab{b}}.
\newblock \doi{10.1038/s41437-021-00411-2}.

\bibitem[Du et~al.(2023)Du, Bernstein, and Hoppe]{du23}
M.~Du, R.~Bernstein, and A.~Hoppe.
\newblock {{The potential of instrumental insemination for sustainable honeybee
  breeding}}.
\newblock \emph{Genes}, 14:\penalty0 1799, 2023.
\newblock \doi{10.3390/genes14091799}.

\bibitem[Du et~al.(2024{\natexlab{a}})Du, Bernstein, and Hoppe]{du24comparison}
M.~Du, R.~Bernstein, and A.~Hoppe.
\newblock {{Comparison of pooled semen insemination and single colony
  insemination as sustainable honeybee breeding strategies}}.
\newblock \emph{R. Soc. Open Sci.}, 11:\penalty0 231556, 2024{\natexlab{a}}.
\newblock \doi{10.1098/rsos.231556}.

\bibitem[Du et~al.(2024{\natexlab{b}})Du, Bernstein, and Hoppe]{du24the}
M.~Du, R.~Bernstein, and A.~Hoppe.
\newblock {{The number of drones to inseminate a queen with has little
  potential for optimization of honeybee breeding programs}}.
\newblock \emph{Hereditas}, 161:\penalty0 28, 2024{\natexlab{b}}.
\newblock \doi{10.1186/s41065-024-00332-0}.

\bibitem[Fisher(1918)]{fisher18}
R.~A. Fisher.
\newblock {{The correlations between relatives on the supposition of Mendelian
  inheritance}}.
\newblock \emph{Trans. Roy. Soc. Edinb.}, 52\penalty0 (2):\penalty0 321--341,
  1918.
\newblock \doi{10.1017/S0080456800012163}.

\bibitem[FAO(2013)]{fao13}
{Food and Agriculture Organization of the United Nations (FAO)}.
\newblock \emph{{{Draft guidelines on in vivo conservation of animal genetic
  resources}}}.
\newblock Number~14 in {{FAO Animal Production and Health Guidelines}}. Rome,
  Italy, 2013.

\bibitem[Gallais(2003)]{gallais03}
A.~Gallais.
\newblock \emph{{{Quantitative genetics and breeding methods in autopolyploid
  plants}}}.
\newblock INRA Editions, Paris, France, 2003.

\bibitem[Gervan et~al.(2005)Gervan, Winston, Higo, and Hoover]{gervan05}
N.~L. Gervan, M.~L. Winston, H.~A. Higo, and S.~E.~R. Hoover.
\newblock {{The effects of honey bee (\textit{Apis mellifera}) queen mandibular
  pheromone on colony defensive behaviour}}.
\newblock \emph{J. Apicult. Res.}, 44:\penalty0 175--179, 2005.
\newblock \doi{10.1080/00218839.2005.11101175}.

\bibitem[Gray et~al.(2023)Gray, Adjlane, Arab, Ballis, Brusbardis, {Bugeja
  Douglas}, Cadahía, Charrière, Chlebo, Coffey, Cornelissen, {Amaro da
  Costa}, Danneels, Danihlík, Dobrescu, Evans, Fedoriak, Forsythe, Gregorc,
  {Ilieva Arakelyan}, Johannesen, Kauko, Kristiansen, Martikkala,
  Martín-Hernández, Mazur, Medina-Flores, Mutinelli, Omar, Patalano,
  Raudmets, {San Martin}, Soroker, Stahlmann-Brown, Stevanovic, Uzunov,
  Vejsnaes, Williams, and Brodschneider]{gray23}
A.~Gray, N.~Adjlane, A.~Arab, A.~Ballis, V.~Brusbardis, A.~{Bugeja Douglas},
  L.~Cadahía, J.-D. Charrière, R.~Chlebo, M.~F. Coffey, B.~Cornelissen,
  C.~{Amaro da Costa}, E.~Danneels, J.~Danihlík, C.~Dobrescu, G.~Evans,
  M.~Fedoriak, I.~Forsythe, A.~Gregorc, I.~{Ilieva Arakelyan}, J.~Johannesen,
  L.~Kauko, P.~Kristiansen, M.~Martikkala, R.~Martín-Hernández, E.~Mazur,
  C.~A. Medina-Flores, F.~Mutinelli, E.~M. Omar, S.~Patalano, A.~Raudmets,
  G.~{San Martin}, V.~Soroker, P.~Stahlmann-Brown, J.~Stevanovic, A.~Uzunov,
  F.~Vejsnaes, A.~Williams, and R.~Brodschneider.
\newblock {{Honey bee colony loss rates in 37 countries using the COLOSS survey
  for winter 2019–2020: the combined effects of operation size, migration and
  queen replacement}}.
\newblock \emph{J. Apicult. Res.}, 62:\penalty0 204--210, 2023.
\newblock \doi{10.1080/00218839.2022.2113329}.

\bibitem[Grimmett(2004)]{grimmett04}
G.~Grimmett.
\newblock {{Stochastic apportionment}}.
\newblock \emph{Am. Math. Monthly}, 111:\penalty0 299--307, 2004.
\newblock \doi{10.1080/00029890.2004.11920078}.

\bibitem[Guichard et~al.(2020)Guichard, Neuditschko, Soland, Fried, Grandjean,
  Gerster, Dainat, Bijma, and Brascamp]{guichard20}
M.~Guichard, M.~Neuditschko, G.~Soland, P.~Fried, M.~Grandjean, S.~Gerster,
  B.~Dainat, P.~Bijma, and E.~W. Brascamp.
\newblock {{Estimates of genetic parameters for production, behaviour, and
  health traits in two Swiss honey bee populations}}.
\newblock \emph{Apidologie}, 51:\penalty0 876--891, 2020.
\newblock \doi{10.1007/s13592-020-00768-z}.

\bibitem[Gutiérrez-Reinoso et~al.(2022)Gutiérrez-Reinoso, Aponte, and
  García-Herreros]{gutierrezreinoso22}
M.~A. Gutiérrez-Reinoso, P.~M. Aponte, and M.~García-Herreros.
\newblock {{A review of inbreeding depression in dairy cattle: current status,
  emerging control strategies, and future prospects}}.
\newblock \emph{J. Dairy Res.}, 89:\penalty0 3--12, 2022.
\newblock \doi{10.1017/S0022029922000188}.

\bibitem[Harbo(1999)]{harbo99the}
J.~R. Harbo.
\newblock {{The value of single-drone inseminations in selective breeding of
  honey bees}}.
\newblock In L.~Connor and R.~Hoopingarner, editors, \emph{{{ Apiculture for
  the 21st Century}}}, pages 1--5. Wicwas Press, Cheshire, United States, 1999.

\bibitem[Henderson(1975)]{henderson75best}
C.~R. Henderson.
\newblock {{Best linear unbiased estimation and prediction under a selection
  model}}.
\newblock \emph{Biometrics}, 31:\penalty0 423--447, 1975.
\newblock \doi{10.2307/2529430}.

\bibitem[Henryon et~al.(2015)Henryon, Ostersen, Ask, Sørensen, and
  Berg]{henryon15}
M.~Henryon, T.~Ostersen, B.~Ask, A.~C. Sørensen, and P.~Berg.
\newblock {{Most of the long-term genetic gain from optimum-contribution
  selection can be realised with restrictions imposed during optimisation}}.
\newblock \emph{Genet. Sel. Evol.}, 47:\penalty0 21, 2015.
\newblock \doi{10.1186/s12711-015-0107-7}.

\bibitem[Hoppe et~al.(2020)Hoppe, Du, Bernstein, Tiesler, Kärcher, and
  Bienefeld]{hoppe20}
A.~Hoppe, M.~Du, R.~Bernstein, F.-K. Tiesler, M.~Kärcher, and K.~Bienefeld.
\newblock {{Substantial genetic progress in the international \textit{Apis
  mellifera carnica} population since the implementation of genetic
  evaluation.}}
\newblock \emph{Insects}, 11:\penalty0 768, 2020.
\newblock \doi{10.3390/insects11110768}.

\bibitem[Jiménez-Mena et~al.(2016)Jiménez-Mena, Schad, Hanna, and
  Lacy]{jimenezmena16}
B.~Jiménez-Mena, K.~Schad, N.~Hanna, and R.~C. Lacy.
\newblock {{Pedigree analysis for the genetic management of group-living
  species}}.
\newblock \emph{Ecol. Evol.}, 6:\penalty0 3067--3078, 2016.
\newblock \doi{10.1002/ece3.1831}.

\bibitem[Kerr et~al.(1998)Kerr, Goddard, and Jarvis]{kerr98}
R.~J. Kerr, M.~E. Goddard, and S.~F. Jarvis.
\newblock {{Maximising genetic response in tree breeding with constraints on
  group coancestry}}.
\newblock \emph{Silv. Genet.}, 47:\penalty0 165--173, 1998.

\bibitem[Kerr et~al.(2012)Kerr, Li, Tier, Dutkowski, and McRae]{kerr12}
R.~J. Kerr, L.~Li, B.~Tier, G.~W. Dutkowski, and T.~A. McRae.
\newblock {{Use of the numerator relationship matrix in genetic analysis of
  autopolyploid species}}.
\newblock \emph{Theor. Appl. Genet.}, 124:\penalty0 1271--1282, 2012.
\newblock \doi{10.1007/s00122-012-1785-y}.

\bibitem[Kistler et~al.(2021)Kistler, Basso, and Phocas]{kistler21}
T.~Kistler, B.~Basso, and F.~Phocas.
\newblock {{A simulation study of a honeybee breeding scheme accounting for
  polyandry, direct and maternal effects on colony performance}}.
\newblock \emph{Genet. Sel. Evol.}, 53\penalty0 (71), 2021.
\newblock \doi{10.1186/s12711-021-00665-8}.

\bibitem[Kohl and Herold(2017)]{kohl17}
S.~Kohl and P.~Herold.
\newblock {{Problemanalyse zur Implementierung der Selektion nach optimierten
  Genbeiträgen in kleinen Populationen}}.
\newblock \emph{Züchtungskunde}, 89:\penalty0 345--358, 2017.

\bibitem[Lange(1997)]{lange97mathematical}
K.~Lange.
\newblock \emph{Mathematical and statistical methods for genetic analysis}.
\newblock Springer, New York, 1997.
\newblock \doi{10.1007/978-1-4757-2739-5}.

\bibitem[Lewis(1942)]{lewis42}
D.~Lewis.
\newblock {{The evolution of sex in flowering plants}}.
\newblock \emph{Biol. Rev.}, 17:\penalty0 46--67, 1942.
\newblock \doi{10.1111/j.1469-185X.1942.tb00431.x}.

\bibitem[Lijphart(2003)]{lijphart03}
A.~Lijphart.
\newblock Degrees of proportionality of proportional representation formulas.
\newblock In B.~Grofman and A.~Lijphart, editors, \emph{Electoral laws and
  their political consequences}, pages 170--179. Algora Publishing, New York,
  NY, United States, 2003.

\bibitem[Lush(1937)]{lush37}
J.~L. Lush.
\newblock \emph{Animal breeding plans}.
\newblock Iowa State Public Press, Ames, 1937.

\bibitem[Lynch and Walsh(1998)]{lynch98}
M.~Lynch and B.~Walsh.
\newblock \emph{{{Genetics and Analysis of Quantitative Traits}}}, volume~1.
\newblock Sinauer, Sunderland, MA, 1998.

\bibitem[Mackensen(1967)]{mackensen67}
O.~Mackensen.
\newblock {{Breeding and genetics of bees}}.
\newblock In {US Agricultural Research Service}, editor, \emph{{{Beekeeping in
  the United States}}}, volume 335 of \emph{Agriculture Handbook}, pages
  68--76. US Government Printing Office, Washington D.C., United States, 1967.

\bibitem[Malécot(1948)]{malecot48}
G.~Malécot.
\newblock \emph{{{Les mathématiques de l'hérédité}}}.
\newblock Masson, Paris, France, 1948.

\bibitem[Mattila and Seeley(2007)]{mattila07}
H.~R. Mattila and T.~D. Seeley.
\newblock {{Genetic diversity in honey bee colonies enhances productivity and
  fitness}}.
\newblock \emph{Science}, 317:\penalty0 362--364, 2007.
\newblock \doi{10.1126/science.1143046}.

\bibitem[Meuwissen(1997)]{meuwissen97}
T.~H.~E. Meuwissen.
\newblock {{Maximizing the response of selection with a predefined rate of
  inbreeding}}.
\newblock \emph{J. Anim. Sci.}, 75\penalty0 (4):\penalty0 934--940, 1997.
\newblock \doi{10.2527/1997.754934x}.

\bibitem[Meuwissen and Sonesson(1998)]{meuwissen98}
T.~H.~E. Meuwissen and A.~K. Sonesson.
\newblock {{Maximizing the response of selection with a predefined rate of
  inbreeding: overlapping generations}}.
\newblock \emph{J. Anim. Sci.}, 76:\penalty0 2575--2583, 1998.
\newblock \doi{10.2527/1998.76102575x}.

\bibitem[Neumann et~al.(1999)Neumann, Moritz, and van Praagh]{neumann99queen}
P.~Neumann, R.~F.~A. Moritz, and J.~van Praagh.
\newblock {{Queen mating frequency in different types of honey bee mating
  apiaries}}.
\newblock \emph{J. Apicult. Res.}, 38:\penalty0 11--18, 1999.
\newblock \doi{10.1080/00218839.1999.11100990}.

\bibitem[Panzani et~al.(2007)Panzani, Rota, Pacini, Vannozzi, and
  Camillo]{panzani07}
D.~Panzani, A.~Rota, M.~Pacini, I.~Vannozzi, and F.~Camillo.
\newblock {{One year old fillies can be successfully used as embryo donors}}.
\newblock \emph{Theriogenology}, 67:\penalty0 367--371, 2007.
\newblock \doi{10.1016/j.theriogenology.2006.08.004}.

\bibitem[Pernal et~al.(2012)Pernal, Sewalem, and Melathopoulos]{pernal12}
S.~F. Pernal, A.~Sewalem, and A.~P. Melathopoulos.
\newblock {{Breeding for hygienic behaviour in honeybees (Apis mellifera) using
  free-mated nucleus colonies}}.
\newblock \emph{Apidologie}, 43:\penalty0 403--16, 2012.
\newblock \doi{10.1007/s13592-011-0105-x}.

\bibitem[Plate et~al.(2019{\natexlab{a}})Plate, Bernstein, Hoppe, and
  Bienefeld]{plate19comparison}
M.~Plate, R.~Bernstein, A.~Hoppe, and K.~Bienefeld.
\newblock {{Comparison of infinitesimal and finite locus models for long-term
  breeding simulations with direct and maternal effects at the example of
  honeybees}}.
\newblock \emph{PLOS One}, 14:\penalty0 e0213270, 2019{\natexlab{a}}.
\newblock \doi{10.1371/journal.pone.0213270}.

\bibitem[Plate et~al.(2019{\natexlab{b}})Plate, Bernstein, Hoppe, and
  Bienefeld]{plate19the}
M.~Plate, R.~Bernstein, A.~Hoppe, and K.~Bienefeld.
\newblock {{The importance of controlled mating in honeybee breeding}}.
\newblock \emph{Genet. Sel. Evol.}, 51:\penalty0 74, 2019{\natexlab{b}}.
\newblock \doi{10.1186/s12711-019-0518-y}.

\bibitem[Plate et~al.(2020)Plate, Bernstein, Hoppe, and Bienefeld]{plate20}
M.~Plate, R.~Bernstein, A.~Hoppe, and K.~Bienefeld.
\newblock {{Long-term evaluation of breeding scheme alternatives for endangered
  honeybee subspecies}}.
\newblock \emph{Insects}, 11:\penalty0 404, 2020.
\newblock \doi{10.3390/insects11070404}.

\bibitem[{R Core Team}(2019)]{readtable}
{R Core Team}.
\newblock \emph{read.table: Data Input}.
\newblock R Foundation for Statistical Computing, Vienna, Austria, 2019.
\newblock URL
  \url{https://www.rdocumentation.org/packages/utils/versions/3.6.2/topics/read.table}.

\bibitem[Tang et~al.(2023)Tang, Ji, Shi, Su, Xue, Xu, Chen, Zhao, and
  Chen]{tang23}
J.~Tang, C.~Ji, W.~Shi, S.~Su, Y.~Xue, J.~Xu, X.~Chen, Y.~Zhao, and C.~Chen.
\newblock {{Survey results of honey bee colony losses in winter in China
  (2009–2021)}}.
\newblock \emph{Insects}, 14:\penalty0 554, 2023.
\newblock \doi{10.3390/insects14060554}.

\bibitem[Tarpy et~al.(2013)Tarpy, vanEngelsdorp, and Pettis]{tarpy13}
D.~R. Tarpy, D.~vanEngelsdorp, and J.~S. Pettis.
\newblock {{Genetic diversity affects colony survivorship in commercial honey
  bee colonies}}.
\newblock \emph{Naturwissenschaften}, 100:\penalty0 723--728, 2013.
\newblock \doi{10.1007/s00114-013-1065-y}.

\bibitem[Tiesler et~al.(2016)Tiesler, Bienefeld, and Büchler]{tiesler16}
F.~K. Tiesler, K.~Bienefeld, and R.~Büchler.
\newblock \emph{{{Selektion bei der Honigbiene}}}.
\newblock Buschhausen, Herten, Germany, 2016.

\bibitem[Uzunov et~al.(2022{\natexlab{a}})Uzunov, Andonov, Dahle, Kovačić,
  Prešern, Aleksovski, Jaman, Jovanovska, Pavlov, Puškadija, Wegener, and
  Büchler]{uzunov22evaluation}
A.~Uzunov, S.~Andonov, B.~Dahle, M.~Kovačić, J.~Prešern, G.~Aleksovski,
  F.~Jaman, M.~Jovanovska, B.~Pavlov, Z.~Puškadija, J.~Wegener, and
  R.~Büchler.
\newblock {{Evaluating the potential for mating control in honey bee breeding
  in three SE European countries (preliminary results)}}.
\newblock In: \emph{Proceedings of the 12th World Congress on Genetics Applied
  to Livestock Production}. Rotterdam, The Netherlands, 2022{\natexlab{a}}.

\bibitem[Uzunov et~al.(2022{\natexlab{b}})Uzunov, Brascamp, Du, and
  Büchler]{uzunov22initiation}
A.~Uzunov, E.~W. Brascamp, M.~Du, and R.~Büchler.
\newblock {{Initiation and implementation of honey bee breeding programs}}.
\newblock \emph{Bee World}, 99:\penalty0 50--55, 2022{\natexlab{b}}.
\newblock \doi{10.1080/0005772X.2022.2031545}.

\bibitem[Uzunov et~al.(2022{\natexlab{c}})Uzunov, Brascamp, Du, and
  Büchler]{uzunov22the}
A.~Uzunov, E.~W. Brascamp, M.~Du, and R.~Büchler.
\newblock {{The relevance of mating control for successful implementation of
  honey bee breeding programs}}.
\newblock \emph{Bee World}, 99:\penalty0 94--98, 2022{\natexlab{c}}.
\newblock \doi{10.1080/0005772X.2022.2088166}.

\bibitem[Uzunov et~al.(2023)Uzunov, Brascamp, Du, Bijma, and
  Büchler]{uzunov23}
A.~Uzunov, E.~W. Brascamp, M.~Du, P.~Bijma, and R.~Büchler.
\newblock {{Breeding values in honey bees}}.
\newblock \emph{Bee World}, 100:\penalty0 9--14, 2023.
\newblock \doi{10.1080/0005772X.2023.2166737}.

\bibitem[Wang et~al.(2017)Wang, Bennewitz, and Wellmann]{wang17}
Y.~Wang, J.~Bennewitz, and R.~Wellmann.
\newblock {{Novel optimum contribution selection methods accounting for
  conflicting objectives in breeding programs for livestock breeds with
  historical migration}}.
\newblock \emph{Genet. Sel. Evol.}, 49\penalty0 (45), 2017.
\newblock \doi{10.1186/s12711-017-0320-7}.

\bibitem[Wellmann(2019)]{wellmann19optimum}
R.~Wellmann.
\newblock {{Optimum contribution selection for animal breeding and
  conservation: the R package optiSel}}.
\newblock \emph{BMC Bioinformatics}, 20:\penalty0 25, 2019.
\newblock \doi{10.1186/s12859-018-2450-5}.

\bibitem[Wellmann(2021)]{wellmann21}
R.~Wellmann.
\newblock {{optiSolve: linear, quadratic, and rational optimization}}, 2021.
\newblock URL \url{https://CRAN.R-project.org/package=optiSolve}.

\bibitem[Wellmann and Bennewitz(2019)]{wellmann19key}
R.~Wellmann and J.~Bennewitz.
\newblock {{Key genetic parameters for population management}}.
\newblock \emph{Front. Genet.}, 10:\penalty0 667, 2019.
\newblock \doi{10.3389/fgene.2019.00667}.

\bibitem[Wellmann and Pfeiffer(2009)]{wellmann09}
R.~Wellmann and I.~Pfeiffer.
\newblock {{Pedigree analysis for conservation of genetic diversity and
  purging}}.
\newblock \emph{Genet. Res., Camb.}, 91:\penalty0 209--219, 2009.
\newblock \doi{10.1017/S0016672309000202}.

\bibitem[Woyke(1965)]{woyke65}
J.~Woyke.
\newblock {{Genetic proof of the origin of drones from fertilized eggs of the
  honeybee}}.
\newblock \emph{J. Apicult. Res.}, 4:\penalty0 7--11, 1965.
\newblock \doi{10.1080/00218839.1965.11100095}.

\bibitem[Wright(1922)]{wright22}
S.~Wright.
\newblock {{Coefficients of inbreeding and relationship}}.
\newblock \emph{Am. Nat.}, 56:\penalty0 330--338, 1922.
\newblock \doi{10.1086/279872}.

\bibitem[Zayed and Packer(2005)]{zayed05}
A.~Zayed and L.~Packer.
\newblock {{Complementary sex determination substantially increases extinction
  proneness of haplodiploid populations}}.
\newblock \emph{Proc. Natl. Acad. Sci. USA}, 102:\penalty0 10742--10746, 2005.
\newblock \doi{10.1073/pnas.0502271102}.

\end{thebibliography}

\end{document}